\documentclass[11pt,a4paper]{article}

\usepackage[utf8]{inputenc}
\usepackage{lmodern} %
\usepackage{multirow} %
\usepackage[T1]{fontenc}

\usepackage[bookmarks,colorlinks,breaklinks]{hyperref} %
\hypersetup{urlcolor=blue, colorlinks=true, citecolor=green!50!black, linkcolor=blue}

\makeatletter
\renewcommand{\paragraph}{%
	\@startsection{paragraph}{4}%
	{\z@}{1.25ex \@plus 1ex \@minus .2ex}{-1em}%
	{\normalfont\normalsize\bfseries}%
}
\makeatother

\providecommand{\tabularnewline}{\\}

\usepackage{comment}

\usepackage{amsmath}
\usepackage{amsfonts}
\usepackage{amssymb}
\usepackage{amsthm}
\usepackage{thm-restate}
\usepackage{thmtools}
\usepackage{pgffor}
\usepackage{xcolor}
\usepackage[lined,boxed,ruled,norelsize,algo2e,linesnumbered,noend]{algorithm2e}
\usepackage{cleveref}
\usepackage{geometry}
\usepackage{graphicx}
\usepackage{enumitem}
\usepackage{url} %
\makeatletter
\newcommand\footnoteref[1]{\protected@xdef\@thefnmark{\labelcref{#1}}\@footnotemark}
\makeatother

\ifdefined\DEBUG
	\usepackage{showlabels}
	\def\jan#1{\textcolor{red}{JAN: #1}}
	\def\richard#1{\textcolor{green}{RICHARD: #1}}
	\def\danupon#1{\textcolor{orange}{DN: #1}}
	\def\thatchaphol#1{\textcolor{purple}{TS: #1}}
	\newcommand{\Zhao}[1]{{\color{red}[Zhao: #1]}}
	\newcommand{\yintat}[1]{{\color{red}[YinTat: #1]}}
	\def\wadi#1{\textcolor{blue}{DW: #1}}
	\definecolor{darkgreen}{rgb}{0.0, 0.42, 0.24}
	\newcommand{\sidford}[1]{{\color{darkgreen}[Aaron: #1]}}
	
\else

	\def\jan#1{}
	\def\richard#1{}
	\def\danupon#1{}
	\def\thatchaphol#1{}
	\newcommand{\Zhao}[1]{}
	\def\wadi#1{}
	\newcommand{\sidford}[1]{}
	\newcommand{\yintat}[1]{}
	
\fi

\declaretheorem[numberwithin=section,refname={Theorem,Theorems},Refname={Theorem,Theorems}]{theorem}
\declaretheorem[numberlike=theorem]{lemma}
\declaretheorem[numberlike=theorem]{fact}
\declaretheorem[numberlike=theorem]{proposition}
\declaretheorem[numberlike=theorem]{corollary}
\declaretheorem[numberlike=theorem]{definition}
\declaretheorem[numberlike=theorem]{claim}
\declaretheorem[numberlike=theorem]{invariant}

\newcommand{\R}{\mathbb{R}}

\newcommand{\N}{\mathbb{N}}
\newcommand{\Z}{\mathbb{Z}}
\renewcommand{\S}{\mathbb{S}}%

\renewcommand{\P}{\mathbb{P}}
\newcommand{\E}{\mathbb{E}}

\renewcommand{\tilde}{\widetilde}
\renewcommand{\hat}{\widehat}
\renewcommand{\bar}{\overline}

\newcommand{\argmin}{\operatorname{argmin}}
\newcommand{\argmax}{\operatorname{argmax}}
\newcommand{\poly}{\operatorname{poly}}

\newcommand{\mdiag}{\mathbf{Diag}} %

\newcommand{\Tr}{\operatorname{Tr}} %

\newcommand{\nnz}{\operatorname{nnz}}%
\newcommand{\cnorm}{C_{\mathrm{norm}}}%

\newcommand{\Ot}{{\tilde O}}
\newcommand{\tO}{{\tilde O}}
\newcommand{\Otil}{{\tilde O}}
\newcommand{\otilde}{\Ot}

\newcommand{\zerovec}{\vec{0}} %
\newcommand{\onevec}{{\vec{1}}} %
\newcommand{\mzero}{\boldsymbol{0}}%
\newcommand{\unitvec}{\vec{e}}

\newcommand{\mproj}{\mathbf{P}}%
\newcommand{\mSigma}{\boldsymbol{\Sigma}}%
\newcommand{\mLambda}{\boldsymbol{\Lambda}}%
\foreach \x in {A,...,Z}{%
	\expandafter\xdef\csname m\x\endcsname{\noexpand\mathbf{\x}}
}
\newcommand{\ma}{\mathbf{A}}%
\newcommand{\mb}{\mathbf{B}}%
\newcommand{\mg}{\mathbf{G}}%
\newcommand{\mj}{\mathbf{J}}%
\newcommand{\ms}{\mathbf{S}}%
\newcommand{\mv}{\mathbf{V}}%
\newcommand{\mw}{\mathbf{W}}%

\newcommand{\otau}{\noexpand{\overline{\tau}}}
\newcommand{\omu}{\noexpand{\overline{\mu}}}
\newcommand{\osigma}{\noexpand{\overline{\sigma}}}%
\newcommand{\os}{\noexpand{\overline{s}}}
\newcommand{\ot}{\noexpand{\overline{t}}}
\newcommand{\og}{\noexpand{\overline{g}}}
\newcommand{\ou}{\noexpand{\overline{u}}}
\newcommand{\ov}{\noexpand{\overline{v}}}
\newcommand{\ow}{\noexpand{\overline{w}}}
\newcommand{\ox}{\noexpand{\overline{x}}}
\newcommand{\oy}{\noexpand{\overline{y}}}
\newcommand{\oz}{\noexpand{\overline{z}}}

\foreach \x in {A,...,Z}{%
	\expandafter\xdef\csname c\x\endcsname{\noexpand\mathcal{\x}}
}

\newcommand{\ttau}{\widetilde{\tau}}

\foreach \x in {a,...,z}{%
	\expandafter\xdef\csname t\x\endcsname{\noexpand\widetilde{\x}}
}

\foreach \x in {A,...,Z}{%
	\expandafter\xdef\csname om\x\endcsname{\noexpand\mathbf{\overline{\x}}}
}
\newcommand{\omv}{\omV}%
\newcommand{\omw}{\omW}

\foreach \x in {A,...,Z}{%
	\expandafter\xdef\csname tm\x\endcsname{\noexpand\mathbf{\widetilde{\x}}}
}

\renewcommand{\t}{^{(t)}}

\newcommand{\tmp}{{\mathrm{(tmp)}}} %
\newcommand{\new}{\mathrm{(new)}}%
\newcommand{\init}{\mathrm{(init)}}%
\newcommand{\target}{\mathrm{(end)}}%
\newcommand{\final}{\mathrm{(final)}}%
\newcommand{\apxfinal}{\mathrm{(apx\text{-}final)}}%
\newcommand{\sample}{\mathrm{sample}} %

\newcommand{\opt}{\mathrm{OPT}}

\newcommand{\defeq}{\stackrel{\mathrm{{\scriptscriptstyle def}}}{=}}
\newcommand{\vones}{\onevec}
\newcommand{\vzero}{\zerovec}

\newcommand{\Comment}[1]{\tcp*[h]{#1}}

\newcommand{\Tau}{\mathbf{T}}

\newcommand{\eps}{\varepsilon}
\renewcommand{\d}{\delta}
\newcommand{\diag}{\mathbf{diag}}
\newcommand{\g}{\nabla}
\newcommand{\pe}{\preceq}
\newcommand{\se}{\succeq}
\newcommand{\Iter}{\mathrm{Iter}}
\newcommand{\tpi}{{\tau+\infty}}

\renewcommand{\mLambda}{\mathbf{\Lambda}}
\renewcommand{\mSigma}{\mathbf{\Sigma}}
\newcommand{\ShortStep}{\textsc{ShortStep}}
\newcommand{\PathFollowing}{\textsc{PathFollowing}}
\newcommand{\hx}{\hat{x}}
\newcommand{\hw}{\hat{w}}

\newcommand{\csample}{C_\sample}
\newcommand{\ls}{\lesssim}

\newcommand{\wt}{\widetilde}

\newcommand{\cvalid}{C_{\mathrm{valid}}}
\newcommand{\Var}{\mathrm{Var}}
\newcommand{\assign}{\leftarrow}
\newcommand{\cstart}{C_{\mathrm{start}}}
\newcommand{\tkpi}{{\tau(\hx^{(k)})+\infty}}
\newcommand{\stab}{\mathrm{stab}}
\newcommand{\inexact}{\mathrm{inexact}}

\global\long\def\policyspace{A^{S}}
\global\long\def\actions{A}
\global\long\def\states{S}
\global\long\def\valuespace{\mathbb{R}^{S}}
\global\long\def\S{|S|}
\global\long\def\A{|A|}
\global\long\def\vopt{v^{*}}

\title{Minimum Cost Flows, MDPs, and $\ell_1$-Regression \\
	in Nearly Linear Time for Dense Instances}

\allowdisplaybreaks

\author{
Jan van den Brand\thanks{\texttt{janvdb@kth.se}. KTH Royal Institute of Technology, Sweden.}
\and
Yin Tat Lee\thanks{\texttt{yintat@uw.edu}. University of Washington and Microsoft Research Redmond, USA.}
\and 
Yang P. Liu\thanks{\texttt{yangpliu@stanford.edu}. Stanford University, USA.}
\and
Thatchaphol Saranurak\thanks{\texttt{saranurak@ttic.edu}. Toyota Technological Institute at Chicago, USA.}
\and
Aaron Sidford\thanks{\texttt{sidford@stanford.edu}. Stanford University, USA.} 
\and
Zhao Song\thanks{\texttt{zhaos@ias.edu}. Institute for Advanced Study, USA.}
\and 
Di Wang\thanks{\texttt{wadi@google.com}. Google Research, USA.}
}

\begin{document}
	
\begin{titlepage}
	\maketitle
	
	\abstract{

		In this paper we provide new randomized algorithms with improved runtimes for solving linear programs with two-sided constraints. In the special case of the minimum cost flow problem on $n$-vertex $m$-edge graphs with integer polynomially-bounded costs and capacities we obtain a randomized method which solves the problem in $\tilde{O}(m + n^{1.5})$ time. This improves upon the previous best runtime of $\tilde{O}(m \sqrt{n})$  \cite{ls14} and, in the special case of unit-capacity maximum flow, improves upon the previous best runtimes of $m^{4/3 + o(1)}$  \cite{ls20_focs,kathuria2020potential} and $\tilde{O}(m \sqrt{n})$ \cite{ls14} for sufficiently dense graphs. 
		
		In the case of $\ell_1$-regression in a matrix with $n$-columns and $m$-rows we obtain a randomized method which computes an $\epsilon$-approximate solution in $\tilde{O}(mn + n^{2.5})$ time. This yields a randomized method which computes an $\epsilon$-optimal policy of a discounted Markov Decision Process with $S$ states and, $A$ actions per state in time $\tilde{O}(S^2 A + S^{2.5})$. These methods improve upon the previous best runtimes of methods which depend polylogarithmically on problem parameters, which were $\tilde{O}(mn^{1.5})$ \cite{LeeS15} and $\tilde{O}(S^{2.5} A)$ \cite{ls14,SidfordWWY18} respectively.
		
		To obtain this result we introduce two new algorithmic tools of possible independent interest. First, we design a new general interior point method for solving linear programs with two sided constraints which combines techniques from \cite{lsz19} and \cite{BrandLN+20} to obtain a robust stochastic method with iteration count nearly the square root of the smaller dimension.  Second, to implement this method we provide dynamic data structures for efficiently maintaining approximations to variants of Lewis-weights, a fundamental importance measure for matrices which generalize leverage scores and effective resistances.
	}
	
	\pagenumbering{roman}
	\newpage
	\setcounter{tocdepth}{2}
	\renewcommand{\baselinestretch}{0.9}
	\tableofcontents
	\renewcommand{\baselinestretch}{1.0}
\end{titlepage}

\newpage
\pagenumbering{arabic}

\ifdefined\DEBUG
\fi

\section{Introduction}\label{sec:intro}

We consider solving linear programs expressed in the following primal/dual form:
\begin{equation}
\label{eq:intro_main_formulation}
(P) = 
\min_{\substack{x\in \R^m : \mA^\top x = b \\ \ell_i \le x_i \le u_i ~ \forall i \in [m]}} c^\top x 
\enspace \text{ and } \enspace
(D) =
\max_{\substack{y \in \R^n, s \in \R^m \\ \mA y + s = c}} b^\top y + \sum_{i \in [n]} \min(\ell_i s_i, u_i s_i).
\end{equation}
where $b \in \R^n$, $c \in \R^m$, $\mA \in \R^{m \times n}$ and each $\ell_i \leq u_i \in \R$. 

\Cref{eq:intro_main_formulation} naturally encompasses prominent continuous and combinatorial optimization problems. When $\ell_i = 0$ and $u_i = \infty$, \eqref{eq:intro_main_formulation} corresponds to the standard primal-dual formulation of linear program. In the case when $\ma$ is the incidence matrix of a graph, i.e. $b = \vzero$, and $\ell_i = 0$ for all $i \in [m]$, \eqref{eq:intro_main_formulation} corresponds to the problem of computing a minimum cost circulation in a directed graph with linear costs $c$ and edge capacities given by $u$ (\Cref{sec:mincostflow}). Further, in the case when each $ u_i = 1$, $\ell_i = - 1$, and $b = 0$, \eqref{eq:intro_main_formulation} corresponds to solving $\ell_1$ regression (\Cref{sec:final_dual}) and can be used to solve Markov Decision Processes (MDPs) \cite{SidfordWWY18} (\Cref{sec:app:mdp}).

Recent advances in interior point methods (IPMs), a prominent class of continuous optimization methods, and data structures have led to nearly linear runtimes for solving fundamental classes of \eqref{eq:intro_main_formulation} to high precision with only a polylogarithmic dependence on problem parameters. In \cite{blss20} an $\tilde{O}(mn + n^{2.5})$ time randomized method was obtained for solving \eqref{eq:intro_main_formulation} when $\ell_i = 0$ and $u_i = \infty$ for all $i \in [m]$, i.e., when the problem is in standard form. Further, in \cite{BrandLN+20} a randomized method was obtained for solving minimum cost perfect matching in bipartite graphs in time $\tO(m+n^{1.5})$, i.e. when the problem is in standard form and $\mA$ is the incidence matrix of a bipartite graph.

Unfortunately, though it is well known that all linear programs can be written in standard form, na\"ively transforming \eqref{eq:intro_main_formulation} to standard form can increase the dimension of the problem, i.e. turn $n$ to $\Omega(m)$. This issue prevents the application of the recent advances in \cite{BrandLN+20,blss20} to \eqref{eq:intro_main_formulation} when both $u_i$ and $\ell_i$ are bounded. Consequently, despite extensive study, obtaining nearly linear running times for solving \eqref{eq:intro_main_formulation}, maximum flow, minimum cost flow, $\ell_1$ regression, and Markov Decision Processes to high-precision with a polylogarithmic dependence on problem parameters in nearly linear time in high-dimensional dense instances has been elusive.

The issue of losing density when reformulating \eqref{eq:intro_main_formulation} in standard form is a known difficulty in obtaining improved runtimes from continuous optimization methods. It arose in \cite{ls14,ls19} and was addressed and leveraged to obtain improved randomized runtimes for minimum cost flow, $\ell_1$-regression \cite{LeeS15}, and MDPs \cite{ls14,LeeS15,SidfordWWY18}.
 These methods provide IPMs which work directly with \eqref{eq:intro_main_formulation}. The IPMs re-weight constraints based on variations of Lewis weights \cite{CohenP15,ls19}, a natural notion of row importance for matrices which generalizes leverages scores, and implement and apply the corresponding methods efficiently.
 
Lewis weight reweighting schemes were also key to the aforementioned nearly linear time algorithms for solving linear programs in standard form \cite{blss20} and computing minimum-cost bipartite perfect matchings \cite{BrandLN+20} on dense instances. Unfortunately, a key technique applied by these results is that when \eqref{eq:intro_main_formulation} is in standard form (i.e. $\ell_i = 0$ and $u_i = \infty$ for all $i \in [n]$) it is possible to leverage the primal-dual structure of the problem to rewrite the optimality conditions of \cite{ls19} in terms of leverage scores, a simple special case of Lewis weights. Leveraging this structure, these papers design IPMs and data-structures for efficiently leveraging and manipulating leverage scores towards achieving their runtimes. Unfortunately, in the case that both $\ell_i$ and $u_i$ are bounded, this same technique doesn't directly apply (see \Cref{sec:overview:ipm}). %
	
In this paper, we show how to overcome this difficulty and directly obtain nearly linear time algorithms for solving \eqref{eq:intro_main_formulation}, minimum cost flow\footnote{Though Lewis weight reweighting was used to achieve state-of-the-art $\otilde(m + n^{1.5})$ runtimes for minimum-cost bipartite perfect matching on $m$-edge $n$-node graphs in \cite{BrandLN+20}, the same work showed that it was not needed to obtain $\otilde(n\sqrt{m})$ runtimes. Consequently,  $\otilde(n\sqrt{m})$ runtimes for minimum cost flow on $m$-edge $n$-node graphs may be achievable without the full range of techniques in this paper.}, $\ell_1$-regression \cite{LeeS15}, and MDPs, in on moderately dense instances. 
First, we provide a new IPM directly tailored to solving \eqref{eq:intro_main_formulation} which directly works with a variant of Lewis weights. Our method applies techniques in robustifying IPMs \cite{cls19,lsz19,b20,blss20,BrandLN+20,sy20,jswz20} and in particular techniques from \cite{lsz19} on robust IPMs for solving empirical risk minimization problems, to a variant of the Lewis-weight based optimality conditions suggested by \cite{ls19}. Further, the method applies and generalizes sampling and analysis techniques from \cite{BrandLN+20}. The combination of these techniques and interaction with Lewis weights induces a number of technical challenges. Interestingly, ultimately, our analysis leverages higher order smoothness properties of barriers for the intervals $\{x : \ell_i \leq x \leq u_i\}$ (\Cref{def:highselfcon} and \Cref{lemma:highselfcon}).

As with many recent results \cite{cls19,lsz19,b20,blss20,BrandLN+20,sy20,jswz20}, this new IPM reduces the challenge of solving \eqref{eq:intro_main_formulation} to solving a sequence of data structure problems and appropriately initializing and rounding the iterates of the IPM. One particularly challenging aspect is that our IPM requires that regularized $\ell_p$-Lewis weights be maintained efficiently throughout the algorithm. Though previous work \cite{blss20,BrandLN+20} provided various results on dynamically maintaining leverage scores, i.e. the special case when $p = 2$, and there are known algorithms for computing Lewis weights efficiently using leverage scores, the complexity of such a na\"ive approach is unclear. We overcome this issue by providing both a more careful reduction from Lewis weight maintenance to leverage score maintenance and a direct reduction from leverage score maintenance to detecting large rows in a dynamically changing matrix, what we call a heavy hitter data structure \cite{BrandLN+20,blss20}.

By considering specialized heavy hitter data structures for the various problems we consider, providing additional data structure as needed, and carefully applying the resulting IPM we obtain the main runtime results of this paper. In the special case of solving linear programs in standard form we provide new sampling data structure for matrices that allow us to improve upon the runtime of \cite{blss20} in certain settings. 

\subsection{Our Results}

Here we provide the main results of this paper, including new nearly linear time algorithms for solving \eqref{eq:intro_main_formulation} in different settings.

\paragraph{Linear Programming with Two-sided Constraints}
Our main contributions are efficient algorithms for solving primal/dual LPs with two-sided constraints \eqref{eq:intro_main_formulation} via IPMs. In the general case, where there is no additional structure on $\ma$ we obtain the following result.
\begin{theorem}
	[Primal solution for general LPs] \label{thm:primallp} Let $\mA\in\R^{m\times n}$,
	$c,\ell,u\in\R^{m}$, and $b\in\R^{n}$. Assume that there is a point
	$x$ satisfying $\mA^{\top}x=b$ and $\ell_{i}\le x_{i}\le u_{i}$
	for all $i \in [m]$. Let $W \defeq \max(\|c\|_{\infty},\|\ma\|_{\infty},\|b\|_{\infty},\|u\|_{\infty},\|\ell\|_{\infty},\frac{\max_{i}(u_{i}-\ell_{i})}{\min_{i}(u_{i}-\ell_{i})})$.
	For any $\d>0$ there is an algorithm running in time $\Otil((mn+n^{2.5})\log (W/\d))$ that with high probability (w.h.p.)
	which computes a vector $x^{\final}$ satisfying
	\begin{align*}
	\|\mA^{\top}x^{\final}-b\|_{\infty}\le\d\enspace\text{ and }\enspace\ell_{i}\le x^\final_{i}\le u_{i} ~ \forall i\enspace\text{ and }\enspace c^{\top}x^{\final}\le\min_{\substack{\mA^{\top}x=b\\
			\ell_{i}\le x_{i}\le u_{i}\forall i
		}
	}c^{\top}x+\d.
	\end{align*}
\end{theorem}

The prior best result runtimes for achieving guarantees comparable to \Cref{thm:primallp} were $\tO(m^{ \max \{\omega,2+1/18\} })$ \cite{jswz20} ($\omega \approx 2.37286$ \cite{Williams12,Gall14a,aw21} is the exponent of current matrix multiplication)
and $\tO((\nnz(\mA)+ n^{2})\sqrt{n})$ \cite{LeeS15}. Whenever $\mA$ is tall and dense, i.e. $m \ge n^{1.5}$ and $\nnz(\mA) = \Omega(mn)$, this corresponds to a nearly linear time algorithm for solving \eqref{eq:intro_main_formulation} to high precision.

Interestingly, even in the special case when $\ell_i = 0$ and $u_i = \infty$ for all $i \in [n]$ this improves upon the previous best runtimes of $\tO(m^\omega)$ and $\tO((\nnz(\mA)+\sqrt{n})\sqrt{n})$ mentioned above, as well as $\tO(mn+n^3)$ \cite{blss20}. In particular, we improve upon \cite{blss20} by improving the additive $\tilde{O}(n^{3})$ term to a $\tilde{O}(n^{2.5})$. This improvement stems from IPM sampling techniques of \cite{BrandLN+20} (as refined in this paper) and new data structures introduced in this paper.

\paragraph{$\ell_1$-Regression and MDPs} 
\Cref{thm:primallp} immediately yields improved runtimes for additional prominent optimization problems. For instance, when $b = \vzero$ and $\ell_i = -1$ and $u_i = 1$ for all $i \in [n]$, the dual formulation in \eqref{eq:intro_main_formulation} encodes to $\ell_1$-regression. Consequently, we obtain the following result.
\begin{theorem}
	[Dual solution for general LPs, $\ell_1$-regression] \label{thm:duallp}Let
	$\mA\in\R^{m\times n}$, $c\in\R^{m}$, and $\delta > 0$.There is an algorithm running in time $\wt O((mn+n^{2.5})\log(W/\d))$ for $W \defeq \max(\|c\|_{\infty},\|\ma\|_{\infty})$ which w.h.p. computes a vector $z\in\R^{n}$ such that 
	\[
	\|\mA z+c\|_{1}\le\min_{z\in\R^{n}}\|\mA z+c\|_{1}+\d.
	\]
\end{theorem}
As with \Cref{thm:primallp} this runtime is nearly linear whenever $\mA$ is tall and dense, i.e. $\nnz(\mA) = \Omega(mn)$ and $m \ge n^{1.5}$. Further, as with \Cref{thm:primallp}, in the regime of high accuracy algorithms, it improves upon the results of \cite{cls19,lsz19,b20,sy20,jswz20,LeeS15} mentioned above.

Further, leveraging a reduction from \cite{SidfordWWY18},  in certain settings this yields to improved running times for solving MDPs, a fundamental mathematical model for reasoning about uncertainty.
An instance of the discounted Markov Decision Process (DMDP) is specified by a tuple
$(S,A,P,r,\gamma)$ where $S$ is a the state space, $A$ is the action
space, $P$ describes state-action-state transition probabilities,
$r$ describes state-action rewards in range $[-M,M]$, and $\gamma\in(0,1)$
is a discount factor. The goal is to compute a policy that maps each
state to an action that that (approximately) maximizes the reward
in a certain sense. (See \Cref{sec:app:mdp} for more precise definition of the problem)
We obtain the following result.
\begin{theorem}
	[Discounted MDP]\label{cor:DMDP}Given a DMDP $(S,A,P,r,\gamma)$,
	there is an algorithm that, with high probability, computes a $\eps$-optimal
	policy $\pi$ in $\Otil((|S|^{2}|A|+|S|^{2.5})\log(\frac{M}{(1-\gamma)\epsilon}))$
	time.
\end{theorem}

Since the input size of the state-action-state transition is $\Omega(|S|^2|A|)$, this algorithms runs in nearly-linear time whenever $|A|\ge \sqrt{|S|}$. Also, this result directly improves upon the previous algorithm with running time $\Otil((|S|^{2.5}|A|)\log(\frac{M}{(1-\gamma)\epsilon}))$ \cite{ls14,SidfordWWY18}.

\paragraph{Minimum Cost Flow}

In the minimum cost flow problem, we are given a connected directed
graph $G=(V,E,u,c)$ with edges capacities $u\in\R_{\ge0}^{E}$ and
costs $c\in\R^{E}$. We call $x\in\R^{E}$ an $s$-$t$ flow for $s,t\in V$
if $x_{e}\in[0,u_{e}]$ for all $e$ in $E$ and for each vertex $v\notin\{s,t\}$
the amount of flow entering $v$, i.e.~$\sum_{e=(a,v)\in E}x_{e}$
equals the amount of flow leaving v, i.e.~$\sum_{e=(v,b)\in E}x_{e}$.
The value of $s$-$t$ flow is the amount of flow leaving $s$ (or
equivalently, entering $t$). The maximum flow problem is to compute
an $s$-$t$ flow of maximum value. In the minimum cost maximum flow
problem, the goal is to compute a maximum $s$-$t$ flow of minimum
cost, $\sum_{e\in E}c_{e}x_{e}=c^{\top}x$. 

Such problems can be expressed in the form of \eqref{eq:intro_main_formulation} by taking $\mA$ as the graph incidence matrix, letting $\ell_i$ and $u_i$ denote edge capacities, and choosing $b$ appropriately. Unfortunately, directly applying \Cref{thm:primallp} would yield a large $\tO((mn+n^{2.5})\log W)$ runtime which does not improve upon previous results. However, applying the techniques developed in this paper along with Laplacian system solvers \cite{SpielmanT04,KoutisMP10,KoutisMP11,KOSZ13,CohenKMPPRX14,PS14,LeePS15,KyngLPSS16,KS16} and data structure ideas from \cite{BrandLN+20} specific to graphs we prove the following theorem.

\begin{theorem}
	[Min cost flow]\label{thm:mincost flow}There is an algorithm that,
	given a $n$-vertex, $m$-edge, directed graph $G=(V,E,u,c)$ integral edge capacities $u\in\Z_{\ge0}^{E}$ and
	costs $c\in\Z^{E}$, with high probability, computes a minimum cost
	maximum flow in $\Otil(m\log(\|u\|_{\infty}\|c\|_{\infty})+n^{1.5}\log^{2}(\|u\|_{\infty}\|c\|_{\infty}))$
	time.
\end{theorem}
Efficiently solving the minimum cost flow problem to high accuracy gives an algorithm for maximum flow in the same runtime on weighted graphs, and we give the reduction formally in \Cref{thm:maxflow} in Section \ref{sec:maxflow}. Using standard capacity scaling methods \cite{ahuja1991distance}, we can improve the $\log^2 W$ dependence to $\log W$ on the $n^{1.5}$ term in this case, where $W = \max(\|u\|_\infty).$ The previous best runtimes for mincost flow were $m^{4/3+o(1)} \log W$ in the case of unit capacity graphs \cite{amv20} and $\tO(m\sqrt{n}\log^{O(1)} W)$ \cite{ls19}. Our algorithm improves on these for dense graphs, and in particular runs in nearly linear time for $m \ge n^{1.5}.$

\begin{table}[h]
	\begin{tabular}{|l|c|p{0.23\textwidth}|c|c|}
		\hline
		\multirow{2}{*}{\textbf{Year}} & \multirow{2}{*}{\textbf{Authors}} & \multirow{2}{*}{\textbf{References}}                   & \multicolumn{2}{c|}{\textbf{Time}}             \\ \cline{4-5} 
		&                                   &                                                        & \textbf{Sparse} & \textbf{Dense} \\ \hline
		1970 &  Dinitz & \cite{Dinic70} &  \multicolumn{2}{c|}{$mn$}\\\hline
		1997             &  Goldberg, Rao   & \cite{GoldbergR98} & $m^{3/2}\log W$ & $mn^{2/3} \log W$          \\ \hline
		2013 & Madry&\cite{m13,m16}	&    $m^{10/7}W^{1/7}$                              &                                 \\ \hline
   		2013 & Lee and Sidford & \cite{ls14} & & $m \sqrt{n} \log^{O(1)} W$ \tabularnewline\hline
		2020& Liu and Sidford & \cite{ls20_stoc} & $m^{11/8} W^{1/4}$ & \\ \hline
		2020 & Liu and Sidford; Kathuria &\cite{ls20_focs, kathuria2020potential} & $m^{4/3} W^{1/3}$ & \\ \hline
		2020 & & \bf This paper & & $(m+n^{1.5})\log W$\\\hline
	\end{tabular}
	\caption{The summary of the results for the {\bf maximum flow} problem. $W$ denotes the maximum capacity. Subpolynomial ($n^{o(1)}$) terms are hidden. For simplicity, we only list exact algorithms which yielded polynomial improvements.}\label{tab:intro:maxflow}
\end{table}

\begin{table}[h]
	\begin{centering}
		\begin{tabular}{|l|p{0.31\textwidth}|p{0.13\textwidth}|c|c|}
		\hline
		\multirow{2}{*}{\textbf{Year}} & \multirow{2}{*}{\textbf{Authors}} & \multirow{2}{*}{\textbf{References}}                   & \multicolumn{2}{c|}{\textbf{Time}}             \\ \cline{4-5} 
		&                                   &                                                        & \textbf{Sparse} & \textbf{Dense} \\ \hline
			\hline 
			1972  & Edmonds and Karp & \cite{ek72} & \multicolumn{2}{c|}{$m^{2}\log W$} \tabularnewline
			\hline 
			1987  & Goldberg and Tarjan & \cite{gt90} &  & $m n \log W$ \tabularnewline
			\hline 
			2008 & Daitch and Spielman & \cite{ds08} & $m^{3/2} \log^{2} W$ & \tabularnewline
			\hline 
			2013 & Lee and Sidford & \cite{ls14} & & $m\sqrt{n} \log^{O(1)} W$\tabularnewline
			\hline
			2017 & * Cohen, Madry, Sankowski, Vladu & \cite{cmsv17}& $m^{10/7}\log W$ & \\
			\hline
			2020 & * Axiotis, Madry, Vladu & \cite{amv20} & $m^{4/3}\log W$ & \\
			\hline
			2020 &  & {\bf This paper} & & $ m\log W+ n^{1.5} \log^{2} W $\tabularnewline
			\hline 
		\end{tabular}
		\par\end{centering}
	\caption{The summary of the results for the {\bf minimum-cost flow} problem. Subpolynomial ($n^{o(1)}$) terms are hidden. $W$ denotes the maximum absolute value of capacities
		and costs. For simplicity, we only list exact algorithms which yielded polynomial improvements.
		Results marked with an asterisk work only on unit-capacity graph.
		}\label{tab:intro:mincost flow}
\end{table}

\subsection{Related Work}

The problems we consider in this paper, e.g. linear programming, dynamic data structures, minimum cost flow, $\ell_{1}$-regression, MDPs, are all incredibly well-studied. Each has an extensive history and numerous results. Here, we provide just a brief summary of the results and tools most directly related to this paper.

\paragraph{Linear Programming IPMs:} There has been significant work towards the design of IPMs for linear programming, starting from \cite{Karmarkar84,Renegar88}. More recently, there have been IPM-based runtime improvements to linear programming by decreasing the number of iterations \cite{ls14}, proving that maintaining approximate primal/dual solutions suffice \cite{cls19,lsz19,b20,blss20,sy20,jswz20,BrandLN+20}, and using data structures to decrease iteration costs \cite{LeeS15,cls19,lsz19,b20,sy20,jswz20,BrandLN+20}.
There is a related line of work on strongly polynomial linear programming \cite{vy96,mt03,mt05,dhnv20,dnv20}, where the goal is to achieve exact solutions with improved parameter dependencies, instead of high accuracy solutions as we do here.
In \cite{lsz20} it was shown how to implement IPMs in the semi-streaming model.

\paragraph{$\ell_1$-Regression:} There have been several algorithms for $\ell_1$-regression in both the low accuracy ($\poly(\eps^{-1})$ dependence) \cite{Clarkson05,Nesterov09,SohlerW11,ClarksonW13,WoodruffZ13,ChinMMP13,YangCRM16,ClarksonDMMMW16,DurfeeLS18,WangW19} and high accuracy ($\log(1/\eps)$ dependence) \cite{LeeS15,cls19} regimes.  The state of the art results in the high-accuracy regime are \cite{LeeS15} which achieves a $\tO((\nnz(\mA) + n^2)\sqrt{n} \log(1/\eps))$ runtime, and \cite{jswz20} which achieves a $\tO(n^{ \max\{\omega,2+1/18\} })$ runtime, improving on $\tO(n^{ \max\{\omega,2+1/6\} })$ \cite{cls19}.

\paragraph{Minimum Cost Flow:} There has been significant work towards combinatorial algorithms for the minimum cost flow problem \cite{ek72,t85,o84,gt88,gt90,o93} in both the strongly polynomial and weakly polynomial regimes. Daitch and Spielman \cite{ds08} showed that one can use a Laplacian system solver to implement steps of an IPM to achieve a more efficient mincost flow algorithm with logarithmic capacity dependence. Since then, runtime improvements have been achieved by reducing the number of iterations of a general IPM to $\tO(\sqrt{n})$ \cite{ls14} and to $O(m^{1/3+o(1)})$ \cite{amv20} for graphs with unit-capacity, improving on $\tO(m^{10/7} \log W)$ \cite{cmsv17}.

\paragraph{Markov Decision Process (MDP):} 
We focus on solving \emph{discounted} MDPs.\footnote{
	Other variants of this problems includes deterministic MDPs (equivlent to the min-mean cycle problem) \cite{DG98,CCGMQ98,Madani02,BrandLN+20} and average-reward MDPs \cite{Mahadevan96,AO06,JOA10}.}
Previous works in the high precision regime (i.e.~logarithmic dependency on error) includes \cite{tseng1990solving,LittmanDK95,SidfordWWY18} with the best running time of $\tO((|S|^2|A|+\frac{|S||A|}{(1-\gamma)^3})\log(\frac{M}{\epsilon}))$. 
Strongly polynomial time exact algorithms are also known \cite{ye2005new,ye2011simplex,scherrer2016improved}.

For algorithms that depend logarithmically on one minus the discount factor $\gamma$, the algorithm by \cite{ls14} implies a $\Otil((|S|^{2.5}|A|)\log(\frac{M}{(1-\gamma)\epsilon}))$ running time (shown in \cite{SidfordWWY18}). Our result in \Cref{cor:DMDP} directly improves this algorithm.

There is another line of work focusing on fast algorithms in the low precision regime with polynomial dependency on the error parameter \cite{KS98a,azar2013minimax,W17i,SidfordWWY18,SidfordWWYY18,Wainwright19,wang2020randomized,AgarwalKY20,Dlwcgc20}. This setting is not directly comparable to our result. 

\paragraph{Dynamic Data Structures:} 

IPMs reduce the task of solving linear programs
to the task of solving linear systems.
Instead of solving these linear systems from scratch in each iteration,
these iterative algorithms can be sped up by using data structures
that efficiently maintain the matrix inverse corresponding to the linear system
\cite{Karmarkar84,%
Vaidya89a,%
NesterovN91,%
LeeS15,%
cls19,%
lsz19,%
sy20,%
jswz20,%
blss20%
}. %
There also exist data structures to efficiently maintain the solution of the linear system
instead of (or in addition to) the corresponding matrix 
\cite{Sankowski04,BrandNS19,b20,jswz20,b21}.
Recently there have also been data structures developed that
are able to efficiently maintain an approximation of the primal/dual solution
\cite{lsz19,blss20,BrandLN+20,jswz20}.
For graph applications these algorithms are based on
the dynamic expander decomposition technique \cite{NanongkaiSW17,SaranurakW19,BernsteinBNPSS20,GoranciRST20}.
For general LPs, the data structure for maintaining approximate primal/dual solutions
are based on heavy hitters and sketching \cite{jl84,knpw11,p13,nn13,kn14,lnnt16,psw17,cjn18,ns19,nsw19}.

\subsection{Organization}
We give the preliminaries in Section \ref{sec:prelim}. Our overview is in \Cref{sec:overview}, split into overview of the IPM in \Cref{sec:overview:ipm} and overview of the data structures in \Cref{sec:overview:ds}. We present our IPM in \Cref{sec:ipm}, and show our regularized Lewis weight maintenance data structure in \Cref{sec:lewis_weight_maintenance}. We show how to use data structures in implement our IPM in \Cref{sec:pathfollowing}. We analyze the runtime of our IPM in the graphical setting and show applications and mincost flow and maxflow in  \Cref{sec:mincostflow}. Finally, we analyze the runtime of our IPM for general linear programs and show applications to Markov Decision Processes in  \Cref{sec:linearprogram}.

Several additional pieces are deferred to the appendix. In Appendix \ref{sec:ipmproofs} we give omitted proofs from Section \ref{sec:ipm}. The remaining sections of the appendix give data structures based on previous methods of \cite{BrandLN+20}. In Section \ref{sec:matrix_data_structures} we give our \textsc{HeavyHitter} and sampling data structures, and in Section \ref{sec:leverage_score_maintenance} we give our leverage score maintenance data structure. We show how to maintain the primal variable and gradient of the centrality potential in Section \ref{sec:gradient_maintenance}, and show how to maintain the dual slack variable in Section \ref{sec:dual_maintenance}. Finally we state the graph specific data structures based on expander decompositions in Section \ref{sec:graph_data_structures}.

\ifdefined\DEBUG
\fi

\section{Preliminaries}
\label{sec:prelim}
We follow similar notation as in \cite{BrandLN+20}.
We let $[n] \defeq \{1,2,...,n\}$ and $\unitvec_i$ denote the $i$-th standard unit vector.
We use $\tilde{O}(\cdot)$ notation to hide $(\log \log W)^{O(1)}$, $\log \epsilon^{-1}$, and $(\log n)^{O(1)}$ factors, where $W$ typically denotes the largest absolute value used for specifying any value in the problem (e.g. demands and edge weights) and $n$ denotes the number of nodes.
When we write \emph{with high probability} (or w.h.p), we mean with probability $1-n^c$ for any constant $c > 0$.
We write $\mathbf{1}_{\text{condition}}$ for the indicator variable, 
which is $1$ if the condition is true and $0$ otherwise.

\paragraph{Diagonal Matrices}
Given a vector $v \in \R^d$ for some $d$,
we write $\mdiag(v)$ for the $d\times d$ diagonal matrix with $\mdiag(v)_{i,i} = v_i$.
For a vector $v$ we also write $\mV$ for the diagonal matrix $\mdiag(v)$
when clear from context. 

\paragraph{Matrix and Vector operations}
Given vectors $u,v \in \R^d$ for some $d$,
we perform arithmetic operations $\cdot,+,-,/,\sqrt{\cdot}$ element-wise.
For example $(u\cdot v)_i = u_i\cdot v_i$ or $(\sqrt{v})_i = \sqrt{v_i}$.
For the inner product we write $\langle u, v \rangle$ and $u^\top v$ instead.
For a vector $v \in \R^d$ and a scalar $\alpha \in \R$ we let $(\alpha v)_i = \alpha v_i$ and $(v + \alpha)_i = v_i + \alpha$.

For symmetric matrices $\mA,\mB\in \R^{n\times n}$ we write $\mA\preceq \mB$ to indicate that $x^\top \mA x \leq x^\top \mB x$ for all $x\in \R^n$ and define $\succ$, $\prec$, and $\succeq$ analogously. We call any matrix (not necessarily symmetric) {\em non-degenerate} if its rows are all non-zero and it has full column rank. 

We use $u\approx_{\epsilon}v$ to denote that $\exp(-\epsilon)v\leq u\leq\exp(\epsilon)v$
entrywise and $\ma\approx_{\epsilon}\mb$ to denote that $\exp(-\epsilon)\mb\preceq\ma\preceq\exp(\epsilon)\mb$.
Note that this notation implies $u \approx_\epsilon v \approx_\delta w$ $\Rightarrow$ $u \approx_{\epsilon + \delta} w$, and $u \approx_\epsilon v$ $\Rightarrow$ $u^\alpha \approx_{\epsilon\cdot |\alpha|} v^\alpha$ for any $\alpha \in \R$.

For any matrix $\ma$ with real entries, let $\nnz(\ma)$ denote the number of non-zero entries in $\ma$
and $\nnz(a_i)$ be the number of non-zero entries in the $i$-th row of $\mA$.

Note that we can express the approximation error of \Cref{lem:solver} 
as some spectral approximation, i.e. there exists some $\mH \approx_{20\epsilon} \mA^\top\mw\mA$ 
such that $\mH \ox = b$ \cite[Section 8]{blss20}.

\paragraph{Leverage Scores and Lewis-Weights}
For any non-degenerate matrix $\ma\in\R^{m\times n}$ 
we let $\sigma(\ma) \in \R^m$ with $\sigma(\mA)_i \defeq (\mA(\mA^\top\mA)^{-1}\mA^\top)_{i,i}$ 
denote $\ma$'s \emph{leverage scores}.
For $p \in (0, \infty)$ and non-degenerate matrix $\mA \in \R^{m\times n}$ we define the \emph{$\ell_p$ Lewis weight} as the solution $w \in \R^m_{>0}$ to the equation $w = \sigma(\mW^{\frac{1}{2} - \frac{1}{p}}\mA)$, where $\mW = \diag(w)$. We use a regularized Lewis weight in our algorithms, and this is defined in Definition \ref{def:lewis}.

\paragraph{Norms}
We write $\| \cdot \|_p$ for the $\ell_p$-norm, i.e. $\|v\|_p := (\sum_i |v_i|^p)^{1/p}$, $\|v\|_\infty = \max_i | v_i |$ and $\| v \|_0$ being the number of non-zero entries of $v$.
For a positive definite matrix $\mM$ we define $\| v \|_\mM = \sqrt{v^\top \mM v}$. 
For a vector $\tau$ we define $\| v \|_\tau := (\sum_i \tau_i v_i^2)^{1/2}$
and $\|v\|_\tpi := \| v \|_\infty + C \log(4m/n) \|v\|_\tau$ for a large constant $C$,
where $m\ge n$ are the dimensions of the constraint matrix of the linear program (we define $\|v\|_\tpi$ again in Definition \ref{def:norm}).

\ifdefined\DEBUG
\fi

\section{Overview of Approach}
\label{sec:overview}
In this section, we give an overview of the major aspects of our algorithm, including the path following IPMs (\Cref{sec:overview:ipm}), the data structures necessary for its implementation (\Cref{sec:overview:ds}), and how to combine them  for our applications (\Cref{sec:overview:together}).

\subsection{IPM}
\label{sec:overview:ipm}

As context and motivation for our method, we start by discussing the IPM of \cite{ls19}. For matrices $\mA \in \R^{m \times n}$, vectors $b \in \R^n$, $c \in \R^m$, and lower and upper bounds $\ell, u \in \R^m$, this IPM solves linear programs of the form
\[ 
\min_{\substack{\mA^\top x = b \\ \ell_i \le x_i \le u_i \forall i \in [m]}} c^\top x 
\] 
to high accuracy. The runtime of this method is dominated by the runtime needed for solving $\tO(\sqrt{n})$ linear systems of the form $\mA^\top \mD \mA$ for non-negative diagonal matrices $\mD$. To achieve this result, for all $i \in [n]$ \cite{ls19} considers $1$-self-concordant barriers $\phi_i: (\ell_i, u_i) \rightarrow \R$ for the intervals for $i$, e.g. $\phi_i(x) = -\log(u_i-x)-\log(x-\ell_i)$ (Lemma \ref{lemma:highselfcon}), and for a path parameter $\mu$, considers the central path of points
\begin{align} x_\mu \defeq \argmin_{\mA^\top x = b} c^\top x + \mu \sum_{i=1}^m w_i \phi_i(x_i), \label{eq:mucentralpath} \end{align} for a well-chosen \emph{weight function} $w \in \R^m$. For standard IPMs such as that of Renegar \cite{Renegar88}, $w_i = 1$ for all $i$, and in \cite{ls19} and our algorithm $w$ is a function of $x$ (though it is often convenient to think of $w$ as fixed during a step). Optimality conditions tell us that for fixed $w$, \eqref{eq:mucentralpath} holds if $\mA^\top x_\mu = b$ and $c+\mA y + \mu \mW \g \Phi(x) = 0$ for some vector $y \in \R^n$, where $\mW = \diag(w)$ is the diagonal matrix of the weights $w$. This optimality condition can be re-written as $s + \mu \mW \g \Phi(x) = 0$ where $s = c+\mA y$ denotes the dual slack variables.

In  both \cite{ls19} and this work we allow $w$ to depend on $x$ and define a \emph{weight function} $w(x): \R^m \rightarrow \R^m_{>0}$. We say that a triple $(x, s, \mu)$ is \emph{central} for weight function $w(x)$ if
\begin{align} \mA^\top x = b \enspace \text{ , } \enspace s = \mA y + c \enspace \text{, and } \enspace s + \mu\mW(x) \g \Phi(x) = 0. \label{eq:centralcondition} \end{align}
A choice of $\mW(x)$ made in \cite{ls19}, and which we used a regularized version of in this paper (Definitions \ref{def:lewis}, \ref{def:cpweights}), is a $\ell_p$ Lewis weight for $p = 1 - \frac{1}{4 \log(4m/n)}$. The $\ell_p$ Lewis weight function $w(x)$ is defined as the solution to
\begin{align}
w(x) = \sigma(\mW(x)^{\frac{1}{2} - \frac{1}{p}} (\g^2\Phi(x))^{-\frac12} \mA), \label{eq:intro:lewisweight}
\end{align}
where $\sigma(\cdot)$ denotes the leverage scores of a matrix.
Using this weight function, \cite{ls19} is able to argue that $\tO(\sqrt{n})$ steps of an IPM suffice to solve the linear program, as opposed to $\tO(\sqrt{m})$ from Renegar's method. Carefully implementing methods based on variants of the $\ell_p$-Lewis weight function solves the maximum flow problem in time $\tO(m\sqrt{n})$ and general LPs in time $\tO(\sqrt{n}(\nnz(A) + n^{\omega}))$, accounting for the runtime needed to solve linear systems. 

To obtain further runtime improvements, there has been significant work towards performing less work per iteration by speeding up the linear system solve times via inverse maintenance \cite{LeeS15}, as well as more recent work showing that such methods can in fact be implemented even with only approximate values for the primal and dual variables $x, s$ \cite{cls19,lsz19,b20,blss20,sy20,jswz20,BrandLN+20}. To illustrate these \emph{robust IPMs}, consider the simple case where $\ell_i = 0$ and $u_i = \infty$ for all $i$, i.e. the condition on $x$ in the linear program is simply $x \ge 0$, and $\phi_i(x_i) = -\log x_i$. In this case the centrality condition, \eqref{eq:centralcondition}, reduces to $xs = w(x)\mu$ and this motivates the following \emph{centrality potential}
\begin{equation}
\label{eq:centrality_potential_one_side}
\Psi(x, s, \mu) \defeq \sum_{i=1}^m \cosh\left(\lambda\left(\frac{x_is_i}{w(x)_i \mu} - 1\right)\right) 
\end{equation}
for $\lambda = \Theta(\log m / \eps)$. Maintaining $\Psi(x, s, \mu) \le \poly(m)$ at all times ensures that $xs \approx_\eps w(x)\mu$. Consequently, these robust IPMs take projected Newton steps that induce gradient descent steps on the potential $\Psi$ to guarantee that it stays small in expectation throughout the algorithm.

The analysis in \cite{blss20,BrandLN+20} critically relied on the fact that only one-sided constraints, $x \ge 0$, were imposed, instead of a two-sided constraint, $\ell \le x \le u$. These works leveraged that the centrality condition $xs = w(x)\mu$ for $\ell_p$ Lewis weights can be written in terms of leverage scores of a slightly different diagonal weighting. Specifically, if $xs = w(x)\mu$, where $w(x)$ is the $\ell_p$ Lewis weight, then for $\alpha = 1-p$ we have that
$xs = \sigma(\mS^{-1/2-\alpha}\mX^{1/2-\alpha}\mA)$. The reliance on this fact impairs extending it to the setting of two-sided constraints.

We bypass this issue by working directly with $\ell_p$ Lewis weights and the centrality condition \eqref{eq:centralcondition} for general $1$-self-concordant functions $\phi$. Interestingly, our analysis requires a fourth derivative condition in Definition \ref{def:highselfcon}, beyond the standard third derivative condition of self-concordance. Formally, we consider the centrality potential (\Cref{def:potential})
\begin{align}
\Psi(x,s,\mu)\defeq\sum_{i=1}^{m}\cosh\left(\lambda\left(\frac{s_{i}+\mu\tau(x)_{i}\phi_{i}'(x_{i})}{\mu\tau(x)_{i}\sqrt{\phi_{i}''(x_{i})}}\right)\right),
\label{eq:over_potential}
\end{align}
where intuitively, the denominator arises from normalizing by the Hessian of the current point $x$. To analyze the progress of the Newton steps towards decreasing the potential, analysis is required to understand and bound derivatives of the $\ell_p$ Lewis weights. Additionally, there are several other technical challenges, including working with a regularized version of the $\ell_p$ Lewis weight to ensure that the Lewis weights are all $\ge n/m$, carefully maintaining approximate feasibility of the primal variable $x$, and using spectral sparsifiers of $\mA^\top \mD\mA$ instead of the true matrix for efficient inverse maintenance. This loss of feasibility requires us to use a sampling procedure on the primal variable, and we build a general theory for valid sampling distributions that allow our algorithm to work (Definition \ref{def:validdistro}). At a high level, we show that any sampling scheme that satisfies various properties, such as bounded variance, maximum, and mean preservation suffices to implement our IPM. The goal of showing that these generalized sampling schemes work is to handle both sampling each coordinate independently and sampling coordinates proportional to weights, so that we can handle the graphical case (as in \cite{BrandLN+20}) and linear programs.

Overall, we show that we can takes steps of size $\Omega(n^{-1/2})$ while maintaining that the expected potential is polynomially bounded, and that all points $x$ we maintain are approximately feasible. In this way, we can implement an IPM for two-sided linear programs that requires $\tO(\sqrt{n})$ steps that only approximately maintains the primal variable $x$ and dual slack $s$.

\subsection{Data Structures}
\label{sec:overview:ds}

As outlined in \Cref{sec:overview:ipm}, in contrast to \cite{blss20, BrandLN+20}, our IPM maintains approximate regularized $\ell_p$ Lewis weights for $p \in [1/2,2)$. 
To efficiently implement the IPM we do not want to recompute the Lewis weights from scratch in every iteration. Instead we seek a data structure that maintains approximate Lewis weights.
Here we describe how such a data structure can be obtained
by reducing to the \textsc{HeavyHitter} data structure problem defined below:
\begin{definition}[Heavy hitter]\label{def:heavyhitter}
For $c\in\R^m$ and $P,Q \in \R_{>0}$ with $nP \ge \|c\|_1 \ge P$ ,
we call a data structure with the following procedures a $(P,c,Q)$-\textsc{HeavyHitter} data structure:
\begin{itemize}
\item \textsc{Initialize}$(\mA \in \R^{m \times n}, g \in \R_{>0}^m)$
	Let $\mA$ be a matrix with $c_i \ge \nnz(a_i))$, $\forall i \in [m]$ and $P \ge \nnz(\mA)$.
	The data structure initializes in $O(P)$ time.
\item \textsc{Scale}$(i \in [m], b \in \R)$:
	Sets $g_i \leftarrow b$ in $O(c_i)$ time.
\item \textsc{QueryHeavy}$(h \in \R^n, \epsilon \in (0,1) )$:
	Returns $I \subset [m]$ containing exactly those $i$ with $|(\mG \mA h)_i| \ge \epsilon$
	in $O( \epsilon^{-2} \| \mG \mA h \|_c^2 + Q )$ time.
\end{itemize}
\end{definition}

A contribution of our work
is to show that, if we have such a data structure for a matrix $\mA$,
then we are able to efficiently maintain the Lewis weights of $\mV \mA$ for a diagonal matrix $\mV$ 
that changes over time (i.e.~$\mV = (\nabla^2 \Phi(x))^{-1/2}$ when used inside our IPM, see \eqref{eq:intro:lewisweight}).

Constructing \textsc{HeavyHitter}-data structures
was key to advances in \cite{BrandLN+20, blss20}.
For the special case where $\mA$ is an edge-vertex incidence matrix, 
\cite{BrandLN+20} constructed a \textsc{HeavyHitter}-data structure 
with complexities $P = \tilde{O}(m)$, 
$c_i = \tilde{O}(1)$ for all $i \in [m]$, 
and $Q = \tilde{O}(n \log W)$, 
where $W$ is a bound on the ratio of the largest to smallest non-zero entry in $\mG$.
This data structure will be useful for our min-cost flow application.
In \cite{blss20} a \textsc{HeavyHitter}-data structure was given for general $m\times n$ matrices,
where $P = \tilde{O}(\nnz(\mA))$, $c_i = \tilde{O}(n)$ for all $i \in [m]$,
and $Q = \tilde{O}(n)$.
This data structure can be used for general linear programs.
Our algorithm for solving general LPs use this data structure from \cite{blss20}
and the algorithms for graph problems such as min-cost flow use the data structure from \cite{BrandLN+20}.

We now outline how to reduce the task of maintaining the Lewis weights $\tau(\mG \mA)$ under updates to $\mG$,
to the \textsc{HeavyHitter} problem.
This reduction is done via the intermediate data structure problem of maintaining the leverage scores $\sigma(\mG\mA)$ under updates to $\mG$.
We show that Lewis weight maintenance can be reduced to leverage score maintenance
and show that leverage score maintenance can be reduced to the \textsc{HeavyHitter} problem.

\subsubsection{Regularized Lewis Weights}
\label{sec:overview:lewis}

We are interested in the regularized $\ell_p$-Lewis weight of a matrix $\mM$ 
which is defined as the value $\tau(\mM)$ that satisfies the recursive equation
$\tau = \sigma(\mT^{1/2-1/p}\mM) + z$ for a given vector $z \in \R^m_{>0}$, and $\mT = \diag(\tau)$.
We want a data structure that maintains an approximation $\otau \approx_\epsilon \tau(\mV \mA)$
for any $p \in [1/2,2)$ and $\epsilon > 0$ under updates to $\mV$.
Note that for the IPM we use $p = 1 - 1/(4\log(4m/n))$; consequently, $p \in [1/2,2)$ but intuitively, may be thought of as an approximate $\ell_1$-Lewis weight.

To outline our data structure, we first want to outline the algorithm of Cohen and Peng \cite{CohenP15}
that can be adapted to compute an approximation the regularized $\ell_p$-Lewis weight
$\tau(\mM)$ in the static setting (i.e.~when the input matrix does not change over time).
Given some matrix $\mM$ we initialize with $w = \onevec_m$ 
and repeatedly set 
\begin{align}
w \leftarrow (w^{2/p-1}(\sigma(\mW^{1/2-1/p}\mM)+z))^{p/2}.
\label{eq:over:lw_recursion}
\end{align}
One can prove (see \Cref{lem:approximate_contraction}) that each iteration reduces the approximation error by a $1-p/2$ factor,
i.e.~if we had $w \approx_\gamma \sigma(\mW^{1/2-1/p}\mM)+z$ before \eqref{eq:over:lw_recursion},
then we have $w \approx_{\gamma(1-p/2)} \sigma(\mW^{1/2-1/p}\mM)+z$ after  \eqref{eq:over:lw_recursion}.
Since $\onevec_m \approx_{O(\log(m))} \sigma(\mM)+z$ (by $n/m \le z \le \poly(m)$), after $\Theta(\log((\log m)/\eps))$ iterations of \eqref{eq:over:lw_recursion} we have $w \approx_\epsilon \tau(\mM)$.

We provide an efficient extension of this analysis to the dynamic setting. One natural idea for doing this would be to initialize $K = \Omega(\log((\log m)/\eps))$
data structures $D_1,\dots,D_K$ that maintain the following approximate leverage scores:
Let $w^{(1)} = \onevec_m$ and define recursively
\begin{align}
\osigma^{(i)} \approx \sigma((\mW^{(i)})^{1/2-1/p}\mV \mA)+z \label{eq:overview:approx_sigma}\\
w^{(i+1)} \leftarrow ((w^{(i)})^{2/p-1}\osigma^{(i)})^{p/2} \label{eq:overview:w_recursion}
\end{align}
where $\osigma^{(i)}$ is maintained a the leverage score data structure $D_i$ discussed in Section \ref{sec:overview:leverage}

If the leverage score data structures are accurate enough 
(i.e.~the approximation in \eqref{eq:overview:approx_sigma} is good enough),
then $w^{(K)}$ for $K = \Omega(\log((\log m)/\eps))$ would be a good approximation of the Lewis weight.
Further, this $w^{(K)}$ can be maintained under updates to $\mV$:
When $\mV$ changes, we update all the leverage score data structures $D_1, \dots ,D_K$.
Likewise, if some $\osigma^{(i)}$ changes, then we update $w^{(i+1)}$
and the data structure $D_{i+1}$ that maintains $\osigma^{(i+1)} \approx \sigma((\mW^{(i+1)})^{1/2-1/p}\mV \mA)+z$.

\paragraph{Problems with this approach}

While one can show that the previously outlined approach 
would indeed allow us to maintain approximate regularized Lewis weights
(assuming the approximate leverage scores $\osigma^{(i)}$ are accurate enough),
we are not able to analyze the time complexity of this process.
This is because an update to some $D_i$ (i.e.~when $w^{(i-1)}$ changes)
causes the output $\osigma^{(i)}$ to change as well, 
thus changing the input to $D_{i+1}$.
This means an update to $D_i$ might propagate through all other $D_{i'}$ for $i'>i$.
The computational cost of this propagation of the updates is difficult to analyze
because updating the $j$-th entry of the input of some data structure $D_i$
requires time proportional to $\osigma^{(i)}_j$.
Now, only for large $i$ do we know that $w^{(i)} \approx \osigma^{(i)}$ and can show that, $w^{(i)}$ is an approximation to the regularized Lewis weights. When this happens, we have bounds on how $w^{(i)}$ changes from guarantees of the IPM and this implies a small time complexity for $D_i$. However, for small $i$, $w^{(i)} \not\approx \osigma^{(i)}$ and the same bounds and complexity analysis does not immediately apply.

One attempt to fix this issue would be to start from a moderately good approximation, i.e.
\begin{align}
w^{(1)} \approx_\gamma \sigma((\mW^{(1)})^{1/2-1/p}\mV \mA)+z \label{eq:overview:moderate_approx}
\end{align}
for $0 < \gamma = O(\epsilon)$.
Then after only $K = O(1)$ recursions we have that $w^{(K)}$ is an $\epsilon$-approximation of the regularized Lewis weights.
Since here we only have $O(1)$ data structures $D_i$ and each $\osigma^{(i)}$ 
is at least an $O(\epsilon)$-approximation of the regularized Lewis weight,
we are able to bound the time for propagating the updates through all $D_i$.

The assumption \eqref{eq:overview:moderate_approx} 
on $w^{(1)}$ can be satisfied as follows.
Assume the input $\mV$ changes to some $\mV'$.
We know by guarantees of the IPM that $\mV \approx_{O(\epsilon)} \mV'$.
Let $\ow^{(K)}$ be the value $w^{(K)}$ we previously returned as $\epsilon$-approximation of the regularized Lewis weight, then we have
$$\ow^{(K)} 
\approx_\epsilon 
\sigma((\mW^{(K)})^{1/2-1/p}\mV \mA)+z 
\approx_{O(\epsilon)} 
\sigma((\mW^{(K)})^{1/2-1/p}\mV' \mA)+z.$$
Thus we can define $w^{(1)} := \ow^{(K)}$ as the required moderately good approximation.

The problem with this approach is that the vector $w^{(K)}$
might change in many entries when switching from $\mV$ to $\mV'$.
Thus we might have to perform many updates to the data structure $D_1$
at the start of each iteration, resulting in a larger than desired complexity.

\paragraph{The final algorithm} Our regularized Lewis weight data structure combines these ideas with one more trick to bound the number of updates to $D_1$; we delay the updates a bit.
We know that we have $w^{(K)} \approx_\epsilon \tau$ where $\tau$ is the exact regularized Lewis-weight.
We additionally maintain some $\ow^{(K)}$ where we set $\ow^{(K)}_j \leftarrow w^{(K)}_j$
whenever the entry $w^{(K)}_j$ changes and $\ow^{(K)}_j \not\approx_{3\epsilon} w^{(K)}_j$.
This way we know $\ow^{(K)}_j$ only changes, 
whenever $\tau_j$ must have changed by at least an $\exp(\epsilon)$ factor.
By guarantees of the IPM (Lemma \ref{lemma:tauchange}) 
we can bound how often entries of $\tau$ change by an $\exp(\epsilon)$ factor,
which then in turn bounds how often entries of $\ow^{(K)}$ are changed.
As we set $w^{(1)} := \ow^{(K)}$ whenever the input $\mV$ changes,
we can now bound the number of updates to $D_1$.
Note that here we still satisfy the requirement \eqref{eq:overview:moderate_approx}
because $\ow^{(K)} \approx_{3\epsilon} w^{(K)}$,
so we can still bound the time spent on propagating the updates to $D_1$ through all other $D_i$.
This way we maintain a good approximation of the Lewis-weight with a low overall complexity bound.

\subsubsection{Leverage Scores}
\label{sec:overview:leverage}

In \Cref{sec:overview:lewis} we outlined how to maintain approximate Lewis-weights,
if we have access to a data structure that can maintain approximate leverage scores.
Here we outline how to efficiently maintain an approximation 
$\osigma \approx_\epsilon \sigma(\mG \mA) + z$
for $\mG = \mdiag(g)$, $g \in \R^m_{>0}$, $z \in \R^m_{>0}$.
Here the vector $g$ is allowed to change over time,
while matrix $\mA$ and vector $z$ are fixed,
and our task is to create a data structure to maintain $\osigma$.
We obtain such a data structure by reducing to the \textsc{HeavyHitter} problem
(\Cref{def:heavyhitter}) for the same matrix $\mA$.
Since variants of this have been considered in prior work, 
we first compare their results and explain why these data structures are not sufficient 
for our reduction to maintain regularized Lewis weights.
We then describe how we obtain our data structure for leverage scores,
but we do not yet optimize the complexity to highlight the general idea.
At last, we outline how to speed-up the resulting data structure.

\paragraph{Comparison to previous work}

The general idea of our leverage score data structure 
is the same as in \cite{blss20} and \cite{BrandLN+20},
the main difference here is how we improve the complexity.
Specifically, the leverage score data structures from \cite{blss20} and \cite{BrandLN+20}
were able to maintain an $\epsilon$-approximation of the leverage scores 
$\osigma \approx_\epsilon \sigma(\mG \mA)$,
if the input $\mG$ was an $O(\epsilon/\log n)$-approximation of some other $\tmG$
that satisfied some stability properties, 
i.e. the diagonal matrix $\tmG$ must change very slowly over time.
For previous IPMs $\mG = \omX^{1/2}\omS^{-1/2}$, where $\ox \approx_{O(\epsilon/ \log n)} x$,
$\os \approx_{O(\epsilon / \log n)} s$, so $\mG$ was an $O(\epsilon/\log n)$-approximation of
$\tmG := \mX^{1/2}\mS^{-1/2}$ where both $x,s$ are stable 
(i.e.~they change slowly over the runtime of the algorithm) 
by guarantees of the IPM.
Thus data structures in \cite{blss20} and \cite{BrandLN+20}
were able to maintain leverage scores efficiently.

Unfortunately, these data structures are not usable for our Lewis-weight reduction.
This is because in order for the recursion in \eqref{eq:overview:approx_sigma} and \eqref{eq:overview:w_recursion}
to yield an $\epsilon$-approximation of the Lewis-weight, 
the leverage scores $\osigma^{(i)}$ must be at least an $\epsilon$-approximation as well.
However, for the old leverage score data structure to return an $\epsilon$-approximation,
the input must be $O(\epsilon/\log n)$-close to some other stable sequence.
We use as input $\mG = (\mW^{(i)})^{1/2-1/p}\mV$ (see \eqref{eq:overview:w_recursion}),
which does not satisfy the required property.
This is because, while the exact Lewis-weight $\tau$ would satisfy the required stability properties by guarantees of the IPM, the input vector $w^{(i)}$ is at best an $\epsilon$-approximation of the exact Lewis-weight $\tau$.
In summary, the complexity bounds of the previous leverage score data structures
do not apply when we use them for our Lewis-weight reduction 
due to the additional precision we require. %

So while the idea of our leverage score data structure is the same as in \cite{blss20} and \cite{BrandLN+20},
we must analyze and optimize the complexity in a different way.

\paragraph{High-level Idea for maintaining Leverage Scores}

Note that the output size (i.e.~dimension of $\osigma$) is $m$
and we want to maintain the leverage scores in $o(m)$ time,
so we can not afford to recompute all entries of $\osigma$ in each iteration.
Instead, the high-level idea is that in each iteration we 
(i) detect a set $I \subset [m]$ of indices $i$ where $\osigma_i \not\approx_\epsilon \sigma(\mV\mA)_i + z_i$,
and (ii) compute a new approximation of $\sigma(\mV\mA)_i + z_i$ for all $i \in I$ and update $\osigma_i$ accordingly.
Thus we split the outline of the data structure into these two parts:

\textit{Computing Few Leverage Scores:}
We start by outlining task (ii) as that one is easier. 
Computing few leverage scores is standard (see e.g.~Spielman-Srivastava \cite{spielman2011graph}) and we explain it briefly as we build on it.
For a matrix $\mX$ let $\mP(\mX) \defeq \mX(\mX^\top\mX)^{-1}\mX^\top$ denote an orthogonal projection matrix.
By $\mP(\mV\mA) = \mP(\mV\mA)\mP(\mV\mA)$ we have that
$\mP(\mV\mA)_{i,i} = \|\unitvec_i^\top \mP(\mV\mA)\|_2^2$,
so we can reformulate maintaining approximate leverage scores 
as maintaining an approximation of these norms for $i=1, \dots ,m$.
Given some set $I \subset [m]$ we can compute this norm for $i \in I$
by using a JL-matrix\footnote{%
A JL-matrix $\mJ$ satisfies $\|\mJ v\|_2 \approx \|v\|_2$ for any fixed vector $v$. 
For example a random Gaussian matrix with $O(\epsilon^{-2} \log n)$ rows 
yields w.h.p~a $(1\pm\epsilon)$-approximation of the norm.} 
$\mJ$ and computing the matrix $\mM := (\mA^\top \mV^2 \mA)^{-1} \mA \mV \mJ^\top$. 
Next, we obtain an approximation of the $i$-th leverage score by computing
\begin{align}
\|\unitvec_i^\top \mV \mA \mM\|_2^2 \approx \|\unitvec_i^\top \mV \mA (\mA^\top \mV^2 \mA)^{-1} \mA^\top \mV \|_2^2 = \|\unitvec_i^\top \mP(\mV\mA)\|_2^2.
\label{eq:over:compute_ls}
\end{align}
Here the complexity will be dominated by computing $\mM$
and computing \eqref{eq:over:compute_ls} for all $i \in I$.
Given that $\mJ$ needs only some $\tilde{O}(1)$ rows to yield a good approximation of the norm,
we only need to solve very few linear systems in $\mA^\top \mV^2 \mA$ to compute $\mM$
and $\mM$ has very few columns, so the norm \eqref{eq:over:compute_ls} can be computed quickly.

\textit{Detecting Leverage Score Changes:}
We now outline how to solve task (i), i.e.~how to detect when $\osigma_i \not\approx \sigma_i(\mV\mA) + z_i$.
For that assume that $\mV'$ changes to $\mV$ and we already had $\osigma_i \approx \sigma(\mV'\mA)_i + z_i$
from the previous iteration.
Then we must detect indices $i$ where $\sigma(\mV'\mA)_i + z_i \not\approx \sigma_i(\mV\mA) + z_i$,
because for those $i$ the previous $\osigma_i$ can no longer be a good approximation.
As the vector $z$ is fixed, it suffices to find indices $i$ with
$|\sigma(\mV'\mA)_i - \sigma(\mV\mA)_i| > \epsilon z_i$ for some small enough $\epsilon > 0$.
Using the interpretation of the leverage scores being the norm of the rows of $\mP$,
we can find such indices $i$ by searching for indices where
$$
\|\unitvec_i^\top    \left(   \mP(\mV'\mA) - \mP(\mV\mA)   \right)   \|_2
\ge
\|\unitvec_i^\top \mP(\mV'\mA) \|_2 - \|\unitvec_i^\top \mP(\mV\mA)\|_2
> \epsilon \sqrt{z_i}
$$
If we simply return all $i$ where $v_i \neq v'_i$,
then the only remaining $i$ we must detect are those with
$$
\|\unitvec_i^\top    \mV\mZ^{-1/2} \mA \underbrace{\left((\mA^\top \mV^2 \mA)^{-1} \mA^\top \mV - (\mA^\top \mV'^2 \mA)^{-1} \mA^\top \mV'\right) \mJ^\top}_{=:\mM'}   \|_2 >  \epsilon
$$
where $\mJ$ is again a JL-matrix and $\mZ = \mdiag(z)$.
Note that because $\mJ^\top$ has few columns, it suffices to look for large entries of the matrix vector products 
$\mV\mZ^{-1/2} \mA \mM' \unitvec_k$ for $k=1,...,\tilde{O}(1)$, which in turn can be solved by the HeavyHitter data structure.

So far we only discussed how to detect indices $i$ where the leverage score changed a lot within a single iteration.
However, it could also happen that a leverage score changes only a little in each iteration
such that after many iteration we have $\osigma_i \not\approx \sigma(\mV \mA)_i + z_i$.
To detect these slowly changing indices, we follow the approach appeared in \cite{blss20,BrandLN+20}. That is,
we perform the same trick used for a single iteration, but instead for each $j=0,...,\log \sqrt{n}$, we check whether the leverage score has changed significantly in the past $2^j$ iterations, i.e., the matrix $\mV'$ now refers to the state of $\mV$ some $2^j$ iterations in the past.

\paragraph{Improving the complexity}

For the data structure we outlined so far, 
the main bottleneck is solving linear systems in $\mA^\top \mV^2 \mA$
and computing the product $\mA^\top \mV \mJ^\top$.
Both of these require $\nnz(\mA) = \Omega(m)$ time,
which is too slow for our purposes.\footnote{This is true, even if we assume access to some preconditioner of $\mA^\top \mV \mJ^\top$.}
To speed this up, we use leverage score sampling \cite{spielman2011graph}
to construct a random sparse diagonal matrix $\mR$ 
with $\tilde{O}(n)$ nonzero entries
and $\mA^\top \mV^2 \mR \mA \approx \mA^\top \mV \mA$.
Careful analysis shows that the algorithm outlined above still works,
when solving systems in $\mA^\top \mV^2 \mR \mA$
and when using $\mA^\top \mV \mR^{1/2} \mJ^\top$ instead of $\mA^\top \mV \mJ^\top$.
Because of the sparsity of $\mR$, the $\nnz(\mA)$ cost decreases to $\tilde{O}(n \cdot \max_i \nnz(a_i))$.

However, this speed-up yields a new problem. If we use two different random matrix $\mR$ and $\mR'$
with $\mA^\top \mV^2 \mR \mA \approx \mA^\top \mV \mA$ and $\mA^\top \mV'^2 \mR' \mA \approx \mA^\top \mV' \mA$,
then the runtime of the \textsc{HeavyHitter} data structure can become very large.
This is because the \textsc{HeavyHitter} data structures
must find large entries of $\mV\mZ^{-1/2} \mA \mM' \unitvec_k$
and by \Cref{def:heavyhitter} the complexity of that task scales in
$\|\mV\mZ^{-1/2} \mA \mM' \unitvec_k\|_2^2$,
so the total cost for all $k$ scales in $\|\mV\mZ^{-1/2} \mA \mM' \|_F^2$.
Without random sampling, this Frobenius-norm can be bounded by 
stability properties of the IPM.
However, the Frobenius norm is very sensitive to spectral changes
which causes the norm to blow-up when using two different $\mR$ and $\mR'$ (i.e.~two different spectral approximations) for $\mA^\top \mV \mA$ and $\mA^\top \mV' \mA$.
Thus we wish to use a single random $\mR$ that yields a valid approximation for both.

To see how to construct such $\mR$, consider classic leverage score sampling first.
If one sets $\mR_{i,i} = 1/p_i$ independently for each $i \in [m]$ with probability $p_i$ (where $p_i \ge \min(1, \sigma(\mV \mA)_i\log n/\eps^2)$) and $\mR_{i,i} = 0$ otherwise,
then $\mA^\top \mV^2 \mR \mA \approx \mA^\top \mV^2 \mA$.
To make sure that $\mR$ also has the property $\mA^\top \mV'^2 \mR \mA$,
we use $p_{i} = 1$ for all $i$ where $\sigma(\mV \mA)_i \not\approx_{1} \sigma(\mV' \mA)_i$
and $p_{i} = 2\osigma_i$ otherwise.
The indices $i$ for which we have to set $p_i = 1$ are simply those where we recently had to change $\osigma_i$.
Further, we need $p_i = 1$ for all $i$ with $v_i \neq v'_i$.
This is because we can then bound 
$\mA^\top \mV^2 \mR \mA - \mA^\top \mV'^2 \mR \mA = \mA^\top \mR (\mV^2 - \mV'^2) \mA$ more easily.
If we did not choose $p_i=1$, and $\mR_{i,i}$ happens to be non-zero because of the sampling,
then $\mR_{i,i} (\mV^2 - \mV'^2)$ would blow-up the difference $(v_i - v'_i)$ by an $1/p_i$ factor.
So by choosing $p_i = 1$ for all $i$ with $v_i \neq v'_i$,
we are able to prove better complexity bounds.

For comparison, in \cite{blss20} the leverage score data structure did not use a single $\mR$ and instead they used two different $\mR$ and $\mR'$. They fixed the issue of the Frobenius-norm being large by carefully updating $\mR$ to $\mR'$ and performing an amortized analysis on the sum of Frobenius-norms over several iterations. However, this analysis required the random $\mR$ to be updated over several iterations, which meant the same randomness had to be re-used in all those iterations, thus resulting in difficulties with handling adaptive adversaries.
We instead use only one random $\mR$ per iteration and this random $\mR$ is resampled in every iteration, thus no randomness is re-used and adaptive adversaries are not an issue.

\subsubsection{Further Data Structures}

We now outline all the other data structures used by our algorithms.
Some of these data structures were developed in \cite{BrandLN+20,blss20}
though we perform small modifications to them in \Cref{sec:gradient_maintenance} and \Cref{sec:dual_maintenance}.
Here we give a brief description of these data structures 
and how they are used to efficiently implement our IPM.
A more detailed overview for how to implement our IPM
can be found in \Cref{sec:implement:outline}
where we provide  the exact statements of the involved data structures.

As mentioned in \Cref{sec:overview:ipm},
our IPM only requires access to approximations $\ox,\os$ of the iterates $x,s$.
The updates to these vectors are roughly of the following form
\begin{align}
s^\new \leftarrow s + \mA \mH^{-1} \mA^\top \Phi''(\ox)^{-1/2}g
\label{eq:overview:update}
\end{align}
for some gradient-vector $g$,
matrix $\mH \approx \mA^\top \omT^{-1} \Phi''(\ox)^{-1} \mA$
for diagonal matrix $\omT = \mdiag(\otau)$ with the approximate Lewis weight $\otau$ on the diagonal. 
Note that naive computation of \eqref{eq:overview:update} would require $O(\nnz(\mA)+n^\omega)$ time per iteration.
This can be sped up via data structures that efficiently maintain partial solutions of this expression.
The task of computing \eqref{eq:overview:update} can be split into three subtasks:
(i) Compute $\mA^\top \Phi''(\ox)^{-1/2}g$, solved by a data structure from \cite{BrandLN+20}
which we modify in \Cref{sec:gradient_maintenance}.
(ii) Multiply the result by $\mH^{-1}$, solved either via Laplacian solver (e.g.~when considering min-cost flow)
or a data structure from \cite{blss20} (restated in \Cref{sec:matrix_data_structures}). We can essentially use the previous data structure, except we sample $\mH$ (the spectral sparsifier) by
the $\ell_p$ Lewis weights that we are already maintaining instead of by leverage score.
(iii) Let $v^{(i)}$ be the vector $\mH^{-1} \mA^\top \Phi''(\ox)^{-1/2}g$ we computed with the previous data structures during iteration number $i$ of the IPM.
Then after the $t$-th iteration of the IPM the vector $s$ is given by $s^\init + \sum_{i=1}^t \mA v^{(i)}$ according to \eqref{eq:overview:update}.
Maintaining an approximation of such matrix vector products is done via a data structure from \cite{BrandLN+20}
which we modify in \Cref{sec:dual_maintenance}.
At last, \Cref{sec:overview:ipm} mentioned that the update to the primal solution must be sampled.
For graph applications such as max-flow we use a data structure from \cite{BrandLN+20} (restated in \Cref{sec:graph_data_structures}) to perform this sampling efficiently.
For general linear programs we construct a new data structure in \Cref{sec:matrix_data_structures}.

\subsection{Putting Everything Together}
\label{sec:overview:together}

Given our new IPM for two-sided constraints (\Cref{sec:overview:ipm}) and data structures for implementing this IPM (\Cref{sec:overview:ds}), we apply
them to obtain our results in a standard way. 

First of all, our IPM needs to start with
a \emph{centered} initial point (i.e.~an initial point with small
centrality potential (\ref{eq:over_potential})). Given an LP instance,
we modify the instance by adding an identity block to the constraints
and corresponding variables so that a centered initial point can be
obtained analytically. This allows us to apply the
IPM which moves the initial point to a near optimal point of the modified
instance in $\Otil((mn+n^{2.5})\log W)$ time (using data structures
from previous sections). The modified instance also guarantees that,
at near-optimal points, the added variables must have value very close
to zero. This allows us to round the near optimal point of the modified
instance to a near optimal point of the original instance by a single
linear system solve which takes $\Otil((\nnz(\mA)+n^{\omega})\log W)$.
This gives an algorithm for solving an LP with two-sided constraints
in \Cref{thm:primallp}. 

As our IPM maintains not only a primal solution but also a dual slack,
by solving a linear system involving the dual slack at the near-optimal
point, this gives a dual solution and solves the $\ell_{1}$-regression
problem as stated in \Cref{thm:duallp}. Given an $\ell_{1}$-regression
algorithm, we immediately obtain an algorithm for discounted MDPs using
a known reduction by \cite{SidfordWWY18} and obtain \Cref{cor:DMDP}.

For a given min-cost flow instance, we modify the instance with the same purpose as above by adding a star. Using graph-based data structures, the path following IPM moves the initial point to a near optimal point of the modified instance in faster $\Otil(m\log W+n^{1.5}\log^{2}W)$
time. Analogously, by solving a Laplacian system in $\Otil(m\log W)$ time, we get a near-optimal flow of the original instance. Moreover, since the LP for min-cost flow is integral and the optimal solution can be assumed to be unique (using the isolation lemma as in \cite{ds08,BrandLN+20}), we can round the flow on each edge to its nearest integer and obtain an exactly optimal flow. This takes time $\Otil(m\log W+n^{1.5}\log^{2}W)$ as promised in \Cref{thm:mincost flow}. For the easier maximum flow problem, we can shave a $\log W$ factor on $n^{1.5}$ using a standard scaling technique by \cite{ahuja1991distance} and obtain \Cref{thm:maxflow}.  %

\ifdefined\DEBUG
\newpage
\fi

\section{IPM}
\label{sec:ipm}
Throughout this section we let $\ma \in \R^{m \times n}$ denote a non-degenerate matrix, let $b \in \R^n$, $c \in \R^m$ and consider the following problem
\begin{align*}
	\min_{ \substack{ x \in \R^m : \ma^\top x = b \\ \ell_i \le x_i \le u_i \forall i \in [m] } } c^\top x. 
\end{align*}
In this paper, we will need for the barrier functions $\phi_i : (\ell_i, u_i) \rightarrow \R$ be \emph{highly $1$-self-concordant} barrier functions on $S_i$, as opposed to only $1$-self-concordant. This is due to needing the fourth derivative in Lemma \ref{lemma:pchange} later.
\begin{definition}[Highly $1$-self-concordance]
\label{def:highselfcon}
For an interval $(\ell, u)$ we say that a function $f: (\ell, u) \rightarrow \R$ is a highly $1$-self-concordant barrier on $(\ell, u)$ if for all $x \in (\ell, u)$ we have $|f'(x)| \le f''(x)^{1/2}$, $|f'''(x)| \le 2f''(x)^{3/2}$, $|f''''(x)| \le 6f''(x)^2$, and $\lim_{x \rightarrow a} f(x) = \lim_{x \rightarrow b} f(x) = +\infty.$
\end{definition}
\begin{restatable}{lemma}{highselfcon}
\label{lemma:highselfcon}
For all $\ell \leq u$ the function $\phi(x) = -\log(x-\ell)-\log(u-x)$ is highly $1$-self-concordant on the interval $(\ell , u)$. 
\end{restatable}
We show this in Section \ref{subsec:proofsbasic}. For the remainder of the section, we fix $\phi_i(x_i) = -\log(x_i-\ell_i)-\log(u_i-x_i)$. This function satisfies a simple bound which is useful for getting the initial and final points.
\begin{fact}
	\label{fact:bound phi} Consider the barrier function $\phi(x)=-\log(x-\ell)-\log(u-x)$
	on the interval $[\ell,u]$. We have $\phi'((\ell+u)/2)=0$ and $\phi''(x)=1/(u-x)^{2}+1/(x-\ell)^{2}\ge1/(u-\ell)^{2}$
	for all $x\in[\ell,u]$.
\end{fact}

During the IPM, we maintain tuples $(x, s, \mu)$. Given a current point $(x,s,\mu)$, we define a define a \emph{weight function} $\tau: \R^m \rightarrow \R^m_{>0}$ that governs the central path. Intuitively, $\tau(x)_i$ is the weight on the $i$-th barrier function $\phi_i$. The choice of weight function $\tau$ we use for this paper and the central path will be a regularized Lewis weight. It will be convenient to choose the regularizing vector $v$ to have weight at least $n/m$ on each coordinate, while still having low $\ell_1$ norm.
\begin{definition}[Regularized Lewis weights for a matrix]
\label{def:matrixlewis}
\label{def:v}
For $p = 1 - \frac{1}{4\log(4m/n)}$, vector $v \in \R^m_{>0}$ with $v_i \ge n/m$ for all $i$ and $\|v\|_1 \le 4n$, and matrix $\mA$ define the \emph{($v$-regularized) $\ell_p$-Lewis weights} $w(\mA) \in \R^m_{>0}$ as the solution to
\begin{align*}
 w(\mA) = \sigma( \mW^{\frac{1}{2} - \frac{1}{p}} \mA ) + v
 \text{ where } 
 \mW \defeq \diag(w(\mA)) ~.
\end{align*}
\end{definition}
When the matrix $\mA$ is clear from context, we suppress the notation of $\mA$ in $w(\cdot)$.
\begin{definition}[Regularized Lewis weights for $c$]
\label{def:lewis}
For $p = 1 - \frac{1}{4\log(4m/n)}$, and $c, v \in \R^m_{>0}$ define the \emph{($v$-regularized) $\ell_p$-Lewis weights} $w(c): \R^m_{>0} \rightarrow \R^m_{>0}$ as $w(c) \defeq w(\mC\mA)$ as in Definition \ref{def:matrixlewis}.
\end{definition}

We collect properties of regularized Lewis weights in  \Cref{subsec:lewis}, e.g. that $\|w(c)\|_1 = n + \| v \|_1$. We implicitly suppress the dependence on $v, p$ as they are fixed throughout the algorithm.
\begin{definition}[Central path weights]
\label{def:cpweights}
Define the central path weights $\tau(x) \defeq w(\phi''(x)^{-\frac{1}{2} })$ for a fixed vector $v$.
\end{definition}
Our algoirthm maintains points $(x, s, \mu)$ satisfying the following centrality guarantee.
\begin{definition}[$\eps$-centered point]
\label{def:centered}
We say that $(x,s,\mu) \in \R^m \times \R^m \times \R^m_{>0}$ is $\eps$-centered for $\eps \in (0,1/80]$ if the following properties hold, where $\cnorm = C / (1-p)$ for a constant $C \ge 100$.
\begin{enumerate}
\item (Approximate centrality) $\left\| \frac{s+\mu \tau(x)\phi'(x)}{\mu\tau(x)\sqrt{\phi''(x)}}\right\|_\infty \le \eps.$
\item (Dual Feasibility) There exists a vector $z \in \R^n$ with $\ma z+s=c$.
\item (Approximate Feasibility) $\| \ma^\top x - b \|_{(\mA^\top(\Tau(x)\Phi''(x))^{-1}\mA)^{-1}} \le \eps\gamma/\cnorm$.
\end{enumerate}
\end{definition}
To maintain approximate centrality in Definition \ref{def:centered}, we will track a centrality potential.
\begin{definition}[Centrality potential]
\label{def:potential}
We track the following centrality potential.
\begin{align*} \Psi(x,s,\mu) \defeq \sum_{i=1}^m \psi\left(\frac{s_i+\mu \tau(x)_i\phi_i'(x_i)}{\mu\tau(x)_i\sqrt{\phi_i''(x_i)}} \right) \end{align*}
for $\psi(x) \defeq \cosh(\lambda x)$, where $\lambda = \Theta(\log(m)/\eps)$.
\end{definition}
At a high level, this potential is derived from noting that $s + \mu\tau(x)\phi'(x) = 0$ for exactly central points $x, s$. The denominators are the Hessians of $x$ with respect to the barriers, and thus capture changes of $x, s$ within a small stable region.

Finally, we need to control both the $\tau$ and $\infty$ norms of the steps in the algorithm, which leads to the following definition.
\begin{definition}[$\tpi$ norm]
\label{def:norm}
Let $\|g\|_{\tau+\infty} \defeq \|g\|_\infty + \cnorm \|g\|_\tau$ for $\cnorm = C/(1-p)$ and \begin{align*} g^{\flat(\tau)} \defeq \argmax_{\|h\|_\tpi = 1} h^\top g. \end{align*} Let the dual norm be $\|g\|_\tpi^* \defeq g^\top g^{\flat(\tau)}$.
\end{definition}

As described, our algorithm will not maintain exact points $(x,s,\mu) \in \R^{m} \times \R^{m} \times \R^{m}_{>0}$ and weights $\tau \in \R^m$, instead they will be \emph{approximate} in the following sense. Precisely, the algorithm maintains the following condition throughout.
\begin{invariant}
\label{invar}
We maintain the following approximations $\bar{x}, \bar{\tau}$ of $x, \tau \in \R^m$ at the start and end of each call to $\ShortStep$~(Algorithm \ref{algo:lsstep}).
\begin{itemize}
\item $\|\Phi''(x)^\frac{1}{2} (\bar{x}-x)\|_\infty \le \eps$.
\item $\|\Tau(\bar{x})^{-1}(\bar{\tau}-\tau(\bar{x}))\|_\infty \le \eps$.
\end{itemize}
\end{invariant}

\paragraph{Constants and approximation notation.} We will use $C$ to denote a large constant, chosen later. It is used in the definition of $\cnorm$ and for the parameters $\eps, \gamma, \lambda$ in Algorithm \ref{algo:lsstep}. For quantities $f, g$ we write $f \ls g$ or $f = O(g)$ if there is a universal constant $Z$ (independent of the constant $C$) such that $|f| \le Z|g|.$ We assume that $C$ is chosen large enough in Algorithm \ref{algo:lsstep} so that for any quantity $f$ written in the analysis satisfying $f \ls \eps$ in fact satisfies $f \le 1/1000.$ Also, if $f \ls \gamma$ then, because $\gamma = \eps/(C\lambda)$, we will assume similarly that in fact $f \le \frac{1}{1000\lambda}$.

We are ready to state the IPM. Algorithm \ref{algo:lsstep} takes a single step, and Algorithm \ref{algo:pathfollowing} takes a sequence of short steps to solve a linear program. Taking a sequence of short steps using Algorithm \ref{algo:lsstep} allows us to solve LPs, assuming we have an initial point. The initial point construction is done formally in Section \ref{subsec:initialpoint}, and final point is computed in Lemma \ref{lemma:finalpoint}, proven in the appendix in Section \ref{subsec:initialpoint}.
\begin{restatable}[Final point]{lemma}{finalpoint}
	\label{lemma:finalpoint} Given an $\eps$-centered
	point $(x,s,\mu)$ where $\eps\le1/80$, we can compute a point $(x^{\final},s^{\final})$
	satisfying 
\begin{enumerate}
	\item $\mA^{\top}x^{\final}=b$, $s^{\final}=\mA y+c$ for some $y$. 
	\item $c^{\top}x^{\final}-\min_{\substack{\mA^{\top}x=b\\
			\ell_i \le x_i \le u_i \forall i
		}
	}c^{\top}x\ls n\mu.$ 
\end{enumerate}

The algorithm takes $O(\nnz(\mA))$ time plus the time for solving
a linear system on $\mA^{\top}\mD\mA$ where $\mD$ is a diagonal
matrix.
\end{restatable}

\begin{algorithm2e}[h]
\caption{Short Step (Lee Sidford Barrier) \label{algo:lsstep}}
\SetKwProg{Proc}{procedure}{}{}
\Proc{\ShortStep$(x,s,\mu,\mu^\new)$}{
	Fix $\tau(x) = w(\phi''(x)^{-\frac{1}{2} })$, where $v$ and $w$ are defined in Definition \ref{def:v} and \ref{def:lewis}. \\
	Let $\alpha = \frac{1}{4\log(4m/n)}, \eps = \frac{\alpha}{C}, \lambda = \frac{C\log(Cm/\eps^2)}{\eps}, \gamma = \frac{\eps}{C\lambda}, r = \frac{\eps\gamma}{\cnorm\sqrt{n}}$. \\
	Assume that $(x,s,\mu)$ is $\eps$-centered 
			and $\d_\mu \defeq \mu^\new-\mu$ satisfies $|\d_\mu| \le r\mu$. 
			\label{line:ipm:centered}\\
	Pick $(\ox, \otau)$ to satisfy Invariant \ref{invar} 
			with respect to $(x,\tau(x))$. 
			\label{line:ipm:oxotau}\\
	Let $y = \frac{s_i+\mu\tau(x)_i\phi_i'(x_i)}{\mu\tau(x)_i\sqrt{\phi''_i(x_i)}}$ 
			and let $\|\bar{y}-y\|_\infty \le \gamma/20$. 
			\label{line:ipm:y}\\
	Let $g = -\gamma\g\Psi(\bar{y})^{\flat(\bar{\tau})}$, 
			where $h^{\flat(\bar{\tau})}$ is defined in Definition \ref{def:norm}. 
			\label{line:ipm:g}\\
	Let $\mH \approx_\gamma \bar{\mA}^\top \bar{\mA} = \mA^\top\bar{\Tau}^{-1}\Phi''(\bar{x})^{-1}\mA$, where $\bar{\mA} = \bar{\Tau}^{-\frac{1}{2}}\Phi''(\bar{x})^{-\frac{1}{2}}\mA$.
			\label{line:ipm:H}\\
	Let $\d_1 = \bar{\Tau}^{-1}\Phi''(\bar{x})^{-\frac{1}{2} }\mA\mH^{-1}\mA^\top \Phi''(\bar{x})^{-\frac{1}{2} }g$ 
			and $\d_2 = \bar{\Tau}^{-1}\Phi''(\bar{x})^{-\frac{1}{2} }\mA\mH^{-1}(\mA^\top x - b)$. \\ 
			Let  $\d_r = \d_1 + \d_2$. 
			\label{line:ipm:delta_r}\\
	Let $\mR \in \R^{m\times m}$ be a $\cvalid$-valid random diagonal matrix 
			for large $\cvalid$ chosen later. 
			\Comment{Definition \ref{def:validdistro}} 
			\label{line:ipm:R}\\
	$\bar{\d}_x \leftarrow \Phi''(\bar{x})^{ -\frac{1}{2} }\left( g - \mR \d_r \right).$ 
			\label{line:ipm:delta_x}\\
	$\bar{\d}_s \leftarrow \mu \bar{\Tau}\Phi''(\bar{x})^{ \frac{1}{2} } \d_1$. 
			\label{line:ipm:delta_s}\\
	$x^\new \leftarrow x + \bar{\d}_x$ 
			and $s^\new \leftarrow s + \bar{\d}_s.$ 
			\label{line:ipm:xnewsnew}\\
	\Return~$(x^\new,s^\new)$.
}
\end{algorithm2e}

\begin{algorithm2e}[h]
\caption{Path Following Meta-Algorithm \label{algo:pathfollowing} for solving $\min_{\mA^\top x = b, \ell_i \le x_i \le u_i \forall i} c^\top x,$ given an initial point $\eps/\cstart$-centered point $(x^\init, s^\init, \mu)$ for large $\cstart$.}
\SetKwProg{Proc}{procedure}{}{}
\Proc{\PathFollowing$(\mA, \ell, u, \mu, \mu^\final)$}{
	Define $r$ as in Algorithm \ref{algo:lsstep}. \\
	\While{$\mu > \mu^\final$}{
		$(x^\new, s^\new) \assign \ShortStep(x, s, \mu, (1-r)\mu)$. \\
		$x \assign x^\new, s \assign s^\new, \mu \assign (1-r)\mu$. \\
	}
	Use Lemma \ref{lemma:finalpoint} to return a point $(x^\final, s^\final)$.
}
\end{algorithm2e}
The main goal of Sections \ref{subsec:overview} to \ref{subsec:together} is to show the following, proven formally at the end of Section \ref{subsec:together}.

\begin{lemma}
\label{lemma:pathfollowing}
Algorithm $\PathFollowing(\mA, b, \ell, u, c, \mu, \mu^\final)$ makes $\tilde{O}(\sqrt{n}\log(\mu/\mu^\final))$ calls to $\ShortStep(\cdot)$, and with probability at least $1-m^{-5}$ satisfies the following conditions at the start and end of each call to \ShortStep~(Algorithm \ref{algo:lsstep}).
\begin{enumerate}
\item(Slack feasibility)  $s = \mA z + c$ for some vector $z \in \R^n$ 
\item(Approximate feasibility) $\| \mA^\top x - b \|_{(\mA^\top(\Tau(x)\Phi''(x))^{-1}\mA)^{-1}} \le \eps\gamma/\cnorm$.
\item(Potential function) $\E[\Psi(x, s, \mu)] \le m^2$, where the expectation is over the randomness of $x, s$.
\item($\eps$-centered) $(x, s, \mu)$ is $\eps$-centered.
\end{enumerate}
For $\mu^\final \le \d/(Cn)$, we have that $\mA x^\final = b$ and $c^\top x^\final \le \min_{\substack{\mA^\top x = b \\ \ell_i \le x_i \le u_i \forall i}} c^\top x + \d.$
\end{lemma}

We sample a random diagonal scaling $\mR$ in our algorithm, and will require some properties of this random matrix to guarantee progress of the IPM. We summarize the necessary properties here. This definition captures distributions such as sampling each coordinate independently as a Bernoulli with probabilities $p_i$, or taking the sum of multiple samples proportional to $p_i$.
\begin{definition}[Valid sampling distribution]
\label{def:validdistro}
Given vector $\d_r, \mA, \bar{x}, \bar{\tau}$ as in \ShortStep~(Algorithm \ref{algo:lsstep}), we say that a random diagonal matrix $\mR \in \R^{m \times m}$ is $\cvalid$-\emph{valid} if it satisfies the following properties, for $\bar{\mA} = \Bar{\Tau}^{-\frac{1}{2} }\Phi''(\bar{x})^{-\frac{1}{2} }\mA$. We assume that $\cvalid \ge \cnorm$.
\begin{itemize}
\item(Expectation) We have that $\E[\mR] = \mI$.
\item(Variance) For all $i \in [m]$, we have that $\Var[\mR_{ii}(\d_r)_i] \le \frac{\gamma|(\d_r)_i|}{\cvalid^2}$.
and $\E[\mR_{ii}^2] \le 2\sigma(\bar{\mA})_i^{-1}$.
\item(Covariance) For all $i \neq j$, we have that $\E[\mR_{ii}\mR_{jj}] \le 2$.
\item(Maximum) With probability at least $1-n^{-10}$ we have that $\|\mR\d_r-\d_r\|_\infty \le \frac{\gamma}{\cvalid^2}$.
\item(Matrix approximation) We have that $\bar{\mA}^\top\mR\bar{\mA} \approx_\gamma \bar{\mA}^\top\bar{\mA}$ with probability at least $1-n^{-10}$.
\end{itemize}
\end{definition}

\subsection{Overview of Analysis}
\label{subsec:overview}
Our proof will show that the expected value of the potential function in Definition \ref{def:potential} is bounded by $\poly(m)$ with high probability throughout the algorithm. This will imply that it is $\eps$-centered at the start and end of each call to \ShortStep~(Algorithm \ref{algo:lsstep}). The main pieces of the analysis are as follows.

\paragraph{Potential function analysis.} In Section \ref{subsec:tools} we set up the analysis of the change in the potential function. Specifically, in Lemma \ref{lemma:potentialhelper} we show that bounding the change in the potential function reduces to bounding first and second order changes of the numerator $y = s+\mu\tau(x)\phi'(x)$ and denominator terms $\mu, \tau(x), \phi''(x)$ in Definition \ref{def:potential}. These are done in Section \ref{subsec:together} in Lemmas \ref{lemma:muchange} (for $\mu$), \ref{lemma:pchange} (for $\mC \defeq \Phi''(x)^{-1/2}$), \ref{lemma:tauchange} (for $\tau$), \ref{lemma:ychange} (for $y$).

The analysis of change in $\tau$ requires several facts of derivatives of regularized Lewis weights with respect to diagonal scalings, and these are done in Section \ref{subsec:lewis}, culminating in Lemmas \ref{lemma:tauchange1}, \ref{lemma:tauchange2} which bound the change in $\tau$ under small changes in the diagonal scaling, which we refer to as $\gamma$-bounded changes (Definition \ref{def:gammabounded}).

\paragraph{Feasibility.} To guarantee efficiency, our algorithm sparsifies matrices $\mA^\top \mD\mA$ to solve systems, which results in potentially infeasible points $x$ during the algorithm, i.e. $\mA^\top x = b$ may fail. However, our algorithm maintains approximate feasibility as discussed in Definition \ref{def:centered}. The analysis is done in Lemma \ref{lemma:feasibility}.

\paragraph{Additional properties.} In Section \ref{subsec:initialpoint} we first show that computing an $\eps$-centered point for small path parameter $\mu$ guarantees small objective error. We then show how to construct an initial $\eps$-centered point for a perturbed linear program, and show that the modified linear program still gives approximate solutions to the original. In Section \ref{subsec:sampling} we show that two sampling schemes are both valid distributions. The one in Lemma \ref{lemma:independent} is used for the graphical setting, and Lemma \ref{lemma:prop} is used for the general linear program setting. Finally, the data structures for maintaining approximate solutions for $x, s, \tau$ require a stability bound of the true $x, s$ which may not hold. However, we can show that there are nearby points that satisfy stronger stability bounds in Lemma \ref{lemma:morestablex}.

\paragraph{Comparison to \cite{BrandLN+20}.} The main difference from the analysis of \cite{BrandLN+20} is our use of general self-concordant functions to handle two-sided barrier constraints, while \cite{BrandLN+20} used logarithmic barriers. This leads to the following differences in the IPM: the gradient optimality for our barrier takes the form in \Cref{def:centered} (Approximate centrality), while in \cite{BrandLN+20} they simply use the form $xs \approx w\mu$. Additionally, we require our weights $w$ to be Lewis weights, while \cite{BrandLN+20} was able to use leverage scores due to the structure of the centrality condition $xs \approx w\mu$. Finally, our analysis deals more generally with valid distributions in Definition \ref{def:validdistro} which allows us to handle both sampling coordinate independently for the graphical setting, and proportional to sampling probability for general linear programs.

\subsection{Analysis Tools and Setup}
\label{subsec:tools}
In this section, we set up the analysis of our IPM, and all omitted proofs are given in Section \ref{subsec:proofsbasic}. First, we  collect some basic properties of $\psi$ to help in the analysis.

\begin{lemma}[Basic properties of $\psi$]
\label{lemma:psibasic}
We have for $\lambda \ge 1$ that
\begin{itemize}
\item $\psi''(x) = \lambda^2\psi(x).$
\item $\psi(x') \le 2\psi(x)$ for $|x'-x| \le \frac{1}{20\lambda}$.
\item $|\psi'(x)| \le \lambda^{-1}\psi''(x).$
\end{itemize}
\end{lemma}
We now state a helper lemma that shows that we can analyze the change in the centrality potential by analyzing second order changes of each contributing piece. This differs from the corresponding \cite[Lemma 4.34]{BrandLN+20} in that our errors are not strictly multiplicative errors.
\begin{restatable}[Potential change bound]{lemma}{potentialhelper}
\label{lemma:potentialhelper}
Define for $u_i^{(j)} \ge 0$ and $y_i$
\begin{align*}
 w_i = \prod_{j\in[k]} (u_i^{(j)})^{c_j} \enspace \text{ and } \enspace w_i^\new = \prod_{j\in[k]} (u_i^{(j)} + \d_i^{(j)})^{c_j} 
\end{align*}
and
\begin{align*}
 v_i = y_i w_i \enspace \text{ and } \enspace v_i^\new = (y_i + \eta_i) w_i^\new 
\end{align*}
where $\|(\mU^{(j)})^{-1}\d^{(j)}\|_\infty \le \frac{1}{50(1+\|c\|_1)}$ for all $j \in [k]$, $\|v\|_\infty \le 1/50$, $\|\mw\eta\|_\infty \le \frac{ 1 }{ 50 \lambda(1 + \| c \|_1) }$, and
\begin{align*} \|v\|_\infty \sum_{j \in [k]} |c_j|\|(\mU^{(j)})^{-1}\d^{(j)}\|_\infty \le \frac{1}{100\lambda}. \end{align*} Then we have that $\|v^\new-v\|_\infty \le \frac{1}{20\lambda}$
\begin{align} 
&\Psi(v^\new) \le \Psi(v) + \psi'(v)^\top \left(\mw\eta + \sum_{j\in[k]}c_j\mv(\mU^{(j)})^{-1}\d^{(j)}\right) \label{eq:firstline}
\\ &+ 8\|\mw\eta\|_{\psi''(v)}^2 + 8(1+\|c\|_1)\|v\|_\infty^2\sum_{j\in[k]}|c_j|\|(\mU^{(j)})^{-1}\d^{(j)}\|_{\psi''(v)}^2 \label{eq:secondline}
\\ &+ 8\|\mw\eta\|_{|\psi'(v)|} \sum_{j \in [k]} |c_j|\|(\mU^{(j)})^{-1}\d^{(j)}\|_{|\psi'(v)|} + 8(1+\|c\|_1)\|v\|_\infty\sum_{j\in[k]}|c_j|\|(\mU^{(j)})^{-1}\d^{(j)}\|_{|\psi'(v)|}^2. \label{eq:thirdline}
\end{align}
\end{restatable}

It is known that the Hessian $\Phi''(x)$ is maintained under small perturbations to $x$.
\begin{lemma}
\label{lemma:selfcon}
$\|\Phi''(x)^\frac{1}{2} (\bar{x}-x)\|_\infty \le \eps$ for $\eps \in [0,1/100]$ then $\Phi''(x)^\frac{1}{2}  \approx_{\eps+2\eps^2} \Phi''(\bar{x})^\frac{1}{2}.$
\end{lemma}
\begin{proof}
Follows directly from $\left| \frac{ \mathrm{d} }{ \mathrm{d} x_i }\phi_i''(x_i)^{ - \frac{1}{2} }\right| \le 1.$
\end{proof}

In order the analyze the potential $\Psi$ under changes, we define the quantities $\tau_i \defeq \tau(x)_i$, $\mu_i = \mu$, $c_i = \phi''_i(x_i)^{-\frac{1}{2} }$, and $y_i \defeq s_i + \mu\tau(x)_i\phi_i'(x_i)$. This allows us to write the potential as
\begin{align*} 
\Psi(x,s,\mu) = \sum_{i=1}^m \psi(y_i\mu_i^{-1}\tau_i^{-1}c_i). 
\end{align*}
Let $\tau^\new,\mu^\new,c^\new,y^\new$ be the corresponding vectors after a step. We will use Lemma \ref{lemma:potentialhelper} to analyze the change in $\Psi$. Define
\begin{align*}
\d_\mu \defeq \mu^\new - \mu, \d_\tau \defeq \tau^\new - \tau, \d_c \defeq c^\new - c, \d_y \defeq y^\new - y.
\end{align*}
The goal of the remainder of the section will be to analyze the changes $\d_\mu, \d_\tau, \d_c, \d_y$ and apply Lemma \ref{lemma:potentialhelper} to analyze the change in the centrality potential. This will complete the correctness proof of the IPM.

\subsection{Regularized Lewis Weights}
\label{subsec:lewis}

\begin{table}[h]
    \centering
    \begin{tabular}{|l|l|l|} \hline
        {\bf Statement} & {\bf Term} & {\bf Comment} \\ \hline
        Definition~\ref{def:fundamental} & $\mW_c, \mP_c, \mSigma_c, \mLambda_c, \mJ_c$ & Fundamental matrices \\ \hline
        Lemma~\ref{lemma:optproblem} & $w(c), f(w,c)$ & Alternative definition of regularized Lewis weights \\ \hline
        Lemma~\ref{lemma:derivlewis} & $\mJ_c$ & Jacobian of Regularized Lewis weight \\ \hline
        Lemma~\ref{lemma:decomp} & $\mD_c, \mK_c$ & Decomposition of $\mJ_c$ \\ \hline
        Lemma~\ref{lemma:decomp2} & $\mD_c', \mK_c'$ & Alternate decomposition \\ \hline
        Lemma~\ref{lemma:lewisapprox} & $\mW_c, \mT$ & Lewis weight approximation \\ \hline
        Definition~\ref{def:gammabounded} & $\mC$ & $\gamma$-boundedness \\ \hline
        Lemma~\ref{lemma:firstbound} & $\Tau^{-1} \delta_{\tau_t}$ & $\| \cdot \|_{\infty}$, $\| \cdot \|_{\tau+\infty}$, $\| \cdot \|_{\mP_t^{(2)}}$; infinitesimal bound \\ \hline
        Lemma~\ref{lemma:tauchange1} & $\Tau^{-1} \delta_{\tau}$ & $\| \cdot \|_{\infty}$, $\| \cdot \|_{\tau+\infty}$ norm \\ \hline
        Lemma~\ref{lemma:tauchange2} & $\Tau^{-1} ( \E[ \delta_{\tau} ] - \mJ \E[\delta_c] ) $ & $\| \cdot \|_{\tau+\infty}$ \\ \hline
    \end{tabular}
    \caption{Summary of Section~\ref{subsec:lewis}}
    \label{tab:lewis}
\end{table}

\begin{table}[h]
    \centering
    \begin{tabular}{|l|l|l|} \hline
    	{\bf Parameters} & {\bf Definition} & {\bf Range} \\ \hline
    	$C$ & Large constant, chosen later & $[200, \infty)$ \\ \hline
    	$\alpha$ & $1/(4\log(4m/n))$ & $< 1/2$ \\ \hline
    	$p$ & $1-\alpha$ & $[2/3, 1)$ \\ \hline
    	$\eps$ & $\alpha/C$ & $[0, 1/100]$ \\ \hline
    	$\lambda$ & $C\log(Cm/\eps^2)/\eps$ & $\ge \eps^{-1}$ \\ \hline
    	$\gamma$ & $\eps/(C\lambda)$ & $(0, \eps)$ \\ \hline
    	$\cnorm$ & $C/\alpha$ & $\geq 100$ \\ \hline
    	$\cvalid$ & Constant in Definition \ref{def:validdistro} & $\geq \cnorm$\\ \hline
    	$r$ & $\eps\gamma/(\cnorm\sqrt{n})$ & $\le n^{-1/2}$ \\ \hline
    \end{tabular}
    \label{tab:parameter_choices}
\end{table}

In this section we collect several facts about the regularized Lewis weights defined in Definition \ref{def:lewis}, many of which are variations of those in \cite{ls19}. All omitted proofs are provided in Section \ref{subsec:proofslev}. Before starting, we set up notation for important matrices throughout.
\begin{definition}[Fundamental matrices]
\label{def:fundamental}
We define the matrices
\begin{itemize}
\item $\mW_c \defeq \diag(w(c))$.
\item For any matrix $\mM$, orthogonal projection matrix $\mP(\mM) \defeq \mM(\mM^\top\mM)^{-1}\mM^\top$.
\item The projection matrix $\mP_c \defeq \mP(\mW_c^{\frac{1}{2} -\frac{1}{p} }\mC\mA)$.
\item $\sigma_c \defeq \sigma(\mW_c^{\frac{1}{2} -\frac{1}{p} }\mC\mA)$ and $\mSigma_c \defeq \diag(\sigma_c)$.
\item $\mLambda_c \defeq \mSigma_c - \mP_c^{(2)}$ and $\overline{\mLambda}_c = \mW_c^{-\frac{1}{2} }\mLambda_c \mW_c^{-\frac{1}{2} }$.
\item $\mJ_c$ as the Jacobian of $w(c)$ with respect to $c$.
\end{itemize}
\end{definition}

We now describe the regularized Lewis weights as the solution to a convex program.
\begin{restatable}[Alternate definition of regularized Lewis weights]{lemma}{optproblem}
\label{lemma:optproblem}
For all non-negative $c$
\[
w(c) = \argmin_{w\in\R^m_{>0}} f(w,c)
\]
where
\begin{align*} 
f(w,c) \defeq -\frac{1}{1-\frac2p}\log\det(\mA^\top \mC\mW^{1-\frac2p}\mC\mA) + \sum_{i=1}^m w_i - \sum_{i=1}^m v_i\log w_i. 
\end{align*}
\end{restatable}
Using the convex program in Lemma \ref{lemma:optproblem} we can compute the Jacobian of changes in the $\ell_p$ regularized Lewis weights with respect to the diagonal scaling.
\begin{restatable}[Jacobian of Regularized Lewis weight]{lemma}{derivlewis}
\label{lemma:derivlewis}
For a fixed vector $v$, we have that
\begin{align*} 
\mJ_c = 2\mW_c\left(\mW_c-\left(1-\frac2p\right)\mLambda_c\right)^{-1}\mLambda_c\mC^{-1}. 
\end{align*}
\end{restatable}

It will be useful to decompose the matrices relating to $\overline{\mLambda}_c$ into two parts -- one of which is diagonal, and one of which is bounded by $\mP_c^{(2)}$ essentially.
\begin{restatable}[Decomposition of $\mJ_c$]{lemma}{decomp}
\label{lemma:decomp}
For any vector $c \in \R^m_{>0}$, there is a diagonal matrix $\mzero \pe \mD_c \pe \mI$ such that for $\mK_c = \mW_c^{-1}\mJ_c\mC - \mD_c$, we have for all vectors $h$ that
\begin{itemize}
\item $\|\mK_ch\|_\infty \ls \|h\|_\infty$.
\item $\|\mK_ch\|_{w(c)} \ls \||h|\|_{\mP_c^{(2)}}$.
\end{itemize}
\end{restatable}
We will need a similar decomposition for a related matrix.

\begin{restatable}[Alternate decomposition]{lemma}{decomptwo}
\label{lemma:decomp2}
In the notation of Lemma \ref{lemma:decomp}, there is a diagonal matrix $\mD'_c$ and matrix $\mK'_c$ such that $\mzero \pe \mD'_c \pe \mI$,
\begin{align*} 
\mW_c^{-\frac{1}{2} }\left(\mI - \left(1-\frac2p\right)\overline{\mLambda}_c\right)^{-1}\mW_c^\frac{1}{2}  = \mD_c' + \mK_c', 
\end{align*}
and for all vectors $h$,
\begin{itemize}
\item $\|\mK'_ch\|_{w(c)} \ls \||h|\|_{\mP_c^{(2)}}$
\item $\|\mK'_ch\|_\infty \ls \|h\|_\infty.$
\end{itemize}
\end{restatable}

We note that small perturbations to the diagonal scaling also have a small effect on our $\ell_p$ Lewis weights. This is useful because our algorithm only maintains approximate $x, s, \tau$.
\begin{restatable}[Lewis weight approximation]{lemma}{lewisapprox}
\label{lemma:lewisapprox}
Let $p \in (0, 4).$ If $\bar{\mC} \approx_\eps \mC$ then $w_p(\mC\mA) \approx_{4\eps} w_p(\bar{\mC}\mA).$
\end{restatable}

\paragraph{Notation for analysis of $\d_\tau$.} The remainder of the section is devoted to understanding the changes in the regularized Lewis weights under (random) changes to $c$. The notation used is as follows. We will consider the definition $\d_c = c^\new-c$ and $c_t = c+t\d_c$, and let $\mSigma_t, \mP_t, \mLambda_t, \mJ_t, \overline{\mLambda}_t$ denote the corresponding fundamental matrices defined in Definition \ref{def:fundamental} for $c := c_t$. Additionally, we write $\Tau_t$ for $\mW_{c_t}$, and let $\mK_t, \mD_t, \mK_t', \mD_t'$ be the matrices resulting in Lemmas \ref{lemma:decomp} and \ref{lemma:decomp2}.

To give good bounds on $\d_\tau$ we will to assume bounds on $\d_c$. Our necessary conditions are summarized as follows.
\begin{definition}[$\gamma$-boundedness]
\label{def:gammabounded}
Let $c \in \R^m_{>0}$ be a deterministic vector, and $c^\new \in \R^m_{>0}$ be a stochastic vector. Let $\d_c = c^\new-c$. We say that the $\d_c$ is \emph{$\gamma$-bounded} if the following conditions hold.
\begin{enumerate}
\item With probability $1$ we have \begin{align} \|\mC^{-1}\d_c\|_\infty \ls \gamma. \label{eq:gammaboundedpart1} \end{align}
\item \begin{align} \left\|\E[(\mC^{-1}\d_c)^2]\right\|_\tpi \ls \gamma^2. \label{eq:gammaboundedpart2} \end{align}
\item For all $t \in [0, 1]$, we have \begin{align} \E[\|\mC^{-1}|\d_c|\|_{\mP_t^{(2)}}] \ls \gamma/\cnorm. \label{eq:gammaboundedpart3} \end{align}
\end{enumerate}
\end{definition}
Note that if $\d_c$ is $\gamma$-bounded, then for $c_t = c + t\d_c$, we have $\mC_t \approx_{1.1} \mC$ for all $0 \le t \le 1$.

Before moving to bounds on $\d_\tau$, we first show some preliminary bounds on the norms of derivatives of $\tau$ locally.
\begin{lemma}[$\d_\tau$ infinitesimal bound]
\label{lemma:firstbound}
Let $\d_c = c^\new-c$ be a $\gamma$-bounded change, $c_t = c + t\d_c$, and $\tau_t = w_p(c_t)$. Let $\d_{\tau_t} = \frac{ \mathrm{d} }{ \mathrm{d} t }  \tau_t$. Then
\begin{itemize}
\item $\|\Tau^{-1}\d_{\tau_t}\|_\infty \ls \gamma$ with probability $1$.
\item $\left\| \E[(\Tau^{-1}\d_{\tau_t})^2] \right\|_\tpi \ls \gamma^2.$
\item $\E[\|\Tau^{-1}|\d_{\tau_t}|\|_{\mP_t^{(2)}}] \ls \gamma/\cnorm$.
\end{itemize}
\end{lemma}
\begin{proof}
Before beginning the proof, we note that $\tau_t \approx \tau$ for all $t \in [0, 1]$ by Lemma \ref{lemma:lewisapprox}, and because $\d_c$ is a $\gamma$-bounded change (Definition \ref{def:gammabounded} \eqref{eq:gammaboundedpart1}), so $c_t \approx_{0.1} c$.

For the first claim we write
\begin{align*}
\|\Tau^{-1}\d_{\tau_t}\|_\infty 
\ls & ~ \|\Tau_t^{-1}\mJ_t\d_c\|_\infty \\
= & ~ \|\Tau_t^{-1}\mJ_t\mC_t\mC_t^{-1}\d_c\|_\infty \\
\ls & ~ \|\mD_t\mC_t^{-1}\d_c\|_\infty+\|\mK_t\mC_t^{-1}\d_c\|_\infty \\
\le & ~ (\|\mD_t\|_\infty+\|\mK_t\|_\infty)\|\mC_t^{-1}\d_c\|_\infty \ls \gamma,
\end{align*}
where the second step follows from $\mI = \mC_t \mC_t^{-1}$, the third step follows from triangle inequality and $\Tau_t^{-1}\mJ_t\mC_t = \mD_t + \mK_t$ as in Lemma \ref{lemma:decomp}, the fourth step follows from $\|\mD_t\mC_t^{-1}\d_c\|_\infty \leq \|\mD_t\|_\infty \cdot \|\mC_t^{-1}\d_c\|_\infty$, and the last step follows from $\|\mD_t\|_\infty, \|\mK_t\|_\infty \lesssim 1 $ and $ \|\mC_t^{-1}\d_c\|_\infty \lesssim \gamma$ (Definition \ref{def:gammabounded} \eqref{eq:gammaboundedpart1}).

Because $\Tau_t^{-1}\mJ_t\mC_t = \mD_t + \mK_t$ as in Lemma \ref{lemma:decomp}, we can write
\begin{align*}
\E[(\Tau^{-1}\d_{\tau_t})_i^2] 
\ls & ~ \E[(\Tau_t^{-1}\mJ_t\d_c)_i^2] \\
\le & ~ \E[(\mD_t\mC_t^{-1}\d_c + \mK_t\mC_t^{-1}\d_c)_i^2] \\
\ls & ~ \E[(\mD_t\mC_t^{-1}\d_c)_i^2] + \E[(\mK_t\mC_t^{-1}\d_c)_i^2].
\end{align*}
For the first term, we bound
\begin{align*} 
\left\|\E[(\mD_t\mC_t^{-1}\d_c)^2]\right\|_\tpi \ls \left\|\E[(\mC^{-1} \d_c)^2]\right\|_\tpi \ls \gamma^2 
\end{align*}
by Lemma \ref{lemma:decomp} ($\mD_t$ is diagonal and $\ls \mI$) and $\gamma$-boundedness (Definition \ref{def:gammabounded} \eqref{eq:gammaboundedpart2}). For the second term, we use Lemma \ref{lemma:decomp} item 1 and $\gamma$-boundedness (Definition \ref{def:gammabounded} \eqref{eq:gammaboundedpart1}) to first bound
\begin{align*} 
|(\mK_t\mC_t^{-1}\d_c)_i| \ls \|\mC_t^{-1}\d_c\|_\infty \ls \gamma. 
\end{align*} 
This handles the $\infty$-norm part directly.
For the $\tau$-norm, we compute
\begin{align*}
\left\|\E[(\mK_t\mC_t^{-1}\d_c)^2]\right\|_\tau 
\ls & ~ \gamma\E[\|\mK_t\mC_t^{-1}\d_c\|_\tau] \\
\ls & ~ \gamma\E[\|\mC_t^{-1}|\d_c|\|_{\mP_t^{(2)}}] \\
\ls & ~ \gamma\E[\|\mC^{-1}|\d_c|\|_{\mP_t^{(2)}}] \ls \gamma^2/\cnorm
\end{align*}
where the first step follows from $\|\mK_t\mC_t^{-1}\d_c\|_\infty \ls \gamma$, the second step follows from Lemma \ref{lemma:decomp} item 2, the third step follows from $c \approx_1 c_t$, the last step follows from $\gamma$-boundedness (Definition \ref{def:gammabounded} \eqref{eq:gammaboundedpart3}).

Summing the contributions gives the desired.

For the third claim, we use $\Tau_t^{-1}\mJ_t\mC_t = \mD_t + \mK_t$ as in Lemma \ref{lemma:decomp} to get
\begin{align*}
\E[\|\Tau^{-1}|\d_{\tau_t}|\|_{\mP_t^{(2)}}] 
\ls & ~ \E[\||\Tau_t^{-1}\mJ_t\mC_t\mC_t^{-1}\d_c|\|_{\mP_t^{(2)}} \\
\le & ~ \E[\|\mD_t|\mC_t^{-1}\d_c|\|_{\mP_t^{(2)}}] + \E[\||\mK_t\mC_t^{-1}\d_c|\|_{\mP_t^{(2)}}] \\
\ls & ~ \gamma/\cnorm + \E[\|\mK_t\mC_t^{-1}\d_c\|_{\tau_t}] \\
\ls & ~ \gamma/\cnorm + \E[\|\mC_t^{-1}|\d_c|\|_{\mP_t^{(2)}}] \ls \gamma/\cnorm,
\end{align*}
where the first step follows from $\tau \approx_1 \tau_t$, the second step follows from the triangle inequality, the third step follows from $\mD_t \ls \mI$ and $\gamma$-boundedness (Definition \ref{def:gammabounded} \eqref{eq:gammaboundedpart3}) and $\mP_t^{(2)} \ls \Tau_t$, the fourth step follows from Lemma \ref{lemma:decomp} item 2, and the last step follows from $\gamma$-boundedness (Definition \ref{def:gammabounded} \eqref{eq:gammaboundedpart3}) again.
\end{proof}

Finally, we give our full analysis of $\d_\tau$. We first control the infinity norm and square of the change, which have shorter analyses.

\begin{lemma}[Sharper bound on changes in $\tau$ part 1]
\label{lemma:tauchange1}
Let $\d_c = c^\new-c$ be a $\gamma$-bounded change, and let $\tau^\new = \tau_1$, and $\d_\tau = \tau^\new - \tau$. We have
\begin{itemize}
\item $\|\Tau^{-1}\d_\tau\|_\infty \ls \gamma$ with probability $1$.
\item $\left\|\E[(\Tau^{-1}\d_\tau)^2]\right\|_\tpi \ls \gamma^2.$
\end{itemize}
\end{lemma}
\begin{proof}
To start, we define $c_t = c + t\d_c$, $\tau_t = w_p(c_t)$. Let $\d_{\tau_t} = \frac{ \mathrm{d} }{ \mathrm{d} t }  \tau_t$.
\paragraph{Bound on $\|\Tau^{-1}\d_\tau\|_\infty$.} Follows from $\gamma$-boundedness (Definition \ref{def:gammabounded} \eqref{eq:gammaboundedpart1}) and Lemma \ref{lemma:lewisapprox}.

\paragraph{Bound on $\|\E[(\Tau^{-1}\d_\tau)^2]\|_\tpi.$} We have
\begin{align*}
\|\E[(\Tau^{-1}\d_\tau)^2]\|_\tpi 
= & ~ \left\|\E\left[\left(\int_0^1 \Tau^{-1}\d_{\tau_t} \mathrm{d}t \right)^2\right]\right\|_\tpi \\
\le & ~ \int_0^1 \E[(\Tau^{-1}\d_{\tau_t})^2] \mathrm{d}t \ls \gamma^2.
\end{align*}
where the first step follows from $\d_\tau = \int_0^1 \d_{\tau_t} dt$ and Cauchy-Schwarz, the second step follows from Cauchy-Schwarz, and the last step follows from Lemma \ref{lemma:firstbound} item 2.
\end{proof}

We now analyze the first order change in $\d_\tau$. We defer the proof to the appendix due to its length.
\begin{restatable}[Sharper bound on changes in $\tau$ part 2]{lemma}{tauchangetwo}
\label{lemma:tauchange2}
Let $\d_c = c^\new-c$ be a $\gamma$-bounded change, and let $\tau^\new = \tau_1$, and $\d_\tau = \tau^\new - \tau$. For $\mJ$ as defined in Lemma \ref{lemma:derivlewis}, we have
\begin{align*} \left\|\Tau^{-1}(\E[\d_\tau] - \mJ\E[\d_c])\right\|_\tpi \ls \gamma^2. \end{align*}
\end{restatable}

\subsection{Bounding $\d_\mu, \d_\tau, \d_c, \d_y$}
\label{subsec:mainbounds}

\begin{table}[h]
    \centering
    \begin{tabular}{|l|l|l|} \hline
        {\bf Statement} & {\bf Term} & {\bf Comment} \\ \hline \hline
        Lemma~\ref{lemma:muchange} & $\delta_{\mu}$ & Bounds on $\delta_{\mu}$ \\ \hline
        Lemma~\ref{lemma:xchange} & $\bar{\delta}_x$ & Bounds on $\bar{\delta}_x$ \\ \hline
        Lemma~\ref{lemma:pchange} & $\delta_c$ & Bounds on $\delta_c$ \\ \hline
        Lemma~\ref{lemma:tauchange} & $\delta_{\tau}$ & Bounds on $\delta_{\tau}$ \\ \hline
        Lemma~\ref{lemma:ychange} & $\delta_y$ & Bounds on $\delta_y$ \\ \hline
        Lemma~\ref{lemma:p2bound} & $\delta_x, \delta_c$ & Bounds on $\delta_x$ and $\delta_c$ with respect to $\mP_t^{(2)}$-norm \\ \hline \hline
        Lemma~\ref{lemma:normbound} & $\Phi''(\bar{x}), \bar{\Tau}, \mH$ & Matrix norm bounds \\ \hline
        Corollary~\ref{cor:xchange} & $\bar{\d}_s, \d_r$ & An application of Lemma~\ref{lemma:xchange} \\ \hline
        Lemma~\ref{lemma:p2totau} & $\mR v$ & $\mP^{(2)}$ norm bound \\ \hline
        Lemma~\ref{lemma:approxsolution} & $g$ & Approximation to direction $g$ \\ \hline
    \end{tabular}
    \caption{Summary of Section~\ref{subsec:mainbounds}}
    \label{tab:mainbounds}
\end{table}

In this section, we bound $\d_\mu, \d_\tau, \d_c, \d_y$. First we bound $\delta_\mu$. Since the change in $\mu$ is a deterministic change, bounding $\delta_\mu$ is straightforward.
\begin{lemma}[Bounds on $\delta_{\mu}$]
\label{lemma:muchange}
Let $\delta_{\mu} = \mu^{\new} - \mu$. 
We have that $\|\mu^{-1}(\mu^\new-\mu)\|_\tpi \ls \eps\gamma.$
\end{lemma}
\begin{proof}
Note that $\|\mu^{-1}(\mu^\new-\mu)\|_\infty \le r$ by definition. Therefore
\begin{align*} 
\|\mu^{-1}(\mu^\new-\mu)\|_\tpi 
\le & ~ r\|\vec{1}\|_\tpi \\
= & ~ r(1+\cnorm(n+\|v\|_1)^{1/2}) \\
\ls & ~ \cnorm rn^{1/2} \le \eps\gamma 
\end{align*}
where the first step follows from $\|\mu^{-1}(\mu^\new-\mu)\|_\infty \le r$, the third step follows from $\cnorm \ge 100$ and $\|v\|_1 \ls \sqrt{n}$, and the last step follows from $\cnorm rn^{1/2} \le \eps\gamma$ (by the choice of $r$).

\end{proof}
The following bounds allow us to relate changes in $x$ to $\|g\|_\tpi$.
\begin{lemma}[Matrix norm bounds]
\label{lemma:normbound}
Let $x, \tau, \bar{x}, \bar{\tau}$ satisfy Invariant \ref{invar}. For any $g \in \R^m$ we have that
\begin{itemize}
\item $\|\Phi''(\bar{x})^{-\frac{1}{2} }\bar{\Tau}^{-1}\mA\mH^{-1}\mA^\top\Phi''(\bar{x})^{-\frac{1}{2} }g\|_\infty \ls \|g\|_\tau.$
\item $\|\Phi''(\bar{x})^{-\frac{1}{2} }\bar{\Tau}^{-1}\mA\mH^{-1}g\|_\tau \ls \|g\|_{(\mA^\top\Tau^{-1}\Phi''(x)^{-1}\mA)^{-1}}.$
\item $\|\Phi''(\bar{x})^{-\frac{1}{2} }\bar{\Tau}^{-1}\mA\mH^{-1}g\|_\infty \ls \|g\|_{(\mA^\top\Tau^{-1}\Phi''(x)^{-1}\mA)^{-1}}.$
\end{itemize}
\end{lemma}
\begin{proof}
Define \begin{align*} 
\mQ = \bar{\Tau}^{-\frac{1}{2} }\Phi''(\bar{x})^{-\frac{1}{2} }\mA\mH^{-1}\mA^\top\bar{\Tau}^{-\frac{1}{2} }\Phi''(\bar{x})^{-\frac{1}{2} }. 
\end{align*}
For the first point, we use the Cauchy-Schwarz inequality to get that
\begin{align*}
\|\Phi''(\bar{x})^{-\frac{1}{2} }\bar{\Tau}^{-1}\mA\mH^{-1}\mA^\top\Phi''(\bar{x})^{-\frac{1}{2} }g\|_\infty 
= & ~ \max_{i\in[m]} \left|e_i^\top \Phi''(\bar{x})^{-\frac{1}{2} }\bar{\Tau}^{-1}\mA\mH^{-1}\mA^\top\Phi''(\bar{x})^{-\frac{1}{2} }g \right| \\
= & ~ \max_{i\in[m]} \left| e_i^\top \bar{\Tau}^{-\frac{1}{2} }\mQ\bar{\Tau}^\frac{1}{2}  g \right| \\
\le & ~ \left| g^\top\bar{\Tau}^\frac{1}{2} \mQ\bar{\Tau}^\frac{1}{2} g \right|^\frac{1}{2}  \max_{i\in[m]} \left| e_i^\top \bar{\Tau}^{-\frac{1}{2} }\mQ\bar{\Tau}^{-\frac{1}{2} }e_i\right|^\frac{1}{2} \\ 
\ls & ~ \|g\|_{\bar{\tau}} \max_{i\in[m]} \bar{\tau_i}^{-1}\sigma\left(\bar{\Tau}^{-\frac{1}{2} }\Phi''(\bar{x})^\frac{1}{2} \right)\\ 
\ls & ~ \|g\|_\tau \max_{i\in[m]} \tau(\bar{x})_i^{-1}\sigma\left(\tau(\bar{x})^{\frac{1}{2} -\frac{1}{p} }\Phi''(\bar{x})^\frac{1}{2} \right) \ls \|g\|_\tau,
\end{align*}
where the third step follows from $| a^\top b | \leq \| a \|_2 \cdot \| b \|_2$, where the fifth step follows from the stability of leverage scores and that $p = 1 - \frac{1}{4\log(4m/n)}$.

Let $\bar{\mH} \defeq \mA^\top\bar{\Tau}^{-1}\bar{\Phi}''(x)^{-1}\mA$, so that $\mH \approx_\eps \bar{\mH}$ and $\bar{\mH} \approx_{2\eps} \mA^\top\Tau^{-1}\Phi''(x)^{-1}\mA$.
\begin{align*}
\|\Phi''(\bar{x})^{-\frac{1}{2} }\bar{\Tau}^{-1}\mA\mH^{-1}g\|_\tau 
= & ~ (g^\top\mH^{-1}\bar{\mH}\mH^{-1}g)^\frac{1}{2}  \\
\ls & ~ (g^\top\bar{\mH}^{-1}g)^\frac{1}{2}  \\
= & ~ \|g\|_{\bar{\mH}^{-1}} \ls \|g\|_{(\mA^\top\Tau^{-1}\Phi''(x)^{-1}\mA)^{-1}},
\end{align*}
where the second step follows from $\mH \approx_\eps \bar{\mH}$, and the last step follows from $\bar{\mH} \approx_{2\eps} \mA^\top\Tau^{-1}\Phi''(x)^{-1}\mA$.

For the third point, we use the Cauchy-Schwarz inequality to get that
\begin{align*}
\|\Phi''(\bar{x})^{-\frac{1}{2} }\bar{\Tau}^{-1}\mA\mH^{-1}g\|_\infty 
= & ~ \max_{i\in[m]} |e_i^\top \Phi''(\bar{x})^{-\frac{1}{2} }\bar{\Tau}^{-1}\mA\mH^{-1}g| \\
\le & ~ \max_{i\in[m]} |e_i^\top \bar{\Tau}^{-\frac{1}{2} }\mQ\bar{\Tau}^{-1}e_i|^\frac{1}{2}  |g^\top \mH^{-1} g|^\frac{1}{2}  \\
= & ~ \|g\|_{\mH^{-1}} \max_{i\in[m]} |e_i^\top \bar{\Tau}^{-\frac{1}{2} }\mQ\bar{\Tau}^{-1}e_i|^\frac{1}{2}  \\
\ls & ~ \|g\|_{(\mA^\top\Tau^{-1}\Phi''(x)^{-1}\mA)^{-1}} \max_{i\in[m]} \left| e_i^\top \bar{\Tau}^{-\frac{1}{2} }\mQ\bar{\Tau}^{-\frac{1}{2} }e_i\right|^\frac{1}{2}  \\
\ls & ~ \|g\|_{(\mA^\top\Tau^{-1}\Phi''(x)^{-1}\mA)^{-1}},
\end{align*}
where the second step follows from $|a^\top b| \leq \| a \|_2 \cdot \| b \|_2$, the fourth step follows from $\mH \approx_{3\eps} \mA^\top\Tau^{-1}\Phi''(x)^{-1}\mA$, and the last step follows from $\max_{i\in[m]} | e_i^\top \bar{\Tau}^{-\frac{1}{2} }\mQ\bar{\Tau}^{-\frac{1}{2} }e_i |^\frac{1}{2} \lesssim 1$.
\end{proof}
We now show that the change in $x$, i.e. $\bar{\d}_x$, is small.
\begin{lemma}[Bounds on $\bar{\d}_x$]
\label{lemma:xchange}
Let $x, \tau, \bar{x}, \bar{\tau}$ satisfy Invariant \ref{invar}. If $\bar{\d}_x$ are defined as in Algorithm \ref{algo:lsstep}, then
\begin{itemize}
\item $\|\Phi''(x)^\frac{1}{2} \E[\bar{\d}_x]\|_\tpi \le \gamma + O((\eps+1/\cnorm)\gamma)$.
\item $\|\Phi''(x)^\frac{1}{2} \bar{\d}_x\|_\infty \ls \gamma$ with probability $1$.
\item $\|\E[\Phi''(x)\bar{\d}_x^2]\|_\tpi \ls \gamma^2$.
\end{itemize}
\end{lemma}
\begin{proof}
We break the proof into the three claims.
\paragraph{Bound on $\|\Phi''(x)^\frac{1}{2} \E[\bar{\d}_x]\|_\tpi$.}
Because $\E[\mR] = \mI$ (Expectation) we have that
\begin{align} \label{eq:mainterm} 
\E[\bar{\d}_x] &= \Phi''(\bar{x})^{-\frac{1}{2} }g - \Phi''(\bar{x})^{-\frac{1}{2} }\bar{\Tau}^{-1}\Phi''(\bar{x})^{-\frac{1}{2} }\mA\mH^{-1}\mA^\top \Phi''(\bar{x})^{-\frac{1}{2} }g \\ &- \Phi''(\bar{x})^{-\frac{1}{2} }\bar{\Tau}^{-1}\Phi''(\bar{x})^{-\frac{1}{2} }\mA\mH^{-1}(\mA^\top x - b) \label{eq:mainterm2}. 
\end{align}
We start by bounding the $\tpi$ norm of the expression in (\ref{eq:mainterm2}). We calculate
\begin{align*}
&\|\Phi''(x)^{\frac{1}{2} }\Phi''(\bar{x})^{-\frac{1}{2} }\bar{\Tau}^{-1}\Phi''(\bar{x})^{-\frac{1}{2} }\mA\mH^{-1}(\mA^\top x - b)\|_\tpi \\
\ls & ~ \cnorm\|\mA^\top x-b\|_{(\mA^\top\Tau^{-1}\Phi''(x)^{-1}\mA)^{-1}} \\
\le & ~ \eps\gamma
\end{align*}
where the first step follows from combining Lemma \ref{lemma:normbound} items 2 and 3, and the last step follows from the fact that $(x, s, \mu)$ is $\eps$-centered (Definition \ref{def:centered}).

Now we bound the $\tau$-norm of the two terms in (\ref{eq:mainterm}). Define 
\begin{align*} 
\mQ = \bar{\Tau}^{-\frac{1}{2} }\Phi''(\bar{x})^{-\frac{1}{2} }\mA\mH^{-1}\mA^\top\bar{\Tau}^{-\frac{1}{2} }\Phi''(\bar{x})^{-\frac{1}{2} }, 
\end{align*} and note that $\mQ \approx_\gamma \widetilde{\mQ}$ for orthogonal projection matrix
\[ \widetilde{\mQ} \defeq \bar{\Tau}^{-\frac{1}{2} }\Phi''(\bar{x})^{-\frac{1}{2} }\mA(\mA^\top\bar{\Tau}^{-1}\Phi''(\bar{x})^{-1}\mA)^{-1}\mA^\top\bar{\Tau}^{-\frac{1}{2} }\Phi''(\bar{x})^{-\frac{1}{2} } \] by the condition in line \ref{line:ipm:H} of Algorithm \ref{algo:lsstep}.
Hence all eigenvalues of $\mQ$ are either $0$ or in $[1-\gamma, 1+\gamma]$, so all eigenvalues of $\mI - \mQ$ are either $1$ or in $[-\gamma, \gamma]$, so $\|\mI - \mQ\|_2 \le 1$.

Using Lemma \ref{lemma:selfcon}
\begin{align}
& ~ \|\Phi''(x)^\frac{1}{2} (\Phi''(\bar{x})^{-\frac{1}{2} }g - \Phi''(\bar{x})^{-\frac{1}{2} }\bar{\Tau}^{-1}\Phi''(\bar{x})^{-\frac{1}{2} }\mA\mH^{-1}\mA^\top \Phi''(\bar{x})^{-\frac{1}{2} }g)\|_\tau \notag\\
\le & ~ e^{O(\eps)}\|g - \bar{\Tau}^{-1}\Phi''(\bar{x})^{-\frac{1}{2} }\mA\mH^{-1}\mA^\top\Phi''(\bar{x})^{-\frac{1}{2} }g\|_{\bar{\tau}} \notag\\
= & ~ e^{O(\eps)}\|(\mI-\mQ)\bar{\Tau}^\frac{1}{2} g\|_2 \notag \\
\le & ~ e^{O(\eps)}\|\bar{\Tau}^\frac{1}{2} g\|_2 = e^{O(\eps)}\|g\|_{\bar{\tau}} \le e^{O(\eps)}\|g\|_\tau \le (1+O(\eps))\|g\|_\tau \label{eq:i-q}
\end{align}
where the first step follows from $\Phi''(x)^\frac{1}{2} \approx_{O(\eps)} \Phi''(\bar{x})^{\frac{1}{2}}$ by Invariant \ref{invar} and Lemma \ref{lemma:selfcon}, the second step follows from the definition of $\mQ$, the third step follows from $\|\mI - \mQ\|_2 \le 1$, and the last step follows from $e^{O(\epsilon)} \leq 1+O(\epsilon)$, $\forall \epsilon \in (0,1)$.

Now we bound the $\infty$ norm. To handle the first term of (\ref{eq:mainterm}), we can use Lemma \ref{lemma:selfcon} to get
\begin{align*} 
\|\Phi''(x)^\frac{1}{2} \Phi''(\bar{x})^{-\frac{1}{2} }g\|_\infty \le e^{O(\eps)}\|g\|_\infty \le (1+O(\eps))\|g\|_\infty,
\end{align*}
where the last step follows from $e^{O(\epsilon)} \leq 1+O(\epsilon)$, $\forall \epsilon \in (0,1)$.

For the second term, we have %
\begin{align*}
\|\Phi''(x)^\frac{1}{2} \Phi''(\bar{x})^{-1}\bar{\Tau}^{-1}\mA\mH^{-1}\mA^\top\Phi''(\bar{x})^{-\frac{1}{2} }g\|_\infty 
\ls & ~ \|\Phi''(\bar{x})^{-\frac{1}{2} }\bar{\Tau}^{-1}\mA\mH^{-1}\mA^\top\Phi''(\bar{x})^{-\frac{1}{2} }g\|_\infty \\
\ls & ~ \|g\|_\tau.
\end{align*}
where the first step follows from Lemma \ref{lemma:selfcon}, and the last step follows from Part 1 of Lemma \ref{lemma:normbound}.

Finally
\begin{align*}
\|\Phi''(x)^\frac{1}{2} \E[\bar{\d}_x]\|_\tpi 
\le & ~ (1+O(\eps))\|g\|_\infty + O(\|g\|_\tau) + \cnorm(1+O(\eps))\|g\|_\tau + O(\eps\gamma) \\
\le & ~ (1+O(\eps+1/\cnorm))\|g\|_\tpi + O(\eps\gamma) \\
\le & ~ \gamma + O(\eps+1/\cnorm)\gamma
\end{align*}
where we have used that $\|g\|_\tpi \le \gamma$.

\paragraph{Bound on $\|\Phi''(x)^\frac{1}{2} \bar{\d}_x\|_\infty$.} First, by Lemma \ref{lemma:selfcon} we have
\begin{align*} 
\|\Phi''(x)^\frac{1}{2} \bar{\d}_x\|_\infty = \|\Phi''(x)^\frac{1}{2} \Phi''(\bar{x})^{-\frac{1}{2} }(g-\mR \d_r)\|_\infty \ls \|g-\mR \d_r\|_\infty \le \|g\|_\infty + \|\mR \d_r\|_\infty. 
\end{align*}
Note that $\|g\|_\infty \le \|g\|_\tpi \le \gamma$, and $\|\mR\d_r\|_\infty \ls \gamma$ by the (Maximum) condition.
Therefore, $\|g-\mR \d_r\|_\infty \ls \gamma$ as desired.

\paragraph{Bound on $\|\E[\Phi''(x)\bar{\d}_x^2]\|_\tpi$.}
First we use Lemma \ref{lemma:selfcon} to get
\begin{align*}
\|\E[\Phi''(x)\bar{\d}_x^2]\|_\tpi 
\ls & ~ \|\Phi''(\bar{x})\E[\bar{\d}_x^2]\|_\tpi \\
\ls & ~ \|\E[(g-\mR \d_r)]^2\|_\tpi \\
\ls & ~ \|g^2\|_\tpi + \|\E[\mR^2\d_r^2]\|_\tpi \\
\ls & ~ \gamma^2 + \|\E[\mR^2\d_r^2]\|_\tpi.
\end{align*}
We bound $\E[(\mR_{ii}(\d_r)_i)^2]$ by the (Variance) condition of Definition \ref{def:validdistro}. This gives
\begin{align*}
&\|\E[\mR^2\d_r^2]\|_\tpi \le \|\d_r^2\|_\tpi + \frac{\gamma}{\cvalid^2}\|\d_r\|_\tpi \ls \gamma^2
\end{align*}
as $\|\d_r\|_\tau \ls \gamma$ by the above, and $\cvalid \ge \cnorm \ge 1$. Summing these gives the result.
\end{proof}

We show a corollary which is straightforward application of Part 1 of Lemma~\ref{lemma:xchange}.
\begin{corollary}\label{cor:xchange}
Let $x, \tau, \bar{x}, \bar{\tau}$ satisfy Invariant \ref{invar}. If $\bar{\d}_x$ are defined as in Algorithm \ref{algo:lsstep}, then
\begin{itemize}
\item $\mu^{-1}\|\Tau^{-1}\Phi''(x)^{-\frac{1}{2} }\bar{\d}_s\|_\tpi \ls \gamma$.
\item $\|\d_r\|_\tpi \le \gamma + O((\eps+1/\cnorm)\gamma).$
\end{itemize}
\end{corollary}
\begin{proof}
Recall the proof of Part 1 of Lemma~\ref{lemma:xchange}. The bounds on $\mu^{-1}\|\Tau^{-1}\Phi''(x)^{-\frac{1}{2} }\bar{\d}_s\|_\tpi$ and $\|\d_r\|_\tpi$ follow analogously, by replacing $\mI-\mQ$ with $\mQ$ in (\ref{eq:i-q}).
\end{proof}

We can use highly $1$-self-concordance to bound $\d_c$ in terms of $\bar{\d}_x$.
\begin{lemma}[Bounds on $\d_c$]
\label{lemma:pchange}
Let $x, \tau, \bar{x}, \bar{\tau}$ satisfy Invariant \ref{invar}, and $\|g\|_\tpi \le \gamma$. Let $c^\new = \Phi''(x^\new)^{-\frac{1}{2} }$ and $\d_c = c^\new-c$. Then
\begin{itemize}
\item $\|\mC^{-1}\d_c\|_\infty \ls \gamma$ with probability $1$.
\item $\|\mC^{-1}\E[\d_c]\|_\tpi \le \gamma + O((\eps+1/\cnorm)\gamma).$
\item $\|\E[\mC^{-2}\d_c^2]\|_\tpi \ls \gamma^2.$
\end{itemize}
\end{lemma}
\begin{proof}
For the first point, simply note that $\mC^{-1} = \Phi''(x)^{\frac{1}{2} }$ and $|\d_c| = |c^\new-c| \le |\bar{\d}_x|$ coordinate-wise by $1$-self-concordance. Therefore, the result follows from Lemma \ref{lemma:xchange}. Now, we get that $c^\new \approx_{O(\gamma)} c$ by Lemma \ref{lemma:selfcon}.

For the second point, we integrate and use highly $1$-self-concordance. Specifically, define $x_t = x+t\bar{\d}_x$ and $c_t = \Phi''(x_t)^{-\frac{1}{2} }.$ Then 
\begin{align*} 
\frac{ \mathrm{d} }{ \mathrm{d} t } c_t = -\frac{1}{2} \frac{\Phi'''(x_t)}{\Phi''(x_t)^\frac32}\bar{\d}_x \enspace \text{ and } \enspace \frac{ \mathrm{d}^2}{ \mathrm{d} t^2}c_t = \left(-\frac{1}{2} \frac{\Phi''''(x_t)}{\Phi''(x_t)^\frac32}+\frac34\frac{\Phi'''(x_t)^2}{\Phi''(x_t)^\frac52}\right)\bar{\d}_x^2. 
\end{align*}
Now, we have by second order expansion, highly $1$-self-concordance, and Lemma \ref{lemma:xchange}
\begin{align*}
\|\mC^{-1}\E[\d_c]\|_\tpi &\le \left\|\Phi''(x)^\frac{1}{2}  \cdot -\frac{1}{2} \frac{\Phi'''(x)}{\Phi''(x)^\frac32}\E[\bar{\d}_x]\right\|_\tpi \\
 + & ~  \frac{1}{2} \int_0^1 \left\|\E\left[\Phi''(x)^\frac{1}{2} \left(-\frac{1}{2} \frac{\Phi''''(x_t)}{\Phi''(x_t)^\frac32}+\frac34\frac{\Phi'''(x_t)^2}{\Phi''(x_t)^\frac52}\right)\bar{\d}_x^2\right] \right\|_\tpi \mathrm{d}t \\ 
\le & ~ \left\|\Phi''(x)^\frac{1}{2} \E[\bar{\d}_x]\right\|_\tpi + O\left(\int_0^1\left\|\E\left[\Phi''(x)\bar{\d}_x^2\right] \right\|_\tpi \mathrm{d} t\right) \\ 
\le & ~ \gamma + O(\eps+1/\cnorm)\gamma + O(\gamma^2) \notag \\ 
\le & ~ \gamma + O(\eps+1/\cnorm)\gamma,
\end{align*}
where the last step follows from $\gamma \in (0,1)$.

For the final point, use again that $|\d_c| \le |\bar{\d}_x|$ pointwise, so by Lemma \ref{lemma:xchange}
\begin{align*} 
\|\E[\mC^{-2}\d_c^2]\|_\tpi \le \|\E[\Phi''(x)\bar{\d}_x^2]\|_\tpi \ls \gamma^2. 
\end{align*}
\end{proof}

The next few lemmas show that $\d_c$ is $\gamma$-bounded, as in Definition \ref{def:gammabounded}. We need a variant of \cite[Lemma 48]{blss20}.
\begin{lemma}[$\mP^{(2)}$ norm bound]
\label{lemma:p2totau}
If $\mR$ is valid (Definition \ref{def:validdistro}) then for any vector $v \in \R^m$ we have that
\begin{align*}  
\E[\||\mR v|\|_{\mP^{(2)}}^2] \ls \|\E[|v|]\|_\tau^2
\text{ and }
 \E[\||\mR v|\|_{\mP^{(2)}}] \ls \|\E[|v|]\|_\tau. 
\end{align*} 
\end{lemma}
\begin{proof}
The second claim follows from the first and Cauchy-Schwarz. For the first claim, we compute using (Variance) and (Covariance) of Definition \ref{def:validdistro} that
\begin{align*}
\E[\||\mR v|\|_{\mP^{(2)}}^2]
= & ~ \sum_{i,j} \E[\mR_{ii}\mR_{jj}]|v_i||v_j|\mP_{ij}^2 \\
\le & ~ 2\sum_{i \in [m]} v_i^2\sigma_i^{-1}\mP_{ii}^2 + 2 \sum_{i \neq j} |v_i||v_j|\mP_{ij}^2 \\
\le & ~ 2\sum_{i \in [m]} v_i^2 \sigma_i + 2\||v|\|_{\mP^{(2)}}^2 \ls \||v|\|_\tau^2.
\end{align*}
\end{proof}
We need the following fact to handle terms involving $\mP^{(2)}$.
\begin{lemma}[Bounds on $\delta_x$ and $\delta_c$ with respect to $\mP_t^{(2)}$ norm]
\label{lemma:p2bound}
We have that $\E[\|\mC^{-1}|\d_x|\|_{\mP_t^{(2)}}] \ls \gamma/\cnorm$ and $\E[\|\mC^{-1}|\d_c|\|_{\mP_t^{(2)}}] \ls \gamma/\cnorm$.
\end{lemma}
\begin{proof}
Note that
\begin{align*} 
\mC^{-1}|\d_x| = |g - \mR \d_r| \le |g| + |\mR \d_r|.
\end{align*}
Note that $\Tau_t^{1/2-1/p}\mC_t \approx_2 \Tau^{1/2-1/p}\mC$. Therefore, we can use $\mP_t^{(2)} \pe \mSigma_t$, $\mSigma_t \approx_2 \mSigma$, Lemma \ref{lemma:xchange}, \cite[Lemma 4.23]{BrandLN+20}, and Lemma \ref{lemma:p2totau} to get
\begin{align*}
\E[\|\mC^{-1}|\d_x|\|_{\mP_t^{(2)}}] 
\le & ~ \|g\|_{\mP_t^{(2)}} + \E[\||\mR\d_r|\|_{\mP_t^{(2)}}]\\
\ls & ~ \gamma/\cnorm + \E[\||\mR\d_r|\|_{\mP^{(2)}}] \\
\le & ~ \gamma/\cnorm + \|\d_r\|_\tau \ls \gamma/\cnorm,
\end{align*}
where the second step follows from $\| g \|_{\mP_t^{(2)}} \ls \gamma/\cnorm$, the third step follows from $\E[\||\mR\d_r|\|_{\mP^{(2)}}] \le \|\d_r\|_\tau$.

The second claim follows directly from $1$-self-concordance and the first claim.
\end{proof}
\begin{lemma}[$\gamma$-boundedness of $\d_c$]
\label{lemma:gammaboundedc}
For $c^\new = \Phi''(x^\new)^{-1/2}$ and $\d_c = c^\new - c$, we have that $\d_c$ is $\gamma$-bounded.
\end{lemma}
\begin{proof}
To show Definition \ref{def:gammabounded} \eqref{eq:gammaboundedpart1} and \eqref{eq:gammaboundedpart2}, use Lemma \ref{lemma:pchange}. Definition \ref{def:gammabounded} \eqref{eq:gammaboundedpart3} follows from Lemma \ref{lemma:p2bound}.
\end{proof}
By Lemmas \ref{lemma:tauchange1} and \ref{lemma:tauchange2}, $\gamma$-boundedness of $\d_c$ allows us to control the change to $\tau$.
\begin{lemma}[Bounds on $\d_\tau$]
\label{lemma:tauchange}
Let $\tau^\new = \tau(x^\new)$, and $\d_\tau = \tau^\new - \tau$. Then we have that
\begin{itemize}
\item $\|\Tau^{-1}\d_\tau\|_\infty \ls \gamma$ with probability $1$.
\item $\left\|\E[(\Tau^{-1}\d_\tau)^2]\right\|_\tpi \ls \gamma^2.$
\item $\left\|\Tau^{-1}(\E[\d_\tau] - \mJ\E[\d_c])\right\|_\tpi \ls \gamma^2.$
\end{itemize}
\end{lemma}
\begin{proof}
Follows directly from the fact that $\d_c$ is $\gamma$-bounded (Lemma \ref{lemma:gammaboundedc}), along with Lemmas \ref{lemma:tauchange1} and \ref{lemma:tauchange2}.
\end{proof}

Now, we show that despite the approximations to $x, s, \tau$, the algorithm still take a step in a direction very close to the desired direction $g$.
\begin{lemma}[Approximation to direction $g$]
\label{lemma:approxsolution}
Let $x, \tau, \bar{x}, \bar{\tau}$ be as in Algorithm \ref{algo:lsstep}. Then
\begin{align*} 
\|\mu^{-1}\Phi''(x)^{-\frac{1}{2} }\Tau^{-1}(\bar{\d}_s+\mu\tau\phi''(x)\E[\bar{\d}_x])-g\|_\tpi \ls \eps\gamma. 
\end{align*}
\end{lemma}
\begin{proof}
At a high level, our proof will gradually change $x$ to $\bar{x}$ and $\tau$ to $\bar{\tau}$ and track the incurred error. Here, $\tau = \tau(x)$ and $\bar{\tau}$ satisfies Invariant \ref{invar}. We let $\d_x = \E[\bar{\d}_x]$.
First, we write
\begin{align*}
\bar{\d}_s+\mu\tau\phi''(x)\E[\bar{\d}_x] 
= & ~ \bar{\d}_s + \mu\bar{\tau}\phi''(x)\d_x + \mu(\tau-\bar{\tau})\phi''(x)\d_x \\ 
= & ~ \bar{\d}_s + \mu\bar{\tau}\phi''(\bar{x})\d_x + \mu(\tau-\bar{\tau})\phi''(x)\d_x + \mu\bar{\tau}(\phi''(x)-\phi''(\bar{x}))\d_x \\ 
= & ~ \mu\bar{\Tau}\Phi''(\bar{x})^\frac{1}{2} (g-\d_2) + \mu(\tau-\bar{\tau})\phi''(x)\d_x + \mu\bar{\tau}(\phi''(x)-\phi''(\bar{x}))\d_x.
\end{align*}
Therefore, we have that
\begin{align*}
& ~ \mu^{-1}\Phi''(x)^{-\frac{1}{2} }\Tau^{-1}(\bar{\d}_s+\mu\tau\phi''(x)\d_x)-g \\
= & ~ (\Phi''(x)^{-\frac{1}{2} }\Tau^{-1}\Phi''(\bar{x})^\frac{1}{2} \bar{\Tau}-\mI)g + \Phi''(x)^{-\frac{1}{2} }\Tau^{-1}(\tau-\bar{\tau})\phi''(x)\d_x \\
+ & ~ \Phi''(x)^{-\frac{1}{2} }\Tau^{-1}\bar{\tau}(\phi''(x)-\phi''(\bar{x}))\d_x - \Phi''(x)^{-\frac{1}{2} }\Tau^{-1}\Phi''(\bar{x})^\frac{1}{2} \bar{\Tau}\d_2.
\end{align*}
We now bound the terms in the above sum. For the second term, we use the approximation condition of Invariant \ref{invar} and Lemma \ref{lemma:lewisapprox} to get that
\begin{align*}
\|\Phi''(x)^{-\frac{1}{2} }\Tau^{-1}(\tau-\bar{\tau})\phi''(x)\d_x\|_\tpi \ls \eps\|\Phi''(x)^\frac{1}{2} \bar{\d}_x\|_\tpi \ls \eps\|g\|_\tpi \ls \eps\gamma
\end{align*}
by Lemma \ref{lemma:xchange}. For the third term, we use the approximation condition of Invariant \ref{invar}, Lemma \ref{lemma:selfcon} to get that 
\begin{align*} 
\|\Phi''(x)^{-1}(\phi''(x)-\phi''(\bar{x}))\|_\infty \ls \eps. 
\end{align*} 
Applying this and $\|\Tau^{-1}\bar{\tau}\|_\infty \ls 1$ using Lemma \ref{lemma:lewisapprox}, we get that
\begin{align*}
\|\Phi''(x)^{-\frac{1}{2} }\Tau^{-1}\bar{\tau}(\phi''(x)-\phi''(\bar{x}))\d_x\|_\tpi 
\ls \eps\|\Phi''(x)^{\frac{1}{2} }\d_x\|_\tpi \ls \eps\|g\|_\tpi \ls \eps\gamma
\end{align*}
by Lemma \ref{lemma:xchange}. For the fourth/last term, we use Lemma \ref{lemma:normbound} to get
\begin{align*}
\|\Phi''(x)^{-\frac{1}{2} }\Tau^{-1}\Phi''(\bar{x})^\frac{1}{2} \bar{\Tau}\d_2\|_\tpi \ls \|\d_2\|_\tpi \ls \cnorm\|\mA^\top x-b\|_{(\mA^\top\Tau^{-1}\Phi''(x)^{-1}\mA)^{-1}} \ls \eps\gamma, 
\end{align*} 
where we have used that our point is $\eps$-centered.

For the first term, we write
\begin{align*}
& ~(\mI-\Phi''(x)^{-\frac{1}{2} }\Tau^{-1}\Phi''(\bar{x})^\frac{1}{2} \bar{\Tau})g \\
= & ~ (\mI-\Phi''(x)^{-\frac{1}{2} }\Phi''(\bar{x})^\frac{1}{2}  + \Phi''(x)^{-\frac{1}{2} }\Tau^{-1}\Phi''(\bar{x})^\frac{1}{2} (\Tau-\bar{\Tau}))g \\
= & ~ (\Phi''(x)^{-\frac{1}{2} }(\Phi''(\bar{x})^\frac{1}{2} -\Phi''(x)^\frac{1}{2} ) + \Phi''(x)^{-\frac{1}{2} }\Tau^{-1}\Phi''(\bar{x})^\frac{1}{2} (\Tau-\bar{\Tau}))g.
\end{align*}
For the first term in the previous expression, we can use the approximation condition of Invariant \ref{invar}, and Lemma \ref{lemma:selfcon} to get
\begin{align*} 
\|\Phi''(x)^{-\frac{1}{2} }(\Phi''(\bar{x})^\frac{1}{2} -\Phi''(x)^\frac{1}{2} )g\|_\tpi \ls \eps\|g\|_\tpi \le \eps\gamma. 
\end{align*}
For the second term, we can use the approximation condition of Invariant \ref{invar}, Lemma \ref{lemma:selfcon}, and Lemma \ref{lemma:lewisapprox} to get
\begin{align*} 
\|\Phi''(x)^{-\frac{1}{2} }\Tau^{-1}\Phi''(\bar{x})^\frac{1}{2} (\bar{\Tau}-\Tau))g\|_\tpi \ls \eps\|g\|_\tpi \le \eps\gamma. 
\end{align*}
Summing over these bounds gives the desired result.
\end{proof}
Combining the above analyses allows us to analyze $\d_y$.
\begin{lemma}[Bounds on $\d_y$]
\label{lemma:ychange}
Let $x, \tau, \bar{x}, \bar{\tau}$ be as in Algorithm \ref{algo:lsstep}, and let 
\begin{align*} 
y^\new = s^\new+\mu^\new\tau(x^\new)\phi'(x^\new). 
\end{align*} 
Then we have that
\begin{align*} 
\|\mu^{-1}\Phi''(x)^{-\frac{1}{2} }\Tau^{-1}\E[\d_y] - g\|_\tpi \le p\gamma + O((\eps+1/\cnorm)\gamma)
\end{align*}
and
\begin{align*}
\|\E[(\mu^{-1}\Phi''(x)^{-\frac{1}{2} }\Tau^{-1}\d_y)^2]\|_\tpi \ls \gamma^2. 
\end{align*}
\end{lemma}
\begin{proof}
We start with the first claim. Let $\d_{\phi'} = \Phi'(x^\new)-\Phi'(x)$.
\begin{align*}
y^\new = & ~ s^\new + \mu^\new\tau^\new \phi'(x^\new) \\
= & ~ s+\bar{\d}_s+\mu^\new\tau \phi'(x^\new) + \mu^\new\d_\tau \phi'(x^\new) \\ 
= & ~ s+\bar{\d}_s+\mu\tau \phi'(x^\new) + \mu^\new\d_\tau \phi'(x^\new) + \d_\mu\tau \phi'(x^\new)\\
= & ~ s+\bar{\d}_s+\mu\tau \phi'(x) + \mu^\new\d_\tau \phi'(x^\new) + \d_\mu\tau \phi'(x^\new) + \mu\tau\d_{\phi'}\\ 
= & ~ y+(\bar{\d}_s+\mu\tau\d_{\phi'}) + \mu^\new\d_\tau \phi'(x^\new) + \d_\mu\tau \phi'(x^\new),
\end{align*}
where for simplicity of notation we have defined vector multiplication coordinate-wise.
Also for $x_t = x + t \bar{\d}_x$ we have that
\begin{align*} \d_{\phi'} = \phi''(x)\bar{\d}_x + \int_0^1 (1-t)\phi'''(x_t)\bar{\d}_x^2 \mathrm{d}t. \end{align*}
Therefore, we have that
\begin{align*} 
\d_y = (\bar{\d}_s+\mu\tau\phi''(x)\bar{\d}_x) + \mu^\new\d_\tau \phi'(x^\new) + \d_\mu\tau \phi'(x^\new) + \mu\tau\int_0^1 (1-t)\phi'''(x_t)\bar{\d}_x^2 \mathrm{d} t. 
\end{align*}
We analyze the four terms in the above term one by one.
The first term can be handled by using Lemma \ref{lemma:approxsolution}. Precisely, we have that
\begin{align*} 
\|\mu^{-1}\Phi''(x)^{-\frac{1}{2} }\Tau^{-1}(\bar{\d}_s+\mu\tau\phi''(x)\E[\bar{\d}_x])-g\|_\tpi \ls \eps\gamma. 
\end{align*}
For the second term, we first rewrite
\begin{align*}
& ~ \|\mu^{-1}\Phi''(x)^{-\frac{1}{2} }\E[\phi'(x^\new)\Tau^{-1}\mu^\new\d_\tau]\|_\tpi \\
\le & ~ (1+r)\|\Phi''(x)^{-\frac{1}{2} }\E[\phi'(x^\new)\Tau^{-1}\d_\tau]\|_\tpi \\
\le & ~ (1+r)(\|\Phi''(x)^{-\frac{1}{2} }\Phi'(x)\E[\Tau^{-1}\d_\tau]\|_\tpi+\|\E[\Phi''(x)^{-\frac{1}{2} }\d_{\phi'}\Tau^{-1}\d_\tau]\|_\tpi).
\end{align*}
For the first of these, we can write
\begin{align*}
(1+r)\|\Phi''(x)^{-\frac{1}{2} }\Phi'(x)\E[\Tau^{-1}\d_\tau]\|_\tpi 
\le & ~ (1+r)\|\E[\Tau^{-1}\d_\tau]\|_\tpi \\
\le & ~ (1+r) ( \| \E[\Tau^{-1} \mJ \d_c] \|_{\tpi} + \| \E[\Tau^{-1}\d_\tau] - \E[\Tau^{-1} \mJ \d_c] \|_{\tpi}  ) \\
\le & ~(1+r) ( \| \E[\Tau^{-1} \mJ \d_c] \|_{\tpi} + O(\gamma^2)) \\
= & ~ (1+r)(\|\Tau^{-1}\mJ\mC\mC^{-1}\E[\d_c]\|_\tpi+O(\gamma^2)) \\
\le & ~ (1+O(r+1/\cnorm))p\|\mC^{-1}\E[\d_c]\|_\tpi + O(\gamma^2) \\
\le & ~ (1+O(r+\eps+1/\cnorm))p\gamma + O(\gamma^2) \\
\le & ~ p\gamma + O(\eps+1/\cnorm)\gamma,
\end{align*}
where the first step follows from $1$-self-concordance (Definition \ref{def:highselfcon}), the second step follows from triangle inequality, the third step follows from Part 3 of Lemma \ref{lemma:tauchange}, the fifth step follows from Lemma \ref{lemma:matrixbound} Part 3, the sixth step follows from Part 2 of Lemma \ref{lemma:pchange}, and the last step follows from the choice of parameters.

For the second, we can write
\begin{align} 
\|\E[\Phi''(x)^{-\frac{1}{2} }\d_{\phi'}\Tau^{-1}\d_\tau]\|_\tpi 
\ls & ~ \|\E[(\Tau^{-1}\d_\tau)^2]\|_\tpi + \|\E[\Phi''(x)^{-1}\d_{\phi'}^2]\|_\tpi \label{eq:boundphi1} \\
\ls & ~ \gamma^2 + \|\E[\Phi''(x)^{-1}\d_{\phi'}^2]\|_\tpi \nonumber \\
= & ~ \gamma^2 + \left\|\E\left[\Phi''(x)^{-1}\left(\int_0^1 \Phi''(x_t)\bar{\d}_x \mathrm{d}t\right)^2\right]\right\|_\tpi \nonumber \\
\ls & ~ \gamma^2 + \int_0^1 \left\|\E[\Phi''(x)^{-1}\Phi''(x_t)^2\bar{\d}_x^2]\right\|_\tpi \mathrm{d}t \nonumber \\
\ls & ~ \gamma^2 + \|\E[\Phi''(x)\bar{\d}_x^2]\|_\tpi \ls \gamma^2, \label{eq:boundphi2}
\end{align}
where the first step follows from AM-GM and the triangle inequality, the second step follows from Lemma \ref{lemma:tauchange} Part 2, the fourth step follows by Cauchy-Schwarz, the fifth step follows from $\Phi''(x) \approx_1 \Phi''(x_t)$, and the final step follows from Lemma \ref{lemma:xchange} Part 3.

Therefore, the total for the second term is at most $p\gamma + O(\eps+1/\cnorm)\gamma,$ as $\gamma \le \eps$.

For the third term, we can bound
\begin{align*}
\|\mu^{-1}\Phi''(x)^{-\frac{1}{2} }\Tau^{-1}\d_\mu\tau \phi'(x^\new)\|_\tpi
= & ~ \mu^{-1}\d_\mu\|\Phi''(x^\new)^{-\frac{1}{2} } \phi'(x^\new)\|_\tpi \\
\ls & ~ r \|\vec{1}\|_\tpi \\
\ls & ~ r \cnorm n^\frac{1}{2}  \ls \eps\gamma.
\end{align*}
where the second step uses 1-self-concordance (Definition \ref{def:highselfcon}), the third step uses $|\mu^{-1}\d_\mu| \ls r$ by the choice of $r$, and the final step uses $\|\vec{1}\|_\tpi = 1 + \cnorm\sqrt{\|v\|_1} \ls \cnorm n^\frac{1}{2}$ by the choice of parameters.

For the fourth term, we have
\begin{align*}
\left\|\mu^{-1}\Phi''(x)^{-\frac{1}{2} }\Tau^{-1}\mu\tau\int_0^1 (1-t)\E[\phi'''(x_t)\bar{\d}_x^2] \mathrm{d} t\right\|_\tpi 
\ls & ~ \int_0^1 \|\E[\Phi''(x_t)^{-\frac{1}{2}}|\Phi'''(x_t)|\bar{\d}_x^2]\|_\tpi \\
\ls & ~ \int_0^1 \|\E[\Phi''(x_t)\bar{\d}_x^2]\|_\tpi \ls \gamma^2
\end{align*}
where the first step follows from $\Phi''(x_t) \approx_1 \Phi''(x)$ and the triangle inequality, the second step follows from $1$-self-concordance (Definition \ref{def:highselfcon}) and the final step follows from Lemma \ref{lemma:xchange} Part 3.

Combining everything gives us that
\begin{align*} 
\|\mu^{-1}\Phi''(x)^{-\frac{1}{2} }\Tau^{-1}\d_y - g\|_\tpi - p\gamma \ls (\eps+1/\cnorm)\gamma. 
\end{align*}

Now, we move on to the second claim. Once again, we write
\begin{align*}
\d_y =  \bar{\d}_s+\mu\tau\d_{\phi'} + \mu^\new\d_\tau \phi'(x^\new) + \d_\mu\tau \phi'(x^\new)
\end{align*} 
and note that
\begin{align*}
\d_y^2 &\ls \bar{\d}_s^2 + (\mu\tau\d_{\phi'})^2 + (\mu^\new\d_\tau \phi'(x^\new))^2 + (\d_\mu\tau \phi'(x^\new))^2.
\end{align*}
We analyze it term by term. Lemma \ref{lemma:xchange} we have
\begin{align*} 
\|(\mu^{-1}\Phi''(x)^{-\frac{1}{2} }\Tau^{-1}\bar{\d}_s)^2\|_\tpi \le \|\mu^{-1}\Phi''(x)^{-\frac{1}{2} }\Tau^{-1}\bar{\d}_s^2\|_\tpi^2 \ls \gamma^2. 
\end{align*}
By (\ref{eq:boundphi1}) and (\ref{eq:boundphi2}) above we have
\begin{align*} 
\|\E(\mu^{-1}\Phi''(x)^{-\frac{1}{2} }\Tau^{-1}\mu\tau\d_{\phi'})^2\|_\tpi \le \|\E[\Phi''(x)^{-1}\d_{\phi'}^2]\|_\tpi \ls \gamma^2. 
\end{align*}
By Lemma \ref{lemma:pchange}, $1$-self-concordance, and Lemma \ref{lemma:tauchange} we have
\begin{align*} 
\|\E(\mu^{-1}\Phi''(x)^{-\frac{1}{2} }\Tau^{-1}\mu^\new\d_\tau \phi'(x^\new))^2\|_\tpi \ls \|\E(\Tau^{-1}\d_\tau)^2\|_\tpi \ls \gamma^2. 
\end{align*}
By Lemma \ref{lemma:muchange}, $1$-self-concordance, and Lemma \ref{lemma:pchange} we have
\begin{align*} 
\|\E(\mu^{-1}\Phi''(x)^{-\frac{1}{2} }\Tau^{-1}\d_\mu\tau \phi'(x^\new))^2\|_\tpi \ls r^2\|\vec{1}\|_\tpi \ls r^2\cnorm n^\frac{1}{2}  = r\eps\gamma \ls \gamma^2 
\end{align*} 
by the choice of $r$.
Combining these gives the desired result.
\end{proof}

\subsection{Feasibility and potential function analysis}
\label{subsec:together}

\begin{table}[h]
    \centering
    \begin{tabular}{|l|l|l|l|} \hline
        {\bf Section} & {\bf Statement} & {\bf Comment} \\ \hline
        Section~\ref{subsec:together} & Lemma~\ref{lemma:feasibility} & Feasibility bound \\ \hline
        Section~\ref{subsec:together} & Lemma~\ref{lemma:firstpotentialdrop} & First potential drop \\ \hline 
        Section~\ref{subsec:together} & Corollary~\ref{cor:finaldrop} & Final potential drop \\ \hline\hline
        Section~\ref{subsec:sampling} & Lemma~\ref{lemma:independent} & Independent sampling \\ \hline
        Section~\ref{subsec:sampling} & Lemma~\ref{lemma:prop} & Proportional sampling \\ \hline
        Section~\ref{subsec:sampling} & Corollary~\ref{cor:sampling_by_a_mixture} & Sampling by a mixture of $\ell_2$ and uniform. \\ \hline \hline
        Section~\ref{subsec:additional} & Lemma~\ref{lemma:morestablex} & Nearby stability of $x$ \\ \hline
        Section~\ref{subsec:additional} & Lemma~\ref{lemma:morestablerest} & Nearby stability of $\phi''$ and $\tau$ \\ \hline
        Section~\ref{subsec:additional} & Lemma~\ref{lemma:paramchange} & Parameter changes along the central path \\ \hline
    \end{tabular}
    \caption{Summary of Section~\ref{subsec:together}, \ref{subsec:sampling}, \ref{subsec:additional}. }
    \label{tab:my_label}
\end{table}

In this section, we analyze the change in feasibility and centrality potential. We start with the feasibility.
\begin{lemma}[Feasibility bound]
\label{lemma:feasibility}
For sufficiently large constants $C$ in Algorithm \ref{algo:lsstep} we have the following. Let $x^\new, s^\new, \mu$ be as in Algorithm \ref{algo:lsstep}, where $\|g\|_\tpi \le \gamma$ and $(x,s,\mu)$ is $\eps$-centered.
Then with probability at least $1-m^{-10}$ we have
\begin{align*} 
\| \ma^\top x^\new - b \|_{(\mA^\top(\Tau(x^\new)\Phi''(x^\new))^{-1}\mA)^{-1}} \le .5\eps\gamma/\cnorm. 
\end{align*}
\end{lemma}
\begin{proof}
Define $\bar{\mA} = \Bar{\Tau}^{-\frac{1}{2} }\Phi''(\bar{x})^{-\frac{1}{2} }\mA$, and note that 
\begin{align*} 
\mA^\top(\Tau(x^\new)\Phi''(x^\new))^{-1}\mA \approx_2 \bar{\mA}^\top \bar{\mA} 
\end{align*} 
by Lemma \ref{lemma:pchange}, Lemma \ref{lemma:tauchange}, and $\mH \approx_\gamma \bar{\mA}^\top \bar{\mA}$ by definition. Define the vector $v = \mA^\top \Phi''(\bar{x})^{-\frac{1}{2} }g + (\mA^\top x - b).$ A direct calculation shows that
\begin{align*} 
\mA^\top x^\new - b = (\mI-\mA^\top\Phi''(\bar{x})^{-\frac{1}{2} }\Bar{\Tau}^{-\frac{1}{2} }\mR\Bar{\Tau}^{-\frac{1}{2} }\Phi''(\bar{x})^{-\frac{1}{2} }\mA\mH^{-1})v = (\mI-\bar{\mA}^\top \mR \bar{\mA}\mH^{-1})v. 
\end{align*}
Therefore, we have by Lemma \ref{lemma:normbound} that
\begin{align*}
&\|\mA^\top x^\new - b\|_{(\mA^\top(\Tau(x^\new)\Phi''(x^\new))^{-1}\mA)^{-1}} \ls \|\mH^{-\frac{1}{2} }(\mI-\bar{\mA}^\top \mR \bar{\mA}\mH^{-1})\mH^\frac{1}{2} (\mH^{-\frac{1}{2} }v)\|_2 \\
&\le \|\mH^{-\frac{1}{2} }(\mH-\bar{\mA}^\top \mR \bar{\mA})\mH^{-\frac{1}{2} }\|_2 \|\mH^{-\frac{1}{2} }v\|_2.
\end{align*}
We have that $\bar{\mA}^\top \mR \bar{\mA} \approx_\gamma \bar{\mA}^\top \bar{\mA} \approx_\gamma \mH$ with probability $1-m^{-10}$ by the (Matrix approximation) condition of Definition \ref{def:validdistro} and line \ref{line:ipm:H} of Algorithm \ref{algo:lsstep}. Therefore, by \cite[Lemma 4.30]{BrandLN+20} we have that
\begin{align*} 
\|\mH^{-\frac{1}{2} }(\mH-\bar{\mA}^\top \mR \bar{\mA})\mH^{-\frac{1}{2} }\|_2 \ls \gamma.
\end{align*}
Also
\begin{align*}
\|\mH^{-\frac{1}{2} }v\|_2 
\le & ~ \|\mH^{-\frac{1}{2} }\mA^\top \Phi''(\bar{x})^{-\frac{1}{2} }g\|_2 + \|\mH^{-\frac{1}{2} }(\mA^\top x-b)\|_2 \\
\ls & ~ \|g\|_\tau + \eps\gamma/\cnorm \ls \gamma/\cnorm.
\end{align*}
Therefore with probability $1-m^{-10}$ we have
\begin{align*} 
\|\mA^\top x^\new - b\|_{(\mA^\top(\Tau(x^\new)\Phi''(x^\new))^{-1}\mA)^{-1}} \ls \gamma \cdot \gamma/\cnorm \le \gamma^2/\cnorm. 
\end{align*}
Let $C_2$ be the universal constant such that
\begin{align*} 
\|\mA^\top x^\new - b\|_{(\mA^\top(\Tau(x^\new)\Phi''(x^\new))^{-1}\mA)^{-1}} \le C_2\gamma^2/\cnorm. 
\end{align*} 
Now, we can choose $C \ge 2C_2$, so that $\gamma \le \eps/C \le \eps/(2C_2)$. Then
\begin{align*} 
\|\mA^\top x^\new - b\|_{(\mA^\top(\Tau(x^\new)\Phi''(x^\new))^{-1}\mA)^{-1}} \le C_2\gamma^2/\cnorm \le .5\eps\gamma/\cnorm 
\end{align*}
 as desired.
\end{proof}

We first bound the effect of the step in Algorithm \ref{algo:lsstep} on the centrality potential.
\begin{lemma}[First potential drop]
\label{lemma:firstpotentialdrop}
Let $x^\new, s^\new, \mu$ be as in Algorithm \ref{algo:lsstep}. For sufficiently large choice of $C$, we have that for $v = \frac{s_i+\mu \tau(x)_i\phi_i'(x_i)}{\mu\tau(x)_i\sqrt{\phi_i''(x_i)}}$ and $\alpha = 1-p$ we have
\begin{align*} 
\E[\Psi(x^\new, s^\new, \mu^\new)] &\le \Psi(x, s, \mu) + \psi'(v)^\top g \\ &+ (1-\alpha/4)\|\psi'(v)\|_\tpi^*\gamma + O(\|\psi''(v)\|_\tpi^*\gamma^2). 
\end{align*}
\end{lemma}
\begin{proof}
We carefully apply Lemma \ref{lemma:potentialhelper} for the choices $u^{(j)} = \mu, \tau, p$ respectively and $y$ as itself.
Note that 
\begin{align*} 
\|\mC^{-1}\E[\d_c]\|_\tpi \ls \gamma \enspace \text{ and } \|\Tau^{-1}\E[\d_\tau]\|_\tpi \ls \gamma 
\end{align*} 
by Lemma \ref{lemma:pchange} and Lemma \ref{lemma:tauchange} respectively. Also, by Lemma \ref{lemma:muchange}
\begin{align*} 
\|\mu^{-1}\d_\mu\|_\tpi \ls \eps\gamma. 
\end{align*}
We now bound each term in the conclusion of Lemma \ref{lemma:potentialhelper}. We start with (\ref{eq:firstline}). For the $\psi'(v)^\top \mW\eta$ term we bound
\begin{align}
\E[\psi'(v)^\top \mu^{-1}\Phi''(x)^{-\frac{1}{2} }\Tau^{-1}\d_y] 
= & ~ \psi'(v)^\top g + \psi'(v)^\top\left(\mu^{-1}\Phi''(x)^{-\frac{1}{2} }\Tau^{-1} \E[\d_y] - g\right) \notag \\
\le & ~ \psi'(v)^\top g + \|\psi'(v)\|_\tpi^* \|\mu^{-1}\Phi''(x)^{-\frac{1}{2} }\Tau^{-1} \E[\d_y] - g\|_\tpi \notag \\
\le & ~ \psi'(v)^\top g + \|\psi'(v)\|_\tpi^* \left(p\gamma + O(\eps+1/\cnorm)\gamma\right) \label{eq:controltheO1} \\
\le & ~ \psi'(v)^\top g + (1-\alpha/2)\|\psi'(v)\|_\tpi^* \gamma \notag
\end{align}
for sufficiently large $C$. Indeed, let $C_3$ be the universal constant such that the bound in (\ref{eq:controltheO1}) is
\begin{align*} 
\psi'(v)^\top g + \|\psi'(v)\|_\tpi^* \left(p\gamma + C_3(\eps+1/\cnorm)\gamma\right). 
\end{align*} 

We can then choose $C \ge 10C_3$ so that $\eps = \alpha/C = 1/\cnorm.$

 For the terms $\psi'(v)^\top \mV(\mU^{(j)})^{-1}\d^{(j)}$ we can bound them as for example
\begin{align*}
\left|\psi'(v)^\top \mV\Tau^{-1}\E[\d_\tau]\right| 
\ls & ~ \|\mV\psi'(v)\|_\tpi^* \|\Tau^{-1}\E[\d_\tau]\|_\tpi \\
\ls & ~ \eps \|\psi'(v)\|_\tpi^* \gamma \\
\le & ~ \frac{\alpha}{12} \|\psi'(v)\|_\tpi^* \gamma
\end{align*}
because $\|v\|_\infty \le \eps$, and that $\eps = \alpha/C$ for a sufficiently large constant $C$. Similar bounds hold for the contributions from $\d_c$ and $\d_\mu$.

Now we bound \eqref{eq:secondline}. First, the contribution of the $16\|\mW\eta\|_{\psi''(v)}^2$ term is at most
\begin{align*}
16\E\left[\|\mu^{-1}\Phi''(x)^{-\frac{1}{2} }\Tau^{-1}\d_y\|_{\psi''(v)}^2\right]  
\ls & ~ \|\psi''(v)\|_\tpi^*\|\E[(\mu^{-1}\Phi''(x)^{-\frac{1}{2} }\Tau^{-1}\d_y)^2]\|_\tpi \\
\ls & ~ \|\psi''(v)\|_\tpi^*\gamma^2
\end{align*}
by Lemma \ref{lemma:ychange}. In \eqref{eq:secondline} we know that $1+\|c\|_1 = 4.$ For the second term in \eqref{eq:secondline} we have
\begin{align*}
\E[\|\Tau^{-1}\d_\tau\|_{\psi''(v)}^2] \le \|\psi''(v)\|_\tpi^*\|\E[(\Tau^{-1}\d_\tau)^2]\|_\tpi \ls \|\psi''(v)\|_\tpi^*\gamma^2
\end{align*}
by Lemma \ref{lemma:tauchange}, and achieve bounds of
\begin{align*} 
\E[\|\mC^{-1}\d_c\|_{\psi''(v)}^2 ] \le \|\psi''(v)\|_\tpi^*\|\E[(\mC^{-1}\d_c)^2]\|_\tpi \ls \|\psi''(v)\|_\tpi^*\gamma^2 
\end{align*}
and
\begin{align*} 
\E[\|\mM^{-1}\d_\mu\|_{\psi''(v)}^2] \ls \|\psi''(v)\|_\tpi^*\gamma^2 
\end{align*} 
similarly, using Lemma \ref{lemma:pchange} and Lemma \ref{lemma:muchange}.
Therefore, the total contribution from \eqref{eq:secondline} and \eqref{eq:thirdline} is $\ls \|\psi''(v)\|_\tpi^*\gamma^2$, where we have used that $|\psi'(v)| \le \psi''(v)$ on all coordinates.
Summing all the previous bounds gives the desired result.
\end{proof}
Applying Lemma \ref{lemma:firstpotentialdrop} allows us to show that the potential is bounded in expectation.
\begin{corollary}
\label{cor:finaldrop}
In the notation of Lemma \ref{lemma:firstpotentialdrop} we have for sufficiently large $C$ that
\begin{align*} 
\E[\Psi(x^\new, s^\new, \mu^\new)] \le \left(1 - \frac{\alpha^2\lambda\gamma}{32C\sqrt{n}}\right)\Psi(x, s, \mu) + m. 
\end{align*}
\end{corollary}
\begin{proof}
We use \cite[Lemma 4.36]{BrandLN+20} and verify that the guarantees of Lemma \ref{lemma:firstpotentialdrop} satisfy the hypotheses.
In that notation, we have that $(1-c_1) = \alpha/4$, $\delta = \gamma/10$, and $c_2$ as the implicit constant in the $O(\|\psi''(v)\|_\tpi^*\gamma^2)$ of Lemma \ref{lemma:firstpotentialdrop}. Note that $\gamma = \frac{\eps}{C\lambda} = \frac{\alpha}{C^2\gamma}$. Therefore, for sufficiently large $C$ we have
\begin{align*} 
2\lambda\delta + c_2\lambda\gamma \le \alpha/8 = .5(1-c_1). 
\end{align*}
Now, $u = \frac{1}{4\cnorm\sqrt{n}} = \frac{\alpha}{4C\sqrt{n}}$, as $\|\vec{1}\|_\tpi \le 2\cnorm(n+\|v\|_1)^\frac{1}{2}  \le 4\cnorm\sqrt{n}$. 
Therefore,
\begin{align*} 
.5(1-c_1)\lambda\gamma u = \frac{\alpha^2\lambda\gamma}{32C\sqrt{n}}. 
\end{align*} as desired.
\end{proof}

We have the necessary lemmas to show Lemma \ref{lemma:pathfollowing}.
\begin{proof}[Proof of Lemma \ref{lemma:pathfollowing}]
The iteration complexity is clear by the definition of $r$ in Algorithm \ref{algo:lsstep}. The conditions on $x^\final$ follow from Lemma \ref{lemma:finalpoint}.

We proceed by induction. (Slack feasibility) follows by the fact that $\mA^\top \bar{d}_s = 0$ in Algorithm \ref{algo:lsstep}. (Approximate feasibility) follows by induction and Lemma \ref{lemma:feasibility} with probability at least $1-m^{-10}$ per step. To show the (Potential function) bound, we first verify the base case. Indeed, because $(x^\init, s^\init, \mu)$ is $\eps/\cstart$-centered, we know that
\begin{align*} \Psi(x, s, \mu) \le m \exp(\lambda\eps/\cstart) \le m^2 \end{align*} for sufficiently large $\cstart$ compared to $C$. The inductive step follows from Corollary \ref{cor:finaldrop}, and that $\frac{32Cm\sqrt{n}}{\alpha^2\lambda\gamma} \le m^2$ for sufficiently large $m, n$. Therefore, by Markov's inequality, with probability $1-m^{-7}$ we have that $\E[\Psi(x, s, \mu)] \le m^{10}$ for all steps. Now, as $\exp(\lambda\eps) > m^{10}$, we know that $(x, s, \mu)$ will be $\eps$-centered.
\end{proof}

\subsection{Sampling Schemes}
\label{subsec:sampling}
In this section, we analyze two sampling schemes. All proofs are deferred to \Cref{subsec:proofssampling}. The first samples each coordinate independently, and is efficiently implementable in the graphical setting.
\begin{restatable}[Independent sampling]{lemma}{independent}
\label{lemma:independent}
Let vector $q \in \R^m_{\ge0}$ satisfy 
\begin{align*}
q_i \ge \cvalid^2\gamma^{-1}|(\d_r)_i|+\csample\sigma(\bar{\Tau}^{-\frac{1}{2} }\Phi''(\bar{x})^{-\frac{1}{2} }\mA)_i\log(m)\gamma^{-2}
\end{align*}
for sufficiently large $\csample$. Then picking $\mR_{ii} = 1/\min(q_i,1)$ with probability $\min(q_i,1)$ and $0$ otherwise is a $\cvalid$-valid (Definition \ref{def:validdistro}).
\end{restatable}

A different sampling scheme is sampling proportional to weights $q_i$ -- this is useful for general linear programs, and is implemented in \Cref{sec:matrix_data_structures}.
\begin{restatable}[Proportional sampling]{lemma}{prop}
\label{lemma:prop}
Let vector $q \in \R^m_{\ge0}$ satisfy 
\begin{align*}
q_i \ge |(\d_r)_i| + \sigma(\bar{\Tau}^{-\frac{1}{2} }\Phi''(\bar{x})^{-\frac{1}{2} }\mA)_i .
\end{align*}
Let $S \ge \sum_i q_i$. Let $X$ be a random variable which equals $q_i^{-1}e_i$ ($e_i$ is the standard basis vector) with probability $q_i/S$ for all $i$, and $\vec{0}$ otherwise. For $C_0 = 100\cvalid^4\gamma^{-2}\log(m)$ let $\mR = C_0^{-1}\sum_{j=1}^{C_0S} X_j$, where $X_j$ are i.i.d. copies of $X$. Then $\mR$ is a $\cvalid$-valid distribution (Definition \ref{def:validdistro}).
\end{restatable}

It is convenient for our sampling data structures to sample by the $\ell_2$ norm in $\d_r$, instead of the $\ell_1$ norm as described in Lemma \ref{lemma:independent} and \ref{lemma:prop}. We can achieve this by the following observation.
\begin{restatable}[Sampling by a mixture of $\ell_2$ and uniform]{corollary}{mixture}
\label{cor:sampling_by_a_mixture}
\label{cor:l2_sampling}
Let $C_1, C_2, C_3$ be constants such that $C_3 \ge 4\csample$ and $C_1C_2 \ge \cvalid^4\gamma^{-2}$ and
\[ p_i = C_1 \sqrt{n}(\delta_r)_i^2 + C_2/\sqrt{n} + C_3 \tau_i\gamma^{-2}\log m.
\]
Then $p_i \ge q_i$ in each of Lemma \ref{lemma:independent} and \ref{lemma:prop}. Hence replacing $q_i$ with $p_i$ in Lemma \ref{lemma:independent} and \ref{lemma:prop} and sampling acoordingly gives a valid distribution (Definition \ref{def:validdistro}).
Additionally
\begin{align}
\sum_{i \in [m]} p_i \le \left((C_1+C_2)\frac{m}{\sqrt{n}} + C_3n\gamma^{-2}\log m \right). \label{eq:sumbound}
\end{align}
\end{restatable}

\subsection{Additional Properties of the IPM}
\label{subsec:additional}
Even though random sampling may make the sequence of points $x$ and weights $\tau$ change more rapidly, we can still argue that there is a nearby sequence of points $\hx, \hat{\tau}$ that is more stable. We start with the stability of $\hx.$ The proofs of the following lemmas are given in \Cref{subsec:proofsadditional}.
\begin{restatable}[Nearby stability of $x$]{lemma}{morestablex}
\label{lemma:morestablex}
Suppose that $\mR$ is sampled from a $\cvalid$-valid distribution for $\cvalid \ge \beta^{-2}\log(mT)$ where $\beta \in (0, \gamma)$. Let $(x^{(k)}, s^{(k)})$ for $k \in [T]$ be the sequence of points found by Algorithm \ref{algo:lsstep}. With probability $1-m^{-10}$, there is a sequence of points $\hx^{(k)}$ from $1 \le k \le T$ such that
\begin{itemize}
\item $\|\Phi''(x^{(k)})^\frac{1}{2} (\hx^{(k)}-x^{(k)})\|_\infty \le \beta/2.$
\item $\|\Phi''(\hx^{(k)})^\frac{1}{2} (\hx^{(k)}-x^{(k)})\|_\infty \le \beta$.
\item $\|\Phi''(x^{(k)})^\frac{1}{2} (\hx^{(k+1)}-\hx^{(k)})\|_\tkpi \le 2\gamma$.
\end{itemize}
\end{restatable}
The stable sequence $\hx$ induces corresponding stable sequences for $\phi''$ and $\tau$.
\begin{restatable}[Nearby stability of $\phi''$ and $\tau$]{lemma}{morestablerest}
\label{lemma:morestablerest}
In the setup of Lemma \ref{lemma:morestablex}, and $\hw^{(k)} = \Phi''(\hx^{(k)})^{-\frac{1}{2} }\tau(\hx^{(k)})^{\frac{1}{2} -\frac{1}{p} }$, we have the following.
\begin{itemize}
\item $\|\Phi''(\hx^{(k)})^\frac{1}{2} (\phi''(\hx^{(k)})^{-\frac{1}{2} } - \phi''(x^{(k)})^{-\frac{1}{2} })\|_\infty \le \beta.$
\item $\|\Phi''(\hx^{(k)})^\frac{1}{2} (\phi''(\hx^{(k+1)})^{-\frac{1}{2} }-\phi''(\hx^{(k)})^{-\frac{1}{2} })\|_\tkpi \ls \gamma.$
\item $\|\Tau(\hx^{(k)})^{-1}(\tau(\hx^{(k+1)}) - \tau(\hx^{(k)}))\|_\tkpi \ls \gamma$.
\item $\|(\hat{\mW}^{(k)})^{-1}(\hw^{(k+1)} - \hw^{(k)})\|_\tkpi \ls \gamma$.
\end{itemize}
\end{restatable}

To bound the bit complexity that our algorithms must maintain throughout, we show that the Hessian of points $x$ encountered along the central path is bounded by polynomial factors in $n, m,$ and $\mu^\final/\mu^\init.$
\begin{restatable}[Parameter changes along central path]{lemma}{paramchange}
\label{lemma:paramchange}
For $\mA \in \R^{m \times n}, b \in \R^n, c \in \R^m$ and $\ell, u \in \R^m$ assume that the point $x^\init = (\ell+u)/2$ is feasible, i.e. $\mA^\top x^\init = b$. Let $W$ be the ratio of the largest to smallest entry of $\phi''(x^\init)^{1/2}$, and let $W'$ be the ratio of the largest to smallest entry of $\phi''(x)^{1/2}$ encountered in Algorithm \ref{algo:pathfollowing}. Then
\[ \log W' = \tO\left(\log W + \log(1/\mu^\final) + \log \|c\|_\infty \right). \]
\end{restatable} 	   %

\section{Maintaining Regularized Lewis-Weights}
\label{sec:lewis_weight_maintenance}

In this section we show how to efficiently maintain an approximation 
of the regularized $\ell_p$-Lewis weight of $\mG \mA$
under updates to $\mG = \mdiag(g)$ for $p\in [1/2,2)$.
Lewis weights are a generalization of leverage scores,
i.e. for $p=2$ the two concepts are identical. Correspondingly, 
we obtain our regularized Lewis weight data structure by reducing to a 
regularized leverage score data structure
presented in \Cref{sec:leverage_score_maintenance}. Our exact result for maintaining regularized Lewis weights is as follows.

\begin{theorem}\label{thm:lewis_weight_maintenance}
Assume there exists a $(P,c,Q)$-\textsc{HeavyHitter} data structure (\Cref{def:heavyhitter})
and let $z \ge n\cdot c / \|c\|_1 + n/m$.
Let $\tau(\mG \mA) \in \R^m$ such that $\tau(\mG \mA) = \sigma(\tau(\mG \mA)^{1/2-1/p} \mG \mA) + z$ for $p \in (0,2]$,
i.e. $\tau(\mG \mA)$ are regularized Lewis weights of $\mG \mA$.
There exists a Monte-Carlo data-structure (\Cref{alg:lewis_weight_maintenance}), 
that works against an adaptive adversary, 
with the following procedures: 
\begin{itemize}
\item \textsc{Initialize}$(\mA \in \R^{m \times n}, g \in \R_{\ge 0}^m, z\in\R_{>0}^m, p \in [1/2,2), 
	\delta >1, \epsilon \in (0,\frac{1}{2^{10}\delta\cdot\log n}])$:
	The data structure initializes for the given matrix $\mA \in \R^{m \times n}$, 
	scaling $g \in \R_{\ge0}^m$,
	regularization parameter $z \in \R^m$,
	Lewis weight parameter $p \in (0,2]$, 
	and target accuracy $\epsilon \in (0, \frac{1}{2^{10}\delta\cdot\log n}]$.
	The parameter $\delta$ is a bound on how much the vector $g$ is allowed to change per iteration:
	Let $g\t$ be the vector $g$ during the $t$-th call to \textsc{Query} (with $g^{(0)}$ during \textsc{Initialization}),
	then we assume that $g\t \approx_{\delta\epsilon} g^{(t-1)}$ for all $t$.
\item \textsc{Scale}$(i \in [m], b \in \R_{\ge0})$: 
	Sets $g_{i} \leftarrow b$.
\item \textsc{Query}$()$: 
	W.h.p.~in $n$ the data structure outputs a vector $\otau \in \R^m$
	with the property $\otau \approx_{\epsilon} \sigma(\otau^{1/2-1/p} \mG \mA) + z$
	and therefore $\otau \approx_{\epsilon} \tau(\mG \mA)$.
	The vector $\otau$ is returned as a pointer and the data structure also returns a set $I \subset [m]$ of indices $i$ where $\otau_i$ has changed compared to the last call to \textsc{Query}.
\end{itemize}
\end{theorem}

The amortized complexities of this data structure depends on the parameters $P,c,Q$ of the \textsc{HeavyHitter} data structure. Further, they require that some additional properties are satisfied.
These properties and the resulting amortized complexities are stated in \Cref{thm:lw:complexity}.

\begin{theorem}\label{thm:lw:complexity}
Consider the data structure of \Cref{thm:lewis_weight_maintenance} (\Cref{alg:lewis_weight_maintenance})
and let $P,c,Q$ be the parameters of the \textsc{HeavyHitter} data structure (\Cref{def:heavyhitter}).
Let $g\t$ be the vectors $g$ in \Cref{thm:lewis_weight_maintenance} during the $t$-th call to query (and $t=0$ during the initialization) and let $\otau\t$ be the returned vector.
Further, assume the following:
\begin{enumerate}
\item \label{item:lw:solver}
For any given $\omw \in \R^{m}_{\ge 0}$ we can solve linear systems in $ (\mA^\top \omW \mA)^{-1}$
with $\epsilon/(64n)$ accuracy 
(i.e. for input $b$ we can output $\omH b$ for some $\omH \approx_{\epsilon/(64 n)} (\mA^\top \omW \mA)^{-1}$)
in $\tilde{O}(P + \Psi + \nnz(\omW\mA))$ time. Further, if 
$$\mA^\top \omW \mA \approx_{1/\log n} \mA^\top (\omT^{(t-1)})^{1-2/p} (\mG^{(t)})^2 \mA$$
for some $t$, then the required time is only $\tilde{O}(\Psi + \nnz(\omW\mA))$. 
\item \label{item:lw:stable}
There exists a sequence $\tg\t$ such that for all $t$ %
\begin{align}
g\t \in (1\pm1/(10^5 \log n))& ~\tg\t \label{eq:lw:nearby_sequence}\\
\| (\tw\t)^{-1}(\tw\t - \tw^{(t+1)})\|_{\tau(\tmG\t \mA)} &= O(1) \label{eq:lw:nearby_sequence_stable}
\end{align}
where we define $\tw\t := \tau(\tmG\t \mA)^{1/2-1/p}\tg\t$.
\end{enumerate}
Then the following time complexities hold:
\begin{itemize}
	\item \textsc{Initialize} takes $\tilde{O}(P + \epsilon^{-2}(\Psi + \nnz(\mA)))$ time.
	\item \textsc{Scale}$(i, \cdot)$ takes $\tilde{O}(
	\frac{\|c\|_1}{n\epsilon^{O(\log \delta)}} \tau(\mG \mA)_i
	)$ amortized time.
	\item \textsc{Query} takes
	$
	\tilde{O}(
	\Psi \epsilon^{-2} \log^3 \delta
	+ \epsilon^{-4} n (\max_i \nnz(a_i)) \log^5 \delta
	+ \epsilon^{-6} \sqrt{P\|c\|_1/n} \log^4 \delta
	+ Q \log \delta
	)
	$
	amortized time.
\end{itemize}
\end{theorem}

The idea of the data structure (i.e.~the reduction to leverage scores) is as follows.
It is known that the map $w \mapsto (w^{2/p-1} \sigma(\mW^{1/2-1/p} \mA))^{p/2}$
moves $w$ closer to the Lewis weight $\tau(\mA)$ \cite{CohenP15}.
We show in \Cref{lem:approximate_contraction} that the same is true for the regularized Lewis weight,
even when using only an approximation $\osigma \approx \sigma(\mW^{1/2-1/p} \mA) + z$.
The proof follows directly from the techniques in \cite{CohenP15}, and we state the proof in \Cref{sec:lw:corectness} for completeness sake.

\begin{lemma}\label{lem:approximate_contraction}
Let $w,z,\osigma \in \R_{>0}^m$, $\epsilon,\gamma > 0$ and $p \in (0,2)$
with $w \approx_\epsilon \sigma(\mW^{1/2-1/p} \mA) + z$
and $\osigma \approx_\gamma \sigma(\mW^{1/2-1/p} \mA) + z$.
Define $w' = (w^{2/p-1} \osigma)^{p/2}$,
then $w' \approx_{(\epsilon + \gamma) |1-p/2| + \gamma} \sigma(\mW'^{1/2-1/p} \mA)$ 
and $w \approx_{(\epsilon+\gamma)p/2} w'$.
\end{lemma}

We leverage \Cref{lem:approximate_contraction} critically in our data structure. To illustrate this, consider a call to \textsc{Query}, et $g$ be the current state of vector $g$ , and let $g'$ be the state of $g$ during the previous call to \textsc{Query}. Further, assume that the previous call to \textsc{Query} returned  
$w \approx_{\epsilon} \sigma(\mW^{1/2-1/p} \mG' \mA) + z$. The assumption $g \approx_{\delta\epsilon} g'$ (see \Cref{thm:lewis_weight_maintenance}) then implies 
then $w \approx_{5\delta\epsilon} \sigma(\mW^{1/2-1/p} \mG \mA) + z$. Consequently, our task is simply to counter-act this decrease in approximation quality. 

For $\ov^{(1)} := w$, 
$\osigma^{(i)} \approx \sigma((\omV^{(i)})^{1/2-1/p} \mG \mA)+z$, 
and $\ov^{(i+1)} = ((\ov^{(i)})^{2/p-1} \osigma^{(i)})^{p/2}$,
we have $\ov^{(L)} \approx_{\epsilon} \sigma((\omV^{(L)})^{1/2-1/p}\mG \mA)$ 
for some $L = O(\log \delta)$.
Thus we can maintain an $(1\pm\epsilon)$-approximation of the regularized Lewis weight 
by maintaining $L$ many leverage scores via the data structure of \Cref{thm:leverage_score_maintenance}.
A formal specification of this algorithm is given in \Cref{alg:lewis_weight_maintenance}, a proof of correctness is given in \Cref{lem:lw:correctness}, and the complexity analysis is given in \Cref{sec:lw:complexity}.

\begin{algorithm2e}[h]
\caption{Algorithm of \Cref{thm:lewis_weight_maintenance} \label{alg:lewis_weight_maintenance}}
\SetKwProg{Members}{members}{}{}
\SetKwProg{Proc}{procedure}{}{}
\Members{}{
	$D_j$ for $j=1,...,O(1/p)$ \Comment{Data structure of \Cref{thm:leverage_score_maintenance}}\\
	$\ov^{(j)}\in \R^m$ for $j=1,...,O(1/p)$ \\ 
	$g \in \R^m$, $p \in [1/2,2)$, $L \in \N$, $\epsilon > 0$
}
\Proc{\textsc{Initialize}$(\mA \in \R^{m \times n}, g \in \R^m, z \in \R^m, p \in [1/2,2), \delta>1, \epsilon \in (0,1/(2^{10}\delta\log n)])$}{
	Compute $\ov^{(1)}$ with $\ov^{(1)} \approx_\epsilon \sigma((\omV^{(1)})^{1/2-1/p} \mG \mA) + z$ \label{line:lw:initialize_v1}\\
	$L \leftarrow \lceil(\log_{4/3}(200\delta)+1\rceil$ \\
	\For{$j = 1,...,L$}{
		$D_j.\textsc{Initialize}(\mA, (\omV^{(j)})^{1/2-1/p} g, z, \epsilon /(40L))$ \label{line:lw:initialize_D}\\
		$\ov^{(j+1)} \leftarrow ((\ov^{(j)})^{2/p-1} D_j.\textsc{Query}())^{p/2}$
	}
	$g \leftarrow g$, $p \leftarrow p$, $\epsilon \leftarrow \epsilon$
}
\Proc{\textsc{Scale}$(i, b)$}{
	$g_i \leftarrow b$ \\
	$D_j.\textsc{Scale}(i, (\ov^{(j)}_i)^{1/2-1/p} b)$ for $j=1,...,L$ \label{line:lw:scale_Dj}
}
\Proc{\textsc{Query}$()$}{
	\tcp{Maintain $\ov^{(j+1)} = ((\ov^{(j)})^{2/p-1} \osigma^{(j)})^{p/2}$ for $j=1,...,L-1$}
	\For{$j=1,...,L-1$}{
		$I_\osigma^{(j)}, \osigma^{(j)} \leftarrow D_j.\textsc{Query}()$ 
		\label{line:lw:DjQuery}\\
		$\ov^{(j+1)}_i \leftarrow ((\ov^{(j)}_i)^{2/p-1} \osigma^{(j)}_i)^{p/2}$ for $i \in I_\osigma^{(j)}$\\
		$D_{j+1}.\textsc{Scale}(i, (\ov^{(j+1)}_i)^{1/2-1/p} g_i)$ for $i \in I_\osigma^{(j)}$ \label{line:lw:update_Dj}\\
	}
	\tcp{Maintain $\ov^{(1)}_i \approx_\epsilon \tau(\mG \mA)_i$ for all $i\in[m]$,}
	\tcp{but update only if $\tau(\mG \mA)_i$ changed sufficiently.}
	\For{$i \in I_\osigma^{(L-1)}$ with $\ov^{(L)}_i \not\approx_{\epsilon/10} \ov^{(1)}_i$}{ \label{line:lw:check_change}
		$\ov^{(1)}_i \leftarrow \ov^{(L)}_i$ \label{line:lw:update_v1}\\
		$D_1.\textsc{Scale}(i, (\ov^{(1)}_i)^{1/2-1/p} g_i)$ \label{line:lw:update_D1}
	}
	\Return $I_\osigma^{(L)}, \ov^{(1)}$
}
\end{algorithm2e}

\subsection{Correctness}
\label{sec:lw:corectness}

We start by proving \Cref{lem:approximate_contraction}.
We then show in \Cref{lem:lw:correctness} 
that \Cref{alg:lewis_weight_maintenance} indeed maintains an $\exp(\pm\epsilon)$-approximation
of the regularized Lewis weight.

\begin{proof}[Proof of \Cref{lem:approximate_contraction}]
We first show that $w$ and $w'$ are close to each other. By the definition of $w'$
\begin{align*}
\frac{w'}{w} 
= 
\frac{(w^{2/p-1}\osigma)^{p/2}}{w}
=
\left( \frac{\osigma}{w}\right)^{p/2} ~.
\end{align*}
Thus by $w \approx_{\epsilon} \sigma(\mW^{1/2-1/p} \mA) + z \approx_{\gamma} \osigma$
we have 
\begin{align}
w' \approx_{(\epsilon + \gamma) p/2} w. \label{eq:lw:w'_approx_w}
\end{align}

Fix any $i \in [m]$. By definition of leverage score we have that
\begin{align*}
\sigma(\mW^{1/2-1/p} \mA)_i + z_i
=
w^{1-2/p}_i (\mA (\mA^\top \mW^{1-2/p} \mA)^{-1} \mA^\top)_{i,i} + z_i
\end{align*}
This allows us to transform $w'$ as follows
\begin{align}
[w'_i]^{2/p}
&=
w^{2/p-1}_i \osigma_i \notag \\
&\approx_\gamma
w^{2/p-1}_i (\sigma(\mW^{1/2-1/p} \mA)_i + z_i) \notag\\
&=
(\mA (\mA^\top \mW^{1-2/p} \mA)^{-1} \mA^\top)_{i,i} + w^{2/p-1}_i z_i \label{eq:lw:w_transform}
\end{align}
At last, we analyze the approximation ratio of $w'$ to its leverage score
\begin{align*}
\frac{\sigma(\mW'^{1/2-1/p} \mA)_i + z_i}{w'_i} 
&=
\frac{w'^{1-2/p}_i (\mA (\mA^\top \mW'^{1-2/p} \mA)^{-1} \mA^\top)_{i,i} + z_i}
{w'_i} \\
&\approx_\gamma
\frac{
	(\mA (\mA^\top \mW'^{1-2/p} \mA)^{-1} \mA^\top)_{i,i} + w'^{2/p-1}_i z_i
}{
	(\mA (\mA^\top \mW^{1-2/p} \mA)^{-1} \mA^\top)_{i,i} + w^{2/p-1}_i z_i} \\
&\approx_{(\epsilon+\gamma)(p/2)|2/p-1|}
\frac{
	(\mA (\mA^\top \mW'^{1-2/p} \mA)^{-1} \mA^\top)_{i,i} + w'^{2/p-1}_i z_i
}{
	(\mA (\mA^\top \mW'^{1-2/p} \mA)^{-1} \mA^\top)_{i,i} + w'^{2/p-1}_i z_i} 
= 1
\end{align*}
where in step 2 we used \eqref{eq:lw:w_transform}
and in step 3 we used \eqref{eq:lw:w'_approx_w}.
Thus $$
w'_i \approx_{(\epsilon+\gamma)|1-p/2|+\gamma} \sigma(\mW'^{1/2-1/p} \mA)_i + z_i.
$$
\end{proof}

Next we show that our sequence of contractions $\ov^{(1)},\ov^{(2)},...,$ 
improve the approximation quality sufficiently 
to counter-act the impact of changing $g$.

\begin{lemma}\label{lem:lw:correctness}
Let $\gamma = \epsilon /(40L)$ be the accuracy 
used for the leverage score data structures in \Cref{line:lw:initialize_D} of \Cref{alg:lewis_weight_maintenance}. After each call to \textsc{Query}
we have 
$$
\ov^{(i)} 
\approx_{5\delta\epsilon (3/4)^{i-1} + \gamma i} 
\sigma((\omV^{(i)})^{1/2-1/p} \mG \mA) + z
$$
for $i > 1$
and $\ov^{(1)} \approx_\epsilon \sigma((\omV^{(1)})^{1/2-1/p} \mG \mA) + z$.
\end{lemma}

\begin{proof}
We start the proof by induction over the number of calls to \textsc{Query}.
Directly after the initialization (i.e.~zero calls to \textsc{Query})
we have
$\ov^{(1)} \approx_\epsilon \sigma((\omV^{(1)})^{1/2-1/p} \mG \mA) + z$
by \Cref{line:lw:initialize_v1}.
Further, each $\ov^{(j+1)}$ for $j \ge 1$ is defined via
$\ov^{(j+1)} = ((\ov^{(j)})^{2/p-1} \osigma^{(j)})^{p/2}$
where $\osigma^{(j)} \approx_\gamma \sigma((\omV^{(j)})^{1/2-1/p} \mG \mA) + z$ 
is the output of the leverage score data structure $D_j$.
Thus by \Cref{lem:approximate_contraction}
we have $
\ov^{(j)} 
\approx_{\epsilon (1-p/2)^{j-1} + \gamma\sum_{i=0}^{j-1} (1-p/2)^{i}} 
\sigma((\omV^{(j)})^{1/2-1/p} \mG \mA)+z
$. We can bound this approximation quality via
$\epsilon (1-p/2)^{j-1} + \gamma\sum_{i=0}^{j-1} (1-p/2)^{i} \le\epsilon(3/4)^{j-1} j\gamma$
since $p \in [1/2,2)$.

Next, we consider a call to \textsc{Query}.
Note that at the start of executing \textsc{Query},
we have $\ov^{(1)} \approx_{5\delta\epsilon} \sigma((\omV^{(1)})^{1/2-1/p} \mG \mA)+z$,
by induction hypothesis and because $g$ can change by at most a $\exp(\pm \delta \epsilon)$ factor
(see \textsc{Initialize} in \Cref{thm:lewis_weight_maintenance}).

By the recursive definition of the $\ov^{(j)}$ we then have
$$
\ov^{(j)} 
\approx_{5\delta\epsilon (3/4)^{j-1} + \gamma j} 
\sigma((\omV^{(j)})^{1/2-1/p} \mG \mA)
$$
for $j > 1$ at the end of \textsc{Query}.
For $L = \lceil(\log_{4/3}(200\delta)+1\rceil$ and $\gamma \le \epsilon/(40L)$ 
we have
\begin{align*}
5\delta\epsilon (4/3)^{L-1} + \gamma L
&\le
\frac{5\delta\epsilon}{200\delta} + \epsilon / 40
\le 
\epsilon / 20.
\end{align*}
Thus
$\ov^{(L)} \approx_{\epsilon/20} \sigma((\omV^{(L)})^{1/2-1/p} \mG \mA) + z$.

The vector $v^{(1)}$ is modified again at the end of \textsc{Query}
in \Cref{line:lw:update_v1}.
This update to $\ov^{(1)}$ is only performed,
if $\ov^{(L)}_i$ and $\ov^{(1)}$ differ by at least an $\exp(\pm\epsilon/10)$ factor (see \Cref{line:lw:check_change}).
Thus 
\begin{align}
\ov^{(1)} \approx_{\epsilon/10} \ov^{(L)}
\approx_{\epsilon/20} \sigma((\omV^{(L)})^{1/2-1/p} \mG \mA) + z
\approx_{6\epsilon/10}
\sigma((\omV^{(1)})^{1/2-1/p} \mG \mA) + z
\end{align}
where we used $p \le 1/2$.
In summary, $\ov^{(1)} \approx_\epsilon \sigma((\omV^{(1)})^{1/2-1/p} \mG \mA) + z$.
\end{proof}

As the data structure returns $\ov := \ov^{(1)}$ we have $\ov \approx_\epsilon \tau(\omV^{1/2-1/p}\mG\mA) + z$,
which concludes the proof of \Cref{thm:lewis_weight_maintenance}.
Next, we analyze the complexity of the data structure in \Cref{sec:lw:complexity}.

\subsection{Complexity}
\label{sec:lw:complexity}

The main difficulty in analyzing the complexity of our data structure
is to bound the complexity impact of all the calls to $D_j.\textsc{Scale}$
that occur during a call to \textsc{Query}.
To bound this complexity we the following partial results:

First, we use that $\osigma^{(L)}$ is a good approximation 
of the exact regularized Lewis weight $\tau(\mG \mA)$.
Thus \Cref{line:lw:update_D1} is only executed for $i$, 
where $\tau(\mG \mA)_i$ changed by a sufficiently large amount,
because of the condition in \Cref{line:lw:check_change}.
Thus we can bound the total complexity impact of all calls to $D_1.\textsc{Scale}(i,\cdot)$ 
for some $i$ via the stability property \eqref{eq:lw:nearby_sequence_stable}.

Second, to bound how often $D_{j+1}.\textsc{Scale}(i, \cdot)$ is called for $j > 0$,
note that we perform such a call whenever the output $\osigma^{(j)}_i$ of the leverage score data structure $D_j$ (\Cref{thm:leverage_score_maintenance}) changes.
Such a bound on how often the output changes is given by \Cref{lem:ls:bound_output} which bounds the number of changes to the output $\osigma^{(j)}$ relative to the number of changes to the input, i.e.~how often $D_{j}.\textsc{Scale}$ was called.
By propagating, we are able to bound the number of calls to any $D_{j+1}.\textsc{Scale}$ relative to the number of calls to $D_{1}.\textsc{Scale}$.

At last, we require complexity bounds for the data structure $D_j$ that maintain leverage scores. 
The complexities are given in \Cref{thm:ls:complexity}.

We start the formal proof of the complexity analysis by showing in \Cref{lem:lw:close_to_exact}
that an approximate regularized Lewis weight $w \approx_\epsilon \sigma(\mW^{1/2-1/p} \mA) + z$
is also close to the exact regularized Lewis weight $w \approx_\epsilon \tau(\mA)$. 
The following \Cref{lem:lw:close_to_exact} follows directly from techniques in \cite{CohenP15}.

\begin{lemma}\label{lem:lw:close_to_exact}
For any $w, z \in \R_{>0}^m$, $\epsilon > 0$ and $p \in (0,2]$
with $w \approx_\epsilon \sigma(\mW^{1/2-1/p} \mA) + z$
we have $w \approx_\epsilon \tau(\mA)$.
\end{lemma}

\begin{proof}
Define $w^{(0)} := w$ and $w^{(k+1)} = ((w^{(k)})^{2/p-1} (\sigma((\mW^{(k)})^{1/2-1/p} \mA) + z))^{p/2}$.
By \Cref{lem:approximate_contraction} we have
$\lim_{k \rightarrow \infty} w^{(k)} = \tau(\mA)$.
Further we have $w^{(k+1)} \approx_{\epsilon |1-p/2|^k p/2} w^{(k)}$.
Thus 
$$\tau(\mA) = \lim_{k\rightarrow\infty} w^{(k)} \approx_{\epsilon (p/2) \sum_{i\ge 0} |1-p/2|^i} w^{(0)},$$
where $\epsilon (p/2) \sum_{i\ge 0} |1-p/2|^i = \epsilon (p/2) /(1-|1-p/2|) = \epsilon$ for $p \le 2$.
\end{proof}

The input to the leverage score data structures $D_j$ is the matrix $(\omV^{(j)})^{1/2-1/p}\mG\mA$,
and the complexity bounds of the leverage score data structure as stated in \Cref{thm:ls:complexity} only hold,
if Condition \ref{item:ls:solver} and \ref{item:ls:stable} (as stated in \Cref{thm:ls:complexity}) are satisfied.
\Cref{lem:lw:complexity_condition} shows that these requirements hold,
if Condition \ref{item:lw:solver} and \ref{item:lw:stable} of \Cref{thm:lw:complexity} are satisfied.

\begin{lemma}\label{lem:lw:complexity_condition}
Let $(\ov^{(j)})\t$ be the vector $\ov^{(j)}$ when performing the $t+1$-th call to
$D_j.\textsc{Query}()$ (i.e. during the $t$-th call to \textsc{Query} of \Cref{alg:lewis_weight_maintenance}).
Let $g\t$ be the vector $g$ during the $t$-th call to \textsc{Query}.
Let $\tg\t$ be the vector assumed by \eqref{eq:lw:nearby_sequence}
and let $\tw\t := \tau(\tmG\t \mA) \tg\t$.
Then we have
\begin{align*}
((\omV^{(j)})\t)^{1/2-1/p}g\t \in (1\pm1/(64 \log n)) \tmW\t,
\end{align*}
i.e.~Condition \ref{item:ls:stable} of \Cref{thm:ls:complexity} is satisfied.

Further, if Condition \ref{item:lw:solver} of \Cref{thm:lewis_weight_maintenance} holds true, 
then Condition \ref{item:ls:solver} of \Cref{thm:ls:complexity} holds true 
for all instances $D_1,...,D_L$ of \Cref{thm:leverage_score_maintenance}.
\end{lemma}

\begin{proof}
We start by analyzing the sequence $\tw^{(0)},\tw^{(1)},...$.
\paragraph{Sequence}
By \Cref{lem:lw:correctness} we have that
$$
(\ov^{(j)})\t \approx_{6\delta\epsilon} \sigma(((\omV^{(j)})\t)^{1/2-1/p} \mG\t \mA) + z
$$
for all $j=1,...,L$.
Thus by \Cref{lem:lw:close_to_exact}, $p \in [1/2,2]$, and assumption $\epsilon \le 1/(2^{10}\delta\log n)$ 
(see \Cref{thm:lewis_weight_maintenance}) we have
$$
((\ov^{(j)})\t)^{1/2-1/p} \approx_{1/(96 \log n)} (\tau(\mg\t\mA))^{1/2-1/p}.
$$
Via \eqref{eq:lw:nearby_sequence} and $1/2-1/p \le 1/4$ we have
$\tau(\mG\t\mA)^{1/2-1/p} \approx_{1/(10^5 \log n)} \tau(\tmG\t\mA)^{1/2-1/p}$
which results in
$$
((\omV^{(j)})\t)^{1/2-1/p}g\t \in (1\pm 1/(64 \log n)) \tw\t.
$$

\paragraph{Solvers}

During the $t$-th call to \textsc{Query},
when the algorithms calls $D_j.\textsc{Query}()$ in \Cref{line:lw:DjQuery}, 
it uses $(\omV^{(j)})^{1/2-1/p}g\t$ as scale-vector.
Here $\ov^{(1)}$ is exactly the vector returned by the previous call to \textsc{Query}.
Thus if we have a solver as assumed in Condition \ref{item:lw:solver} of \Cref{thm:lewis_weight_maintenance},
then we also have a solver as assumed in Condition \ref{item:ls:solver} of \Cref{thm:ls:complexity}
for instance $D_1$ of \Cref{thm:leverage_score_maintenance}.

For the other instances $D_2,...,D_L$ consider the following.
By \Cref{lem:lw:correctness} we have for $j>1$ that 
$$\ov^{(j)} \approx_{6\delta\epsilon} \sigma((\omV^{(j)})^{1/2-1/p}\mG\t \mA) + z.$$
As $g$ changes by at most an $\exp(\pm\epsilon\delta)$-factor, we have
$$\ov^{(j)} 
\approx_{6\delta\epsilon} 
\tau(\mG\t\mA) 
\approx_{4\delta\epsilon} 
\tau(\mG^{(t-1)}\mA)
\approx_{\epsilon} 
\ov^{(1)}.$$
Thus each $\ov^{(j)} \approx_{11\delta\epsilon} \ov^{(1)}$,
and by $p \in [1/2,2]$ we have
$$
\mA^\top (\omV^{(j)})^{1-2/p}(\mG\t)^2\mA 
\approx_{33\delta\epsilon} 
\mA^\top (\omV^{(j)})^{1-2/p}(\mG\t)^2\mA.
$$
Note that by the upper bound on $\epsilon$ in \Cref{thm:lewis_weight_maintenance},
we have $33\delta\epsilon + 1/(64\log n) < 1/\log n$,
so if we have a solver that satisfies Condition \ref{item:lw:solver} of \Cref{thm:lewis_weight_maintenance},
then we also have a solver that satisfies Condition \ref{item:ls:solver} of \Cref{thm:ls:complexity} 
for instance $D_j$ of \Cref{thm:leverage_score_maintenance}.

\end{proof}

We can now analyze the amortized complexity of \Cref{thm:lewis_weight_maintenance}, i.e.~prove \Cref{thm:lw:complexity}.
This is done by analyzing how often $D_j.\textsc{Scale}$ and $D_j.\textsc{Query}$ 
(instances of \Cref{thm:leverage_score_maintenance})
are called.

\begin{proof}[Proof of \Cref{thm:lw:complexity}]
To bound the complexity of our regularized Lewis weight data structure \Cref{alg:lewis_weight_maintenance}
we first bound the total time complexity after $T$ iterations
and then charge some of the terms to \textsc{Scale} and \textsc{Query}.

For that we will first state the complexities for the internal data structures $D_j$ used by \Cref{alg:lewis_weight_maintenance}.

\paragraph{Complexity of Leverage-Score Data Structures}

Our regularized Lewis weight data structure (\Cref{alg:lewis_weight_maintenance}) 
uses $L$ instances $D_1,...,D_L$ of \Cref{thm:leverage_score_maintenance}.
The complexity bounds of \Cref{thm:leverage_score_maintenance} hold,
if condition \eqref{eq:ls:nearby_sequence} and \eqref{eq:ls:nearby_sequence_stable} are satisfied.
These conditions are satisfied by \Cref{lem:lw:complexity_condition} 
and assumption \eqref{eq:lw:nearby_sequence_stable}.

Since the conditions are satisfied, a call to $D_j.\textsc{Scale}(i, \cdot)$ 
has amortized cost (by \Cref{thm:leverage_score_maintenance})
$$
\tilde{O}\left(
\frac{\|c\|_1 \log^4 \delta}{n\epsilon^4} \sigma((\mV^{(j)})^{1/2-1/p} \mG \mA)_i + \frac{c_i \log^2 \delta}{\epsilon^2}
\right)
=
\tilde{O}\left(
\frac{\|c\|_1 \log^4 \delta}{n\epsilon^4} \tau(\mG \mA)_i
\right)
$$
because we use an accuracy parameter of $\epsilon/(40L) = \Omega(\epsilon/\log \delta)$ in the initialization of each $D_j$,
and since $\tau(\mG \mA) \approx \ov^{(j)} \approx \sigma((\omV^{(j)})^{1/2-1/p} \mG \mA) + z$
by \Cref{lem:lw:correctness} and \Cref{lem:lw:close_to_exact},
and because $z \ge n\cdot c/\|c\|_1$.
Similarly, a call to $D_j.\textsc{Query}()$ has amortized cost
$$
\tilde{O}\left(
\Psi \epsilon^{-2} \log^2 \delta
+ \epsilon^{-4} n (\max_i \nnz(a_i)) \log^4 \delta
+ \epsilon^{-2} \sqrt{P\|c\|_1/n} \log^2 \delta
+ Q
\right).
$$

\paragraph{Total time complexity}
We now bound the total time spent after $T$ iterations. We will then later charge some of the cost as amortized cost to \textsc{Scale} and \textsc{Query}.

When calling \textsc{Scale}$(i,\cdot)$, 
the data structure calls $D_j.\textsc{Scale}(i,\cdot)$
for $j=1,...,L$ with $L = O(\log \delta)$.
Thus we incur the total cost 
\begin{align}
\tilde{O}\left(
\sum_{t=1}^T \sum_{i\in[m]} \left(
\frac{\|c\|_1 \log^5 \delta}{n\epsilon^4} \tau(\mG^{(t-1)} \mA)_i
\right) \mathbf{1}_{g\t_i \neq g^{(t-1)}_i}
\right).
\label{eq:lw:cost:scale}
\end{align}

When calling \textsc{Query}, the function $D_j.\textsc{Query}()$ is called for every $j=1,...,L$
with $L = O(\log \delta)$.
Thus we must add
\begin{align}
\tilde{O}\left(
T \cdot \left(
\Psi \epsilon^{-2} \log^3 \delta
+ \epsilon^{-4} n (\max_i \nnz(a_i)) \log^5 \delta
+ \epsilon^{-2} \sqrt{P\|c\|_1/n} \log^3 \delta
+ Q \log \delta
\right)
\right)
\label{eq:lw:cost:query}
\end{align}
to the total time complexity.

Further, the data structure calls $D_1.\textsc{Scale}(i,\cdot)$
in \Cref{line:lw:update_D1}.
To bound that complexity impact, we observe that $D_1.\textsc{Scale}(i,\cdot)$ is only called
whenever $\ov^{(L)}$ changed by at least an $\exp(\epsilon/10)$-factor
(see \Cref{line:lw:check_change}).
Since 
$$
v^{(L)}_i \approx_{\epsilon/20} \tau(\mG)_i,
$$
this means that $v^{(1)}_i$ is only updated if $\tau(\mG)_i$ changed at least an $\exp(\epsilon/10)$-factor.
By \eqref{eq:lw:nearby_sequence} this means $\tau(\tmG)_i$ must have changed 
by at least an $\exp(\epsilon/2000)$-factor. 
Using the fact that $\tau(\tmG)$ changes slowly \eqref{eq:lw:nearby_sequence_stable} we can thus bound
\begin{align*}
\sum_{t=1}^T \sum_{i\in[m]} \tau(\mG\t \mA)_i \mathbf{1}_{(v^{(1)}_i)\t \neq (v^{(1)}_i)^{(t-1)}}
=
O\left(
T^2 / \epsilon^2
\right).
\end{align*}
The time complexity incurred by calling $D_1.\textsc{Scale}(i,\cdot)$ in \Cref{line:lw:update_D1}
is thus bounded by
\begin{align}
&~
\tilde{O}\left(
\sum_{t=1}^T \sum_{i\in[m]}
\left(
\frac{\|c\|_1 \log^4 \delta}{n\epsilon^4} \tau(\mG^{(t-1)} \mA)_i
\right)
\mathbf{1}_{(v^{(1)}_i)\t \neq (v^{(1)}_i)^{(t-1)}}
\right) %
\le%
\tilde{O}\left( \frac{\|c\|_1 \log^4 \delta}{n\epsilon^6} T^2 \right). \label{eq:lw:cost:D1scale}
\end{align}

At last, we are left with bounding the impact of calling
$D_j.\textsc{Scale}$ for $j > 1$ in \Cref{line:lw:update_Dj}.
Note that $D_j.\textsc{Scale}(i, \cdot)$ is called whenever $\ov^{(j-1)}_i$ changed.
By \Cref{lem:ls:bound_output} 
we can bound for every $j$ how often often any entry of $\ov^{(j)}$ changes as follows
\begin{align*}
&~
\sum_{t=1}^T \sum_{i\in[m]} \tau(\mG\t \mA)_i \mathbf{1}_{(\ov^{(j)}_i)\t \neq (\ov^{(j)}_i)^{(t-1)}} \\
\le&~
O\left( \epsilon^{-1} \sum_{t=1}^T \sum_{i\in[m]} \tau(\mG\t \mA)_i (\mathbf{1}_{((\ov^{(j-1)}_i)\t \neq (\ov^{(j-1)}_i)^{(t-1)}} + \mathbf{1}_{g\t_i \neq g^{(t-1)}_i}) \right) \\
\le&~
O\left(\sum_{i=1}^j \epsilon^{-i} \sum_{t=1}^T \sum_{i\in[m]} \tau(\mG\t \mA)_i \mathbf{1}_{g\t_i \neq g^{(t-1)}_i} \right) \\
\le&~
O\left(\epsilon^{-j} \sum_{t=1}^T \sum_{i\in[m]} \tau(\mG\t \mA)_i \mathbf{1}_{g\t_i \neq g^{(t-1)}_i} \right)
\end{align*}
where the second step comes from repeatedly applying the first step.
Finally, with $j \le L = O(\log \delta)$ this leads to a complexity cost of
\begin{align}
\tilde{O}\left(
\frac{\|c\|_1}{n\epsilon^{O(\log \delta)}} \cdot
\left(
 \sum_{t=1}^T \sum_{i\in[m]} \tau(\mG\t \mA)_i \mathbf{1}_{g\t_i \neq g^{(t-1)}_i}
\right)
\right). \label{eq:lw:cost:Djscale}
\end{align}
In summary, the total cost after $T$ iterations is bounded by
\begin{align*}
\tilde{O}\big(
\underbrace{T \cdot \left(
\Psi \epsilon^{-2} \log^3 \delta
+ \epsilon^{-4} n (\max_i \nnz(a_i)) \log^5 \delta
+ \epsilon^{-2} \sqrt{P\|c\|_1/n} \log^3 \delta
+ Q \log \delta
\right)}_{\eqref{eq:lw:cost:query}}\\
+
\underbrace{\frac{\|c\|_1\log^4 \delta}{n\epsilon^6}T^2}_{\eqref{eq:lw:cost:D1scale}}
+
\underbrace{
\frac{\|c\|_1}{n\epsilon^{O(\log \delta)}} \cdot
\left(
 \sum_{t=1}^T \sum_{i\in[m]} \tau(\mG\t \mA)_i \mathbf{1}_{g\t_i \neq g^{(t-1)}_i}
\right)}_{
\eqref{eq:lw:cost:scale},\eqref{eq:lw:cost:Djscale}
}
\big)
\end{align*}

\paragraph{Amortized Cost}

We previously bounded the total time spent after $T$ iterations, we now charge some of the terms as amortized cost to \textsc{Scale} and \textsc{Query}.

The amortized cost of \textsc{Scale} is
$$
\tilde{O}(
\frac{\|c\|_1}{n\epsilon^{O(\log \delta)}} \tau(\mG \mA)_i
)
$$
which covers the terms depending on $\mathbf{1}_{g\t_i \neq g^{(t-1)}_i}$ in \eqref{eq:lw:cost:scale} and \eqref{eq:lw:cost:Djscale}.
The amortized cost of \textsc{Query} is
\begin{align*}
&~
\tilde{O}(
\frac{\|c\|_1 \log^4 \delta}{n\epsilon^6} T
+
\Psi \epsilon^{-2} \log^3 \delta
+ \epsilon^{-4} n (\max_i \nnz(a_i)) \log^5 \delta
+ \epsilon^{-2} \sqrt{P\|c\|_1/n} \log^3 \delta
+ Q \log \delta
) \\
=&~
\tilde{O}(
\Psi \epsilon^{-2} \log^3 \delta
+ \epsilon^{-4} n (\max_i \nnz(a_i)) \log^5 \delta
+ \epsilon^{-6} \sqrt{P\|c\|_1/n} \log^4 \delta
+ Q \log \delta
)
\end{align*}
which covers the remaining terms in \eqref{eq:lw:cost:query}, and \eqref{eq:lw:cost:D1scale},
when we bound $T \le \sqrt{Pn/\|c\|_1}$ 
by restarting the data structure after $\sqrt{Pn/\|c\|_1}$ calls to \textsc{Query}.
(The reinitialization cost every $T$ iterations is subsumed by the terms above.)

\paragraph{Initialization}

The initialization requires us to compute $\ov^{(1)}$. 
This is done by performing the contraction $w \leftarrow (w^{2/p-1} (\sigma(\mW \mG \mA)+z)^{2/p}$ 
a total of $O(\log_{|1-p/2|} \epsilon^{-1}) = O(\log \epsilon^{-1}) = \tilde{O}(1)$ times.
So one could compute this contraction by just initializing 
$\tilde{O}(1)$ instances of the leverage score data structure (\Cref{thm:leverage_score_maintenance})
to compute the leverage scores required for the contraction, 
and then immediately discarding these instances again.
The cost for computing this initial regularized Lewis weight $\ov^{(1)}$
is subsumed by initializing the data structures $D_j$ for $j=1,...,L$.
Initializing these data structures requires
$\tilde{O}(P \log \delta + (\Psi + \nnz(\mA))\epsilon^{-2} \log^3 \delta)$ time.

\end{proof}

\section{Path Following}
\label{sec:pathfollowing}

In this section we show how to efficiently implement our IPM
which was given by \Cref{algo:pathfollowing,algo:lsstep} in \Cref{sec:ipm}.
Note that \Cref{algo:pathfollowing,algo:lsstep} only specify which steps must be performed,
but not how they must be implemented.
For example \Cref{line:ipm:oxotau} of \Cref{algo:lsstep} specifies that one should
pick an approximation $\ox$ of the primal solution $x$,
but it is not specified how this approximation must be obtained.
Here we show how all the steps of \Cref{algo:pathfollowing} can be performed efficiently,
if we assume the existence of certain data structures. 

These data structure may differ depending on the application,
for example the \textsc{HeavyHitter}-problem (\Cref{def:heavyhitter})
can be solved more efficiently if the LP is a min-cost flow instance (\Cref{lem:graph:heavyhitter})
than when the problem is a general LP (\Cref{lem:lp:heavyhitter}).
However, while these data structures have different complexities,
they implement the same interfaces (e.g.~\Cref{def:heavyhitter}).
Thus the correctness proof, 
i.e.~showing that \Cref{algo:pathfollowing} can be implemented 
by using these data structures, 
is the same for min-cost flow and for general LPs.
This is why we perform this proof in a generalized way.
More accurately, we can show the following theorem.

\begin{theorem}
\label{thm:implement:general}
Assume there exists a $(P,c,Q)$-\textsc{HeavyHitter} (\Cref{def:heavyhitter}),
a $(P,c,Q)$-\textsc{InverseMaintenance} (\Cref{def:inverse_maintenance}),
and a $(P,c,Q)$-\textsc{HeavySampler} (\Cref{def:heavysampler}).
Then we can implement the IPM given by  \Cref{alg:implement:pathfollowing} (\PathFollowing, \Cref{lemma:pathfollowing})
such that the total time of \textsc{PathFollowing} can be bounded by
$$
\tilde{O}\left(
	\left(
	\sqrt{P\|c\|_1} + 
	\sqrt{n}\left(Q + n \cdot \max_i \nnz(a_i)\right)
	\right)
	\log \frac{\mu^\init}{\mu^\target} 
\right).
$$
The implementation is given by \Cref{alg:implement:pathfollowing} and \Cref{alg:implement:short_step}.
\end{theorem}

In \Cref{sec:mincostflow} we state the resulting complexity when we use data structures optimized for the min-cost flow problem. Data structures for general LPs are a bit slower and the resulting LP solver and their complexity is stated in \Cref{sec:linearprogram}.

We start proving \Cref{thm:implement:general}
by giving a general outline of our implementation in \Cref{sec:implement:outline}
and listing the assumed data structures that we require.
In \Cref{sec:implement:correctness} we then show that \Cref{alg:implement:pathfollowing} and \Cref{alg:implement:short_step} do indeed implement the IPM given by \Cref{alg:implement:pathfollowing}.
In \Cref{sec:implement:complexity} we analyze the resulting complexity,
which concludes the proof of \Cref{thm:implement:general}.

\subsection{Outline}
\label{sec:implement:outline}

Recall that our IPM consists of \Cref{algo:pathfollowing} which is esssentially a \textsc{while}-loop that repeatedly calls \Cref{algo:lsstep}. Consequently, we focus on the implementation of  \Cref{algo:lsstep}. In order to implement \Cref{algo:lsstep} (\Cref{line:ipm:oxotau})
we must maintain approximations $\ox$ and $\otau$
that satisfy Invariant \ref{invar}.
Here $\otau$ is an approximation of the regularized Lewis weight $\tau(\ox)$
and will be maintained via the data structure of \Cref{thm:lewis_weight_maintenance} 
presented in \Cref{sec:lewis_weight_maintenance}.
When $D^{(\tau)}$ is an instance of the Lewis weight data structure (\Cref{thm:lewis_weight_maintenance}),
then all we have to do is call $D^{(\tau)}.\textsc{Scale}(i, \ox_i)$
whenever some entry $\ox_i$ changes, and then $D^{(\tau)}.\textsc{Query}()$
will return the desired approximation $\otau$.

We now explain how to obtain the approximation $\ox$ via the following data structure:

\begin{restatable*}[Primal/Gradient Maintenance]{theorem}{gradientMaintenance}
\label{thm:gradient_maintenance} 
There exists a deterministic data-structure
that supports the following operations
\begin{itemize}
\item $\textsc{Initialize }(\mA\in\R^{m\times n}, x^{\init} \in \R^m, g\in\R^{m}, \ttau\in\R^{m}, z\in\R^{m}, w \in [0,1]^m, \epsilon>0)$:
	The data-structure preprocesses the given matrix $\mA\in\R^{m\times n}$,
	vectors $x^{\init},g,\ttau,z\in\R^{m}$, and the accuracy parameters $w \in [0,1]^m$ and $\epsilon>0$
	in $\tilde{O}(\nnz(\mA))$ time. We denote $\mG$ the diagonal matrix $\mdiag(g)$. 
	The data-structure assumes $0.5\le z\le2$ and $n/m\le\ttau\le2$.
\item $\textsc{Update}(i \in [m], a \in \R, b \in \R, c \in \R)$: 
	Sets $g_{i}\leftarrow a$, $\ttau_{i} \leftarrow b$ and $z_i \leftarrow c$ in $O(\nnz(a_i)+\log n)$ time. %
	The data-structure assumes $0.5\le b\le2$ and $n/m\le c\le2$. 
\item $\textsc{SetAccuracy}(i, \delta)$
	Sets $w_i \leftarrow \delta$ in $O(\log n)$ time. %
\item $\textsc{QueryProduct}()$: 
	Returns $\mA^{\top}\mG\nabla\Psi(\oz)^{\flat(\otau)} \in \R^n$ for some $\otau \in \R^m$, $\oz \in \R^m$ 
	with $\otau \approx_\epsilon \ttau$ and $\|\oz-z\|_{\infty}\le \epsilon$,
	where 
	\begin{align*}
	x^{\flat(\otau)} := \argmax_{\|w\|_{\otau + \infty} \leq 1} \langle x, w \rangle.
	\end{align*}
	Every call to \textsc{QueryProduct} must be followed by a call to \textsc{QuerySum},
	and we bound their complexity together (see \textsc{QuerySum}).
\item $\textsc{QuerySum}(h \in \R^m)$:
	Let $v^{(\ell)}$ be the vector $\mG\nabla\Psi(\oz)^{\flat(\otau)}$ used for the result of the $\ell$-th call to \textsc{QueryProduct}.
	Let $h^{(\ell)}$ be the input vector $h$ given to the $\ell$-th call to \textsc{QuerySum}.
	We define 
	\begin{align*}
	  x^{(t)} := x^{\init} + \sum_{\ell=1}^{t} \left( v^{(\ell)} + h^{(\ell)} \right).  
	\end{align*}
	Then the $t$-th call to \textsc{QuerySum} returns a vector $\ox \in \R^m$ with
	$$\|w^{-1}(\ox - x^{(t)})\|_\infty \le \epsilon.$$

	Assuming the input vector $h$ is given in a sparse representation (e.g. a list of non-zero entries), 
	then after $T$ calls to \textsc{QuerySum} and \textsc{QueryProduct} 
	the total time for all calls together is bounded by
	\begin{align*}
	O\left(
		T n \epsilon^{-2} \log n
		+ \log n \cdot \sum_{\ell=0}^T \|h^{(\ell)}\|_0 
		+ T \log n \cdot \sum_{\ell=1}^T \|v^{(\ell)}/w^{(\ell-1)}\|_2^2 / \epsilon^2
	\right)
	\end{align*}%
	The output $\ox \in \R^m$ is returned in a compact representation to reduce the size. In particular, the data-structure returns a pointer to $\ox$ 
and a set $J \subset [m]$ of indices which specifies which entries of $\ox$ have changed 
	between the current and previous call to \textsc{QuerySum}.
\item $\textsc{ComputeExactSum}()$:
	Returns the exact $x^{(t)}$ in $O(m\log n)$ time.
\item $\textsc{Potential}()$:
	Returns $\Psi(\oz)=\sum_i \cosh(\lambda\oz_i)$ in $O(1)$ time for some  $\oz$ with $\|\oz-z\|_{\infty}\le \epsilon$.  
\end{itemize}
\end{restatable*}

A variant of this data structure was proven in \cite{BrandLN+20} to obtain  an element-wise $\ox \approx_\epsilon x$ approximation. Here we instead need to obtain $\|\Phi''(x)^{1/2}(\ox -x)\|_\infty \le \epsilon$. 
We show in \Cref{sec:gradient_maintenance} how \Cref{thm:gradient_maintenance} 
is obtained via a small modification to data structure of \cite{BrandLN+20}.

We now explain how \Cref{thm:gradient_maintenance} 
can be used to maintain the approximate $\ox$.
Note that by \Cref{line:ipm:delta_x} of the IPM (\Cref{algo:lsstep}) 
the vector $x$ changes via
$$x^\new \leftarrow x + \Phi''(\ox)^{-1/2}(\gamma\nabla\Psi(\oy)^{\flat(\otau)} - \mR \delta_r).$$
Here $\mR\delta_r$ can be written as some vector $h$
and $\Phi''(\ox)^{-1/2}\gamma\nabla\Psi(\oy)^{\flat(\otau)}$ can be written as 
$v := \mG \nabla\Psi(\oy)^{\flat(\otau)}$
for $\mG := \gamma\Phi''(\ox)^{-1/2}$,
so the update to $x$ becomes
$$x^\new \leftarrow x + v + h.$$
Note that an approximation of such $x$ satisfying \Cref{invar} is returned by \textsc{QuerySum} of \Cref{thm:gradient_maintenance} when also calling \textsc{SetAccuracy} appropriately.

\Cref{line:ipm:y} of the IPM (\Cref{algo:lsstep}) 
requires an approximation $\oy$ of
$$y = \frac{s+\mu\tau\phi'(x)}{\mu\tau\sqrt{\phi''(x)}}.$$
Such an approximation can be obtained by having a vector 
$\os$ with small enough $\|(\mu\tau\sqrt{\phi''(x)})^{-1}(\os - s)\|_\infty$
and then replacing all $\tau,s,x$ in the definition of $y$
by the approximate $\otau,\oy,\ox$. This is proven in \Cref{lem:y_approx}.
The required approximation $\os$ can be obtained via the following data structure:

\begin{restatable*}[Dual Maintenance]{theorem}{dualSlackMaintenance}
\label{thm:dual_maintenance}
Assuming a $(P,z,Q)$-HeavyHitter data structure as in \Cref{def:heavyhitter},
there exists a data-structure (\Cref{alg:dual_maintenance}) that supports the following operations. Note in the bounds we use $\tilde{O}$ to hide polynomials in $\log (nP/\|z\|_1)$ in addition to $\log n$ factors, and in our instantiations of the data structure the former factor will be bounded by $\log n$.
\begin{itemize}
\item \textsc{Initialize($\mA\in\R^{m\times n}, v^{\init}\in \R^m, w^{\init} \in [0,1]^m, \epsilon \in [0,1]$)}
The data-structure preprocesses the given matrix $\mA \in \R^{m \times n}$,
the vector $v^{\init} \in \R^m$
and accuracy vector $0 < w^{\init} \le 1$
in $\tilde{O}(P)$ time.
\item \textsc{SetAccuracy}($i, \delta$):
Sets $w_i \leftarrow \delta$ in $\tilde{O}(z_i)$ amortized time.
\item \textsc{Add($h\in\R^n$)}:
Suppose this is the $t$-th time the {\sc Add} operations is called, 
and let $h^{(k)}$ be the vector $h$ given when the {\sc Add} operation is called for the $k^{th}$ time. 
Define $v^{(t)}\in\R^m$ to be the vector
$$
v^{(t)} = 
v^{\init} + \mA \sum_{k=1}^t h^{(k)}.
$$
Then the data structure returns a vector $\ov^{(t)} \in \R^m$ such that
$
\|w^{-1}(\ov^{(t)} - v^{(t)}) \|_\infty \le \epsilon.
$
The output will be in a compact representation to reduce the size. 
In particular, the data-structure returns a pointer to $\ov$ 
and a set $I \subset [m]$ of indices $i$
where $\ov^{(t)}_i$ is changed compared to $\ov^{(t-1)}_i$, i.e., the result of the previous call to \textsc{Add}. 
The amortized time for the $t$-th call to \textsc{Add} is
$$
\tilde{O}\left(
Q + \sqrt{nP/\|z\|_1} \cdot \| (v^{(t)}-v^{(t-1)})/w^{(t)}\|_z^2 \epsilon^{-2} + \sqrt{\|z\|_1P/n}
\right).
$$
\item \textsc{ComputeExact()}: 
Returns $v^{(t)}\in \R^m$ in $O(\nnz(\mA))$ time, 
where $t$ is the number of times \textsc{Add} is called so far 
(i.e., $v^{(t)}$ is the state of the exact vector $v$ after the most recent call to \textsc{Add}).
\end{itemize}
\end{restatable*}
A variant of \Cref{thm:dual_maintenance} was proven in \cite{BrandLN+20}
where $\os \approx_\epsilon s$ was maintained. Since we need a slightly different type of approximation
for $\os$, we added the \textsc{SetAccuracy} method.
This is only small modification and the correctness is proven in \Cref{sec:dual_maintenance}.

By \Cref{line:ipm:delta_s} the exact $s$ is defined via
$$s^\new \leftarrow s + \mu\omT\Phi''(\ox)^{1/2}\delta_1 = s + \mA h$$
for some vector $h$.
This vector $h$ is exactly the input given to \textsc{Add} of \Cref{thm:dual_maintenance},
so we can use the data structure to maintain an approximation $\os$ of $s$.

\Cref{line:ipm:g} of the IPM (\Cref{algo:lsstep})
requires $g = -\gamma\nabla\Psi(\oy)^{\flat(\otau)}$
which in the next line is multiplied by $\mA \Phi''(\ox)^{-1/2}$.
This product is can be obtained by \textsc{QueryProduct}
of \Cref{thm:gradient_maintenance}.

\Cref{line:ipm:H} asks us to approximately solve a linear system in $\mA^\top \omT\Phi''(\ox)^{-1} \mA$
which can be done via the following data structure:
\begin{definition}\label{def:inverse_maintenance}
We call a data structure a $(P,c,\Psi)$-\textsc{InverseMaintenance},
if it supports the following operations:
\begin{itemize}
\item \textsc{Initialize}$(\mA, v, \osigma)$
	Initializes in $O(P)$ time
	for $\osigma \ge \frac{1}{2}\sigma(\mV^{1/2} \mA)$
	and $\|\osigma\|_1 = O(n)$.
\item \textsc{Update}$(i, a, b)$
	Set $v_i \leftarrow a$ and $\osigma_i \leftarrow b$ in $O(c_i)$ amortized time.
\item \textsc{Solve}$(\ov, b, \epsilon)$
	Assume $\osigma \ge 1/2 \sigma(\mV^{1/2} \mA)$
	and the given $\ov$ satisfies $\mA^\top \mV \mA \approx_{1/2} \mA^\top \omV \mA$.
	Then \textsc{Solve} returns $\mH^{-1} b$ 
	for $\mH \approx_\epsilon \mA^\top \omV \mA$ 
	in $O(Q + \nnz(\omV, \mA) \log \epsilon^{-1})$ time.
Furthermore, for the same $\ov$ and $\epsilon$, the algorithm uses the same $\mH$.
\end{itemize}

For the complexity bounds one may further assume the following stability assumption:
Let $\ov^{(1)},\ov^{(2)}, \dots$ be the sequence of inputs given to \textsc{Solve},
then there exists a sequence $\tv^{(1)},\tv^{(2)}, \dots$ such that for all $t>0$ %
\begin{align*}
\ov\t \in \left(1\pm 1/(100 \log n)\right) \tv\t \text{ and }%
\|(\tv\t)^{-1} (\tv\t - \tv^{(t+1)})\|_{\osigma} = O(1)
\end{align*}
\end{definition}

At last, we are only left with implementing \Cref{line:ipm:R} 
of the IPM (\Cref{algo:lsstep}),
because all further lines (which update $x$ and $s$) were already covered 
when we discussed how to maintain $\ox$ and $\os$.
\Cref{line:ipm:R} wants us to sample a random diagonal matrix $\mR$
according to some distribution that satisfied \Cref{def:validdistro}.
This can be done by assuming the existence of the following data structure:

\begin{definition}\label{def:heavysampler}
We call a data structure a $(P,c,Q)$-\textsc{HeavySampler} data structure if it supports the following operations:
\begin{itemize}
\item \textsc{Initialize}$(\mA \in \R^{m \times n}, g \in \R^m_{>0}, \otau \in \R^m_{>0})$
	Let $\mA$ be a matrix with $c_i \ge \nnz(a_i)$.
	The data structure initializes in $O(P)$ time.
\item \textsc{Scale}$(i, a, b)$:
	Sets $g_i \leftarrow a$ and $\otau_i \leftarrow b$ in $O(c_i)$ amortized time.
\item \textsc{Sample}$(h \in \R^m)$:
	Returns a random diagonal matrix $\mR \in \R^{m \times m}$
	that satisfies \Cref{def:validdistro}
	for $\delta_r = \mG\mA h$ 
	with $\|\delta_r\|_2 \le m/n$
	and $\otau \approx_{1/2} \sigma(\omA)$
	in $O(Q)$ expected time.
	Further we have $\E[\nnz(\mR\mA)] = O(Q)$.
\end{itemize}
\end{definition}

In summary, we can implement all steps of IPM (\Cref{algo:lsstep}) efficiently via the assumed data structures.
While so far we only outlined how to use these data structures,
\Cref{sec:implement:correctness} proves in \Cref{lem:implement:short_step_log_barrier} 
the correctness of these claims in detail.
At last, \Cref{sec:implement:complexity} analyses the resulting complexity in \Cref{lem:pathfollowing_complexity}.
\Cref{lem:implement:short_step_log_barrier} and \Cref{lem:pathfollowing_complexity}
together form the proof of \Cref{thm:implement:general}.

\subsection{Correctness}
\label{sec:implement:correctness}

We now proceed with the correctness proof of \Cref{thm:implement:general}
by proving in \Cref{lem:implement:short_step_log_barrier} 
that \Cref{alg:implement:pathfollowing} and \Cref{alg:implement:short_step} 
do indeed implement our IPM.

\begin{algorithm2e}[h]
\caption{Implementation of \Cref{algo:pathfollowing} \label{alg:implement:pathfollowing}}
\SetKwProg{Globals}{global variables}{}{}
\SetKwProg{Proc}{procedure}{}{}
\Globals{}{
	$D^{(x,\nabla)}$ instance of primal/gradient maintenance (\Cref{thm:gradient_maintenance}) using $\gamma/2^{12}$ accuracy \\
	$D^{(s)}$ instance of dual maintenance (\Cref{thm:dual_maintenance}) using $\gamma/2^{12}$ accuracy\\
	$D^{(\tau)}$ instance of Lewis weight data structure (\Cref{thm:lewis_weight_maintenance}) with accuracy $\gamma/2^{12}$ \\
	$D^{(\sample)}$ instance of \textsc{HeavySampler} (\Cref{def:heavysampler}) \\
	$D^{(-1)}$ instance of \textsc{InverseMaintenance} (\Cref{def:inverse_maintenance}) \\
	$\otau \in \R^m$ element-wise approximation of $\tau(\ox)$ (multiplicative error)\\
	$\ox \in \R^m$ element-wise approximation of $x$ (error relative to $\Phi''(\ox)$) \\
	$\os \in \R^m$ element-wise approximation of $s$ (multiplicative error) \\
	$\Delta \in \R^n$ (Infeasibility $\Delta = A^\top x - b$) \\
	$\omu \in \R$ approximation of $\mu$ \\
	\tcc{Parameters where $C$ is a sufficiently large constant}
	$\alpha \leftarrow \frac{1}{4\log(4m/n)}, \eps \leftarrow \frac{\alpha}{C}, \lambda \leftarrow \frac{C\log(Cm/\eps^2)}{\eps}, \gamma \leftarrow \frac{\eps}{C\lambda}, r \leftarrow \frac{\eps\gamma}{\cnorm\sqrt{n}}$
}
\Proc{\textsc{PathFollowing}$(\mA, x^\init, s^\init, \mu^\init, \mu^\target)$}{
	$\ox \leftarrow x^\init$, $\os \leftarrow s^\init$, $\omu \leftarrow \mu^\init$, $\mu \leftarrow \mu^\init$, $\Delta \leftarrow 0$ \\
	Let $c$ be the parameter assumed in \Cref{def:inverse_maintenance}, \Cref{def:heavyhitter}, and \Cref{def:heavysampler}, then define $z \leftarrow n/m + nc/\|c\|_1$. \\
	$\otau \leftarrow D^{(\tau)}.\textsc{Initialize}(\mA, \phi''(\ox)^{-1/2}, z, 1- 1/(4\log(4m/n)), 2^{12}C_{\ref*{lemma:pchange}}, \gamma / 2^{16})$ \Comment{$C_{\ref*{lemma:pchange}}$ is the constant suppressed by the first item of Lemma \ref{lemma:pchange}} \\
	$D^{(x, \nabla)}.\textsc{Initialize}(\mA, x^\init, -\gamma \phi(\ox)''^{-1/2}, \otau, \frac{\os+\omu\otau\phi'(\ox)}{\omu\otau\sqrt{\phi''(\ox)}},\phi''(\ox)^{-1/2}, \gamma/2^{16})$\\
	$D^{(s)}.\textsc{Initialize}(\mA, s^\init, \omu\otau\phi''(\ox)^{1/2}, \gamma/2^{16})$ \\
	
	$D^{(\sample)}.\textsc{Initialize}(\mA, \otau^{-1}\phi''(\ox)^{-1/2}, \otau)$\\
	$D^{(-1)}.\textsc{Initialize}(\mA, \otau^{-1} \phi''(\ox)^{-1}, \otau)$\\
	\While{$\mu > \mu^\target$}{
		\textsc{ShortStep}$(\mu)$ (\Cref{alg:implement:short_step}) \\
		$\mu \leftarrow (1-r)\mu$ \\
	}
	\Return $D^{(x,\nabla)}.\textsc{ComputeExact}()$, $D^{(s)}.\textsc{ComputeExact}()$
}
\end{algorithm2e}

\begin{algorithm2e}[p!]
\caption{Implementation of \Cref{algo:lsstep} \label{alg:implement:short_step}}
\SetKwProg{Globals}{global variables}{}{}
\SetKwProg{Proc}{procedure}{}{}
\Globals{}{
	Same variables as in \Cref{alg:implement:pathfollowing}.
}
\Proc{\textsc{ShortStep}$(\mu^{\new}>0)$}{
	\If{$\omu \not\approx_{\gamma/2^{12}} \mu^\new$}{ \label{line:step:update_mu}
		$\omu \leftarrow \mu^\new$ \\
		\For{$i \in [m]$}{
			$D^{(x,\nabla)}.\textsc{Update}(
				i, 
				-\gamma \phi''_i(\ox_i)^{-1/2}, 
				\otau_i,
				(\os_i+\omu\otau_i\phi'_i(\ox_i))/(\omu\otau_i\sqrt{\phi''_i(\ox_i)}))$ \\
			$D^{(s)}.\textsc{SetAccuracy}(i, \omu\otau_i\phi''(\ox_i)^{1/2})$ \\
		}
	}
	$h' \leftarrow D^{(x,\nabla)}.\textsc{QueryProduct()}$ \label{line:step:gradient}
			\Comment{$h' = -\gamma \mA^\top \Phi''(\ox)^{-1/2} \nabla\Phi(\oy)^{\flat(\otau)}$} \\
       \tcc{Leverage score sampling gives $\mH \defeq \mA^\top \omV \mA \approx_{\gamma/2} \mA^\top \omT^{-1} \Phi''(\ox)^{-1} \mA$}
	$\ov_{i}\leftarrow\begin{cases}
\frac{1}{\min(1,100\otau\log(n)/\gamma^{2})} & \text{with probability}\min(1,100\otau\log(n)/\gamma^{2})\\
0 & \text{otherwise}
\end{cases}$.  \label{line:step:H}\\
      \tcc{$h'' = \mH^{-1} (h' + (\mA^\top x - b))$, $\delta_r = \omT^{-1}\Phi''(\ox)^{-1/2} \mA h''$}
	$h'' \leftarrow D^{(-1)}.\textsc{Solve}(\ov, h' + \Delta, \gamma/2)$ \label{line:step:h''}\\
	$\mR \leftarrow D^{(\sample)}.\textsc{Sample}(h'')$\\
	$x^{\tmp}, I_x \leftarrow D^{(x,\nabla)}.\textsc{QuerySum}(-\mR \omT^{-1}\Phi''(\ox)^{-1} \mA h'')$ \label{line:step:xtmp}\\
	$\Delta \leftarrow \Delta + h' - \mA^\top \mR \omT^{-1}\Phi''(\ox)^{-1} \mA h''$ 
			\Comment{Maintain $\Delta = \mA^\top x - b$} \label{line:step:delta}\\
	\For{$i \in I_x$}{ \label{line:step:update_x}
		\If{$|\sqrt{\phi''_i(x^{\tmp}_i)} (x^{\tmp}_i - \ox_i)| > \gamma / 2^{12} $}{
			$\ox_i \leftarrow x^{\tmp}_i$ \\
			$D^{(x,\nabla)}.\textsc{Update}(
				i,
				-\gamma \phi''_i(\ox_i)^{-1/2}, 
				\otau_i,
				(\os_i+\omu\otau_i\phi'_i(\ox_i))/(\omu\otau_i\sqrt{\phi''_i(\ox_i)})
				)$ \\
			$D^{(x,\nabla)}.\textsc{SetAccuracy}(i, \phi''(\ox_i)^{-1/2})$ \\
			$D^{(\tau)}.\textsc{Scale}(i, \phi''(\ox_i)^{-1/2})$ \label{line:step:update_tau_for_x}\\
			$D^{(\sample)}.\textsc{Scale}(i, \otau^{-1}_i \phi''(\ox_i)^{-1/2})$\\
			$D^{(-1)}.\textsc{Update}(i, \otau^{-1}_i \phi''(\ox_i)^{-1}, \otau_i)$\\
			$D^{(s)}.\textsc{SetAccuracy}(i, \omu\otau_i\phi''(\ox_i)^{1/2})$\\
		}
	}
	$\tau^{\tmp}, I_\tau \leftarrow D^{(\tau)}.\textsc{Query}()$ \label{line:step:tautmp}\\
	\For{$i \in I_\tau$}{ \label{line:step:update_tau}
		\If{$\tau^{\tmp}_i \not\approx_{\gamma/2^{10}} \otau_i$}{
			$\otau_i \leftarrow \tau^{\tmp}_i$ \\
			$D^{(x,\nabla)}.\textsc{Update}(i,
				-\gamma \phi''_i(\ox_i)^{-1/2}, 
				\otau_i,
				(\os_i+\omu\otau_i\phi'_i(\ox_i))/(\omu\otau_i\sqrt{\phi''_i(\ox_i)})
				)$ \\
			$D^{(\sample)}.\textsc{Scale}(i, \otau^{-1}_i \phi''(\ox_i)^{-1/2})$\\
			$D^{(-1)}.\textsc{Update}(i, \otau^{-1}_i \phi''(\ox_i)^{-1}, \otau_i)$\\
			$D^{(s)}.\textsc{SetAccuracy}(i, \omu\otau_i\phi''(\ox_i)^{1/2})$\\
		}
	}
	$s^{\tmp}, I_s \leftarrow D^{(s)}.\textsc{Add}(\mu D^{(-1)}.\textsc{Solve}(\ov,h', \gamma/2))$ \label{line:step:stmp}\\
	\For{$i \in I_s$}{ \label{line:step:update_s}
		\If{$|\omu^{-1}\otau^{-1}\Phi''(\ox)^{-1/2}(s^\tmp_i-\os_i)| > \gamma/2^{10}$}{
			$\os_i \leftarrow s^{\tmp}_i$ \\
			$D^{(x,\nabla)}.\textsc{Update}(
				i,
				-\gamma \phi''_i(\ox_i)^{-1/2}, 
				(\os_i+\omu\otau_i\phi'_i(\ox_i))/(\omu\otau_i\sqrt{\phi''_i(\ox_i)})
				)$ \\
		}
	}
}
\end{algorithm2e}

\begin{lemma}\label{lem:implement:short_step_log_barrier}
\Cref{alg:implement:short_step} (\ShortStep) and \Cref{alg:implement:pathfollowing} (\PathFollowing)
implement \Cref{algo:lsstep} (\ShortStep) and \Cref{algo:pathfollowing} (\PathFollowing) respectively.
\end{lemma}

\begin{proof}

As \Cref{alg:implement:short_step} uses many different data structures
we will use the notation $D.var$
to refer to variable $var$ of data structure $D$.
For example $D^{(\tau)}.g$ refers to variable $g$
of data structure $D^{(\tau)}$ 
(which is an instance of Lewis weight data structure \Cref{thm:lewis_weight_maintenance}).

\Cref{algo:pathfollowing} consists only of a while loop
that calls \textsc{ShortStep} (\Cref{algo:lsstep}).
This while loop is also present in \Cref{alg:implement:pathfollowing},
which calls \textsc{ShortStep} (\Cref{alg:implement:short_step}).
Thus we need to prove that \Cref{alg:implement:short_step}
implements \Cref{algo:lsstep}.

In addition to the while loop \Cref{alg:implement:pathfollowing}
initializes data structure such that the following assumptions hold
during the first call to \textsc{ShortStep} (\Cref{alg:implement:short_step}).

\begin{definition}[Parameter assumptions]
\label{def:assumptions}
We assume that the following assumptions hold true
at the start of the first call to \textsc{ShortStep}
(\Cref{alg:implement:short_step}).
\begin{align}
\|\Phi''(\ox)^{1/2}(\ox - x)\|_\infty &\le {\gamma/2^9}, %
\|\omu^{-1}\otau^{-1}\Phi''(\ox)^{-1/2} (\os - s)\|_\infty \le \gamma/2^9, \label{eq:step:approx_x_s}\\
\omu &\approx_{\gamma/2^9} \mu, %
\otau \approx_{\gamma/2^9} \tau(\ox) \label{eq:step:approx_mu_tau}\\
\Delta &= \mA^\top x - b \label{eq:step:Delta}\\
D^{(x,\nabla)}.g = -\gamma \phi''(\ox)^{-1/2},~
D^{(x,\nabla)}.z &= \frac{\os+\omu\otau\phi'(\ox)}{\omu\otau\sqrt{\phi''(\ox)}},~
D^{(x,\nabla)}.w = \phi''(\ox)^{-1/2} \label{eq:step:Dxnabla}\\
D^{(\sample)}.g = \otau^{-1} \phi''(\ox)^{-1/2},~ %
D^{(s)}.w &= \omu\otau\phi''(\ox)^{1/2},~ %
D^{(\tau)}.v = \phi''(\ox)^{-1/2} \label{eq:step:Dsample_Ds_Dtau} \\
D^{(-1)}.v &= \otau^{-1}\phi''(\ox)^{-1},~ D^{(-1)}.\osigma = \otau \label{eq:step:Dinverse}
\end{align}
\end{definition}
We show by induction that these assumptions 
then also holds true for all subsequent calls to \textsc{ShortStep}. 

Before we prove that these assumptions hold true 
for all subsequent calls to \textsc{ShortStep} (\Cref{alg:implement:short_step})
we will first prove that \Cref{alg:implement:short_step}
performs the computations required by the IPM of \Cref{algo:lsstep},
i.e. we show that \Cref{alg:implement:short_step} does indeed implement \Cref{algo:lsstep}.

\paragraph{\Cref{alg:implement:short_step} implements \Cref{algo:lsstep}}

We argue the correctness line by line of \Cref{algo:lsstep}.
\Cref{line:ipm:oxotau} (\Cref{algo:lsstep}) requires that we satisfy Invariant~\ref{invar}.
This is given by the assumption on $\ox$ in \eqref{eq:step:approx_x_s}
and $\otau$ in \eqref{eq:step:approx_mu_tau}.

\Cref{line:ipm:y} (\Cref{algo:lsstep}) requires to find a $\oy$ 
with $\|\oy-y\|_\infty \le \gamma / 20$ for
\[ y = \frac{s+\mu\tau\phi'(x)}{\mu\tau\sqrt{\phi''(x)}}. \]
We have this $\oy$ implicitly by replacing $x$, $s$, $\mu$, and $\tau$ in the definition of $y$
by $\ox$, $\os$, $\omu$, and $\otau$.
We now prove that this $\oy$ satisfies the slightly stronger guarantee $\|\oy-y\|_\infty \le \gamma / 40$.
\begin{lemma}[Approximation of $y$]
\label{lem:y_approx}
Under the assumptions in Definition \ref{def:assumptions} for $\oy \defeq \frac{\os+\omu\otau\phi'(\ox)}{\omu\otau\sqrt{\phi''(\ox)}}$ we have that $\|y-\oy\|_\infty \le \gamma/40$.
\end{lemma}
\begin{proof}
Using the approximations above, we get that
\begin{align*}
&\left\|\frac{s+\mu\tau\phi'(x)}{\mu\tau\sqrt{\phi''(x)}} - \frac{\os+\omu\otau\phi'(\ox)}{\omu\otau\sqrt{\phi''(\ox)}} \right\|_\infty \\
&\le \left\|\frac{s-\os}{\omu\otau\sqrt{\phi''(\ox)}}\right\|_\infty + \left\|\frac{s+\mu\tau\phi'(x)}{\mu\tau\sqrt{\phi''(x)}} - \frac{s+\omu\otau\phi'(\ox)}{\omu\otau\sqrt{\phi''(\ox)}} \right\|_\infty \\
&\le \gamma/2^9 + \left\|\frac{s+\mu\tau\phi'(x)}{\mu\tau\sqrt{\phi''(x)}} - \frac{s+\mu\otau\phi'(\ox)}{\omu\otau\sqrt{\phi''(\ox)}} \right\|_\infty + 
\left\|\frac{(\mu-\omu)\phi'(\ox)}{\omu\sqrt{\phi''(\ox)}}\right\| \\
&\le \gamma/2^8 + \left\|\frac{s+\mu\tau\phi'(x)}{\mu\tau\sqrt{\phi''(x)}} - \frac{s+\mu\otau\phi'(\ox)}{\omu\otau\sqrt{\phi''(\ox)}} \right\|_\infty,
\end{align*}
where we used $1$-self-concordance, specifically that $|\phi'(\ox)| \le \sqrt{\phi''(\ox)}$ in the last step. Now, we calculate the errors resulting from $\otau$ and $\phi'(\ox)$, which gives that
\begin{align*}
&\left\|\frac{s+\mu\tau\phi'(x)}{\mu\tau\sqrt{\phi''(x)}} - \frac{s+\mu\otau\phi'(\ox)}{\omu\otau\sqrt{\phi''(\ox)}} \right\|_\infty \le \left\|\frac{s+\mu\tau\phi'(x)}{\mu\tau\sqrt{\phi''(x)}} - \frac{s+\mu\tau\phi'(\ox)}{\omu\otau\sqrt{\phi''(\ox)}} \right\|_\infty + \left\|\frac{\mu(\tau-\otau)\phi'(\ox)}{\omu\otau\sqrt{\phi''(\ox)}}\right\|_\infty \\ &\le 1.1\gamma/2^9 + \left\|\frac{s+\mu\tau\phi'(x)}{\mu\tau\sqrt{\phi''(x)}} - \frac{s+\mu\tau\phi'(x)}{\omu\otau\sqrt{\phi''(\ox)}} \right\|_\infty + \left\|\frac{\mu\tau(\phi'(\ox)-\phi'(x))}{\omu\otau\sqrt{\phi''(\ox)}}\right\|_\infty \\
&\le 2.2\gamma/2^9 + \left\|\frac{s+\mu\tau\phi'(x)}{\mu\tau\sqrt{\phi''(x)}} - \frac{s+\mu\tau\phi'(x)}{\omu\otau\sqrt{\phi''(\ox)}} \right\|_\infty
\end{align*}
where we used $1$-self-concordance in the last step. Note that $\phi''(x)^{1/2} \approx_{1.1\gamma/2^9} \phi''(\ox)^{1/2}$ by self-concordance (Lemma \ref{lemma:selfcon}), hence $\mu\tau\sqrt{\phi''(x)} \approx_{\gamma/2^7} \omu\otau\sqrt{\phi''(\ox)}$. This gives us
\begin{align*}
\left\|\frac{s+\mu\tau\phi'(x)}{\mu\tau\sqrt{\phi''(x)}} - \frac{s+\mu\tau\phi'(x)}{\omu\otau\sqrt{\phi''(\ox)}} \right\|_\infty \le \left\|\frac{s+\mu\tau\phi'(x)}{\mu\tau\sqrt{\phi''(x)}}\right\|_\infty \left\|1 - \frac{\mu\tau\sqrt{\phi''(x)}}{\omu\otau\sqrt{\phi''(\ox)}}\right\|_\infty \le \eps\gamma/2^6 \le \gamma/2^9
\end{align*}
because $(x, s, \mu)$ is $\eps$-centered. Combining everything, we have that the total error is $\gamma/2^8 + 2.2\gamma/2^9 + \gamma/2^9 \le \gamma/40$.
\end{proof}
\Cref{line:ipm:g} (\Cref{algo:lsstep})
asks us to compute
$$
g \leftarrow - \gamma\nabla\Phi(\oz)^{\flat(\otau)}
$$
for $\|\oz - y\|_\infty \le \gamma/20$.
While we do not compute $g$,
our implementation does compute $\mA^\top \Phi''(\ox)^{-1/2} g$ as follows:
\Cref{line:step:gradient} of \Cref{alg:implement:short_step}
computes
\begin{align*}
h' 
= 
-\gamma \mA^\top \Phi''(\ox)^{-1/2} \nabla\Phi(\oz)^{\flat(\otau)}
= 
\mA^\top \Phi''(\ox)^{-1/2} g
\end{align*}
for some $\oz \approx_{\gamma/2^{10}} D^{(x,\nabla)}.z = \oy \approx_{\gamma/40} y$
by the guarantees of the primal/gradient data structure \Cref{thm:gradient_maintenance}
and assumption \eqref{eq:step:Dxnabla}.

By \Cref{line:ipm:H} of \Cref{algo:lsstep}
we must obtain a matrix $\mH \approx_\epsilon \mA^\top \omT^{-1}\Phi''(\ox)^{-1} \mA$
which is done in \Cref{line:step:H} of \Cref{alg:implement:short_step}
with higher accuracy $\mH \approx_{\gamma/2} \mA^\top \omT^{-1}\Phi''(\ox)^{-1} \mA$.
Next, since $D^{(-1)}.v = \otau^{-1}\phi''(\ox)^{-1}$
and $D^{(-1)}.\osigma = \otau$ by \eqref{eq:step:Dinverse}
we have $\osigma \ge \frac{1}{2}\sigma((\omT^{-1}\Phi''(\ox)^{-1})^{1/2} \mA)$
and a call to $D^{(-1)}.\textsc{Solve}(\cdot,\cdot,\gamma/2)$ (e.g. in \Cref{line:step:h''}) 
is equivalent to multiplying by some matrix $\mH^{-1}$ with
\begin{align*}
\mH^{-1} \approx_{\gamma} (\mA^\top \omT^{-1}\Phi''(\ox)^{-1} \mA)^{-1}
\end{align*}
which with $\gamma \le \epsilon$ is accurate enough.

\Cref{line:ipm:delta_r} of \Cref{algo:lsstep}
wants us to compute
\begin{align*}
\delta_r
=&~
\omT^{-1} \Phi''(\ox)^{-1/2} \mA \mH^{-1} \mA^\top (\Phi''(\ox)^{-1/2}g + \mA^\top x - b)
\end{align*}
This is done implicitly in \Cref{line:step:h''}
of our implementation \Cref{alg:implement:short_step}.
By assumption \eqref{eq:step:Delta} we have $\Delta = \mA^\top x - b$,
thus \Cref{line:step:h''} computes $h''$ with 
\begin{align*}
h'' = \mH^{-1}(h' + \mA^\top x-b).
\end{align*} 
The vector $\delta_r$ can be represented via $h''$ by
\begin{align*}
\delta_r
=
\omT^{-1} \Phi''(\ox)^{-1/2}\mA \mH^{-1}(h' + \mA^\top x-b) 
=
\omT^{-1} \Phi''(\ox)^{-1/2}\mA h''
\end{align*}
because of $h' = \mA^\top \Phi''(\ox)^{-1} g$.

\Cref{line:ipm:R} of \Cref{algo:lsstep}
wants us to compute a random diagonal matrix $\mR$
that satisfies the conditions of \Cref{def:validdistro}.
By \Cref{def:heavysampler} (HeavySampler), assumption \eqref{eq:step:Dsample_Ds_Dtau} 
on $D^{(\sample)}$,
and assumption \eqref{eq:step:approx_mu_tau} on $\otau$,
such a random matrix can be obtained
via $D^{(\sample)}.\textsc{Sample}(h'')$.
Note that for $\delta_r := D^{(\sample)}.\mG \mA h''$
we have by \Cref{cor:xchange} (item 2, the bound on $\d_r$) and $\tau \ge n/m$ that 
\[ \|\d_r\|_2^2 \le \frac{m}{n}\|\d_r\|_\tau^2 \le \frac{m}{n}. \]
So the requirements for the \textsc{Sample} procedure stated in \Cref{def:heavysampler} are satisfied.

\Cref{line:ipm:delta_x} and \Cref{line:ipm:xnewsnew} of \Cref{algo:lsstep} want us to compute
$$x^\new \leftarrow x + \Phi''(\ox)^{-1/2}(g - \mR\delta_r).$$
By guarantees of the primal/gradient maintenance (\Cref{thm:gradient_maintenance}) and assumption $\eqref{eq:step:Dxnabla}$,
\Cref{line:step:xtmp} of our implementation \Cref{alg:implement:short_step}
computes $x^\tmp$ with
\begin{align}
&~
\|\Phi''(\ox)^{1/2} (x^\new - x^\tmp) \|_\infty \le \gamma/2^{12} \label{eq:step:xtmp} 
\end{align}

Our implementation also computes an
$\tau^\tmp$ in \Cref{line:step:tautmp}.
By \Cref{line:step:update_tau_for_x} our implementation makes sure that
$D^{(\tau)}.g = \phi''(\ox)^{-1/2}$.
Thus the vector $\tau^\tmp$ in \Cref{line:step:tautmp} satisfies
$\tau^\tmp \approx_{\gamma/2^{10}} \tau(\ox)$ for the new $\ox$
and after \Cref{line:step:update_tau} we have 
\begin{align}
\otau \approx_{\gamma/2^{10}} \tau(\ox) \label{eq:step:otau}
\end{align}
for the new $\ox$.

\Cref{line:ipm:delta_s} and \Cref{line:ipm:xnewsnew} of \Cref{algo:lsstep} asks us to compute
\begin{align*}
s^\new \leftarrow s + \mu \mA \mH^{-1} \mA^\top \Phi''(\ox)^{-1/2}g
\end{align*}
By dual maintenance \Cref{thm:dual_maintenance} and $D^{(s)}.w = \omu\otau\Phi''(\ox)^{-1/2}$ by \eqref{eq:step:Dsample_Ds_Dtau}, 
\Cref{line:step:stmp} of our implementation \Cref{alg:implement:short_step}
computes $s^\tmp$ with
\begin{align}
&~
\|\omu^{-1}\otau^{-1}\Phi''(\ox)^{-1/2} (s^\new - s^\tmp)\|_\infty \le \gamma / 2^{12} \label{eq:step:stmp}
\end{align}
Note that $\ox$ and $\otau$ in \eqref{eq:step:stmp}
refer to the new values of $\ox$ and $\otau$
as they were changed in \Cref{line:step:update_x} and \Cref{line:step:update_tau}.

\paragraph{Assumptions on $\ox$, $\os$:}

To argue assumption \eqref{eq:step:approx_x_s} on $\ox$, 
define $\ox$ to be the value of $\ox$ at the start of \textsc{ShortStep}
and $\ox^\new$ to be the new value at the end of \textsc{ShortStep}.
We have $\|\Phi''(\ox)^{1/2} (x^\new - x^\tmp) \|_\infty \le \gamma/2^{10}$
by \eqref{eq:step:xtmp}.
If $\ox_i = \ox^\new_i$, then by \Cref{line:step:update_x}
we have $|\phi''(\ox_i)^{1/2} (x_i^\tmp - \ox_i) | \le \gamma/2^{12}$.
Thus 
\begin{align*}
|\phi''(\ox^\new_i)^{1/2} (x_i^\new - \ox_i^\new) |
=
|\phi''(\ox_i)^{1/2} (x_i^\new - \ox_i) | \\
\le
|\phi''(\ox_i)^{1/2} (x_i^\tmp - \ox_i) | + |\phi''(\ox_i)^{1/2} (x^\new - x_i^\tmp) |
\le 
\gamma/2^{10} + \gamma/2^{12} \le \gamma/2^9.
\end{align*}
On the other hand, if $\ox_i \neq \ox^\new_i$, then $\ox^\new_i = x^\tmp_i$, so
\begin{align*}
|\phi''(\ox^\new_i)^{1/2} (x_i^\new - \ox_i^\new)|
\le
1.1|\phi''(\ox_i)^{1/2} (x_i^\new - x_i^\tmp)|
\le
\gamma/2^9
\end{align*}
by self-concordance (Lemma \ref{lemma:selfcon}) and $\ox^\new_i \approx_\gamma \ox_i$.
In summary, we have $\|\Phi(\ox^\new)^{1/2}(\ox^\new - x)\|_\infty \le \gamma / 2^9$, 
so at the start of the next call of \textsc{ShortStep} 
we satisfy the assumption on $\ox$ in \eqref{eq:step:approx_x_s} again.

For the assumption in \eqref{eq:step:approx_x_s} on $\os$, 
note that we have 
$$\| \omu^{-1}(\otau^{\new})^{-1}\Phi''(\ox^\new)^{-1/2}(s^\tmp - s) \|_\infty \le \gamma / 2^{12}$$
by \eqref{eq:step:stmp}.
So if $\os^\new_i = s^\tmp_i$ then 
$$| \omu^{-1}(\otau^{\new}_i)^{-1}\phi''(\ox^\new_i)^{-1/2} (s^\new_i - \os^\new_i)| \le \gamma / 2^{12}.$$ 
Alternatively, if $\os^\new_i \neq s^\tmp_i$, 
then $\os^\new_i = \os_i$ and by the condition in \Cref{line:step:update_s}
we have 
\begin{align*}
&~
| \omu^{-1}(\otau^{\new}_i)^{-1}\phi''(\ox^\new_i)^{-1/2} (s^\new_i - \os^\new_i)| \\
\le&~ 
| \omu^{-1}(\otau^{\new}_i)^{-1}\phi''(\ox^\new_i)^{-1/2} (s^\new_i - s^\tmp_i)|
+
| \omu^{-1}(\otau^{\new}_i)^{-1}\phi''(\ox^\new_i)^{-1/2} (s^\tmp_i - \os^\new_i)| \\
\le&~
\gamma / 2^{12} + \gamma / 2^{10} \le \gamma/2^9
\end{align*}
Since $\omu$ changes by at most an $\exp(\gamma/2^{12})$ factor at the start of \textsc{ShortStep},
assumption \eqref{eq:step:approx_x_s} will be satisfied during the next call to \textsc{ShortStep}.

\paragraph{Assumptions on $\otau$ and $\omu$}

We already argued in \eqref{eq:step:otau} that 
$\otau \approx_{\gamma/2^{10}} \tau(\ox)$.
Further $\mu \approx_{\gamma/2^{12}} \omu$ is verified at the start of each call to \textsc{ShortStep}
in \Cref{line:step:update_mu}.
Thus the assumptions of \eqref{eq:step:approx_mu_tau} are satisfied.

\paragraph{Assumption on $D^{-1}$, $D^{(\tau)}$, $D^{(\sample)}$, and $D^{(x,\nabla)}$:}

All data structures that depend on $\omu$, $\ox$, $\os$, or $\otau$ are updated, 
whenever $\omu$ or an entry of $\ox$, $\os$, or $\otau$ changes 
(see \Cref{line:step:update_mu}, \Cref{line:step:update_x}, 
\Cref{line:step:update_tau}, and \Cref{line:step:update_s}).
So the assumptions in \eqref{eq:step:Dxnabla}, \eqref{eq:step:Dsample_Ds_Dtau}, and \eqref{eq:step:Dinverse}
are always satisfied.

\paragraph{Assumption $\Delta = \mA^\top x - b$:}
$\Delta = \mA^\top x - b$ initially because the input $x$ is feasible.
Then after \Cref{line:step:delta} we have $\Delta = \mA^\top x^{\new} - b$ 
for $x^{\new} = x + \Phi''(\ox)^{-1/2} g - \mR\omT^{-1}\Phi''(\ox)^{-1/2}\mA h'' $.
Thus we always maintain $\mA^\top x - b$, whenever $x$ changes.

\end{proof}

\subsection{Complexity}
\label{sec:implement:complexity}

We now analyze in \Cref{lem:pathfollowing_complexity}
the complexity of \Cref{alg:implement:pathfollowing}
which together with \Cref{lem:implement:short_step_log_barrier}
concludes the proof of \Cref{thm:implement:general}.

\begin{lemma}\label{lem:pathfollowing_complexity}
Assume $(P,c,Q)$ heavy hitter,
$(P,c,Q)$ inverse maintenance,
and $(P,c,Q)$ heavy sampler.
Then the total time of \textsc{PathFollowing} can be bounded by
$$
\tilde{O}\left(
	\left(
	\sqrt{P\|c\|_1} + 
	\sqrt{n}\left(Q + n \cdot \max_i \nnz(a_i)\right)
	\right)
	\log \frac{\mu^\init}{\mu^\target} 
\right)
$$
\end{lemma}

\begin{proof}
We analyze the complexity of \textsc{PathFollowing} in 
multiple parts:
First we analyze the initialization of all data structures,
i.e.~the time spent until the first call of \textsc{ShortStep}.
Then we analyze the total time spent on all calls to \textsc{ShortStep}.

\paragraph{Initialization}

Initializing $D^{(x,\nabla)}$ takes 
$\tilde{O}(\nnz(\mA)) = \tilde{O}(P)$ time
by \Cref{thm:gradient_maintenance} and $\nnz(\mA) \le P$ (\Cref{def:heavyhitter}).
The initialization of $D^{(s)}$, $D^{(\sample)}$, and $D^{(-1)}$ take $\tilde{O}(P)$ time each 
by \Cref{thm:dual_maintenance},
\Cref{def:heavysampler} and \Cref{def:inverse_maintenance}.

The initialization of $D^{(\tau)}$ takes $\tilde{O}(P+Q+n (\max_i \nnz(a_i)) + \sqrt{P \|c\|_1 / n})$ time
by \Cref{thm:lw:complexity}, because we can solve any linear system of the form $\mA^\top \mV \mA x = b$
by initializing an instance of \Cref{def:inverse_maintenance} for that $w$ 
and then solving the system via $D.\textsc{Solve}$.

\paragraph{ShortStep}

Let $T$ %
be the number of calls to \textsc{ShortStep} (\Cref{alg:implement:short_step}).
We start by bounding how often $\omu$, $\ox$, $\os$, and $\otau$
are modified.
Let $\omu^{(t)}$, $\ox^{(t)}$, $\os^{(t)}$, $\otau^{(t)}$
refer to the respective variables during iteration number $t$.
We will prove the following bounds.
\begin{align}
 \sum_{t=2}^T \mathbf{1}_{\omu^{(t)} \neq \omu^{(t-1)}} 
&\le \tilde{O}(T/\sqrt{n})
\label{eq:change:mu}\\
 \sum_{t=2}^T \sum_{i=1}^m \tau(\ox^{(t)})_i \mathbf{1}_{\ox_i^{(t)} \neq \ox_i^{(t-1)}} 
&\le \tilde{O}(T^2)
\label{eq:change:x}\\
 \sum_{t=2}^T \sum_{i=1}^m \tau(\ox^{(t)})_i \mathbf{1}_{\os_i^{(t)} \neq \os_i^{(t-1)}} 
&\le \tilde{O}(T^2)
\label{eq:change:s}\\
 \sum_{t=2}^T \sum_{i=1}^m \tau(\ox^{(t)})_i \mathbf{1}_{\otau_i^{(t)} \neq \otau_i^{(t-1)}} 
&\le \tilde{O}(T^2)
\label{eq:change:tau}
\end{align}
The bound on $\omu$ follows because $\mu$ changes by an $(1-r)$ factor for $r = \tilde{O}(1/\sqrt{n})$
in each iteration,
and we update $\omu$ whenever $\mu \not\approx_{\gamma/2^{12}} \omu$.
Thus is takes $\tilde{\Omega}(\sqrt{n})$ iterations until we have to change $\omu$
which results in the $\tilde{O}(T / \sqrt{n})$ bound.

For the bound on $\ox$ note that we update $\ox_i \leftarrow x^\tmp_i$ 
whenever $|\phi''(x^\tmp_i)^{1/2}(x^\tmp_i - \ox_i)| > \gamma/2^{12}$.
Let $t_1,t_2$ be two iterations where we update $\ox_i$, 
and let $\hx^{(1)}, \hx^{(2)}, \dots$ be the sequence from \Cref{lemma:morestablex}
for $\beta = \gamma / 2^{15}$.
Then we have 
\begin{align*}
|\phi''((x^\tmp)^{(t_2)}_i)^{1/2} (\hx^{(t_1)} - \hx^{(t_2)})_i | 
&\ge
|\phi''((x^\tmp)^{(t_2)}_i)^{1/2} (x^{(t_1)} - x^{(t_2)})_i | - \gamma / 2^{14} \\
&\ge 
|\phi''((x^\tmp)^{(t_2)}_i)^{1/2} ((x^\tmp)^{(t_1)} - (x^\tmp)^{(t_2)})_i |
- \gamma/2^{13} \\
&\ge
\gamma /2^{12}
\end{align*}
where we used 
$\| \Phi''(\hx)^{1/2} (\hx - x) \|_\infty \le \gamma / 2^{15}$,
$\| \Phi''(x^\tmp)^{1/2} (x^\tmp - x) \|_\infty \le \gamma / 2^{16}$,
and the fact that we only update when $x^\tmp_i$ changed by at least $\gamma/2^{12}$.

Thus we can bound
\begin{align*}
\sum_{t=2}^T \sum_{i=1}^m \tau(\ox^{(t)})_i \mathbf{1}_{\ox_i^{(t)} \neq \ox_i^{(t-1)}} 
=&~
\tilde{O}\left( 
\sum_{t=2}^T \sum_{i=1}^m \tau(\ox^{(t)})_i \left(\phi''(\hx^{(t-1)}_i)^{1/2} |\hx_i^{(t)} - \hx_i^{(t-1)}| \right)^2
\right) \\
=&~
\tilde{O}\left( 
\left(\sum_{t=2}^T \| \Phi''(\hx^{(t)})^{1/2} (\hx^{(t)} - \hx^{(t-1)}) \|_{\tau(\hx^{(t-1)})}\right)^2
\right) = \tilde{O}\left( T^2 \right).
\end{align*}

In a similar way we can bound
\begin{align*}
\sum_{t=2}^T \sum_{i=1}^m \tau(\ox^{(t)})_i  \mathbf{1}_{\os_i^{(t)} \neq \os_i^{(t-1)}} 
&=
\tilde{O}\left( 
\sum_{t=2}^T \sum_{i=1}^m \tau(\ox^{(t)})_i \left( (\mu^{-1}\tau(\ox^{(t)})_i \phi''(\ox_i^{(t)})^{1/2})^{-1}|s^{(t)}_i - s^{(t-1)}_i| \right)^2
\right) \\
&=
\tilde{O}\left( 
\left( \sum_{t=2}^T \|(\mu^{-1}\tau(\ox^{(t)}) \Phi''(\ox^{(t)})^{1/2})^{-1}(s^{(t)} - s^{(t-1)}) \|_{\tau(\ox^{(t-1)})}\right)^2 
\right)\\
&=
\tilde{O}\left( T^2 \right)
\end{align*}
by using $s^{(t)} - s^{(t-1)} = \d_s$ and Corollary \ref{cor:xchange} Part 1.
At last, consider the bound on $\otau$, where we have
\begin{align*}
\sum_{t=2}^T \sum_{i=1}^m \tau(\ox^{(t)})_i  \mathbf{1}_{\otau_i^{(t)} \neq \otau_i^{(t-1)}} 
&=
\tilde{O}\left( 
\sum_{t=2}^T \sum_{i=1}^m \tau(\hx^{(t)})_i \left(\tau(\hx^{(t)})_i^{-1}|\tau(\hx^{(t)})_i - \tau(\hx^{(t-1)})_i| \right)^2
\right) \\
&=
\tilde{O}\left( 
\left( \sum_{t=2}^T \|\tau(\hx^{(t)})^{-1}(\tau(\hx^{(t)}) - \tau(\hx^{(t-1)})) \|_{\tau(\hx^{(t-1)})}\right)^2 
\right)=\tilde{O}\left( T^2 \right)
\end{align*}
by \Cref{lemma:morestablerest}.

\paragraph{Cost of $D^{(x,\nabla)}$, $D^{(s)}$, $D^{(\sample)}$}

The time spent per call to 
$D^{(s)}.\textsc{Scale}(i, \cdot)$,
$D^{(s)}.\textsc{SetAccuracy}(i, \cdot)$,
$D^{(x,\nabla)}.\textsc{Update}(i, \cdot)$,
and $D^{(\sample)}.\textsc{Scale}(i, \cdot)$
is bounded by $\tilde{O}(c_i) = \tilde{O}(\tau_i \cdot \|c\|_1/n)$,
because of $\tau \ge nc/\|c\|_1$.
These functions are called whenever $\otau_i$, $\ox_i$, $\os_i$, or $\omu$
are changed, so by the previous bounds on \eqref{eq:change:mu}-\eqref{eq:change:tau} we can bound the total time
spent on calling these functions by $\tilde{O}(T^2 \|c\|_1/n + \|c\|_1 T / \sqrt{n})$.

In each iteration we call $D^{(s)}.\textsc{Add}$
and the total time for these calls is bounded by
\begin{align*}
&~
\tilde{O}(
TQ + T \sqrt{nP/\|c\|_1} \|\overline{\delta}_s / (\mu\tau\phi''(x)^{1/2}) \|_c^2 + T\sqrt{\|c\|_1P/n}
) \\
=&~
\tilde{O}(
TQ + T \sqrt{nP/\|c\|_1} (\|c\|_1/n) \|\overline{\delta}_s / (\mu\tau\phi''(x)^{1/2}) \|_\tau^2 + T\sqrt{\|c\|_1P/n}
) \\
=&~
\tilde{O}(
T (Q + \sqrt{P\|c\|_1/n})
)
\end{align*}
by \Cref{thm:dual_maintenance}, \Cref{cor:xchange} Part 1, and $\tau \ge nc/\|c\|_1$.

Likewise, we can bound the time spent on all calls to
$D^{(x,\nabla)}.\textsc{QueryProduct}$ and $D^{(x,\nabla)}.\textsc{QuerySum}$ by 
\begin{align*}
&~
\tilde{O}(Tn + T(\max_i \nnz(a_i)) + \E[\|\mR\mA h''\|_0] + T \sum_{t=1}^T \| \phi''(\ox\t)^{-1/2} g\t / \phi''(\ox\t)^{-1/2} \|_2^2) \\
=&~
\tilde{O}(Tn + T(\max_i \nnz(a_i)) + \E[\nnz(\mR\mA)] + (Tm/n) \sum_{t=1}^T \| g\t \|_{\tau(\ox\t)}) \\
=&~
\tilde{O}(T(n + (\max_i \nnz(a_i)) + \E[\nnz(\mR\mA)]) + T^2 m/n)
\end{align*}
where $\E[\nnz(\mR\mA)]$ is the expected number of entries returned by \Cref{def:heavysampler}
and we used that $g$ by definition has $\|g\|_{\tau+\infty} \le 1$ and that $\tau \geq n/m$.

We can bound the expected time of a call to $D^{(\sample)}.\textsc{Sample}$
and the expected size of $\E[\nnz(\mR\mA)]$ by $Q$ using Definition \ref{def:heavysampler}.

\paragraph{Cost of $D^{(\tau)}$}

The complexity bounds of $D^{(\tau)}$ stated in \Cref{thm:lewis_weight_maintenance}
only holds, if Condition \ref{item:lw:solver} and Condition \ref{item:lw:stable} are satisfied.
The first condition requires us to be able to solve linear systems
in $(\mA^\top \omV \mA)^{-1}$ for $\mA^\top \omV \mA \approx \mA^\top \omT^{1-2/p} \Phi''(\ox)^{-1} \mA$
in $\tilde{O}(Q + \nnz(\omV \mA))$ time.
Such a solver by $D^{-1}$ (\Cref{def:inverse_maintenance}).

Next, we require the existence of a sequence $\hx^{(1)},\hx^{(2)},...$
where for all $t \in [T]$ we have
$$\phi''(\ox^{(t)})^{-1/2} \in (1 \pm 1/(10^5 \log n)) \phi''(\hx^{(t)})^{-1/2}$$
$$\|\tau(\hx\t)^{\frac{1}{p}-\frac{1}{2}}\phi''(\hx\t)^{-\frac{1}{2}}(\tau(\hx\t)^{\frac{1}{2}-\frac{1}{p}}\phi''(\hx\t)^{-\frac{1}{2}} - \tau(\hx^{(t+1)})^{\frac{1}{2}-\frac{1}{p}}\phi''(\hx^{(t+1)})^{-\frac{1}{2}}) \|_{\tau(\hx\t)} = O(1).$$
for $p = 1 - 1/(4\log(4m/n))$.
This sequence is given by \Cref{lemma:morestablex} and \Cref{lemma:morestablerest}.
Thus in summary, the complexity bounds stated in \Cref{thm:lewis_weight_maintenance} apply.

We can now bound the total cost of all calls to $D^{(\tau)}.\textsc{Scale}(i, \cdot)$
by $\tilde{O}(T^2 \|c\|_1/n)$.
This is because on such function call has amortized cost $\tilde{O}(\frac{\|c\|_1}{n}\tau_i)$
and the function is called whenever $\ox_i$ changes.
By the bound on \eqref{eq:change:x} we then obtain the total cost over $T$ iterations of \textsc{ShortStep}.

At last, the time requires for all calls to $D^{(\tau)}.\textsc{Query}$
is bounded by
$$
\tilde{O}\left(
T (Q + n (\max_i \nnz(a_i)) + \sqrt{P\|c\|_1/n})
\right).
$$

\paragraph{Cost of $D^{(-1)}$}

The stability assumption stated in \Cref{def:inverse_maintenance}
is given by \Cref{lemma:morestablex} and \Cref{lemma:morestablerest},
so we can use the complexities stated in \Cref{def:inverse_maintenance}.

A call to $D^{(-1)}.\textsc{Update}(i, \cdot)$ takes $O(c_i)$
amortized time and occurs whenever $\otau_i$ or $\ox_i$ changes.
By \eqref{eq:change:tau} and \eqref{eq:change:x} we can thus bound the cost by
$\tilde{O}(T^2 \|c\|_1/n)$.

Performing the calls to
$D^{(-1)}.\textsc{Solve}(\ov, h' + \Delta, \gamma/2)$
and
$D^{(-1)}.\textsc{Solve}(\ov, h', \gamma/2)$
takes $O(Q + n (\max_i \nnz(a_i)))$ time,
which is subsumed by the cost of $D^{(\tau)}.\textsc{Query}$.

\paragraph{Total Cost}

When calling \textsc{ShortStep} a total of $T$ times,
the total cost is
\begin{align*}
&\tilde{O}\left((T (Q + n (\max_i \nnz(a_i)) + \sqrt{P\|c\|_1/n} + \|c\|_1/\sqrt{n})
+
T^2 \|c\|_1/n
+
T^2 m/n
\right) \\
= & \tilde{O}\left((T (Q + n (\max_i \nnz(a_i)) + \sqrt{P\|c\|_1/n} + \|c\|_1/\sqrt{n})
+
T^2 \|c\|_1/n
\right)
\end{align*}
where we used $\|c\|_1 \geq m$.
 
In general the algorithm will call \textsc{ShortStep}
a total of $T = \tilde{O}(\sqrt{n} \log \frac{\mu^\init}{\mu\target})$ times.
However, we can speed-up the algorithm
by stopping the loop after $T' = \Theta((Pn/\|c\|_1)^{1/2})$ iterations %
and then calling \textsc{PathFollowing} again
for the obtained solution $x,s$, the current value of $\mu$,
and the desired target value $\mu^\target$.
This way we call \textsc{PathFollowing}
for a total of $\tilde{O}((\|c\|_1/P)^{1/2}\log \frac{\mu^\init}{\mu\target} )$
times. Hence, the term $T^2 \|c\|_1/n$ becomes $\tilde{O}(\sqrt{P \|c\|_1} \log \frac{\mu^\init}{\mu\target} )$

Note that by assumption $P \le \|c\|_1$, 
we can perform this reset before ever changing $\omu$ in \Cref{line:step:update_mu}.
Thus we can remove the $\tilde{O}(\|c\|_1/\sqrt{n})$ term from our cost
and we obtain the total cost
$$
\tilde{O}\left(
	\left(
		\sqrt{P\|c\|_1} 
		+ \sqrt{n} Q
		+ n^{3/2} (\max_i \nnz(a_i))
	\right)
	\log \frac{\mu^\init}{\mu^\target}
\right).
$$

\end{proof}

\section{Minimum Cost Flow and Applications}
\label{sec:mincostflow}

In this section, we prove \Cref{thm:mincost flow}. Mincost flow with general demands can be reduced to
single source / sink mincost flow by adding a super-source and super-sink to collect positive / negative demands respectively
and assigning edges outwards with the proper capacities. We assume that we know the maximum flow value $F$ as well, as we
can perform a binary search on $F$, and add a $s$-$t$ edge of sufficiently large demand and capacity.
Therefore, the minimum cost flow problem we consider can be precisely formulated as 
\begin{equation}
\min_{\substack{\mA^{\top}x=Fe_{s,t}\\
		0\le x_{e}\le u_{e}\forall e\in E
	}
}c^{\top}x\label{eq:LP mincost}
\end{equation}
where $\mA\in\{-1,0,1\}^{E\times V}$ is an incidence matrix of $G$ where $\mA_{e,u}=-1,\mA_{e,v}=1$ for every edge $e=(u,v)\in E$.
We denote the optimal value of the above LP by $\opt(G)$.

To prove \Cref{thm:mincost flow}, there are two main parts. For the
first part, we show that, given an initial point $(x^{\init},s^{\init})$
where $x^{\init}$ is an initial primal feasible solution and $s^{\init}$
is an initial dual slack of (\ref{eq:LP mincost}), the path following
algorithm (\Cref{alg:implement:short_step}) returns $(x^{\target},s^{\target})$
where $x^{\target}$ is a near-optimal and near-feasible flow in $\Otil(m+n^{1.5})$
time. This is proved in \Cref{sec:pathfollowing_graph} by putting
together the tools developed in previous sections. 

For the second part, we describe how to obtain an initial point $(x^{\init},s^{\init})$
and how to obtain an exactly optimal and feasible flow instead of
a near-optimal and near-feasible flow, which gives \Cref{thm:mincost flow}.
This part follows using standard techniques and, for completeness,
we show how to do these tasks in \Cref{sec:initfinal_graph}.

Our algorithms in this section will use the following linear system solver.
\begin{lemma}[See e.g.~\cite{SpielmanT04,KS16}]
	\label{lem:solver}There is an algorithm that, 
	given a matrix $\mA\in\R^{m\times n}$ with at most two non-zero entries per row,
	a diagonal non-negative weight matrix $\mD\in\R^{m\times m}$ and
	a vector $b\in\R^{n}$ 
	such that $\mA^\top\mD\mA$ is a symmetric diagonally dominant (SDD) matrix\footnote{A SDD matrix $\mA \in \R^{n \times n}$ is a matrix that is symmetric, and for all $i \in [n]$ satisfies $\mA_{ii} \ge \sum_{j\in[n]} |\mA_{ij}|$.}
	and there exists a vector $x\in\R^{n}$
	where $(\mA^{\top}\mD\mA)x=b$, w.h.p. returns a vector $\overline{x}$
	such that $\|\overline{x}-x\|_{\mA^{\top}\mD\mA}\le\eps\|x\|_{\mA^{\top}\mD\mA}$ in $\Otil(\nnz(\mA)\log\eps^{-1})$ time. 
\end{lemma}

\subsection{Path Following for Graph Problems}

In Appendix \ref{sec:graph_data_structures} we show that \textsc{HeavyHitter}, \textsc{InverseMaintenance}, and \textsc{HeavySampler} data structures exist for graphical problems, and leverage these to show the following.
\label{sec:pathfollowing_graph}
\begin{lemma}
	\label{lem:pathfollowing_graph}Consider a linear program
	\[
	\Pi:\min_{\substack{\mA^{\top}x=b\\
			\ell_{i}\le x_{i}\le u_{i}\forall i
		}
	}c^{\top}x
	\]
	where $\mA$ is obtained by removing one column (corresponding to one vertex) from an incidence matrix of a graph with $n$ vertices
	and $m$ edges. Let $\eps=1/(4C\log(m/n))$ for a large enough constant
	$C$. Given an $\eps$-centered initial point $(x^{\init},s^{\init},\mu^{\init})$
	for $\Pi$ and a target $\mu^{\target}$, \Cref{alg:implement:short_step}
	(when using graph-based data structures) returns an $\eps$-centered point $(x^{\target},s^{\target},\mu^{\target})$
	in time 
	\[
	\tilde{O}\left((m+n^{1.5} \cdot(\log W''+|\log\mu^{\init}/\mu^{\target}|))\cdot|\log\mu^{\init}/\mu^{\target}|\right)
	\]
	where $W''$ is the ratio of largest to smallest
	entry in the vector $\phi''(x^{\init})=\frac{1}{(u-x^{\init})^{2}}+\frac{1}{(x^{\init}-\ell)^{2}}$.
\end{lemma}

\begin{proof}
We use \Cref{alg:implement:pathfollowing}, 
which is an implementation of \Cref{algo:pathfollowing}.
By \Cref{lemma:pathfollowing} the algorithm returns 
an $\epsilon$-centered $(x^\target, s^\target, \mu^\target)$.
For the complexity, note that we have a 
$(P,c,Q)$-\textsc{HeavyHitter},
$(P,c,Q)$-\textsc{HeavySampler},
and $(P,c,Q)$-\textsc{InverseMaintenance}
for $P = \tilde{O}(m)$, 
$c_i = \tilde{O}(1)$ for $i \in [m]$,
$Q = \tilde{O}(m/\sqrt{n} + n\log W')$
according to \Cref{lem:graph:heavyhitter}, \Cref{lem:graph:heavysampler},
and \Cref{lem:graph:inversemaintenance}.
Here $W'$ is the ratio of the largest to smallest entry of any $\Phi''(\ox)$
maintained in \Cref{alg:implement:pathfollowing}.
By \Cref{lemma:paramchange} this is bounded by $\tilde{O}(\log W''+|\log\mu^{\init}/\mu^{\target}|)$.
Thus by using \Cref{lem:pathfollowing_complexity} the time complexity of \Cref{alg:implement:pathfollowing}
is bounded by $$
\tilde{O}\left((m+n^{1.5}\cdot(\log W''+|\log\mu^{\init}/\mu^{\target}|))\cdot|\log\mu^{\init}/\mu^{\target}|\right).
$$
\end{proof}

\subsection{Initial and Final Points}
\label{sec:initfinal_graph}

In this section we discuss how to construct an initial flow which is on the central path. To do this, we augment the graph with a star which allows us to route a feasible flow. This corresponds to adding an identity block in the case of linear programs. Then we discuss how to extract a final point from our central path flow for sufficiently small path parameter $\mu$.

Throughout this subsection, let $\eps=1/(4C\log(m/n))$ for a large
enough constant $C$. We assume that $u_{e}$ and $c_{e}$ are integral
and let $W$ be the maximum absolute value of $u_{e}$ and $c_{e}$
over all edges $e$. 

\paragraph{Defining Modified Graph and Linear Program with Trivial Initial Points.}

Recall that our path following algorithm (\Cref{alg:implement:short_step})
requires an initial point $(x^{\init},s^{\init},\mu^{\init})$ which
is $\eps$-centered (as defined in \Cref{def:centered}). However,
it is not clear how to quickly obtain such $\eps$-centered point
for an arbitrary graph $G$. Therefore, we will modify the graph $G$
to another graph $\wt G$ where an $\eps$-centered initial point
can be easily defined. %

Given $G=(V,E)$, let $\wt G=(V\cup\{z\},E\cup\wt E)$ where $\wt E=\{(v,z),(z,v)\mid v\in V\}$.
That is, we add a bi-directional star rooted at $z$ into $G$. We
let $\begin{bmatrix}x^{\init}\\
\wt x^{\init}
\end{bmatrix}$ denote an initial flow in $\wt G$ where $x^{\init}$ and $\wt x^{\init}$
specify the flow values on $E$ and $\wt E$, respectively. For each
$e\in E$, we set $x_{e}^{\init}\defeq u_{e}/2$. For all star-edges
$e\in\wt E$, we first set $\wt x_{e}^{\init}=1$ (just to prevent
them from having flow value too close to zero). Then, we additionally route the excess
at each vertex induced by $x^{\init}$ using the star-edges from $\wt E$.
More formally, for each $v\in V$, if the flow excess at $v$ is $[\mA^{\top}x^{\init}-Fe_{s,t}]_{v}>0$,
then we set $\wt x_{(v,z)}^{\init}\defeq1+[\mA^{\top}x^{\init}-Fe_{s,t}]_{v}$
and $\wt x_{(z,v)}^{\init}=1$. Otherwise, the flow deficit at $v$
is $[Fe_{s,t}-\mA^{\top}x^{\init}]_{v}\ge0$ and we set $\wt x_{(z,v)}^{\init}=1+[Fe_{s,t}-\mA^{\top}x^{\init}]_{v}$
and $\wt x_{(v,z)}^{\init}\defeq1$. 
The capacity $u_{e}$ and cost
$c_{e}$ of each original edge $e\in E$ stay the same. For each star-edge
$e\in\wt E$, we set the capacity $\wt u_{e}=2\wt x_{e}^{\init}$
and $\wt c_{e}\defeq50m\|u\|_{\infty}\|c\|_{\infty}$. Consider the
modified linear program for minimum cost flow on $\wt G$
\begin{align}
\min_{\substack{\wt{\mA}^{\top}\footnotesize{\begin{bmatrix}x\\
		\wt x
		\end{bmatrix}}=Fe_{s,t}\\
		0\le x_{e}\le u_{e}\forall e\in E\\
		0\le\wt x_{e}\le\wt u_{e}\forall e\in\wt E
	}
}c^{\top}x+\wt c^{\top}\wt x.\label{eq:modifylp_flow}
\end{align}
where $\wt{\mA}$ is an incidence matrix of $\wt G$. We denote the
optimal value of the above LP by $\opt(\wt G)$. As $\wt G$ is obtained
from $G$ by adding some edges, we obtain this simple fact:
\begin{fact}
	\label{fact:OPT modified graph}$\opt(\wt G)\le\opt(G)$.
\end{fact}

For a small technical reason, we need to further modify the linear program since an incidence matrix is degenerate and our path following algorithm only works on a full rank matrix. Let $\wt{\mA}_z$ be obtained from $\wt{\mA}$ by removing a single column corresponds to the vertex $z$.\footnote{The modification will actually work if we remove an arbitrary vertex. We choose the vertex $z$ for convenience.}

\begin{fact}
	\label{fact:remove z}We have the following: 
	\begin{enumerate}
		\item $\wt{\mA}_z$ has full rank, and 
		\item For any $\footnotesize{\begin{bmatrix}x\\
			\wt x
			\end{bmatrix}}$, $\wt{\mA}_z^{\top}\footnotesize{\begin{bmatrix}x\\
			\wt x
			\end{bmatrix}}=Fe_{s,t}$ if and only if $\wt{\mA}^{\top}\footnotesize{\begin{bmatrix}x\\
			\wt x
			\end{bmatrix}}=Fe_{s,t}$.
	\end{enumerate}
\end{fact}

\begin{proof}
	(1) This is well-known and straightforward to verify,
	(2) Let $\footnotesize{\begin{bmatrix}x\\
		\wt x
		\end{bmatrix}}$ satisfy $\wt{\mA}_z^{\top}\footnotesize{\begin{bmatrix}x\\
		\wt x
		\end{bmatrix}}=Fe_{s,t}$ which is the same as $\wt{\mA}^{\top}\footnotesize{\begin{bmatrix}x\\
		\wt x
		\end{bmatrix}}=Fe_{s,t}$ except that we do not require that the excess flow at $z$ is zero.
	However, once we fix the excess at every vertex except at $z$, when
	we can determine the excess at $z$. In our case, the excesses are
	all zero except at $s$ and $t$. This implies that the excess at
	$z$ is zero. So $\wt{\mA}{}^{\top}\footnotesize{\begin{bmatrix}x\\
		\wt x
		\end{bmatrix}}=Fe_{s,t}$. Proving the claim in the reverse direction is trivial.
\end{proof}

By \Cref{fact:remove z}(2), the following linear program is equivalent to (\ref{eq:modifylp_flow}).

\begin{align}
\min_{\substack{\wt{\mA}_z^{\top}\footnotesize{\begin{bmatrix}x\\
			\wt x
			\end{bmatrix}}=Fe_{s,t}\\
		0\le x_{e}\le u_{e}\forall e\in E\\
		0\le\wt x_{e}\le\wt u_{e}\forall e\in\wt E
	}
}c^{\top}x+\wt c^{\top}\wt x.\label{eq:modifylp_flow_remove_z}
\end{align}

\begin{lemma}
	[Initial point for the modified min cost flow]\label{lem:initial point flow}Let
	$\begin{bmatrix}s^{\init}\\
	\wt s^{\init}
	\end{bmatrix}=\begin{bmatrix}c\\
	\wt c
	\end{bmatrix}$ and $\mu^{\init}=100m^{2}W^{3}\eps^{-1}$. Then, the point $\left(\begin{bmatrix}x^{\init}\\
	\wt x^{\init}
	\end{bmatrix},\begin{bmatrix}s^{\init}\\
	\wt s^{\init}
	\end{bmatrix},\mu^{\init}\right)$ is $\eps$-centered w.r.t. the linear program in (\ref{eq:modifylp_flow}) and (\ref{eq:modifylp_flow_remove_z}).
\end{lemma}

\begin{proof}
	As $\wt x$ is defined such that all excess induced by $x$ is canceled
	except at the source $s$ and the sink $t$, we have that $\wt{\mA}^{\top}\begin{bmatrix}x\\
	\wt x
	\end{bmatrix}=Fe_{s,t}$. That is, the flow is exactly feasible in $\wt G$. Also, the vector
	$\begin{bmatrix}c\\
	\wt c
	\end{bmatrix}$ is clearly dual feasible for $z = 0$ in \Cref{def:centered}.
	
	Now we bound the centrality.
	
	By the choices of $x^{\init}$ and $\wt x^{\init}$ and \Cref{fact:bound phi},
	we have the following. For each original edge $e\in E$, we have $\phi'(x_{e}^{\init})=0$
	and $\phi''(x_{e}^{\init})\ge1/u_e^2$. For each star-edge $e\in\wt E$,
	we have $\phi'(\wt x_{e}^{\init})=0$ and $\phi''(\wt x_{e}^{\init})\ge 1/\wt{u}_e^2 = 1/(2\wt x_{e}^{\init}){}^{2}\ge1/(2\|u\|_{\infty}n){}^{2}$
	because $\wt x_{e}^{\init}\le\|u\|_{\infty}n$. 
	
	Therefore, for $x=\begin{bmatrix}x^{\init}\\
	\wt x^{\init}
	\end{bmatrix}$ and $s=\begin{bmatrix}c\\
	\wt c
	\end{bmatrix}$, because $\tau(x)\ge\frac{n}{m}$, we can bound 
	\[
	\left\Vert \frac{s+\mu^{\init}\tau(x)\phi'(x)}{\mu^{\init}\tau(x)\sqrt{\phi''(x)}}\right\Vert _{\infty}\le\left\Vert \frac{50m\|u\|_{\infty}\|c\|_{\infty}}{\mu^{\init}(n/m)/(2\|u\|_{\infty}n)}\right\Vert _{\infty}\le\frac{100m^{2}W^{3}}{\mu^{\init}}=\eps.
	\]
	
\end{proof}

\paragraph{Following the Path: Near-optimal Near-Feasible Flow for the Modified
	Graph.}

Given an initial point from \Cref{lem:initial point flow}, we can run the path following algorithm and
obtain a near optimal point as formalized below:
\begin{lemma}
	\label{lem:near-opt near-feasible}Consider an $\eps$-centered initial
	point 
	\[
	\left(\begin{bmatrix}x^{\init}\\
	\wt x^{\init}
	\end{bmatrix},\begin{bmatrix}s^{\init}\\
	\wt s^{\init}
	\end{bmatrix},\mu^{\init}\right)
	\]
	of (\ref{eq:modifylp_flow_remove_z}) obtained from \Cref{lem:initial point flow}.
	Let $\mu^{\target}=1/\poly(mW)$ where $\poly(mW)$ is an arbitrarily
	large polynomial on $mW$. By invoking \Cref{lem:pathfollowing_graph},
	we obtain an $\eps$-centered point 
	\[
	\left(\begin{bmatrix}x^{\target}\\
	\wt x^{\target}
	\end{bmatrix},\begin{bmatrix}s^{\target}\\
	\wt s^{\target}
	\end{bmatrix},\mu^{\target}\right)
	\]
	of (\ref{eq:modifylp_flow_remove_z}) in $\Otil(m\log W+n^{1.5}\log^{2}W)$ time. 
\end{lemma}

\begin{proof}
	Note that we can invoke \Cref{lem:pathfollowing_graph} because the constraint matrix $\wt{A}_z$ of  (\ref{eq:modifylp_flow_remove_z}) is exactly an incidence matrix with one column removed. 
	
	Now, we only need to analyze the running time. This follows because $\wt{\mA}_z$ corresponds to the graph $\wt G$ with $O(n)$ vertices and $O(m)$ edges. Also,
	$\mu^{\init}/\mu^{\target}=\poly(mW)$. It remains to bound the ratio
	of largest to smallest entry in $\phi''(\begin{bmatrix}x^{\init}\\
	\wt x^{\init}
	\end{bmatrix})$ denoted by $W$. It suffices to show that $W''=\poly(W)$. 
	
	By definition of $x^{\init}$ and $\wt x^{\init}$ and \Cref{fact:bound phi},
	we have $\phi''(x_{e}^{\init})=\frac{2}{(u_{e}/2)^{2}}$ for each
	$e\in E$ and $\phi''(\wt x_{e}^{\init})=\frac{2}{(\wt x_{e}^{\init})^{2}}$
	where $1\le\wt x_{e}^{\init}\le n\|u\|_{\infty}$ by definition of
	$\wt x_{e}^{\init}$. Therefore, we have that the ratio $W''$ is
	at most $\poly(n\|u\|_{\infty})=\poly(mW)$ as desired. 
\end{proof}

\paragraph{Near-optimal Feasible Flow for the Modified Graph.}

As the point from \Cref{lem:near-opt near-feasible} may not be even
feasible. In particular, there might exist a vertex $u$ where total
flow value entering $u$ may not equal the one leaving $u$. Now,
we show how to obtain a near-optimal \emph{feasible} flow on the modified
graph $\wt G$. Moreover, we additionally guarantee the flow value
on every star-edge incident to the dummy vertex $z$ is small. Intuitively,
this is because the cost of every star-edge is very large. This is
formalized in the following lemma: 
\begin{lemma}
	\label{lem:near optimal flow}Given an $\eps$-centered point $\left(\begin{bmatrix}x\\
	\wt x
	\end{bmatrix},\begin{bmatrix}s\\
	\wt s
	\end{bmatrix},\mu\right)$ of (\ref{eq:modifylp_flow_remove_z}) where $\mu<1/n$, 
	there exists $\begin{bmatrix}x^{\final}\\
	\wt x^{\final}
	\end{bmatrix}$ which is a feasible flow for $\wt G$ where $c^{\top}x^{\final}+\wt c^{\top}\wt x^{\final}\le\opt(\wt G)+\mu n$
	in $\Otil(m \log W)$ time. Moreover, if there exists a feasible solution
	to the original LP (\ref{eq:LP mincost}), then $\|\wt x^{\final}\|_{\infty}<0.1$.
	
	We can compute in $\Otil(m\log W)$ time $\begin{bmatrix}x^{\apxfinal}\\
	\wt x^{\apxfinal}
	\end{bmatrix}$ whose each entry differs from $\begin{bmatrix}x^{\final}\\
	\wt x^{\final}
	\end{bmatrix}$ by at most $1/(mW)^{10}$.
\end{lemma}

\begin{proof}
	\Cref{lemma:finalpoint} implies the existence of $\begin{bmatrix}x^{\final}\\
	\wt x^{\final}
	\end{bmatrix}$ which is also an exactly feasible solution of (\ref{eq:modifylp_flow_remove_z}). By \Cref{fact:remove z}(2), it is exactly feasible for (\ref{eq:modifylp_flow}) and so it is a feasible flow for $\wt{G}$. 
	
	For the ``moreover''
	part. If there exists a feasible solution to (\ref{eq:LP mincost}),
	then $\opt(G)\le m\|u\|_{\infty}\|c\|_{\infty}$. By \Cref{fact:OPT modified graph},
	we also have $\opt(\wt G)\le m\|u\|_{\infty}\|c\|_{\infty}$. However,
	if $\|\wt x^{\final}\|_{\infty}\ge0.1$, then $\wt c^{\top}\wt x^{\final}\ge0.1\cdot\min_{e\in\wt E}\wt c_{e}\ge5m\|u\|_{\infty}\|c\|_{\infty}$.
	This contradicts the fact that
	\[
	\wt c^{\top}\wt x^{\final}\le\opt(\wt G)+\mu n-c^{\top}x^{\final}\le2m\|u\|_{\infty}\|c\|_{\infty}+1.
	\]
	
	Finally, as \Cref{lemma:finalpoint} says that $\begin{bmatrix}x^{\final}\\
	\wt x^{\final}
	\end{bmatrix}$ can be obtained in time $O(m)$ plus the time for solving a Laplacian system exactly.
	We can use the approximate solver from \Cref{lem:solver} to obtain the desired $\begin{bmatrix}x^{\apxfinal}\\
	\wt x^{\apxfinal}
	\end{bmatrix}$ in $\Otil(m \log W)$ time.
\end{proof}

\paragraph{Optimal Feasible Flow for the Original Graph.}

\Cref{lem:near optimal flow} only gives us a near-optimal flow for
$\wt G$, but our goal is to obtain an exactly optimal flow for $G$.
To achieve this goal, we will use a convenient lemma based on the
Isolation Lemma below.
\begin{lemma}
	[\cite{ds08}, or Lemma 8.10 of \cite{BrandLN+20}]\label{thm:isolation}Let
	$\Pi=(G,b,c)$ be an instance for minimum-cost flow problem where
	$G$ is a directed graph with $m$ edges, the demand vector $b\in\{-W,\dots,W\}^{V}$
	, the cost vector $c\in\{-W,\dots,W\}^{E}$ and the capacity vector $u\in\{0,\dots,W\}^{E}$.
	
	Let the perturbed instance $\Pi'=(G,b,c')$ be such that $c'_{e}=c_{e}+z_{e}$
	where $z_{e}$ is a random number from the set $\left\{ \frac{1}{4m^{2}W^{2}},\dots,\frac{2mW}{4m^{2}W^{2}}\right\} $.
	Let $x'$ be a feasible flow for $\Pi'$ whose cost is at most $\opt(\Pi')+\frac{1}{12m^{2}W^{3}}$
	where $\opt(\Pi')$ is the optimal cost for problem $\Pi'$. 
	With probability at least $1/2$, there exists an optimal feasible and integral flow $x$ for $\Pi$ such that $\| x- x' \|_\infty \le 1/3$.
\end{lemma}
Combining these pieces allows us to prove \Cref{thm:mincost flow}.
\begin{proof}[Proof of \Cref{thm:mincost flow}]
Given the input graph $G=(V,E,u,c)$, we first perturb the edge costs
$c$ to $c'$ according \Cref{thm:isolation}. Let $G'=(V,E,u,c')$
denote the perturbed graph. Then, we construct the modified graph
$\wt G'$ and the initial point for $\wt G'$ according to \Cref{lem:initial point flow}
(instead of the initial point for $\wt G$ as we did before). Then,
we again invoke the path following algorithm in $\Otil(m\log W+n^{1.5}\log^{2}W)$
time using \Cref{lem:near-opt near-feasible} with $\mu^{\target}=\frac{1}{12m^{2}W^{3}n}$,
and then invoke \Cref{lem:near optimal flow} for $\wt G'$ in $\Otil(m \log W)$
time. This process gives us $\begin{bmatrix}(x')^{\apxfinal}\\
(\wt x')^{\apxfinal}
\end{bmatrix}$ whose each entry differs by at most $1/(mW)^{10}$ from a feasible flow $\begin{bmatrix}(x')^{\final}\\
(\wt x')^{\final}
\end{bmatrix}$ for $\wt G'$ where $c^{\top}(x')^{\final}+\wt c^{\top}(\wt x')^{\final}\le\opt(\wt G')+\mu^{\target}n\le\opt(\wt G')+\frac{1}{12m^{2}W^{3}}$.
Moreover, $\|(\wt x')^{\final}\|_{\infty}<0.1$ (otherwise, we declare
that, with
this specific choice of flow value $F$, there is no feasible solution for $G'$ and hence for $G$). 

Now, let $\begin{bmatrix}x^{\final}\\
\wt x{}^{\final}
\end{bmatrix}$ be obtained from $\begin{bmatrix}(x')^{\apxfinal}\\
(\wt x')^{\apxfinal}
\end{bmatrix}$ by rounding the flow on each edge to the nearest integer. \Cref{thm:isolation}
guarantees that, with probability at least 1/2, there exists
an integral optimal flow $\begin{bmatrix}x^{\opt}\\
\wt x^{\opt}
\end{bmatrix}$ that entry-wise differs from $\begin{bmatrix}x'^{\final}\\
\wt x'^{\final}
\end{bmatrix}$ by at most $1/3$. So $\begin{bmatrix}x^{\opt}\\
\wt x^{\opt}
\end{bmatrix}$ differs from $\begin{bmatrix}x'^{\apxfinal}\\
\wt x'^{\apxfinal}
\end{bmatrix}$ entry-wise by less than $1/2$. So, with probability at least 1/2,
$\begin{bmatrix}x^{\final}\\
\wt x^{\final}
\end{bmatrix}=\begin{bmatrix}x^{\opt}\\
\wt x^{\opt}
\end{bmatrix}$ is indeed optimal solution for $\wt G$.
Moreover, as $\|(\wt x')^{\final}\|_{\infty}<0.1$,
we have that $\wt x{}^{\final}=0$. Therefore, $x^{\final}$ is also
a feasible flow for $G$ with cost $c^{\top}x^{\final}=c^{\top}x^{\final}+\wt c^{\top}\wt x{}^{\final}=\opt(\wt G).$
As $\opt(\wt G)\le\opt(G)$ by \Cref{fact:OPT modified graph}, we
conclude that $x^{\final}$ is an optimal feasible flow for $G$.
Finally, we can boost the success probability to hold with high probability
by repeating the algorithm $O(\log n)$ times. This concludes the
proof of \Cref{thm:mincost flow}.
\end{proof}

\subsection{Application: Maximum Flow}
\label{sec:maxflow}

In this section, we provide our main results regarding computing a maximum flow. In particular we prove the following corollary.

\begin{corollary}
	[Maximum flow]\label{thm:maxflow}There is an algorithm that, given
	a directed graph $G=(V,E,u,c)$ with $n$ vertices and $m$ edges
	that have integral capacities $u\in\Z_{\ge0}^{E}$, with high probability,
	computes a maximum flow in $\Otil((m+n^{1.5})\log\|u\|_{\infty})$
	time.
\end{corollary}

\Cref{thm:maxflow} is an immediate corollary of our min-cost flow
algorithm from \Cref{thm:mincost flow} with modifications to decrease the $\log^{2}\|u\|_{\infty}$
to $\log\|u\|_{\infty}$. This is done using a standard
scaling technique (e.g.~Section 6 of \cite{ahuja1991distance} and Chapter 2.6 of \cite{williamson2019network}) and we provide the following proof of \Cref{thm:maxflow} only for
completeness. 

\begin{proof}[Proof of Theorem \ref{thm:maxflow}]
For any graph $G$ and flow $f$ in $G$, let $G_{f}$ be the residual
graph of $G$ w.r.t.~$f$. Moreover, let $G_{f}(\Delta)$ denote
the graph obtained $G_{f}$ by removing all edges with capacity in
$G_{f}$ less than $\Delta$. %
Now, consider the following algorithm.
\begin{itemize}
	\item $\Delta=2^{\left\lfloor \log_{2}\|u\|_{\infty}\right\rfloor }$
	\item While $\Delta\ge1$, 
	\begin{itemize}
		\item find a maximum flow $f'$ in $G_{f}(\Delta)$. 
		\item set $f\gets f+f'$ and $\Delta\gets\Delta/2$. 
	\end{itemize}
	\item Return $f$. 
\end{itemize}
Clearly, this algorithm correctly computes a maximum flow because,
after the last iteration when $\Delta=1$, there is no augmenting
path left in $G_{f}$. Observe the following:
\begin{proposition}
	In the beginning of each iteration in the while loop, the value of
	a maximum flow in $G_{f}$ is at most $2m\Delta$. 
\end{proposition}

\begin{proof}
	This is trivially true in the first iteration. For other iterations,
	we know from the previous iteration that the value of a maximum flow
	in $G_{f}(2\Delta)$ is zero. As $G_{f}$ can be obtained from $G_{f}(2\Delta)$
	by adding at most $m$ edges each of which has capacity less than $2\Delta$, the
	value of a maximum flow in $G_{f}$ is at most $2m\Delta$. 
\end{proof}
The above observation implies that, to compute a maximum flow in $G_{f}(\Delta)$,
we can safely cap the capacity each edge to be at most $2m\Delta$.
(That is, if the capacity of $e$ in $G_{f}$ is more than $2m\Delta$,
set it to be $2m\Delta$.) As the ratio between the maximum and minimum
capacity is at most $2m$, we can compute the maximum flow in $G_{f}(\Delta)$
in $\Otil(m+n^{1.5})$ using \Cref{thm:mincost flow}. Since there
are $\left\lfloor \log_{2}\|u\|_{\infty}\right\rfloor $ iterations,
the total running time of algorithm is $\Otil((m+n^{1.5})\log\|u\|_{\infty})$. This completes the proof of \Cref{thm:maxflow}.
\end{proof}          %
\section{General Linear Programs}
\label{sec:linearprogram}

In this section, we prove \Cref{thm:primallp} and \Cref{thm:duallp}.
There are two main parts. For the first part, we show that, given
an initial point $(x^{\init},s^{\init})$, the path following algorithm
(\Cref{alg:implement:short_step}) returns $(x^{\target},s^{\target})$
where $x^{\target}$ is a near-optimal point in $\Otil(mn+n^{2.5})$
time. This is proved in \Cref{sec:pathfollowing_graph} by putting
together the tools developed in previous sections. 

For the second part, we describe how to obtain an initial point $(x^{\init},s^{\init})$
and how to obtain a near optimal primal solution in \Cref{sec:final_primal}
hence proving \Cref{thm:primallp}. In \Cref{sec:final_dual}, we then
show how to extract a near optimal dual solution hence proving \Cref{thm:duallp}.
This part uses standard techniques and we show how to do these tasks
for completeness. 

Our algorithms in this section will use the following linear system solver.
\begin{lemma}[\cite{nn13,lmp13,clmmps15}]
	\label{lem:solver general}There is an algorithm that, given a matrix $\mA\in\R^{m\times n}$,
	a diagonal non-negative weight matrix $\mD\in\R^{m\times m}$, and
	a vector $b\in\R^{n}$ such that there exists a vector $x\in\R^{n}$
	where $(\mA^{\top}\mD\mA)x=b$, w.h.p. returns a vector $\overline{x}$
	such that $\|\overline{x}-x\|_{\mA^{\top}\mD\mA}\le\eps\|x\|_{\mA^{\top}\mD\mA}$ in $\Otil((\nnz(\mA)+n^{\omega})\log\eps^{-1})$ time. 
\end{lemma}
\subsection{Path Following for General LPs}

In Appendix \ref{sec:matrix_data_structures} we give efficient \textsc{HeavyHitter}, \textsc{InverseMaintenance}, and \textsc{HeavySampler} data structures for general linear programs, and use these to show the following.
\begin{lemma}
	\label{lem:pathfollowing_LP}Consider a linear program
	\[
	\Pi:\min_{\substack{\mA^{\top}x=b\\
			\ell_{i}\le x_{i}\le u_{i}\forall i
		}
	}c^{\top}x
	\]
	where $\mA\in\R^{m\times n}$. Let $\eps=1/(4C\log(m/n))$ for a large
	enough constant $C$. Given an $\eps$-centered initial point $(x^{\init},s^{\init},\mu^{\init})$
	for $\Pi$ and a target $\mu^{\target}$, \Cref{alg:implement:short_step}
	(when using matrix-based data structures) returns an $\eps$-centered point $(x^{\target},s^{\target},\mu^{\target})$
	in time 
	\[
	\tilde{O}\left((mn+n^{2.5})\cdot|\log\mu^{\init}/\mu^{\target}|\right).
	\]
\end{lemma}

\begin{proof}
We use \Cref{alg:implement:pathfollowing}, 
which is an implementation of \Cref{algo:pathfollowing}.
By \Cref{lemma:pathfollowing} the algorithm returns 
an $\epsilon$-centered $(x^\target, s^\target, \mu^\target)$.
For the complexity, note that we have a 
$(P,c,Q)$-\textsc{HeavyHitter},
$(P,c,Q)$-\textsc{HeavySampler},
and $(P,c,Q)$-\textsc{InverseMaintenance}
for $P = \tilde{O}(\nnz(\mA) + n^\omega)$, 
$c_i = \tilde{O}(n)$ for $i \in [m]$, 
$Q = \tilde{O}(n^2 + m\sqrt{n})$
according to \Cref{lem:lp:heavyhitter}, \Cref{lem:lp:heavysampler},
and \Cref{lem:lp:inversemaintenance}.
By using \Cref{lem:pathfollowing_complexity} the time complexity of \Cref{alg:implement:pathfollowing}
is bounded by $$
\tilde{O}\left((mn + n^{2.5})\cdot|\log\mu^{\init}/\mu^{\target}|\right).
$$
\end{proof}

\subsection{Initial and Final Primal Solutions}

\label{sec:final_primal}
In this section, we prove \Cref{thm:primallp} using \Cref{lem:pathfollowing_LP}.
Throughout this section, let $\d'=\frac{\delta}{10mW^{2}}$ and let
$\eps=1/(4C\log(m/n))$ for a large enough constant $C$. The path
following algorithm, and therefore \Cref{lem:pathfollowing_LP}, requires an initial point which is $\eps$-centered. However, it is not clear how to obtain such point efficiently. We first
show how to modify the linear program so that it is easy to compute
an $\eps$-centered initial point similar to how we did in \Cref{sec:initfinal_graph}. As described in \Cref{sec:mincostflow}, we construct initial point by adding an identity block to the constraint matrix $\mA$. Interestingly, this is simpler than in previous work \cite{blss20,BrandLN+20} because we are able to handle two-sided constraints.

Given matrix $\mA\in\R^{m\times n}$, vector $b\in\R^{n}$, vector
$c\in\R^{m}$, an accuracy parameter $\delta$, and the linear program
\begin{align}
\min_{\substack{\mA^{\top}x=b\\
		\ell_{i}\le x_{i}\le u_{i}\forall i
	}
}c^{\top}x,\label{eq:originallp}
\end{align}
define matrix $\wt{\mA}\defeq\begin{bmatrix}\mA\\
\beta\mI_{n}
\end{bmatrix}$, $\beta=\|b-\mA^{\top}x^{\init}\|_{\infty}/\Xi$, $\Xi\defeq\max_{i}|u_{i}-\ell_{i}|$,
$x_{i}^{\init}\defeq(\ell_{i}+u_{i})/2$, $\wt x^{\init}\defeq\frac{1}{\beta}(b-\mA^{\top}x^{\init})$.
By flipping the signs of columns of $\mA$, we can assume $\wt x^{\init}\geq0$.
If $\wt x_{i}^{\init}=0$, enforce that coordinate to be $0$ always
by removing the variable $\wt x_{i}^{\init}$ (as it is unnecessary for constructing the initial point), and otherwise define
$\wt{\ell}_{i}=-\Xi$ and $\wt u_{i}=2\wt x_{i}^{\init}+\Xi$ (the terms
$-\Xi$ and $+\Xi$ are just to ensure that $\wt u_{i}>\wt{\ell}_{i}$).
We define $\wt c\defeq2\|c\|_{1}/\delta'$. We overload notation and
let $\wt c$ be the vector in $\R^{n}$ with value $\wt c$ in all
coordinates.

Consider the modified linear program 
\begin{align}
\min_{\substack{\wt{\mA}^{\top}\footnotesize{\begin{bmatrix}x\\
			\wt x
			\end{bmatrix}}=b\\
		\ell_{i}\le x_{i}\le u_{i}\forall i\\
		\wt{\ell}_{i}\le\wt x_{i}\le\wt u_{i}\forall i
	}
}c^{\top}x+\wt c^{\top}\wt x.\label{eq:modifylp}
\end{align}
\begin{lemma}
	[Initial point for the modified LP] \label{lem:initialpointmodified}
	For the linear program in (\ref{eq:modifylp}), the point 
	\[
	\left(\begin{bmatrix}x^{\init}\\
	\wt x^{\init}
	\end{bmatrix},\begin{bmatrix}c\\
	\wt c
	\end{bmatrix},\mu\right)
	\]
	is $\eps$-centered for $\mu=\frac{8m\|c\|_{1}\Xi}{\eps\delta'}$. 
\end{lemma}

\begin{proof}
	By the definition of $\wt x^{\init}=\beta^{-1}(b-\mA^{\top}x^{\init})$
	it is clear that $\wt{\mA}\begin{bmatrix}x\\
	\wt x
	\end{bmatrix}=b$, so the point is exactly feasible. Also, the vector $\begin{bmatrix}c\\
	\wt c
	\end{bmatrix}$ is clearly dual feasible as in \Cref{def:centered} by taking $z = 0$.
	
	Now we bound the centrality. Recall that the barrier function $\phi(x)$
	on the interval $[\ell,u]$ is given by $\phi(x)=-\log(x-\ell)-\log(u-x)$.
	Therefore, $\phi'((\ell+u)/2)=0$, and $\phi''(x)\ge1/(u-\ell)^{2}.$
	In particular, this is lower bounded by $1/\Xi^{2}$ for all original
	constraints. For the new constraints, similarly, we have that $|\wt{u_{i}}-\wt{\ell_{i}}|\leq\|\frac{1}{\beta}(b-\mA^{\top}x^{\init})\|_{\infty}\leq 4 \Xi$.
	Hence, the Hessian $\phi''(x_i)$ is lower bounded by $1/(16\Xi^{2})$ for all constraints.
	
	Therefore, for $x=\begin{bmatrix}x^{\init}\\
	\wt x^{\init}
	\end{bmatrix}$ and $s=\begin{bmatrix}c\\
	\wt c
	\end{bmatrix}$, because $\tau(x)\ge\frac{n}{m}$, we can bound 
	\begin{align*}
	\left\Vert \frac{s+\mu\tau(x)\phi'(x)}{\mu\tau(x)\sqrt{\phi''(x)}}\right\Vert _{\infty}\le 4 mn^{-1}\mu^{-1}\Xi\|s\|_{\infty}\le\eps
	\end{align*}
	by the choice of $\mu$. 
\end{proof}
Next, we show that a solution to the modified LP can be used to round
to nearly feasible and optimal solutions to the original LP. 

\begin{lemma}
	[Final Point for the Original LP] \label{lem:finalpoint_LP}Assume
	that the the linear program in (\ref{eq:originallp}) has a feasible
	solution. Given an $\eps$-centered point for the modified LP in (\ref{eq:modifylp})
	for $\mu=\d'\|c\|_{1}\Xi/Cn$ for some large enough constant $C$
	and some $\d\leq1$, we can in $\Otil((\nnz(\mA)+n^{\omega})\log(W/\d))$ time compute a point $x^{\final}$ that satisfies 
	\begin{align*}
	c^{\top}x^{\final}\le\min_{\substack{\mA^{\top}x=b\\
			\ell_{i}\le x_{i}\le u_{i}\forall i
		}
	}c^{\top}x+\d\text{ and }\|\mA^{\top}x^{\final}-b\|_{\infty}\le\d
	\end{align*}
	and $\ell_{i}\le x_{i}^{\final}\le u_{i}$ for all $i$.
\end{lemma}

\begin{proof}
	By \Cref{lemma:finalpoint} we can compute an exactly feasible
	point $\begin{bmatrix}x^{\final}\\
	\wt x^{\final}
	\end{bmatrix}$, and error $\ls n\mu\le\d'\|c\|_{1}\Xi\le\delta$ off optimal with
	a single solve to $\mA^{\top}\mD\mA$. We discuss at the end how to deal with inexact solvers.
	We claim that $x^{\final}$
	satisfies the necessary conditions. Note that the optimum for (\ref{eq:originallp})
	is at least that of (\ref{eq:modifylp}) as long as (\ref{eq:originallp})
	is feasible. Therefore, we have that 
	\begin{align*}
	c^{\top}x^{\final} & \le\min_{\substack{\mA^{\top}x=b\\
			\ell_{i}\le x_{i}\le u_{i}\forall i
		}
	}c^{\top}x-\wt c^{\top}\wt x^{\final}+\delta\le\min_{\substack{\mA^{\top}x=b\\
			\ell_{i}\le x_{i}\le u_{i}\forall i
		}
	}c^{\top}x+\d.
	\end{align*}
	Finally, we show that $\|\mA^{\top}x^{\final}-b\|_{\infty}\le\d'\|\mA^{\top}x^{\init}-b\|_{\infty}$,
	this suffices for the first claim. 
	
	Because $x^{\final}$ is feasible, note that $\|\mA^{\top}x^{\final}-b\|_{\infty}=\beta\|\wt x^{\final}\|_{\infty}.$
	Assume $\|\mA^{\top}x^{\final}-b\|_{\infty}\le\d'\|\mA^{\top}x^{\init}-b\|_{\infty}$
	is false, using $\beta=\|\mA^{\top}x^{\init}-b\|_{\infty}/\Xi$,
	then there is some $i$ such that $\wt x_{i}^{\final}>\d'\Xi$. In
	this case, we get that 
	\begin{align*}
	c^{\top}x^{\final} & =\left(c^{\top}x^{\final}+\wt c^{\top}\wt x^{\final}\right)-\wt c^{\top}\wt x^{\final}\le\min_{\substack{\mA^{\top}x=b\\
			\ell_{i}\le x_{i}\le u_{i}\forall i
		}
	}c^{\top}x+\d'\|c\|_{1}\Xi-\wt c^{\top}\wt x^{\final}\\
	& <\sum_{i}\min(c_{i}\ell_{i},c_{i}u_{i})+2\|c\|_{1}\Xi-\wt c\d'\Xi=\sum_{i}\min(c_{i}\ell_{i},c_{i}u_{i})
	\end{align*}
	This is clearly a contradiction, as $c^{\top}x^{\final}\ge\sum_{i}\min(c_{i}\ell_{i},c_{i}u_{i})$
	because $x^{\final}$ is feasible. Therefore, we have 
	\[
	\|\mA^{\top}x^{\final}-b\|_{\infty}\le\d'\|\mA^{\top}x^{\init}-b\|_{\infty}\le\delta
	\]
	because $\|\mA^{\top}x^{\init}-b\|_{\infty}\le m\|\mA\|_{\infty}\max\{\|u\|_{\infty},\|\ell\|_{\infty}\}+\|b\|_{\infty}\}\le2mW^{2}$.
	
	Finally, we have an approximate instead of exact solver for $\begin{bmatrix}x^{\final}\\ \wt x^{\final} \end{bmatrix}$. However, we can compute the vector where each entry differs from $\begin{bmatrix}x^{\final}\\ \wt x^{\final}\end{bmatrix}$ at most $\d_\inexact$ in time $\Otil((\nnz(\mA)+n^\omega)\log(W/\d_\inexact))$ using \Cref{lem:solver general}. As long as we can guarantee that $\d_\inexact \le 1/\poly(mW\lambda_{\min}(\wt{\mA}^\top \mD \wt{\mA}))$, all entries will differ by an arbitrarily small polynomial, so this approximate vector satisfies all the conditions of \Cref{lem:finalpoint_LP} as well. Recall that in Lemma \ref{lemma:finalpoint} that $\mD = \bar{\Tau}^{-1}\Phi''(x)^{-1}$, and $\log \Phi''(x)^{-1} \ge -\tO(\log W + \log \mu^\final + \log \|c\|_\infty)$ by Lemma \ref{lemma:paramchange}. Hence
\begin{align*}
\log\lambda_{\min}(\mA^{\top}\mD\mA) & \ge-\tO(\log W+\log(1/\mu^{\final})+\log\|c\|_{\infty})+\log\lambda_{\min}(\wt{\mA}^{\top}\wt{\mA})\\
 & \geq-\tO(\log(W/\delta))
\end{align*}
where we used that $\lambda_{\min}(\wt{\mA}^{\top}\wt{\mA}) \geq \beta^2 \geq \poly(1/(mW))$ and $\mu^{\final} \geq \delta \poly(1/(mW))$.
\end{proof}

Now, we are ready to prove \Cref{thm:primallp} by combining the
two lemmas above.

\begin{proof}[Proof of \Cref{thm:primallp}]

Using \Cref{lem:initialpointmodified}, we can obtain an $\eps$-centered
initial point 
\[
\left(\begin{bmatrix}x^{\init}\\
\wt x^{\init}
\end{bmatrix},\begin{bmatrix}c\\
\wt c
\end{bmatrix},\mu^{\init}\right)
\]
for the modified linear program (\ref{eq:modifylp}) in $O(\nnz(\mA))$
time where $\mu^{\init}=4m\|c\|_{1}\Xi/\eps\delta'$. 

Let $\mu^{\target}=\d'\|c\|_{1}\Xi/Cn$ where $C$ is a large enough
constant. We invoke \Cref{lem:pathfollowing_LP} which returns $\left(\begin{bmatrix}x^{\target}\\
\wt x^{\target}
\end{bmatrix},\begin{bmatrix}s^{\target}\\
\wt s^{\target}
\end{bmatrix},\mu^{\target}\right)$ in time $\tilde{O}((mn+n^{2.5})\cdot\log^{2}W)$. This running time
follows because $\mu^{\init}/\mu^{\target}=\poly(mW/\d)$. It remains
to bound the ratio $W''$ of largest to smallest entry in $\phi''(\begin{bmatrix}x^{\init}\\
\wt x^{\init}
\end{bmatrix})$. For each entry in $\wt x^{\init}$, we have $\phi''(\wt x_{i}^{\init})=2/(\wt x_{i}^{\init}+\Xi)^{2}$.
Note that $\Xi\le\wt x_{i}^{\init}+\Xi\le2\Xi\ls W$, where $\Xi = \max_{i}|u_{i}-\ell_{i}|$ was set previously
For each entry in $x^{\init}$, we have $\phi''(x_{i}^{\init})=\Theta(1/(u_{i}-\ell_{i})^{2})$.
As $W\ge\frac{\max_{i}(u_{i}-\ell_{i})}{\min_{i}(u_{i}-\ell_{i})}$,
we have that indeed the ratio $W''\le\poly(mW)$ as desired. 

As the last step, we give $\left(\begin{bmatrix}x^{\target}\\
\wt x^{\target}
\end{bmatrix},\begin{bmatrix}s^{\target}\\
\wt s^{\target}
\end{bmatrix},\mu^{\target}\right)$ as an input to \Cref{lem:finalpoint_LP} and obtain a final point
$x^{\final}$ which satisfies all condition of \Cref{thm:primallp}
in time $\Otil((\nnz(\mA)+n^{\omega})\log (W/\d))$. In total, the running time of the algorithm is
$\tO((mn+n^{2.5})\log(W/\d))$. 
\end{proof}

\subsection{Final Dual Solutions}

\label{sec:final_dual}

In this section, we prove \Cref{thm:duallp}. Before proving it, we
show the key subroutine for extracting a dual solution from the
primal LP as formalized below. By scaling $x$ and $\mA$, we can focus on the $\ell = -1, u = 1$ case.
Additionally, this formulation suffices for our application of solving MDPs in \Cref{cor:DMDP}.

\begin{lemma}[Dual solution bound]
	 \label{lemma:duallp} Given an $\eps$-centered
	point $(x,s,\mu)$ where $\eps\le1/(C \log(m/n))$ and $\mu\le\d/Cn$ for sufficiently
	large $C$ to the LP 
	\begin{equation}
	\min_{\substack{\mA^{\top}x=0\\
			-1\le x_{i}\le1\forall i
		}
	}c^{\top}x,\label{eq:l1regression_LP}
	\end{equation}
	we can compute 
	in time $\Otil((\nnz(\mA)+n^{\omega})\log(W/\d))$
	a vector $z\in\R^{n}$ satisfying 
	\[
	\|\mA z+c\|_{1}\le\min_{z\in\R^{n}}\|\mA z+c\|_{1}+\d.
	\]
\end{lemma}

\begin{proof}
	Define 
	\[
	\opt=\min_{\substack{\mA^{\top}x=0\\
			-1\le x_{i}\le1\forall i
		}
	}c^{\top}x.
	\]
	By Sion's minimax theorem, we have that 
	\begin{align*}
	& \min_{\substack{\mA^{\top}x=0\\
			-1\le x_{i}\le1\forall i
		}
	}c^{\top}x=\min_{\|x\|_{\infty}\le1}\max_{z\in\R^{n}}c^{\top}x+z^{\top}\mA^{\top}x\\
	& =\max_{z\in\R^{n}}\min_{\|x\|_{\infty}\le1}(\mA z+c)^{\top}x=\max_{z\in\R^{n}}-\|\mA z+c\|_{1}=-\min_{z\in\R^{n}}\|\mA z+c\|_{1}.
	\end{align*}
	Hence, we have that $\min_{z\in\R^{n}}\|\mA z+c\|_{1}=-\opt.$ Now,
	find $x^{\final}$ and $s^{\final}=\mA z+c$
	satisfying the conclusions of \Cref{lemma:finalpoint}. Specifically, $z$ may be computed as
	$z = (\mA^\top \mA)^{-1}\mA^\top(s^\final - c)$ in time $\tO((\nnz(A)+n^\omega)\log \d^{-1})$ 
	by \Cref{lem:solver general}.

	We now prove that this $z$ satisfies the conclusions of this lemma: $\|\mA z+c\|_{1}\le\min_{z\in\R^{n}}\|\mA z+c\|_{1}+\d.$
	To simplify notation, we will write $x$ for $x^{\final}$
	and $s$ for $s^{\final}$. Note that because $\mA^{\top}x=0$, 
	\[
	\|\mA z+c\|_{1}=(\mA z)^{\top}x+\|s\|_{1}=-c^{\top}x+x^{\top}s+\|s\|_{1}\le-\opt+x^{\top}s+\|s\|_{1},
	\]
	as $c^{\top}x\ge\opt$ by definition. Therefore, it suffices to bound
	$x^{\top}s+\|s\|_{1}$. We will show that $x_{i}s_{i}+|s_{i}|\ls\mu\tau_{i}$.
	This suffices, as then 
	\[
	x^{\top}s+\|s\|_{1}\ls\sum_{i}\mu\tau_{i}\ls n\mu\le\d
	\]
	for sufficiently large choice of $C$ in the definition of $\mu$.
	
	To show this, define $y=\frac{s+\mu\tau\phi'(x)}{\mu\tau\sqrt{\phi''(x)}}$,
	so that $\|y\|_{\infty}\le\frac{1}{10}$ by Claim \ref{claim:finalpoint_technical}, which shows approximate centrality of the final point returned.
	We have that
	\begin{align}
	y_{i}\sqrt{\phi_{i}''(x_{i})}-\phi_{i}'(x_{i})=\mu^{-1}\tau_{i}^{-1}s_{i}. \label{eq:l1_y}
	\end{align}
If $|x_{i}|\leq\frac{1}{2}$, then both $\phi''_{i}(x_{i})$ and $|\phi'_{i}(x_{i})|$
are $O(1)$. Hence, \cref{eq:l1_y} shows that $x_{i}s_{i}+|s_{i}|\ls|s_{i}|\ls\mu\tau_{i}$. If $|x_{i}|\geq\frac{1}{2}$, we have
\[
y_{i}\sqrt{\phi''_{i}(x_{i})}=y_{i}\sqrt{\frac{1}{(1-x_{i})^{2}}+\frac{1}{(1+x_{i})^{2}}}\leq\frac{2y_{i}}{1 - |x_{i}|}\leq\frac{1}{5(1 - |x_{i}|)}\leq |\phi_{i}'(x_{i})|.
\]
Hence, \cref{eq:l1_y} shows that $0\geq\mu^{-1}\tau_{i}^{-1}s_{i}\geq-\frac{1}{1-x_{i}}$
and
\[
x_{i}s_{i}+|s_{i}|=-(1-x_{i})s_{i}\leq\mu\tau_{i}.
\] The proof for $x_{i}\leq-\frac{1}{2}$ is similar. Therefore, we proved the claim  $x_{i}s_{i}+|s_{i}|\ls\mu\tau_{i}$ in all the cases.
Handling the inexactness of solvers may be done as in the proof of Lemma \ref{lem:finalpoint_LP}.
\end{proof}
With the above lemma at hand, we are ready to prove \Cref{thm:duallp}.

\begin{proof}[Proof of \Cref{thm:duallp}]
First, observe that we can easily define an initial solution $$(x^{\init},s^{\init},\mu^{\init})\defeq(0,c,\eps^{-1}\|c\|_{\infty}m/n)$$
of (\ref{eq:l1regression_LP}). Note that $(x^{\init},s^{\init},\mu^{\init})$
is $\eps$-centered because $x^{\init}=0$ and $s^{\init}=c$ are
exactly feasible. It remains to bound the centrality. As the barrier
function is defined on $[-1,1]$, we have $\phi'(0)=0$, and $\phi''(0)=2$
by \Cref{fact:bound phi}. As $\tau(x^{\init})\ge\frac{n}{m}$, we
have
\[
\left\Vert \frac{s^{\init}+\mu^{\init}\tau(x^{\init})\phi'(x^{\init})}{\mu^{\init}\tau(x^{\init})\sqrt{\phi''(x^{\init})}}\right\Vert _{\infty}\le\frac{m\|c\|_{\infty}}{\mu n}\le\eps
\]
as desired.

Then, we invoke \Cref{lem:pathfollowing_LP} to obtain an $\eps$-centered
point $(x^{\target},s^{\target},\mu^{\target})$ of (\ref{eq:l1regression_LP})
where $\mu^{\target}=\d/Cn$. \Cref{lem:pathfollowing_LP} takes time
$\tilde{O}\left((mn+n^{2.5})\cdot\log(W/\delta)\right)$ because $\mu^{\init}/\mu^{\target}=\poly(mW/\delta)$.

Lastly, we invoke \Cref{lemma:duallp} to obtain a vector $z$ which
satisfies the guarantee of \Cref{thm:duallp} in time $\Otil((\nnz(\mA)+n^{\omega})\log(W/\delta))$
by using \Cref{lem:solver general}. 
Therefore, the total running time is $\tO\left((mn+n^{2.5})\log(W/\d)\right)$.
\end{proof}

\subsection{Application: Discounted Markov Decision Process}
\label{sec:app:mdp}

In this section, we show an algorithm (\Cref{cor:DMDP}) for solving
the discounted Markov Decision Process problem in time $\Otil((|S|^{2}|A|+|S|^{2.5})\log(\frac{M}{(1-\gamma)\epsilon}))$
(these parameters are defined below). 
See \Cref{tab:DMDP} for the comparison with previous results. Below, we formally define
the problem and prove this bound.

\begin{table}
\scriptsize
	\begin{centering}
		\begin{tabular}{|c|c|c|c|c|}
			\hline 
			{\bf Year} & {\bf Authors} & {\bf Refs} & \textbf{Algorithm} & \textbf{Running Time up to } $\widetilde{O}(\cdot)$ \tabularnewline
			\hline 
			\hline 
			1990 & Tseng, Littman, Dean, Kaelbling & \cite{tseng1990solving,LittmanDK95} & Value Iteration & $ \S^{2}\A \frac{1}{1-\gamma} $\tabularnewline
			\hline 
			2014 & Lee, Sidford & \cite{ls14,LeeS15,SidfordWWY18} & IPM & $\S^{2.5}\A$\tabularnewline
			\hline 
			2018 & Sidford, Wang, Wu, Yang & \cite{SidfordWWY18} & High Precision RVI & $|S|^{2}|A|+\frac{|S||A|}{(1-\gamma)^{3}}$\tabularnewline
			\hline 
			2020 & & This paper & Robust IPM & $S|^{2}|A|+|S|^{2.5}$\tabularnewline
			\hline 
		\end{tabular}
		\par\end{centering}
	\caption{\label{tab:DMDP} High-precision algorithms for discounted MDPs.
		This table gives only the fastest high precision results, i.e. those that depend polylogarithmically on $\epsilon$.
		We ignore the $\widetilde{O}$ in the running time, and suppress factors of $\log\left(\max(|S|, |A|, \eps^{-1}, (1-\gamma)^{-1}, M)\right).$
		RVI denotes Randomized Value Iteration. \cite{SidfordWWY18} showed that discounted MDPs could be solved via linear programming, so previous results towards faster linear programming immediately gave results for discounted MDPs.}

\end{table}

\paragraph{Problem Definition.}

A discounted Markov Decision Process (DMDP) is specified by a tuple
$(S,A,P,r,\gamma)$ where $S$ is a the finite state space, $A$ is
the finite action space, $P=\{p_{a}\}_{a\in A}$ is the collection
of state-action-state transition probabilities where $p_{a}(i,j)$
denote the probability of going to state $j$ from state $i$ when
taking action $a$, $r$ is the collection of state-action rewards
where $r_{a}(i)\in[-M,M]$ is the collected reward when we are currently
in state $i$ and take action $a$, and $\gamma\in(0,1)$ is a discount
factor.

 Given a DMDP $(S,A,P,r,\gamma)$, the \emph{value operator}
$T:\mathbb{R}^{S}\rightarrow\mathbb{R}^{S}$ is defined for all $u\in\mathbb{R}^{S}$
and $i\in S$ by
\[
T(u)_{i}=\max_{a\in A}[r_{a}(i)+\gamma\cdot p_{a}(i)^{\top}u]
\]
where $p_{a}(i)\in\mathbb{R}^{S}$ with $(p_{a}(i))_{j}=p_{a}(i,j)$.
It is known that there is a unique vector $v^{*}$ such that $T(v^{*})=v^{*}$. 

A vector $\pi\in A^{S}$ that tells the actor which action to take
from any state is called a\emph{ policy }and $\pi_{i}$ denotes the
action prescribed by $\pi$ to be taken at state $i\in S$. The \emph{value
	operator associated }$T_{\pi}$ \emph{with $\pi$} is defined for
all $u\in\mathbb{R}^{S}$ and $i\in S$ by 
\[
T_{\pi}(u)_{i}=r_{\pi_{i}}(i)+\gamma\cdot p_{\pi_{i}}(i)^{\top}u.
\]
Note that $T_{\pi}$ can be viewed as the value operator for the modified
DMDP where the only available action from each state is given by the
policy $\pi$. Let $v_{\pi}$ denote the unique vector such that $T_{\pi}(v_{\pi})=v_{\pi}$. 

We says that values $u\in\mathbb{R}^{S}$ are $\eps$-optimal if $\|v^{*}-u\|{}_{\infty}\le\eps$
and we say that a policy $\pi\in A^{S}$ is $\eps$-optimal if $\|v^{*}-v_{\pi}\|_{\infty}\le\eps$,
i.e.~the values of the policy $\pi$ are $\eps$-optimal. Our goal
is to find an $\eps$-optimal policy $\pi$. 

\paragraph{Reductions to Solving LP and $\ell_{1}$ Regression.}

In Section B of \cite{SidfordWWY18}, the authors show how to reduce
the problem of finding an $\eps$-optimal policy to finding $\epsilon$-approximate
solution of the following LP. We leverage this reduction along with \Cref{thm:duallp} ($\ell_1$ regression) to prove our main result \Cref{cor:DMDP}.
\begin{definition}
	[DMDP linear program] We call the following linear program the \emph{DMDP
		LP} 
	\begin{equation}
	\min_{\mA v \ge r} v^\top \vones. \label{eq:optimal_lp}
	\end{equation}
	where $r\in\R^{(S\times A)}$ is the vector of rewards, i.e. $r_{i,a}=r_{a}(i)$
	for all $i\in S$ and $a\in A$ and $\ma=\mE-\gamma\mproj$ where
	$\mE\in\R^{(S\times A)\times S}$ is the matrix where for all $i,j\in S$
	and $a\in A$ we have that the $j$-th entry of row $(i,a)$ of $\mE$
	is $1$ if $i=j$ and $0$ otherwise, and $\mproj\in\R^{(S\times A)\times S}$
	is a matrix where for all $i\in S$ and $a\in A$ we have that row
	$(i,a)$ of $\mproj$ is $p_{a}(i)$. We call a vector $v\in\valuespace$
	an \emph{$\epsilon$-approximate DMDP LP solution} if $\ma v\geq r-\epsilon\vones$
	and $v^{\top}\vones\leq OPT+\epsilon$ where $OPT$ is the optimal
	value of the DMDP LP, \Cref{eq:optimal_lp}.
\end{definition}

The reduction is stated as follows:
\begin{lemma}
	[Lemma B.3 of \cite{SidfordWWY18}] \label{lem:apx_policy} If $v$
	is an $\epsilon$-approximate DMDP LP solution and if $\pi\in\policyspace$
	is defined with $\pi_{i}=\argmax_{a\in\actions}r_{a}+\gamma\cdot p(i,a)^{\top}v$
	for all $i\in\states$ then $\pi$ is an $8\epsilon|S|(1-\gamma)^{-2}$-optimal
	policy. 
\end{lemma}

\cite{SidfordWWY18} further reduces the problem of finding an $\epsilon$-approximate
DMDP LP solution to finding a $\epsilon$-optimal solution of the
following instance of the $\ell_{1}$ regression problem.

\begin{definition}
	[DMDP $\ell_1$-regression] For a given DMDP $(S,A,P,r,\gamma)$ and
	a parameter $\alpha$ we call the following $\ell_{1}$ regression
	problem the \emph{DMDP $\ell_{1}$ problem}
	\begin{align}
	\min_{v}f(v) & =\left|\alpha\left(\frac{\S M}{1-\gamma}+\vones^{\top}v\right)\right|+\left\Vert \ms^{-1}\ma v-\ms^{-1}b-\vones\right\Vert _{1}+\left\Vert \ms^{-1}\ma v-\ms^{-1}b+\vones\right\Vert _{1}\label{eq:new_formulation_lp}
	\end{align}
	where $s=\frac{1}{2}(\frac{2M}{1-\gamma}\vones-r)$ and $b=\frac{1}{2}(\frac{2M}{1-\gamma}\vones+r)$,
	and $\ms=\diag(s)$. We let $\vopt_{f}$ denote the optimal solution
	to this $\ell_{1}$ regression problem and we call $v$ an $\epsilon$-optimal
	solution to $f$ if $f(v)\leq f(\vopt_{f})+\epsilon$. 
\end{definition}

Observe that (\ref{eq:new_formulation_lp}) is indeed an instance of
the $\ell_{1}$ regression problem from \Cref{thm:duallp} where the
input matrix and the input vector $c$ are defined as 
\[
\left[\begin{array}{c}
\ms^{-1}\ma\\
\ms^{-1}\ma\\
\alpha\vones^{\top}
\end{array}\right]\in\R^{(2|S\times A|+1)\times|S|}\text{ and }c=\left[\begin{array}{c}
\ms^{-1}b+\vones\\
\ms^{-1}b-\vones\\
\alpha\frac{\S M}{1-\gamma}
\end{array}\right]\in\R^{(2|S\times A|+1)}.
\]

The next reduction is stated as follows:
\begin{lemma}
	[Lemma B.7 of \cite{SidfordWWY18}]\label{lem:l1_regress_qual}Suppose
	that $v$ is an $\epsilon$-approximate solution to the DMDP $\ell_{1}$
	problem then $v$ is an $\epsilon'$-approximate DMDP LP solution
	for 
	\[
	\epsilon'\leq\max\left\{ \frac{\epsilon}{\alpha}~,~\frac{2\alpha|S|M^{2}}{(1-\gamma)^{2}}+\frac{\epsilon M}{(1-\gamma)}\right\} .
	\]
\end{lemma}

By combining the above two reductions and the $\ell_{1}$-regression
algorithm from \Cref{thm:duallp}, we obtain the new algorithm for
solving DMDP, proving \Cref{cor:DMDP}.

\begin{proof}[Proof of \Cref{cor:DMDP}]
	To obtain an $\epsilon$-optimal policy $\pi$, by \Cref{lem:apx_policy},
	it suffices to compute an $\ensuremath{\eps_{2}}$-approximate DMDP
	LP solution $v$ where $\eps_{2}=\frac{\epsilon(1-\gamma)^{2}}{8\S}$.
	Let $v$ be a $\eps_{3}$-approximate solution to the DMDP $\ell_{1}$-regression
	instance in (\ref{eq:new_formulation_lp}) where $\ensuremath{\alpha=\frac{\eps_{2}(1-\gamma)^{2}}{4|S|M^{2}}}$
	and $\eps_{3}=\min\{\alpha\eps_{2},\frac{\eps_{2}(1-\gamma)}{2M}\}$.
	By \Cref{lem:l1_regress_qual}, $v$ is $\eps'$-approximate DMDP LP
	solution where 
	\[
	\eps'\le\max\left\{ \frac{\epsilon_{3}}{\alpha}~,~\frac{2\alpha|S|M^{2}}{(1-\gamma)^{2}}+\frac{\epsilon_{3}M}{(1-\gamma)}\right\} \le\max\left\{ \eps_{2}~,~\eps_{2}/2+\eps_{2}/2\right\} =\eps_{2}
	\]
	as desired. To solve the $\ell_{1}$-regression instance in (\ref{eq:new_formulation_lp}),
	we invoke the algorithm from \Cref{thm:duallp}. As $\eps_{3}\ge\poly(\frac{\eps\cdot(1-\gamma)}{M|S|})$
	and the input matrix has size $(2|S\times A|+1)\times|S|$ and the
	input vector has size $(2|S\times A|+1)$, with high probability,
	we obtain $v$ in time $\Otil((|S|^{2}|A|+|S|^{2.5})\log(\frac{M}{(1-\gamma)\epsilon}))$.
\end{proof}
Note that the running time above is nearly linear time because the
size of the input is $\Omega(|S|^{2}|A|)$. Note that \Cref{cor:DMDP}
gives the first nearly linear time algorithm for which dependencies
on $M,\frac{1}{\eps},\frac{1}{(1-\gamma)}$ are all logarithmic.  %

\section*{Acknowledgment}

This project has received funding from the European Research Council (ERC) under the European
Unions Horizon 2020 research and innovation programme under grant agreement No 715672. 
Jan van den Brand is partially supported by the Google PhD Fellowship Program.
Aaron Sidford is supported by a Microsoft Research Faculty Fellowship, NSF CAREER Award CCF-1844855,  NSF Grant CCF-1955039, a PayPal research award, and a Sloan Research Fellowship. Yin Tat Lee is supported by NSF awards CCF-1749609, CCF-1740551, DMS-1839116, DMS-2023166, Microsoft Research Faculty Fellowship, Sloan Research Fellowship, Packard Fellowships.
Di Wang did part of this work while at Georgia Tech,
and was partially supported by NSF grant CCF-1718533.
Yang P. Liu was supported by the Department of Defense (DoD) through the National Defense
Science and Engineering Graduate Fellowship (NDSEG) Program. We sincerely thank Danupon Nanongkai and Richard Peng for helpful discussions on this project. Also, we thank Danupon Nanongkai for his help on improving the quality in this paper.

\ifdefined\DEBUG
\newpage
\fi

\bibliographystyle{alpha}
\bibliography{ref}

\newcommand{\etalchar}[1]{$^{#1}$}
\begin{thebibliography}{CTCG{\etalchar{+}}98}

\bibitem[AKY20]{AgarwalKY20}
Alekh Agarwal, Sham~M. Kakade, and Lin~F. Yang.
\newblock Model-based reinforcement learning with a generative model is minimax
  optimal.
\newblock In {\em {COLT}}, volume 125 of {\em Proceedings of Machine Learning
  Research}, pages 67--83. {PMLR}, 2020.

\bibitem[AMK13]{azar2013minimax}
Mohammad~Gheshlaghi Azar, R{\'e}mi Munos, and Hilbert~J Kappen.
\newblock Minimax pac bounds on the sample complexity of reinforcement learning
  with a generative model.
\newblock {\em Machine learning}, 91(3):325--349, 2013.

\bibitem[AMV20]{amv20}
Kyriakos Axiotis, Aleksander Madry, and Adrian Vladu.
\newblock Circulation control for faster minimum cost flow in unit-capacity
  graphs.
\newblock In {\em {FOCS}}, pages 93--104. {IEEE}, 2020.

\bibitem[AO91]{ahuja1991distance}
Ravindra~K Ahuja and James~B Orlin.
\newblock Distance-directed augmenting path algorithms for maximum flow and
  parametric maximum flow problems.
\newblock {\em Naval Research Logistics (NRL)}, 38(3):413--430, 1991.

\bibitem[AO06]{AO06}
Peter Auer and Ronald Ortner.
\newblock Logarithmic online regret bounds for undiscounted reinforcement
  learning.
\newblock In Bernhard Sch{\"{o}}lkopf, John~C. Platt, and Thomas Hofmann,
  editors, {\em Advances in Neural Information Processing Systems 19,
  Proceedings of the Twentieth Annual Conference on Neural Information
  Processing Systems, Vancouver, British Columbia, Canada, December 4-7, 2006},
  pages 49--56. {MIT} Press, 2006.

\bibitem[AW21]{aw21}
Josh Alman and Virginia~Vassilevska Williams.
\newblock A refined laser method and faster matrix multiplication.
\newblock In {\em SODA}. \url{https://arxiv.org/pdf/2010.05846.pdf}, 2021.

\bibitem[BBG{\etalchar{+}}20]{BernsteinBNPSS20}
Aaron Bernstein, Jan van~den Brand, Maximilian~Probst Gutenberg, Danupon
  Nanongkai, Thatchaphol Saranurak, Aaron Sidford, and He~Sun.
\newblock Fully-dynamic graph sparsifiers against an adaptive adversary.
\newblock {\em CoRR}, abs/2004.08432, 2020.
\newblock \url{https://arxiv.org/pdf/2004.08432}.

\bibitem[BLN{\etalchar{+}}20]{BrandLN+20}
Jan van~den Brand, Yin~Tat Lee, Danupon Nanongkai, Richard Peng, Thatchaphol
  Saranurak, Aaron Sidford, Zhao Song, and Di~Wang.
\newblock Bipartite matching in nearly-linear time on moderately dense graphs.
\newblock In {\em {FOCS}}, pages 919--930. {IEEE}, 2020.

\bibitem[BLSS20]{blss20}
Jan van~den Brand, Yin~Tat Lee, Aaron Sidford, and Zhao Song.
\newblock Solving tall dense linear programs in nearly linear time.
\newblock In {\em {STOC}}, pages 775--788. {ACM}, 2020.
\newblock \url{https://arxiv.org/pdf/2002.02304.pdf}.

\bibitem[BNS19]{BrandNS19}
Jan van~den Brand, Danupon Nanongkai, and Thatchaphol Saranurak.
\newblock Dynamic matrix inverse: Improved algorithms and matching conditional
  lower bounds.
\newblock In {\em {FOCS}}, pages 456--480. {IEEE} Computer Society, 2019.

\bibitem[Bra20]{b20}
Jan van~den Brand.
\newblock A deterministic linear program solver in current matrix
  multiplication time.
\newblock In {\em SODA}, pages 259--278. {SIAM}, 2020.

\bibitem[Bra21]{b21}
Jan van~den Brand.
\newblock Unifying matrix data structures: Simplifying and speeding up
  iterative algorithms.
\newblock In {\em {SOSA}}, pages 1--13. {SIAM}, 2021.

\bibitem[CDM{\etalchar{+}}16]{ClarksonDMMMW16}
Kenneth~L. Clarkson, Petros Drineas, Malik Magdon{-}Ismail, Michael~W. Mahoney,
  Xiangrui Meng, and David~P. Woodruff.
\newblock The fast cauchy transform and faster robust linear regression.
\newblock {\em {SIAM} J. Comput.}, 45(3):763--810, 2016.

\bibitem[CJN18]{cjn18}
Michael~B Cohen, TS~Jayram, and Jelani Nelson.
\newblock Simple analyses of the sparse johnson-lindenstrauss transform.
\newblock In {\em 1st Symposium on Simplicity in Algorithms (SOSA)}. Schloss
  Dagstuhl-Leibniz-Zentrum fuer Informatik, 2018.

\bibitem[CKM{\etalchar{+}}14]{CohenKMPPRX14}
Michael~B. Cohen, Rasmus Kyng, Gary~L. Miller, Jakub~W. Pachocki, Richard Peng,
  Anup~B. Rao, and Shen~Chen Xu.
\newblock Solving sdd linear systems in nearly $m\log^{1/2}n$ time.
\newblock In {\em Proceedings of the 46th Annual ACM Symposium on Theory of
  Computing (STOC)}, pages 343--352, 2014.

\bibitem[Cla05]{Clarkson05}
Kenneth~L Clarkson.
\newblock Subgradient and sampling algorithms for l 1 regression.
\newblock In {\em SODA}, pages 257--266, 2005.

\bibitem[CLM{\etalchar{+}}15]{clmmps15}
Michael~B Cohen, Yin~Tat Lee, Cameron Musco, Christopher Musco, Richard Peng,
  and Aaron Sidford.
\newblock Uniform sampling for matrix approximation.
\newblock In {\em Proceedings of the 2015 Conference on Innovations in
  Theoretical Computer Science (ITCS)}, pages 181--190, 2015.

\bibitem[CLS19]{cls19}
Michael~B. Cohen, Yin~Tat Lee, and Zhao Song.
\newblock Solving linear programs in the current matrix multiplication time.
\newblock In {\em {STOC}}, pages 938--942. {ACM}, 2019.
\newblock \url{https://arxiv.org/pdf/1810.07896}.

\bibitem[CMMP13]{ChinMMP13}
Hui~Han Chin, Aleksander Madry, Gary~L. Miller, and Richard Peng.
\newblock Runtime guarantees for regression problems.
\newblock In Robert~D. Kleinberg, editor, {\em Innovations in Theoretical
  Computer Science, {ITCS} '13, Berkeley, CA, USA, January 9-12, 2013}, pages
  269--282. {ACM}, 2013.

\bibitem[CMSV17]{cmsv17}
Michael~B Cohen, Aleksander Madry, Piotr Sankowski, and Adrian Vladu.
\newblock Negative-weight shortest paths and unit capacity minimum cost flow in
  ${O}(m^{10/7} \log {W})$ time.
\newblock In {\em Proceedings of the Twenty-Eighth Annual ACM-SIAM Symposium on
  Discrete Algorithms (SODA)}, pages 752--771. SIAM, 2017.

\bibitem[CP15]{CohenP15}
Michael~B. Cohen and Richard Peng.
\newblock $\ell_p$ row sampling by lewis weights.
\newblock In {\em STOC}, pages 183--192. {ACM}, 2015.

\bibitem[CTCG{\etalchar{+}}98]{CCGMQ98}
Jean Cochet-Terrasson, Guy Cohen, St{\'e}phane Gaubert, Michael McGettrick, and
  Jean-Pierre Quadrat.
\newblock Numerical computation of spectral elements in max-plus algebra.
\newblock {\em IFAC Proceedings Volumes}, 31(18):667--674, 1998.

\bibitem[CW13]{ClarksonW13}
Kenneth~L. Clarkson and David~P. Woodruff.
\newblock Low rank approximation and regression in input sparsity time.
\newblock In {\em {STOC}}, pages 81--90. {ACM}, 2013.

\bibitem[DG98]{DG98}
Ali Dasdan and Rajesh~K. Gupta.
\newblock Faster maximum and minimum mean cycle algorithms for
  system-performance analysis.
\newblock {\em {IEEE} Trans. Comput. Aided Des. Integr. Circuits Syst.},
  17(10):889--899, 1998.

\bibitem[DHNV20]{dhnv20}
Daniel Dadush, Sophie Huiberts, Bento Natura, and L{\'{a}}szl{\'{o}}~A.
  V{\'{e}}gh.
\newblock A scaling-invariant algorithm for linear programming whose running
  time depends only on the constraint matrix.
\newblock In Konstantin Makarychev, Yury Makarychev, Madhur Tulsiani, Gautam
  Kamath, and Julia Chuzhoy, editors, {\em Proccedings of the 52nd Annual {ACM}
  {SIGACT} Symposium on Theory of Computing, {STOC} 2020, Chicago, IL, USA,
  June 22-26, 2020}, pages 761--774. {ACM}, 2020.

\bibitem[Din70]{Dinic70}
Efim~A Dinic.
\newblock Algorithm for solution of a problem of maximum flow in networks with
  power estimation.
\newblock In {\em Soviet Math. Doklady}, volume~11, pages 1277--1280, 1970.

\bibitem[DLS18]{DurfeeLS18}
David Durfee, Kevin~A. Lai, and Saurabh Sawlani.
\newblock $\ell_1$ regression using lewis weights preconditioning and
  stochastic gradient descent.
\newblock In S{\'{e}}bastien Bubeck, Vianney Perchet, and Philippe Rigollet,
  editors, {\em Conference On Learning Theory, {COLT} 2018, Stockholm, Sweden,
  6-9 July 2018}, volume~75 of {\em Proceedings of Machine Learning Research},
  pages 1626--1656. {PMLR}, 2018.

\bibitem[DNV20]{dnv20}
Daniel Dadush, Bento Natura, and L{\'{a}}szl{\'{o}}~A. V{\'{e}}gh.
\newblock Revisiting tardos's framework for linear programming: Faster exact
  solutions using approximate solvers.
\newblock {\em CoRR}, abs/2009.04942, 2020.

\bibitem[DS08]{ds08}
Samuel~I Daitch and Daniel~A Spielman.
\newblock Faster approximate lossy generalized flow via interior point
  algorithms.
\newblock In {\em Proceedings of the fortieth annual ACM symposium on Theory of
  computing (STOC)}, pages 451--460, 2008.

\bibitem[EK72]{ek72}
Jack Edmonds and Richard~M Karp.
\newblock Theoretical improvements in algorithmic efficiency for network flow
  problems.
\newblock {\em Journal of the ACM (JACM)}, 19(2):248--264, 1972.

\bibitem[Gal14]{Gall14a}
Fran{\c{c}}ois~Le Gall.
\newblock Powers of tensors and fast matrix multiplication.
\newblock In {\em {ISSAC}}, pages 296--303. {ACM}, 2014.

\bibitem[GR98]{GoldbergR98}
Andrew~V. Goldberg and Satish Rao.
\newblock Beyond the flow decomposition barrier.
\newblock {\em J. {ACM}}, 45(5):783--797, 1998.
\newblock announced at FOCS'97.

\bibitem[GRST20]{GoranciRST20}
Gramoz Goranci, Harald R{\"{a}}cke, Thatchaphol Saranurak, and Zihan Tan.
\newblock The expander hierarchy and its applications to dynamic graph
  algorithms.
\newblock {\em CoRR}, abs/2005.02369, 2020.

\bibitem[GT88]{gt88}
Zvi Galil and {\'{E}}va Tardos.
\newblock An ${O} (n^2 (m + n \log n) \log n)$ min-cost flow algorithm.
\newblock {\em J. {ACM}}, 35(2):374--386, 1988.
\newblock announced at FOCS'86.

\bibitem[GT90]{gt90}
Andrew~V Goldberg and Robert~E Tarjan.
\newblock Finding minimum-cost circulations by successive approximation.
\newblock {\em Mathematics of Operations Research}, 15(3):430--466, 1990.

\bibitem[JL84]{jl84}
William~B. Johnson and Joram Lindenstrauss.
\newblock Extensions of {L}ipschitz mappings into a {H}ilbert space.
\newblock In {\em Conference in modern analysis and probability ({N}ew {H}aven,
  {C}onn., 1982)}, volume~26 of {\em Contemp. Math.}, pages 189--206. Amer.
  Math. Soc., Providence, RI, 1984.

\bibitem[JOA10]{JOA10}
Thomas Jaksch, Ronald Ortner, and Peter Auer.
\newblock Near-optimal regret bounds for reinforcement learning.
\newblock {\em J. Mach. Learn. Res.}, 11:1563--1600, 2010.

\bibitem[JSWZ20]{jswz20}
Shunhua Jiang, Zhao Song, Omri Weinstein, and Hengjie Zhang.
\newblock Faster dynamic matrix inverse for faster lps.
\newblock {\em CoRR}, abs/2004.07470, 2020.

\bibitem[Kar84]{Karmarkar84}
Narendra Karmarkar.
\newblock A new polynomial-time algorithm for linear programming.
\newblock {\em Combinatorica}, 4(4):373--396, 1984.
\newblock Announced at STOC'84.

\bibitem[Kat20]{kathuria2020potential}
Tarun Kathuria.
\newblock A potential reduction inspired algorithm for exact max flow in almost
  o(m\^4/3) time.
\newblock In {\em FOCS}. \url{https://arxiv.org/pdf/2009.03260.pdf}, 2020.

\bibitem[KLP{\etalchar{+}}16]{KyngLPSS16}
Rasmus Kyng, Yin~Tat Lee, Richard Peng, Sushant Sachdeva, and Daniel~A.
  Spielman.
\newblock Sparsified cholesky and multigrid solvers for connection laplacians.
\newblock In {\em {\em STOC'16: Proceedings of the 48th Annual ACM Symposium on
  Theory of Computing}}, 2016.

\bibitem[KMP10]{KoutisMP10}
Ioannis Koutis, Gary~L. Miller, and Richard Peng.
\newblock Approaching optimality for solving {SDD} systems.
\newblock In {\em Proceedings of the 51st Annual IEEE Symposium on Foundations
  of Computer Science (FOCS)}, pages 235--244, 2010.

\bibitem[KMP11]{KoutisMP11}
Ioannis Koutis, Gary~L. Miller, and Richard Peng.
\newblock A nearly $m \log n$-time solver for {SDD} linear systems.
\newblock In {\em Proceedings of the 52nd Annual IEEE Symposium on Foundations
  of Computer Science (FOCS)}, pages 590--598, 2011.

\bibitem[KN14]{kn14}
Daniel~M Kane and Jelani Nelson.
\newblock Sparser johnson-lindenstrauss transforms.
\newblock {\em Journal of the ACM (JACM)}, 61(1):1--23, 2014.

\bibitem[KNPW11]{knpw11}
Daniel~M Kane, Jelani Nelson, Ely Porat, and David~P Woodruff.
\newblock Fast moment estimation in data streams in optimal space.
\newblock In {\em Proceedings of the forty-third annual ACM symposium on Theory
  of computing (STOC)}, pages 745--754, 2011.

\bibitem[KOSA13]{KOSZ13}
Jonathan~A. Kelner, Lorenzo Orecchia, Aaron Sidford, and Zeyuan {Allen Zhu}.
\newblock A simple, combinatorial algorithm for solving {SDD} systems in
  nearly-linear time.
\newblock In {\em Symposium on Theory of Computing Conference, STOC'13, Palo
  Alto, CA, USA, June 1-4, 2013}, pages 911--920, 2013.

\bibitem[KS98]{KS98a}
Michael~J. Kearns and Satinder~P. Singh.
\newblock Finite-sample convergence rates for q-learning and indirect
  algorithms.
\newblock In Michael~J. Kearns, Sara~A. Solla, and David~A. Cohn, editors, {\em
  Advances in Neural Information Processing Systems 11, {[NIPS} Conference,
  Denver, Colorado, USA, November 30 - December 5, 1998]}, pages 996--1002. The
  {MIT} Press, 1998.

\bibitem[KS16]{KS16}
Rasmus Kyng and Sushant Sachdeva.
\newblock Approximate gaussian elimination for laplacians - fast, sparse, and
  simple.
\newblock In {\em {IEEE} 57th Annual Symposium on Foundations of Computer
  Science, {FOCS} 2016, 9-11 October 2016, Hyatt Regency, New Brunswick, New
  Jersey, {USA}}, pages 573--582, 2016.

\bibitem[LDK95]{LittmanDK95}
Michael~L. Littman, Thomas~L. Dean, and Leslie~Pack Kaelbling.
\newblock On the complexity of solving markov decision problems.
\newblock In {\em {UAI} '95: Proceedings of the Eleventh Annual Conference on
  Uncertainty in Artificial Intelligence, Montreal, Quebec, Canada, August
  18-20, 1995}, pages 394--402, 1995.

\bibitem[LMP13]{lmp13}
Mu~Li, Gary~L. Miller, and Richard Peng.
\newblock Iterative row sampling.
\newblock In {\em 54th Annual {IEEE} Symposium on Foundations of Computer
  Science, {FOCS} 2013, 26-29 October, 2013, Berkeley, CA, {USA}}, pages
  127--136. {IEEE} Computer Society, 2013.

\bibitem[LNNT16]{lnnt16}
Kasper~Green Larsen, Jelani Nelson, Huy~L Nguyen, and Mikkel Thorup.
\newblock Heavy hitters via cluster-preserving clustering.
\newblock In {\em 57th Annual Symposium on Foundations of Computer Science
  (FOCS)}, pages 61--70. IEEE, \url{https://arxiv.org/pdf/1604.01357}, 2016.

\bibitem[LPS15]{LeePS15}
Yin~Tat Lee, Richard Peng, and Daniel~A. Spielman.
\newblock Sparsified cholesky solvers for {SDD} linear systems.
\newblock {\em CoRR}, abs/1506.08204, 2015.

\bibitem[LS14]{ls14}
Yin~Tat Lee and Aaron Sidford.
\newblock Path finding methods for linear programming: Solving linear programs
  in ${O}( \sqrt{rank} )$ iterations and faster algorithms for maximum flow.
\newblock In {\em 55th Annual IEEE Symposium on Foundations of Computer Science
  (FOCS)}, pages 424--433. \url{https://arxiv.org/pdf/1312.6677.pdf},
  \url{https://arxiv.org/pdf/1312.6713.pdf}, 2014.

\bibitem[LS15]{LeeS15}
Yin~Tat Lee and Aaron Sidford.
\newblock Efficient inverse maintenance and faster algorithms for linear
  programming.
\newblock In {\em {FOCS}}, pages 230--249. {IEEE} Computer Society, 2015.

\bibitem[LS19]{ls19}
Yin~Tat Lee and Aaron Sidford.
\newblock Solving linear programs with $\sqrt{rank}$ linear system solves.
\newblock In {\em arXiv preprint}. \url{https://arxiv.org/pdf/1910.08033.pdf},
  2019.

\bibitem[LS20a]{ls20_focs}
Yang~P Liu and Aaron Sidford.
\newblock Faster divergence maximization for faster maximum flow.
\newblock In {\em FOCS}. \url{https://arxiv.org/pdf/2003.08929.pdf}, 2020.

\bibitem[LS20b]{ls20_stoc}
Yang~P Liu and Aaron Sidford.
\newblock Faster energy maximization for faster maximum flow.
\newblock In {\em STOC}. \url{https://arxiv.org/pdf/1910.14276.pdf}, 2020.

\bibitem[LSZ19]{lsz19}
Yin~Tat Lee, Zhao Song, and Qiuyi Zhang.
\newblock Solving empirical risk minimization in the current matrix
  multiplication time.
\newblock In {\em {COLT}}, volume~99 of {\em Proceedings of Machine Learning
  Research}, pages 2140--2157. {PMLR}, 2019.
\newblock \url{https://arxiv.org/pdf/1905.04447}.

\bibitem[LSZ20]{lsz20}
S.~Cliff Liu, Zhao Song, and Hengjie Zhang.
\newblock Breaking the n-pass barrier: {A} streaming algorithm for maximum
  weight bipartite matching.
\newblock {\em CoRR}, abs/2009.06106, 2020.

\bibitem[LWC{\etalchar{+}}20]{Dlwcgc20}
Gen Li, Yuting Wei, Yuejie Chi, Yuantao Gu, and Yuxin Chen.
\newblock Breaking the sample size barrier in model-based reinforcement
  learning with a generative model.
\newblock {\em CoRR}, abs/2005.12900, 2020.

\bibitem[Mad02]{Madani02}
Omid Madani.
\newblock Polynomial value iteration algorithms for detrerminstic mdps.
\newblock In {\em UAI}, pages 311--318, 2002.

\bibitem[Mad13]{m13}
Aleksander Madry.
\newblock Navigating central path with electrical flows: From flows to
  matchings, and back.
\newblock In {\em {FOCS}}, pages 253--262. {IEEE} Computer Society, 2013.

\bibitem[Mad16]{m16}
Aleksander Madry.
\newblock Computing maximum flow with augmenting electrical flows.
\newblock In {\em 2016 IEEE 57th Annual Symposium on Foundations of Computer
  Science (FOCS)}, pages 593--602. IEEE, 2016.

\bibitem[Mah96]{Mahadevan96}
Sridhar Mahadevan.
\newblock Average reward reinforcement learning: Foundations, algorithms, and
  empirical results.
\newblock {\em Mach. Learn.}, 22(1-3):159--195, 1996.

\bibitem[MT03]{mt03}
Renato D.~C. Monteiro and Takashi Tsuchiya.
\newblock A variant of the vavasis--ye layered-step interior-point algorithm
  for linear programming.
\newblock {\em {SIAM} J. Optim.}, 13(4):1054--1079, 2003.

\bibitem[MT05]{mt05}
Renato D.~C. Monteiro and Takashi Tsuchiya.
\newblock A new iteration-complexity bound for the {MTY} predictor-corrector
  algorithm.
\newblock {\em {SIAM} J. Optim.}, 15(2):319--347, 2005.

\bibitem[Nes98]{Nes98}
Yurii Nesterov.
\newblock Introductory lectures on convex programming volume i: Basic course.
\newblock {\em Lecture notes}, 3(4):5, 1998.

\bibitem[Nes09]{Nesterov09}
Yurii~E. Nesterov.
\newblock Unconstrained convex minimization in relative scale.
\newblock {\em Math. Oper. Res.}, 34(1):180--193, 2009.

\bibitem[NN91]{NesterovN91}
Yurii~E. Nesterov and Arkadii Nemirovskii.
\newblock Acceleration and parallelization of the path-following interior point
  method for a linearly constrained convex quadratic problem.
\newblock {\em {SIAM} J. Optim.}, 1(4):548--564, 1991.

\bibitem[NN13]{nn13}
Jelani Nelson and Huy~L Nguy{\^e}n.
\newblock {OSNAP}: Faster numerical linear algebra algorithms via sparser
  subspace embeddings.
\newblock In {\em 54th Annual IEEE Symposium on Foundations of Computer Science
  (FOCS)}, pages 117--126. IEEE, \url{https://arxiv.org/pdf/1211.1002}, 2013.

\bibitem[NS19]{ns19}
Vasileios Nakos and Zhao Song.
\newblock Stronger l\({}_{\mbox{2}}\)/l\({}_{\mbox{2}}\) compressed sensing;
  without iterating.
\newblock In {\em {STOC}}, pages 289--297. {ACM}, 2019.

\bibitem[NSW17]{NanongkaiSW17}
Danupon Nanongkai, Thatchaphol Saranurak, and Christian Wulff{-}Nilsen.
\newblock Dynamic minimum spanning forest with subpolynomial worst-case update
  time.
\newblock In {\em 58th {IEEE} Annual Symposium on Foundations of Computer
  Science, {FOCS} 2017, Berkeley, CA, USA, October 15-17, 2017}, pages
  950--961, 2017.

\bibitem[NSW19]{nsw19}
Vasileios Nakos, Zhao Song, and Zhengyu Wang.
\newblock (nearly) sample-optimal sparse fourier transform in any dimension;
  ripless and filterless.
\newblock In {\em {FOCS}}, pages 1568--1577. {IEEE} Computer Society, 2019.

\bibitem[Orl84]{o84}
James~B Orlin.
\newblock Genuinely polynomial simplex and non-simplex algorithms for the
  minimum cost flow problem.
\newblock 1984.

\bibitem[Orl93]{o93}
James~B Orlin.
\newblock A faster strongly polynomial minimum cost flow algorithm.
\newblock {\em Operations research}, 41(2):338--350, 1993.

\bibitem[Pag13]{p13}
Rasmus Pagh.
\newblock Compressed matrix multiplication.
\newblock {\em ACM Transactions on Computation Theory (TOCT)}, 5(3):1--17,
  2013.

\bibitem[PS14]{PS14}
Richard Peng and Daniel~A. Spielman.
\newblock An efficient parallel solver for {SDD} linear systems.
\newblock In David~B. Shmoys, editor, {\em Symposium on Theory of Computing,
  {STOC} 2014, New York, NY, USA, May 31 - June 03, 2014}, pages 333--342.
  {ACM}, 2014.

\bibitem[PSW17]{psw17}
Eric Price, Zhao Song, and David~P Woodruff.
\newblock Fast regression with an $ ell\_infty $ guarantee.
\newblock In {\em 44th International Colloquium on Automata, Languages, and
  Programming (ICALP)}. Schloss Dagstuhl-Leibniz-Zentrum fuer Informatik, 2017.

\bibitem[Ren88]{Renegar88}
James Renegar.
\newblock A polynomial-time algorithm, based on newton's method, for linear
  programming.
\newblock {\em Math. Program.}, 40(1-3):59--93, 1988.

\bibitem[San04]{Sankowski04}
Piotr Sankowski.
\newblock Dynamic transitive closure via dynamic matrix inverse (extended
  abstract).
\newblock In {\em {FOCS}}, pages 509--517. {IEEE} Computer Society, 2004.

\bibitem[Sch16]{scherrer2016improved}
Bruno Scherrer.
\newblock Improved and generalized upper bounds on the complexity of policy
  iteration.
\newblock {\em Mathematics of Operations Research}, 41(3):758--774, 2016.

\bibitem[SS08]{SS08}
Daniel~A. Spielman and Nikhil Srivastava.
\newblock Graph sparsification by effective resistances.
\newblock In Cynthia Dwork, editor, {\em Proceedings of the 40th Annual {ACM}
  Symposium on Theory of Computing, Victoria, British Columbia, Canada, May
  17-20, 2008}, pages 563--568. {ACM}, 2008.

\bibitem[SS11]{spielman2011graph}
Daniel~A Spielman and Nikhil Srivastava.
\newblock Graph sparsification by effective resistances.
\newblock {\em SIAM Journal on Computing}, 40(6):1913--1926, 2011.

\bibitem[ST04]{SpielmanT04}
Daniel~A. Spielman and Shang{-}Hua Teng.
\newblock Nearly-linear time algorithms for graph partitioning, graph
  sparsification, and solving linear systems.
\newblock In {\em {STOC}}, pages 81--90. {ACM}, 2004.

\bibitem[SW11]{SohlerW11}
Christian Sohler and David~P. Woodruff.
\newblock Subspace embeddings for the l\({}_{\mbox{1}}\)-norm with
  applications.
\newblock In {\em {STOC}}, pages 755--764. {ACM}, 2011.

\bibitem[SW19]{SaranurakW19}
Thatchaphol Saranurak and Di~Wang.
\newblock Expander decomposition and pruning: Faster, stronger, and simpler.
\newblock In {\em Proceedings of the Thirtieth Annual {ACM-SIAM} Symposium on
  Discrete Algorithms, {SODA} 2019, San Diego, California, USA, January 6-9,
  2019}, pages 2616--2635, 2019.

\bibitem[SWW{\etalchar{+}}18]{SidfordWWYY18}
Aaron Sidford, Mengdi Wang, Xian Wu, Lin Yang, and Yinyu Ye.
\newblock Near-optimal time and sample complexities for solving markov decision
  processes with a generative model.
\newblock In {\em Advances in Neural Information Processing Systems 31: Annual
  Conference on Neural Information Processing Systems 2018, NeurIPS 2018, 3-8
  December 2018, Montr{\'{e}}al, Canada}, pages 5192--5202, 2018.

\bibitem[SWWY18]{SidfordWWY18}
Aaron Sidford, Mengdi Wang, Xian Wu, and Yinyu Ye.
\newblock Variance reduced value iteration and faster algorithms for solving
  markov decision processes.
\newblock In {\em Proceedings of the Twenty-Ninth Annual {ACM-SIAM} Symposium
  on Discrete Algorithms, {SODA} 2018, New Orleans, LA, USA, January 7-10,
  2018}, pages 770--787, 2018.

\bibitem[SY20]{sy20}
Zhao Song and Zheng Yu.
\newblock Oblivious sketching-based central path method for solving linear
  programming problems.
\newblock In {\em manuscript}.
  \url{https://openreview.net/forum?id=fGiKxvF-eub}, 2020.

\bibitem[Tar85]{t85}
{\'E}va Tardos.
\newblock A strongly polynomial minimum cost circulation algorithm.
\newblock {\em Combinatorica}, 5(3):247--255, 1985.

\bibitem[Tro11]{Tropp11}
Joel Tropp.
\newblock Freedman's inequality for matrix martingales.
\newblock {\em Electronic Communications in Probability}, 16:262--270, 2011.

\bibitem[Tse90]{tseng1990solving}
Paul Tseng.
\newblock Solving h-horizon, stationary markov decision problems in time
  proportional to log (h).
\newblock {\em Operations Research Letters}, 9(5):287--297, 1990.

\bibitem[Vai89]{Vaidya89a}
Pravin~M. Vaidya.
\newblock Speeding-up linear programming using fast matrix multiplication
  (extended abstract).
\newblock In {\em FOCS}, pages 332--337. {IEEE} Computer Society, 1989.

\bibitem[VY96]{vy96}
Stephen~A. Vavasis and Yinyu Ye.
\newblock A primal-dual interior point method whose running time depends only
  on the constraint matrix.
\newblock {\em Math. Program.}, 74:79--120, 1996.

\bibitem[Wai19]{Wainwright19}
Martin~J. Wainwright.
\newblock Variance-reduced q-learning is minimax optimal.
\newblock {\em CoRR}, abs/1906.04697, 2019.

\bibitem[Wan17]{W17i}
Mengdi Wang.
\newblock Randomized linear programming solves the discounted markov decision
  problem in nearly-linear running time.
\newblock {\em CoRR}, abs/1704.01869, 2017.

\bibitem[Wan20]{wang2020randomized}
Mengdi Wang.
\newblock Randomized linear programming solves the markov decision problem in
  nearly linear (sometimes sublinear) time.
\newblock {\em Mathematics of Operations Research}, 45(2):517--546, 2020.

\bibitem[Wil12]{Williams12}
Virginia~Vassilevska Williams.
\newblock Multiplying matrices faster than coppersmith-winograd.
\newblock In {\em STOC}, pages 887--898. {ACM}, 2012.

\bibitem[Wil19]{williamson2019network}
David~P Williamson.
\newblock {\em Network Flow Algorithms}.
\newblock Cambridge University Press, 2019.

\bibitem[WW19]{WangW19}
Ruosong Wang and David~P. Woodruff.
\newblock Tight bounds for lp oblivious subspace embeddings.
\newblock In {\em {SODA}}, pages 1825--1843. {SIAM}, 2019.

\bibitem[WZ13]{WoodruffZ13}
David~P. Woodruff and Qin Zhang.
\newblock Subspace embeddings and
  {\textbackslash}({\textbackslash}ell{\_}p{\textbackslash})-regression using
  exponential random variables.
\newblock In {\em {COLT}}, volume~30 of {\em {JMLR} Workshop and Conference
  Proceedings}, pages 546--567. JMLR.org, 2013.

\bibitem[YCRM16]{YangCRM16}
Jiyan Yang, Yinlam Chow, Christopher R{\'{e}}, and Michael~W. Mahoney.
\newblock Weighted {SGD} for $\ell_p$ regression with randomized
  preconditioning.
\newblock In Robert Krauthgamer, editor, {\em Proceedings of the Twenty-Seventh
  Annual {ACM-SIAM} Symposium on Discrete Algorithms, {SODA} 2016, Arlington,
  VA, USA, January 10-12, 2016}, pages 558--569. {SIAM}, 2016.

\bibitem[Ye05]{ye2005new}
Yinyu Ye.
\newblock A new complexity result on solving the markov decision problem.
\newblock {\em Mathematics of Operations Research}, 30(3):733--749, 2005.

\bibitem[Ye11]{ye2011simplex}
Yinyu Ye.
\newblock The simplex and policy-iteration methods are strongly polynomial for
  the markov decision problem with a fixed discount rate.
\newblock {\em Mathematics of Operations Research}, 36(4):593--603, 2011.

\end{thebibliography}

\appendix

\section{IPM Proofs}
\label{sec:ipmproofs}
\subsection{Basic Analysis Tools}
\label{subsec:proofsbasic}

\highselfcon*
\begin{proof}
Note that $\phi''(x) = (u-x)^{-2} + (x-\ell)^{-2}.$ For odd $n$, note that $\phi^{(n)}(x) = (n-1)!\left((u-x)^{-n} - (x-\ell)^{-n}\right)$, so
\begin{align*}  
|\phi^{(n)}(x)| \le (n-1)!\min(u-x,x-\ell)^{-n} \le (n-1)!\phi''(x)^{n/2}. 
\end{align*} 
This shows the first two items. For even $n \ge 2$, we have $\phi^{(n)}(x) = (n-1)!\left((u-x)^{-n} + (x-\ell)^{-n}\right)$, and
\begin{align*}  
(n-1)!\left((u-x)^{-n} + (x-\ell)^{-n}\right) \le (n-1)!\left((u-x)^{-2} + (x-\ell)^{-2}\right)^{n/2} \le (n-1)!\phi''(x)^{n/2}. 
\end{align*} 
This shows the last item.
\end{proof}

\potentialhelper*
\begin{proof}
For $0 \le t \le 1$ define
\begin{align*}
w_i^{(t)} = \prod_{j\in[k]} (u_i^{(j)} + t \cdot \d_i^{(j)})^{c_j} \enspace \text{ and } \enspace v_i^{(t)} = (y_i + t \cdot \eta_i)w_i^{(t)}. 
\end{align*}
We first calculate
\begin{align} \label{eq:dtw} 
\frac{ \mathrm{d} }{ \mathrm{d} t } w_i^{(t)} = w_i^{(t)} \sum_{j\in[k]} \frac{c_j\d_i^{(j)}}{u_i^{(j)} + t \cdot \d_i^{(j)}} 
\end{align}
and
\begin{align} \label{eq:dtv} 
\frac{ \mathrm{d} }{ \mathrm{d} t } v_i^{(t)} = y_i \frac{ \mathrm{d} }{ \mathrm{d} t }w_i^{(t)} + \eta_iw_i^{(t)} = w_i^{(t)} \left(\eta_i + y_i\sum_{j\in[k]} \frac{c_j\d_i^{(j)}}{u_i^{(j)} + t \cdot \d_i^{(j)}}\right). 
\end{align}
Applying \eqref{eq:dtw} gives that
\begin{align}
|\log ( w_i^\new ) - \log ( w_i ) | 
\le & ~ \left|\int_0^1 (w_i^{(t)})^{-1}\frac{d}{ \mathrm{d} t } w_i^{(t)}\right| \mathrm{d} t \notag \\
\le & ~ \int_0^1 \left|\sum_{j\in[k]}\frac{c_j\d_i^{(j)}}{u_i^{(j)} + t \cdot \d_i^{(j)}}\right| \mathrm{d} t \notag \\ 
\le & ~ 2\sum_{j \in [k]} |c_j|\|(\mU^{(j)})^{-1}\d^{(j)})\|_\infty \le \frac{1}{25}. \label{eq:cjbound}
\end{align}
Therefore, $\|\mw^{-1}(w^\new-w)\|_\infty \le \exp(|\log(w_i^\new) - \log(w_i)|)-1 \le 2|\log(w_i^\new) - \log(w_i)|.$
Now we can use $v = \mw y$ and \eqref{eq:cjbound} to get that
\begin{align*}
\|v^\new-v\|_\infty 
= & ~ \left\|(\mw^\new-\mw)(y+\eta) + \mw\eta \right\|_\infty \\ 
\le & ~ \|\mw^{-1}(w^\new-w)\|_\infty(\|\mw y\|_\infty + \| \mw \eta \|_\infty) + \|\mw\eta\|_\infty\\ 
\le & ~ \|\mw^{-1}(w^\new-w)\|_\infty\|v\|_\infty + \|\mw^{-1}(w^\new-w)\|_\infty\| \mw \eta \|_\infty + \|\mw\eta\|_\infty \\
\le & ~ 2\|v\|_\infty \sum_{ j \in [k] } |c_j|\|(\mU^{(j)})^{-1}\d^{(j)})\|_\infty + \frac{1}{25} \cdot \frac{1}{50\lambda} + \frac{1}{50\lambda} \\
\le & ~ \frac{1}{50\lambda} + \frac{1}{2500\lambda} + \frac{1}{50\lambda} \le \frac{1}{20\lambda}
\end{align*}
as desired. We now proceed to controlling the potential $\Psi$. To this end, define $f(t) = \psi(v_i^{(t)}).$ By Taylor's theorem we know that there is a $\zeta \in [0,1]$ so that $f(1) = f(0) + f'(0) + \frac{1}{2} f''(\zeta).$ Now compute using \eqref{eq:dtv} that
\begin{align*}
f'(t) = \psi'(v_i^{(t)}) \frac{ \mathrm{d} }{ \mathrm{d} t } v_i^{(t)} = \psi'(v_i^{(t)}) w_i^{(t)} \left(\eta_i + y_i\sum_{j\in[k]} \frac{c_j\d_i^{(j)}}{u_i^{(j)} + t \cdot \d_i^{(j)}}\right)
\end{align*}
 and
\begin{align}
f''(t) &= \psi''(v_i^{(t)}) \frac{ \mathrm{d} }{ \mathrm{d} t } v_i^{(t)} w_i^{(t)} \left(\eta_i + y_i\sum_{j\in[k]} \frac{c_j\d_i^{(j)}}{u_i^{(j)} + t \cdot \d_i^{(j)}}\right) \notag \\ 
&+ \psi'(v_i^{(t)}) \frac{ \mathrm{d} }{ \mathrm{d} t } w_i^{(t)} \left(\eta_i + y_i\sum_{j\in[k]} \frac{c_j\d_i^{(j)}}{u_i^{(j)} + t \cdot \d_i^{(j)}}\right) - \psi'(v_i^{(t)}) w_i^{(t)} y_i \sum_{j\in[k]} \frac{c_j(\d_i^{(j)})^2}{(u_i^{(j)} + t \cdot \d_i^{(j)})^2} \notag \\ 
&= \psi''(v_i^{(t)}) (w_i^{(t)})^2 \left(\eta_i + y_i\sum_{j\in[k]} \frac{c_j\d_i^{(j)}}{u_i^{(j)} + t \cdot \d_i^{(j)}}\right)^2 \label{eq:part1}
\\ 
&+ \psi'(v_i^{(t)}) w_i^{(t)} \left(\sum_{j\in[k]} \frac{c_j\d_i^{(j)}}{u_i^{(j)} + t \cdot \d_i^{(j)}}\right) \left(\eta_i + y_i\sum_{j\in[k]} \frac{c_j\d_i^{(j)}}{u_i^{(j)} + t \cdot \d_i^{(j)}}\right) \label{eq:part2}
\\ 
&- \psi'(v_i^{(t)}) w_i^{(t)} y_i \sum_{j\in[k]} \frac{c_j(\d_i^{(j)})^2}{(u_i^{(j)} + t \cdot \d_i^{(j)})^2}. \label{eq:part3}
\end{align}
We now bound \eqref{eq:part1}, \eqref{eq:part2}, \eqref{eq:part3}. We will heavily use Lemma \ref{lemma:psibasic} and the fact that $|w_i^{(t)}y_i| \le 2|w_iy_i| = 2|v_i| \le 2\|v\|_\infty \le 1/10.$ To bound \eqref{eq:part1} use Cauchy-Schwarz to get that
\begin{align*} 
& ~ \left|\psi''(v_i^{(t)}) (w_i^{(t)})^2 \left(\eta_i + y_i\sum_{j\in[k]} \frac{c_j\d_i^{(j)}}{u_i^{(j)} + t \cdot \d_i^{(j)}}\right)^2\right|
\\ 
\le & ~ 2\psi''(v_i^{(t)})(w_i^{(t)}\eta_i)^2 + 2\psi''(v_i^{(t)}) (w_i^{(t)}y_i)^2 \left(\sum_{j\in[k]} \frac{c_j\d_i^{(j)}}{u_i^{(j)} + t \cdot \d_i^{(j)}}\right)^2
\\ 
\le & ~ 8\psi''(v_i)(w_i\eta_i)^2 + 8\psi''(v_i)\|c\|_1\|v\|_\infty^2\sum_{j\in[k]}|c_j|((u_i^{(j)})^{-1}\d_i^{(j)})^2.
\end{align*}
To bound \eqref{eq:part2} we compute that
\begin{align*}
& ~ \left| \psi'(v_i^{(t)}) w_i^{(t)} \left(\sum_{j\in[k]} \frac{c_j\d_i^{(j)}}{u_i^{(j)} + t \cdot \d_i^{(j)}}\right) \left(\eta_i + y_i\sum_{j\in[k]} \frac{c_j\d_i^{(j)}}{u_i^{(j)} + t \cdot \d_i^{(j)}}\right)\right|
\\ 
\le & ~ 8|\psi'(v_i)|w_i|\eta_i|\sum_{j\in[k]} |c_j||(u_i^{(j)})^{-1}\d_i^{(j)}| + 4|\psi'(v_i)||w_i^{(t)}y_i| \left(\sum_{j\in[k]} \frac{c_j\d_i^{(j)}}{u_i^{(j)} + t \cdot \d_i^{(j)}}\right)^2
\\ 
\le & ~ 8|\psi'(v_i)|w_i|\eta_i|\sum_{j\in[k]} |c_j||(u_i^{(j)})^{-1}\d_i^{(j)}| + 8|\psi'(v_i)|\|c\|_1\|v\|_\infty\sum_{j\in[k]}|c_j|((u_i^{(j)})^{-1}\d_i^{(j)})^2
\end{align*}
To bound \eqref{eq:part3} we compute that
\begin{align*} 
\left|\psi'(v_i^{(t)}) w_i^{(t)} y_i \sum_{j\in[k]} \frac{c_j(\d_i^{(j)})^2}{(u_i^{(j)} + t \cdot \d_i^{(j)})^2}\right|  
\le & ~ 2|\psi'(v_i)| |w_i^{(t)}y_i| \sum_{j\in[k]} \frac{|c_j|(\d_i^{(j)})^2}{(u_i^{(j)} + t \cdot \d_i^{(j)})^2} \\
\le & ~ 4|\psi'(v_i)| \|v\|_\infty\sum_{j\in[k]}|c_j|((u_i^{(j)})^{-1}\d_i^{(j)})^2. 
\end{align*}

Note that 
\begin{align*} 
f'(0) = \psi'(v_i)w_i\eta_i + \psi'(v_i)v_i\sum_{j\in[k]}c_j(u_i^{(j)})^{-1}\d_i^{(j)}. 
\end{align*}
Summing the previous bounds over all $i$ and using that $f(1) \le f(0) + f'(0) + \frac{1}{2}  \max_{\zeta\in[0,1]}f''(\zeta)$ gives us that
\begin{align*}
&\Psi(v^\new) \le \Psi(v^\new) + \psi'(v)^\top \left(\mw\eta + \mv\sum_{j\in[k]}c_j(\mU^{(j)})^{-1}\d^{(j)}\right)
\\ &+ 8 \sum_{i=1}^m \psi''(v_i)(w_i\eta_i)^2  + 8 \|c\|_1 \|v\|_\infty^2\sum_{i=1}^m \psi''(v_i) \sum_{j\in[k]}|c_j|((u_i^{(j)})^{-1}\d_i^{(j)})^2
\\ &+ 8\sum_{i=1}^m \psi''(v_i)w_i|\eta_i|\sum_{j\in[k]} |c_j||(u_i^{(j)})^{-1}\d_i^{(j)}| + 8(1+\|c\|_1)\|v\|_\infty\sum_{i=1}^m|\psi'(v_i)|\sum_{j\in[k]}|c_j|((u_i^{(j)})^{-1}\d_i^{(j)})^2
\\ &= \Psi(v^\new) + \psi'(v)^\top \left(\mw\eta + \mv\sum_{j\in[k]}c_j(\mU^{(j)})^{-1}\d^{(j)}\right)
\\ &+ 8\|\mw\eta\|_{\psi''(v)}^2 + 8\|c\|_1\|v\|_\infty^2\sum_{j\in[k]}|c_j|\|(\mU^{(j)})^{-1}\d^{(j)}\|_{\psi''(v)}^2 \\ &+ 8\left\langle \mw|\eta|, \sum_{j\in[k]} |c_j||(\mU^{(j)})^{-1}||\d^{(j)}| \right\rangle_{|\psi'(v)|} +8(1+\|c\|_1)\|v\|_\infty \sum_j |c_j|\|(\mU^{(j)})^{-1}\d^{(j)}\|_{|\psi'(v)|}
\\ &\le \Psi(v^\new) + \psi'(v)^\top \left(\mw\eta + \mv\sum_{j\in[k]}c_j(\mU^{(j)})^{-1}\d^{(j)}\right)
\\&+ 8\|\mw\eta\|_{\psi''(v)}^2 + 8(1+\|c\|_1)\|v\|_\infty^2\sum_{j\in[k]}|c_j|\|(\mU^{(j)})^{-1}\d^{(j)}\|_{\psi''(v)}^2
\\&+ 8\|\mw\eta\|_{|\psi'(v)|}\sum_j |c_j|\|(\mU^{(j)})^{-1}\d^{(j)}\|_{|\psi'(v)|} + 8(1+\|c\|_1)\|v\|_\infty\sum_{j\in[k]}|c_j|\|(\mU^{(j)})^{-1}\d^{(j)}\|_{|\psi'(v)|}^2
\end{align*}
by the Cauchy-Schwarz inequality as desired.
\end{proof}

\subsection{Leverage Scores and Fundamental Matrix Proofs}
\label{subsec:proofslev}
The analysis and notation throughout this section is broadly based on that of \cite{ls14}.

\optproblem*
\begin{proof}
By \cite[Lemma 23]{ls19} we can compute that
\begin{align*} 
\g_w f(w,c) = -\mW^{-1}\sigma(\mW^{\frac{1}{2} -\frac1p}\mC\mA) + \vec{1} - \mW^{-1}v.
\end{align*}
Therefore, we have that
\begin{align*} 
-\mW_c^{-1}\sigma(\mW_c^{\frac{1}{2} -\frac1p}\mC\mA) + \vec{1} - \mW_c^{-1}v = 0, 
\end{align*} 
which is equivalent to the desired bound of 
$
w(c) = \sigma(\mW_c^{\frac{1}{2} -\frac1p}\mC\mA) + v
$. 
\end{proof}

\derivlewis*
\begin{proof}
By Lemma \ref{lemma:optproblem} we have that
\begin{align*} 
\g_w f(w(c),c) = 0. 
\end{align*} 
Differentiating with respect to $c$ gives us
\begin{align*} 
\g^2_{wc}f(w(c),c) + \g^2_{ww}f(w(c),c)\mJ_c = 0 . 
\end{align*}
As in the proof of \cite[Lemma 24]{ls19} we have that
\begin{align*} 
\g^2_{wc}f(w,c) = -2\mW^{-1}\mLambda_c\mC^{-1} 
\end{align*}
and
\begin{align*} 
\g^2_{ww}f(w,c) &= \mw^{-1}\left(\mSigma_c-\left(1-\frac2p\right)\mLambda_c\right)\mW^{-1} + \mW^{-1}\mV\mW^{-1} \\ &= \mW^{-1}\left(\mSigma_c+\mV-\left(1-\frac2p\right)\mLambda_c\right)\mW^{-1}.
\end{align*}
Now we can solve to get
\begin{align*} 
\mJ_c &= 2\mW_c\left(\mSigma_c+\mV-\left(1-\frac2p\right)\mLambda_c\right)^{-1}\mLambda_c\mC^{-1} \\ &= 2\mW_c\left(\mW_c-\left(1-\frac2p\right)\mLambda_c\right)^{-1}\mLambda_c\mC^{-1} 
\end{align*}
where we have used that $\mW_c = \mSigma_c + \mV$ by definition.
\end{proof}

Before continuing, we collect various properties of projection matrices, proven in \cite[Lemma 47]{ls19}.
\begin{lemma}[Facts related to $\Tau^{-1} \mP^{(2)}$]
\label{lemma:lemma47}
We have the following for a vector $h \in \R^m$, where $\mP$ is an orthogonal projection matrix and $\Tau$ is a diagonal matrix with the same diagonal as $\mP$.
\begin{enumerate}
\item $\mP^{(2)} \pe \Tau$.
\item $\|\Tau^{-1}\mP^{(2)}\|_\infty \le 1$.
\item $\|\Tau^{-1}\mP^{(2)}h\|_\tau \le \|h\|_{\mP^{(2)}} \le \|h\|_\tau$.
\item $\|\Tau^{-1}\mP^{(2)}h\|_\infty \le \|h\|_\tau$.
\end{enumerate}
\end{lemma}

We now show that the operator $\mW_c^{-1}\mJ_c\mC$ is bounded in the $\tpi$ norm.
\begin{lemma}[Facts related to $\mW_c^{-1}\mJ_c\mC$]
\label{lemma:matrixbound}
Define $\|g\|_{w(c)+\infty} \defeq \|g\|_\infty + \cnorm\|g\|_{w(c)}$. If $p \in [2/3,1]$ then for all vectors $h$ we have
\begin{itemize}
\item $\|\mW_c^{-1}\mJ_c\mC h\|_{w(c)} \le p\|h\|_{w(c)}.$
\item $\|\mW_c^{-1}\mJ_c\mC h\|_\infty \le p\|h\|_\infty + 2\|h\|_{w(c)}$.
\item $\|\mW_c^{-1}\mJ_c\mC h\|_{w(c)+\infty} \le p(1+3/\cnorm)\|h\|_{w(c)+\infty}.$
\end{itemize}
\end{lemma}
\begin{proof}
Define $\overline{\mLambda}_c = \mW_c^{-\frac{1}{2} }\mLambda_c \mW_c^{-\frac{1}{2} }.$ Compute that
\begin{align*}
\mW_c^{-\frac{1}{2} }\mJ_c\mC &= 2\mW_c^\frac{1}{2} \left(\mW_c-\left(1-\frac2p\right)\mLambda_c\right)^{-1}\mLambda_c\\ 
= & ~ 2\left(\mI-\left(1-\frac2p\right)\overline{\mLambda}_c\right)^{-1}\mW_c^{-\frac{1}{2} }\mLambda_c \\
= & ~ 2\left(\mI-\left(1-\frac2p\right)\overline{\mLambda}_c\right)^{-1}\overline{\mLambda}_c\mW_c^\frac{1}{2} .
\end{align*}
Recall that $\mLambda_c \pe \mSigma_c \pe \mW_c$, so $0 \pe \overline{\mLambda}_c \pe \mI$. Therefore $\left(\mI-\left(1-\frac2p\right)\overline{\mLambda}_c\right)^{-1}\overline{\mLambda}_c$ is a positive semidefinite matrix with eigenvalues at most 
\begin{align*} 
\max_{0 \le \lambda \le 1} \frac{\lambda}{1-\left(1-\frac2p\right)\lambda} = \frac{p}{2}. 
\end{align*}
Thus, we have
\begin{align*}
\|\mW_c^{-1}\mJ_c\mC h\|_{w(c)} &= \|\mW_c^{-\frac{1}{2} }\mJ_c\mC h\|_2
\\ &\le 2\left\|\left(\mI-\left(1-\frac2p\right)\overline{\mLambda}_c\right)^{-1}\overline{\mLambda}_c\right\|_2\|\mW_c^\frac{1}{2} h\|_2 \le p\|h\|_{w(c)}.
\end{align*}
Define $\mS_c = \mW_c^{-\frac{1}{2} }\mSigma_c\mW_c^{-\frac{1}{2} }.$ Note that $\mS_c$ is diagonal and $0 \pe \mS_c \pe \mI$. Define $\mQ_c = \mW_c^{-\frac{1}{2} }\mP_c^{(2)}\mW_c^{-\frac{1}{2} }$. By definition, $\bar{\mLambda_c} = \mS_c-\mQ_c$. Define $\mD_c = \frac{2\mS_c}{\mI+\left(\frac2p-1\right)\mS_c}$, and note that $\mD_c$ is diagonal and $0 \pe \mD_c \pe p \mI$. Note from the above formula that
\begin{align}
&\mW_c^{-1}\mJ_c\mC h - \mD_c h = 2\mW_c^{-\frac{1}{2} }\overline{\mLambda}_c\left(\mI-\left(1-\frac2p\right)\overline{\mLambda}_c\right)^{-1}\mW_c^\frac{1}{2} h - \mW_c^{-\frac{1}{2} }\mD_c\mW_c^\frac{1}{2} h \label{eq:matrixdiff1}
\\ &= \mW_c^{-\frac{1}{2} }\left(2\overline{\mLambda}_c-\mD_c\left(\mI-\left(1-\frac2p\right)\overline{\mLambda}_c\right)\right)\left(\mI-\left(1-\frac2p\right)\overline{\mLambda}_c\right)^{-1}\mW_c^\frac{1}{2} h \notag
\\ &= \mW_c^{-\frac{1}{2} }\left(2\mS_c-2\mQ_c-\frac{2\mS_c}{\mI+\left(\frac2p-1\right)\mS_c}\left(\mI-\left(1-\frac2p\right)(\mS_c-\mQ_c)\right)\right)\left(\mI-\left(1-\frac2p\right)\overline{\mLambda}_c\right)^{-1}\mW_c^\frac{1}{2} h \notag
\\ &= -2\mW_c^{-\frac{1}{2} }\left(\mI-\left(1-\frac2p\right)\mS_c\right)^{-1}\mQ_c\left(\mI\left(1-\frac2p\right)\overline{\mLambda}_c\right)^{-1}\mW_c^\frac{1}{2} h \notag
\\ &= -2\mW_c^{-\frac{1}{2} }\left(\mI-\left(1-\frac2p\right)\mS_c\right)^{-1}\mW_c^{-\frac{1}{2} }\mP_c^{(2)}\mW_c^{-\frac{1}{2} }\left(\mI-\left(1-\frac2p\right)\overline{\mLambda}_c\right)^{-1}\mW_c^\frac{1}{2} h. \label{eq:matrixdiff2}
\end{align}
By Lemma \ref{lemma:lemma47} and $\mSigma_c \pe \mW_c$ and both are diagonal matrices, we have that 
\begin{align*} 
\|\mW_c^{-1}\mP_c^{(2)}x\|_\infty \le \|x\|_{\mSigma_c} \le \|x\|_{w(c)}. 
\end{align*} 
Also, $\mS_c$ is non-negative and diagonal. Therefore,
\begin{align*}
& ~ \|\mW_c^{-1}\mJ_c\mC h\|_\infty \\ 
\le & ~ \|\mD_c h\|_\infty + \|2\mW_c^{-\frac{1}{2} }\left(\mI-\left(1-\frac2p\right)\mS_c\right)^{-1}\mW_c^{-\frac{1}{2} }\mP_c^{(2)}\mW_c^{-\frac{1}{2} }\left(\mI-\left(1-\frac2p\right)\overline{\mLambda}_c\right)^{-1}\mW_c^\frac{1}{2} h\|_\infty \\
\le & ~ \|\mD_c h\|_\infty + \|2\mW_c^{-1}\mP_c^{(2)}\mW_c^{-\frac{1}{2} }\left(\mI-\left(1-\frac2p\right)\overline{\mLambda}_c\right)^{-1}\mW_c^\frac{1}{2} h\|_\infty \\
\le & ~ p\|h\|_\infty + 2\|\mW_c^{-\frac{1}{2} }\left(\mI-\left(1-\frac2p\right)\overline{\mLambda}_c\right)^{-1}\mW_c^\frac{1}{2} h\|_{w(c)} \\
\le & ~ p\|h\|_\infty + 2\left\|\left(\mI-\left(1-\frac2p\right)\overline{\mLambda}_c\right)^{-1}\right\|_2\|\mW_c^\frac{1}{2} h\|_2 \\
\le & ~ p\|h\|_\infty + 2\|h\|_{w(c)}.
\end{align*}
At the end we have used that $\overline{\mLambda}_c$ is a positive semidefinite matrix so $\left(\mI+\left(\frac2p-1\right)\overline{\mLambda}_c\right)$ has all eigenvalues at least $1$, hence $0 \pe \left(\mI+\left(\frac2p-1\right)\overline{\mLambda}_c\right)^{-1} \pe \mI$.

Finally, note that
\begin{align*}
\|\mW_c^{-1}\mJ_c\mC h\|_{w(c)+\infty} 
\le & ~ p\|h\|_\infty + 2\|h\|_{w(c)} + \cnorm p\|h\|_{w(c)} \\ 
\le & ~ p\left(1+2/(\cnorm p)\right)\|h\|_{w(c)+\infty} \\
\le & ~ p(1+3/\cnorm)\|h\|_{w(c)+\infty}
\end{align*}
as $p \ge 2/3.$
\end{proof}

We will need one additional bound on the $\infty$ norm of a matrix that appears in the expression for the Jacobian in Lemma \ref{lemma:derivlewis}.
\begin{lemma}[Matrix $\infty$-norm bound]
\label{lemma:infbound}
In the notation of Lemma \ref{lemma:matrixbound} we have that for any vector $h$ that
\begin{align*} 
\left\|\mW_c^{-\frac{1}{2} }\left(\mI - \left(1-\frac2p\right)\overline{\mLambda}_c\right)^{-1}\mW_c^\frac{1}{2} h\right\|_\infty \le 3\|h\|_\infty. 
\end{align*}
\end{lemma}
\begin{proof}
By Lemma \ref{lemma:lemma47} we know $\|\mSigma_c^{-1}\mP_c^{(2)}\|_\infty \le 1.$ Define $\mM_c = \frac{p}{2}\left(\mW_c + (\frac2p-1)\mSigma_c\right) \se \mSigma_c$. Let $\ell = \frac{p}{2}\left(\frac2p-1\right).$ We calculate
\begin{align*}
\left\|\mW_c^{-\frac{1}{2} }\left(\mI - \left(1-\frac2p\right)\overline{\mLambda}_c\right)^{-1}\mW_c^\frac{1}{2} \right\|_\infty
= & ~ \left\|\left(\mW_c - \left(1-\frac2p\right)\mLambda_c\right)^{-1}\mW_c \right\|_\infty \\
= & ~ \left\|\left(\frac2p\mM_c - \left(\frac2p-1\right)\mP_c^{(2)}\right)^{-1}\mW_c \right\|_\infty \\
= & ~ \frac{p}{2}\left\|\mM_c^{-\frac{1}{2} }\left(\mI-\ell\mM_c^{-\frac{1}{2} }\mP_c^{(2)}\mM_c^{\frac{1}{2} }\right)^{-1}\mM_c^{-\frac{1}{2} }\mW_c \right\|_\infty \\
= & ~ \frac{p}{2}\left\|\mM_c^{-\frac{1}{2} }\sum_{i\ge0} \ell^i(\mM_c^{-\frac{1}{2} }\mP_c^{(2)}\mM_c^{-\frac{1}{2} })^i \mM_c^{-\frac{1}{2} }\mW_c \right\|_\infty \\
= & ~ \frac{p}{2}\left\|\sum_{i\ge0} \ell^i(\mM_c^{-1}\mP_c^{(2)})^i \mM_c^{-1}\mW_c \right\|_\infty \\
\le & ~ \frac{p}{2}\sum_{i\ge0} \ell^i \left\|\mSigma_c^{-1}\mP_c^{(2)}\right\|_\infty^i \left\| \mM^{-1}\mW_c \right\|_\infty \\
\le & ~ \left\| \mM_c^{-1}\mW_c \right\|_\infty \le \frac2p \le 3
\end{align*}
where we have used that if $\mX$ is a symmetric matrix with $\|\mX\|_2 < 1$ then $(\mI-\mX)^{-1} = \sum_{i\ge0} \mX^i$, and $\mM_c \se \frac{p}{2}\mW_c$ as diagonal matrices.
\end{proof}

\decomp*
\begin{proof}
We follow the notation of the proof of Lemma \ref{lemma:matrixbound}. We then have by (\ref{eq:matrixdiff1}) and (\ref{eq:matrixdiff2}) that
\begin{align*}
\mK_c = \mW_c^{-1}\mJ_c\mC - \mD_c = 2\mW_c^{-\frac{1}{2} }\left(\mI-\left(1-\frac2p\right)\mS_c\right)^{-1}\mW_c^{-\frac{1}{2} }\mP_c^{(2)}\mW_c^{-\frac{1}{2} }\left(\mI-\left(1-\frac2p\right)\overline{\mLambda}_c\right)^{-1}\mW_c^\frac{1}{2} .
\end{align*}
We first bound the $\infty$-norm of $\mK_c$. By Lemma \ref{lemma:lemma47} and Lemma \ref{lemma:infbound} we get
\begin{align*}
& ~ 2\left\|\mW_c^{-\frac{1}{2} }\left(\mI-\left(1-\frac2p\right)\mS_c\right)^{-1}\mW_c^{-\frac{1}{2} }\mP_c^{(2)}\mW_c^{-\frac{1}{2} }\left(\mI-\left(1-\frac2p\right)\overline{\mLambda}_c\right)^{-1}\mW_c^\frac{1}{2} \right\|_\infty \\
\le & ~ 2\|\mW_c^{-1}\mP_c^{(2)}\mW_c^{-\frac{1}{2} }\left(\mI-\left(1-\frac2p\right)\overline{\mLambda}_c\right)^{-1}\mW_c^\frac{1}{2} \|_\infty \ls 1.
\end{align*}
For the $\tau$-norm, we use that $\mW_c \se \mP_c^{(2)}$ and bound
\begin{align*}
& ~ 2\left\|\mW_c^{-\frac{1}{2} }\left(\mI-\left(1-\frac2p\right)\mS_c\right)^{-1}\mW_c^{-\frac{1}{2} }\mP_c^{(2)}\mW_c^{-\frac{1}{2} }\left(\mI-\left(1-\frac2p\right)\overline{\mLambda}_c\right)^{-1}\mW_c^\frac{1}{2} h\right\|_\tau \\
= & ~ 2\left\|\mW_c^{-\frac{1}{2} }\left(\mI-\left(1-\frac2p\right)\overline{\mLambda}_c\right)^{-1}\mW_c^{-\frac{1}{2} }\mP_c^{(2)}\mW_c^{-\frac{1}{2} }\left(\mI-\left(1-\frac2p\right)\mS_c\right)^{-1}\mW_c^\frac{1}{2} h\right\|_\tau \\
\le & ~ 2\left\|\left(\mI-\left(1-\frac2p\right)\overline{\mLambda}_c\right)^{-1}\mW_c^{-\frac{1}{2} }\mP_c^{(2)}\mW_c^{-\frac{1}{2} }\left(\mI-\left(1-\frac2p\right)\mS_c\right)^{-1}\mW_c^\frac{1}{2} h\right\|_2 \\
\le & ~ 2\left\|\left(\mI-\left(1-\frac2p\right)\mS_c\right)^{-1}h\right\|_{\mP_c^{(2)}\mW_c^{-1}\mP_c^{(2)}} \\
\le & ~ 2\left\|\left(\mI-\left(1-\frac2p\right)\mS_c\right)^{-1}h\right\|_{\mP_c^{(2)}} \le 2\left\||h|\right\|_{\mP_c^{(2)}}.
\end{align*}
In the third line (the equality) we have used that $\mW_c^\frac{1}{2}\mK_c\mW_c^{-\frac{1}{2}}$ is a symmetric matrix, as both $\mW_c^{-\frac{1}{2}}\mJ_c\mC\mW_c^{-\frac{1}{2}}$ and $\mD_c$ are, so it equals its transpose.
\end{proof}

\decomptwo*
\begin{proof}
By Lemma \ref{lemma:decomp} we can write
\begin{align*}
\mW_c^{-\frac{1}{2} }\left(\mI - \left(1-\frac2p\right)\overline{\mLambda}_c\right)^{-1}\mW_c^\frac{1}{2}  
= & ~ \mI + \left(1-\frac2p\right)\mW_c^{-\frac{1}{2} }\left(\mI - \left(1-\frac2p\right)\overline{\mLambda}_c\right)^{-1}\overline{\mLambda}_c\mW_c^\frac{1}{2}  \\
= & ~ \mI + \left(\frac{1}{2} -\frac1p\right)\mW_c^{-1}\mJ_c\mC \\
= & ~ \mI + \left(\frac{1}{2} -\frac1p\right)\mD_c + \left(\frac{1}{2} -\frac1p\right)\mK_c.
\end{align*}
We set $\mD_c' = \mI + \left(1-\frac2p\right)\mD_c$ and $\mK'_c = \left(\frac{1}{2} -\frac1p\right)\mK_c.$ This immediately implies the two claims about the $\tau$ and $\infty$ norms of $\mK_c'h.$ Finally, as in the proof of Lemma \ref{lemma:matrixbound} we know that
\begin{align*} 
\mI + \left(\frac{1}{2} -\frac1p\right)\mD_c = \mI + \left(1-\frac2p\right)\frac{\mS_c}{\mI+\left(\frac2p-1\right)\mS_c} = \left(\mI+\left(\frac2p-1\right)\mS_c\right)^{-1}. 
\end{align*} 
Therefore, $\mD_c'$ is diagonal and $0 \pe \mD_c' \pe \mI$ as desired.
\end{proof}

\lewisapprox*
\begin{proof}
Define $w_0 = w_p(\mC\mA)$, and
\begin{align*} 
w_{i+1} = \Iter(w_i, \bar{\mC}) \enspace \text{ for } \enspace \Iter(u, \mC)_k \defeq \left(c_k^2a_k^\top(\mA^\top \mC\mU^{1-\frac2p}\mC\mA)^{-1}a_k + u_k^{\frac2p-1}v_k\right)^\frac{p}{2}. 
\end{align*}
Note that if $\bar{\mC} \approx_\eps \mC$ then $\Iter(u, \mC) \approx_{2p\eps} \Iter(u, \bar{\mC})$ for all vectors $u$. Also, if $u \approx_\eps u'$ then $\Iter(u, \mC) \approx_{\left(1-\frac{p}{2}\right)\eps} \Iter(u', \mC)$ for all diagonal matrices $\mC$.
Now 
\begin{align*} 
w_0 = \Iter(w_0, \mC) \approx_{2p\eps} \Iter(w_0, \bar{\mC}) = w_1. 
\end{align*} 
If $w_i \approx_\delta w_{i+1}$ then
\begin{align*} 
w_{i+1} = \Iter(w_i, \bar{\mC}) \approx_{\left(1-\frac{p}{2}\right)\delta} \Iter(w_{i+1}, \bar{\mC}) = w_{i+2}. 
\end{align*}
Therefore 
\begin{align*} 
w_0 \approx_{\sum_{i \ge 0} 2\left(1-\frac{p}{2}\right)^ip\eps} \lim_{k\rightarrow\infty} w_k = w_p(\bar{\mC}\mA) 
\end{align*}
and 
\begin{align*} 
2\sum_{i \ge 0} \left(1-\frac{p}{2}\right)^ip\eps = 4\eps 
\end{align*}
 as desired.
\end{proof}

Matrices such as $\mP \circ (\mP\mX\mP)$ appears in the derivative of $\mP^{(2)}$, and can be bounded as follows.
\begin{lemma}[Projection matrix facts]
\label{lemma:projfact}
Let $\mP$ be a projection matrix with diagonal given by $\Tau$. For all vectors $x, v$ and $\mX = \diag(x), \mV = \diag(v)$ we have that
\begin{align*} 
|\mP \circ (\mP \mX \mP)v| \le \frac{1}{2} \mP^{(2)}(x^2) + \frac{1}{2} \mP^{(2)}(v^2) 
\end{align*} 
as vectors coordinate-wise. In particular, we also have that 
\begin{itemize}
\item $\left\|\Tau^{-1}\mP \circ (\mP \mX \mP)v\right\|_\infty \le \frac{1}{2}(\|x\|_\infty^2 + \|v\|_\infty^2)$.
\item $\left\|\Tau^{-1}\mP \circ (\mP \mX \mP)v\right\|_\tau \le \frac{1}{2}(\|x^2\|_{\mP^{(2)}} + \|v^2\|_{\mP^{(2)}}) \le \frac{1}{2}(\|x^2\|_{\tau} + \|v^2\|_{\tau})$.
\end{itemize}
\end{lemma}
\begin{proof}
By the inequality $|a^\top b| \le \frac{1}{2} a^\top a + \frac{1}{2} b^\top b$ we have
\begin{align*}
|e_i^\top \mP \circ (\mP \mX \mP)v| 
= & ~ \left|e_i^\top \mP \mX \mP \mV \mP e_i\right| \\
\le & ~ \frac{1}{2} e_i^\top \mP \mX \mP \mX \mP e_i + \frac{1}{2} e_i^\top \mP \mV \mP \mV \mP e_i \\
\le & ~ \frac{1}{2} e_i^\top \mP \mX^2 \mP e_i + \frac{1}{2} e_i^\top \mP \mV^2 \mP e_i \\
= & ~ \frac{1}{2} e_i^\top\mP^{(2)}(x^2) + \frac{1}{2} e_i^\top\mP^{(2)}(v^2).
\end{align*}
Also, we then can use Lemma \ref{lemma:lemma47} that $\|\Tau^{-1}\mP^{(2)}\|_\infty \le 1$ to get
\begin{align*} 
\left\|\Tau^{-1}\mP \circ (\mP \mX \mP)v\right\|_\infty \le \frac{1}{2} \|\Tau^{-1}\mP^{(2)}(x^2)\|_\infty + \frac{1}{2} \|\Tau^{-1}\mP^{(2)}(v^2)\|_\infty \le \frac{1}{2}(\|x\|_\infty^2 + \|v\|_\infty^2). 
\end{align*}
Finally, we have that
\begin{align*} 
\left\|\Tau^{-1}\mP \circ (\mP \mX \mP)v\right\|_{\tau} 
\le & ~ \frac{1}{2} \|\Tau^{-1}\mP^{(2)}(x^2)\|_{\tau} + \frac{1}{2} \|\Tau^{-1}\mP^{(2)}(v^2)\|_{\tau} \\
= & ~ \frac{1}{2}(\|x^2\|_{\mP^{(2)}\Tau^{-1}\mP^{(2)}} + \|v^2\|_{\mP^{(2)}\Tau^{-1}\mP^{(2)}}) \\
\le & ~ \frac{1}{2}(\|x^2\|_{\mP^{(2)}} + \|v^2\|_{\mP^{(2)}})
\end{align*}
as $\mP^{(2)} \pe \Tau$, so $\Tau^{-1} \pe (\mP^{(2)})^{-1}$.
\end{proof}

\tauchangetwo*
\begin{proof}
We adopt the same notation as Lemma \ref{lemma:firstbound}, and also define \\ $\mN_t = 2\left(\mI-\left(1-\frac2p\right)\overline{\mLambda_t}\right)^{-1}\overline{\mLambda_t}.$ 
We have
\begin{align}\label{eq:taylorswift} 
 \Tau^{-1}(\E[\d_\tau] - \mJ\E[\d_c]) = \int_0^1 (1-t) \Tau^{-1}\E\left[\frac{ \mathrm{d} }{ \mathrm{d} t }  \mJ_t \d_c\right]. 
\end{align}
Recall that $\mJ_t = \Tau_t^\frac{1}{2}  \mN_t \Tau_t^\frac{1}{2}  \mC_t^{-1}$, so
\begin{align*} 
\frac{ \mathrm{d} }{ \mathrm{d} t }  \mJ_t 
& =  \left(\frac{ \mathrm{d} }{ \mathrm{d} t } \Tau_t^\frac{1}{2} \right) \mN_t \Tau_t^\frac{1}{2}  \mC_t^{-1} + \Tau_t^\frac{1}{2}  \left(\frac{ \mathrm{d} }{ \mathrm{d} t } \mN_t\right) \Tau_t^\frac{1}{2}  \mC_t^{-1} \\
&+ \Tau_t^\frac{1}{2}  \mN_t \left(\frac{ \mathrm{d} }{ \mathrm{d} t }  \Tau_t^\frac{1}{2}  \right) \mC_t^{-1} + \Tau_t^\frac{1}{2}  \mN_t \Tau_t^\frac{1}{2}  \left(\frac{ \mathrm{d} }{ \mathrm{d} t }  \mC_t^{-1} \right).
\end{align*}
We will bound all four terms separately. Define $\Delta_{\tau_t} = \diag(\d_{\tau_t})$, and note that $\frac{ \mathrm{d} }{ \mathrm{d} t }  \Tau_t = \Delta_{\tau_t}$. Therefore, we may use Lemma \ref{lemma:firstbound} to bound the first term with
\begin{align*}
\left\|\E\Tau^{-1}\left(\frac{ \mathrm{d} }{ \mathrm{d} t } \Tau_t^\frac{1}{2} \right) \mN_t \Tau_t^\frac{1}{2}  \mC_t^{-1}\d_c\right\|_\tpi 
= & ~ \frac{1}{2} \left\|\E\Tau^{-1}\Tau_t^{-1}\Delta_{\tau_t}\Tau_t^\frac{1}{2} \mN_t \Tau_t^\frac{1}{2}  \mC_t^{-1}\d_c\right\|_\tpi \\
= & ~ .5\left\|\E[\Tau^{-1}\Tau_t^{-1}\Delta_{\tau_t}\mJ_t\d_c]\right\|_\tpi \\
\ls & ~ \left\|\E[(\Tau^{-1}\d_{\tau_t})^2]\right\|_\tpi \ls \gamma^2.
\end{align*}
For the third term we first define $v_t = \Tau_t^{-1}\Delta_{\tau_t}\mC_t^{-1}\d_c$ and use Lemma \ref{lemma:decomp} to get
\begin{align*}
\left\|\E\Tau^{-1}\Tau_t^\frac{1}{2}  \mN_t \left(\frac{ \mathrm{d} }{ \mathrm{d} t } \Tau_t^\frac{1}{2} \right)\mC_t^{-1}\d_c\right\|_\tpi 
\ls & ~ \left\|\E\Tau^{-1}\Tau_t\Tau_t^{-1} \Tau_t^\frac{1}{2}  \mN_t \Tau_t^\frac{1}{2}  \Tau_t^{-1}\Delta_{\tau_t}\mC_t^{-1}\d_c\right\|_\tpi \\
= & ~ \left\|\E\Tau^{-1}\Tau_t\Tau_t^{-1}\mJ_t\mC_tv_t\right\|_\tpi \\
\le & ~ \|\E\Tau^{-1}\Tau_t\mD_tv_t\|_\tpi + \|\E\Tau^{-1}\Tau_t\mK_tv_t\|_\tpi.
\end{align*}
Note that 
\begin{align*} 
|v_t| \ls (\Tau_t^{-1}\d_{\tau_t})^2 + (\mC_t^{-1}\d_c)^2. 
\end{align*} 
Therefore, we use Lemma \ref{lemma:firstbound} and $\gamma$-boundedness (Definition \ref{def:gammabounded} \eqref{eq:gammaboundedpart2}) to bound
\begin{align*} 
\|\E\Tau^{-1}\Tau_t\mD_tv_t\|_\tpi \ls \|\E(\Tau_t^{-1}\d_{\tau_t})^2\|_\tpi+\|\E(\mC_t^{-1}\d_c)^2\|_\tpi \ls \gamma^2. 
\end{align*}
For the term with $\mK_t$, we use Lemma \ref{lemma:decomp} to first bound
\begin{align*} 
\|\Tau^{-1}\Tau_t\mK_tv_t\|_\infty \ls \|\mK_tv_t\|_\infty \ls \|\Tau_t^{-1}\d_{\tau_t}\|_\infty\|\mC_t^{-1}\d_c\|_\infty \ls \gamma^2. 
\end{align*}
Also, we can use Lemma \ref{lemma:decomp}, Lemma \ref{lemma:firstbound}, and $\gamma$-boundedness (Definition \ref{def:gammabounded} \eqref{eq:gammaboundedpart3}) to get
\begin{align*}
\E\|\Tau^{-1}\Tau_t\mK_tv_t\|_\tau 
\ls & ~ \E\|\mK_tv_t\|_\tau \\
\le & ~ \E\||\Tau_t^{-1}\Delta_{\tau_t}\mC_t^{-1}\d_c|\|_{\mP_t^{(2)}} \\
\ls & ~ \gamma\E\|\mC^{-1}|\d_c|\|_{\mP_t^{(2)}} \ls \gamma^2/\cnorm.
\end{align*}
Therefore, the total contribution from the third term is $O(\gamma^2)$.
For the fourth term, we use Lemma \ref{lemma:decomp} to write
\begin{align*}
\left\|\E \Tau^{-1} \Tau_t^\frac{1}{2}  \mN_t \Tau_t^\frac{1}{2}  \left(\frac{ \mathrm{d} }{ \mathrm{d} t }  \mC_t^{-1} \right)\d_c\right\|_\tpi 
= & ~ \left\|\E \Tau^{-1}\Tau_t\Tau_t^{-1}\mJ_t\mC_t(\mC_t^{-1}\d_c)^2 \right\|_\tpi \\
\le & ~ \|\E \Tau^{-1}\Tau_t\mD_t(\mC_t^{-1}\d_c)^2\|_\tpi + \|\E \Tau^{-1}\Tau_t\mK_t(\mC_t^{-1}\d_c)^2\|_\tpi.
\end{align*}
For the first piece, use $\gamma$-boundedness (Definition \ref{def:gammabounded} \eqref{eq:gammaboundedpart2}) to bound
\begin{align*} 
\|\E \Tau^{-1}\Tau_t\mD_t(\mC_t^{-1}\d_c)^2\|_\tpi \ls \|\E(\mC_t^{-1}\d_c)^2\|_\tpi \ls \gamma^2. 
\end{align*}
For the second piece, we use Lemma \ref{lemma:decomp} and $\gamma$-boundedness (Definition \ref{def:gammabounded} \eqref{eq:gammaboundedpart2}) to first bound
\begin{align*} 
\|\E \Tau^{-1}\Tau_t\mK_t(\mC_t^{-1}\d_c)^2\|_\infty \ls \|(\mC_t^{-1}\d_c)^2\|_\infty \ls \gamma^2. 
\end{align*}
For the $\tau$-norm we can use Lemma \ref{lemma:decomp}, $\gamma$-boundedness (Definition \ref{def:gammabounded} \eqref{eq:gammaboundedpart2} and \eqref{eq:gammaboundedpart3}) to bound
\begin{align*} 
\E\|\Tau^{-1}\Tau_t\mK_t(\mC_t^{-1}\d_c)^2\|_\tau \ls \E\|(\mC_t^{-1}\d_c)^2\|_{\mP_t^{(2)}} \ls \gamma^2/\cnorm. 
\end{align*}
Therefore, the total contribution from this case is at most $O(\gamma^2)$.

Now we bound the contribution from the $\frac{ \mathrm{d} }{ \mathrm{d} t } \mN_t$ term.
\paragraph{Analysis of $\frac{ \mathrm{d} }{ \mathrm{d} t } \mN_t$.} By the chain rule, we have that
\begin{align*} 
\frac{ \mathrm{d} }{ \mathrm{d} t }  \mN_t 
= & ~ 2\left(\mI-\left(1-\frac2p\right)\overline{\mLambda_t}\right)^{-1} \frac{ \mathrm{d} }{ \mathrm{d} t }  \overline{\mLambda_t} \\ 
+ & ~ 2\left(1-\frac2p\right) \left(\mI-\left(1-\frac2p\right)\overline{\mLambda_t}\right)^{-1} \frac{ \mathrm{d} }{ \mathrm{d} t }  \overline{\mLambda_t} \left(\mI-\left(1-\frac2p\right)\overline{\mLambda_t}\right)^{-1} \overline{\mLambda_t} \\
= & ~ 2\left(\mI-\left(1-\frac2p\right)\overline{\mLambda_t}\right)^{-1} \frac{ \mathrm{d} }{ \mathrm{d} t }  \overline{\mLambda_t} \left(\mI-\left(1-\frac2p\right)\overline{\mLambda_t}\right)^{-1}.
\end{align*}
Recall that $\overline{\mLambda_t} = \Tau_t^{-1}\mSigma_t - \Tau_t^{-\frac{1}{2} }\mP_t^{(2)}\Tau_t^{-\frac{1}{2} }$. Define the vectors $z_t = \Tau_t^{-1}\sigma_t, \ell_t = \Tau_t^{\frac{1}{2} -\frac1p}c_t$, and the diagonal matrices $\mZ_t = \diag(z_t) = \Tau_t^{-1}\mSigma_t$, $\mL_t = \Tau_t^{\frac{1}{2} -\frac1p}\mC_t$, so that $\mP_t = \mP(\mL_t\mA) = \mA^\top\mL_t(\mA^\top\mL_t^2\mA)^{-1}\mL_t\mA$. In this way, $\overline{\mLambda_t} = \mZ_t - \Tau_t^{-\frac{1}{2} }\mP(\mL_t\mA)^{(2)}\Tau_t^{-\frac{1}{2} }.$

Define $\d_{z_t} = \frac{ \mathrm{d} }{ \mathrm{d} t } z_t, \d_{\ell_t} = \frac{ \mathrm{d} }{ \mathrm{d} t } \ell_t$, and $\Delta_{z_t} = \diag(\d_{z_t})$, $\Delta_{\ell_t} = \diag(\d_{\ell_t})$. We have
\begin{align*} 
\Delta_{z_t} = -\Tau_t^{-2}\mSigma_t \Delta_{\tau_t} + \Tau_t^{-1}\Delta_{\tau_t} = \Tau_t^{-1}\Delta_{\tau_t}(\mI - \Tau_t^{-1}\mSigma_t). 
\end{align*}
and
\begin{align*} 
\Delta_{\ell_t} = \left(\frac{1}{2} -\frac1p\right)\Tau_t^{-\frac{1}{2} -\frac1p}\mC_t\Delta_{\tau_t} + \Tau_t^{\frac{1}{2} -\frac1p}\d_c 
\end{align*}
so
\begin{align*} 
\mL_t^{-1}\Delta_{\ell_t} = \left(\frac{1}{2} -\frac1p\right)\Tau_t^{-1}\Delta_{\tau_t} + \mC_t^{-1}\d_c. 
\end{align*}
Therefore, using Lemma \ref{lemma:firstbound}, $\gamma$-boundedness (Definition \ref{def:gammabounded} \eqref{eq:gammaboundedpart1}, \eqref{eq:gammaboundedpart2}, \eqref{eq:gammaboundedpart3}) we get that
\begin{align} \label{eq:z} 
\|\d_{z_t}\|_\infty \ls \gamma \enspace \text{ and } \enspace \|\E[\d_{z_t}^2]\|_\tpi \ls \gamma^2 \enspace \text{ and } \enspace \E\||\d_{z_t}|\|_{\mP_t^{(2)}} \ls \gamma/\cnorm 
\end{align}
and
\begin{align} 
\|\mL_t^{-1}\d_{\ell_t}\|_\infty \ls \gamma \enspace \text{ and } \enspace \|\E[(\mL_t^{-1}\d_{\ell_t})^2]\|_\tpi \ls \gamma^2 \enspace \text{ and } \enspace \E\|\mL_t^{-1}|\d_{\ell_t}|\|_{\mP_t^{(2)}} \ls \gamma/\cnorm. \label{eq:l} 
\end{align}
A direct calculation gives that
\begin{align}
\frac{ \mathrm{d} }{ \mathrm{d} t }  \overline{\mLambda_t} &= \Delta_{z_t} + \frac{1}{2} \Tau_t^{-\frac32}\Delta_{\tau_t}\mP_t^{(2)}\Tau_t^{-\frac{1}{2} } + \frac{1}{2} \Tau_t^{-\frac{1}{2} }\mP_t^{(2)}\Tau_t^{-\frac32}\Delta_{\tau_t} - 2\Tau_t^{-\frac{1}{2} } \mP(\mL_t\mA) \circ \left(\frac{ \mathrm{d} }{ \mathrm{d} t }  \mP(\mL_t\mA)\right)\Tau_t^{-\frac{1}{2} } \notag \\
&= \Delta_{z_t} + \frac{1}{2} \Tau_t^{-\frac32}\Delta_{\tau_t}\mP_t^{(2)}\Tau_t^{-\frac{1}{2} } + \frac{1}{2} \Tau_t^{-\frac{1}{2} }\mP_t^{(2)}\Tau_t^{-\frac32}\Delta_{\tau_t} \label{eq:bound1} \\
&- 2\Tau_t^{-\frac{1}{2} }\mP_t^{(2)}\mL_t^{-1}\Delta_{\ell_t}\Tau_t^{-\frac{1}{2} } - 2\Tau_t^{-\frac{1}{2} }\mL_t^{-1}\Delta_{\ell_t}\mP_t^{(2)}\Tau_t^{-\frac{1}{2} } 
+ 4\Tau_t^{-\frac{1}{2} }(\mP_t \circ (\mP_t \mL_t^{-1}\Delta_{\ell_t} \mP_t))\Tau_t^{-\frac{1}{2} }. \label{eq:bound2}
\end{align}
There are six terms here and we analyze them one by one. At a high level, for each of the six terms, we will use Lemma \ref{lemma:decomp2} to handle the $\left(\mI-\left(1-\frac2p\right)\overline{\mLambda_t}\right)^{-1}$ terms in $\frac{ \mathrm{d} }{ \mathrm{d} t }  \mN_t$. For simplicity, we will omit the factor of $2$. For the $\Delta_{z_t}$ term in (\ref{eq:bound1}) we get
\begin{align}
& ~ \Tau^{-1}\Tau_t^\frac{1}{2} \left(\mI-\left(1-\frac2p\right)\overline{\mLambda_t}\right)^{-1} \Delta_{z_t} \left(\mI-\left(1-\frac2p\right)\overline{\mLambda_t}\right)^{-1} \Tau_t^\frac{1}{2}  \mC_t^{-1} \d_c \notag \\
= & ~ \Tau^{-1}\Tau_t(\mD_t'+\mK_t')\Delta_{z_t}(\mD_t'+\mK_t')\mC_t^{-1}\d_c \notag\\
= & ~ \Tau^{-1}\Tau_t\mD_t'\Delta_{z_t}\mD_t'\mC_t^{-1}\d_c + \Tau^{-1}\Tau_t\mK_t'\Delta_{z_t}\mD_t'\mC_t^{-1}\d_c \label{eq:bound11} \\
+ & ~ \Tau^{-1}\Tau_t\mD_t'\Delta_{z_t}\mK_t'\mC_t^{-1}\d_c + \Tau^{-1}\Tau_t\mK_t'\Delta_{z_t}\mK_t'\mC_t^{-1}\d_c. \label{eq:bound12}
\end{align}
We know by (\ref{eq:z}) and $\gamma$-boundedness (Definition \ref{def:gammabounded} \eqref{eq:gammaboundedpart1})
\begin{align*} 
\left\|\Tau^{-1}\Tau_t^\frac{1}{2} \left(\mI-\left(1-\frac2p\right)\overline{\mLambda_t}\right)^{-1} \Delta_{z_t} \left(\mI-\left(1-\frac2p\right)\overline{\mLambda_t}\right)^{-1} \Tau_t^\frac{1}{2}  \mC_t^{-1} \d_c\right\|_\infty \ls \gamma^2, 
\end{align*} 
so we focus on the $\tau$-norm. First we use (\ref{eq:z}) and $\gamma$-boundedness (Definition \ref{def:gammabounded} \eqref{eq:gammaboundedpart2}) to get
\begin{align*} 
\|\E[\Tau^{-1}\Tau_t\mD_t'\Delta_{z_t}\mD_t'\mC_t^{-1}\d_c]\|_\tau \le \|\E[\d_{z_t}^2]\|_\tau + \|\E[(\mC_t^{-1}\d_c)^2]\|_\tau \ls \gamma^2/\cnorm. 
\end{align*}
For the remaining three terms in (\ref{eq:bound11}), (\ref{eq:bound12}) we can use Lemma \ref{lemma:decomp2} to bound
\begin{align*} 
\|\E[\Tau^{-1}\Tau_t\mK_t'\Delta_{z_t}\mD_t'\mC_t^{-1}\d_c]\|_\tau \ls \|\d_{z_t}\|_\infty\E\|\mC_t^{-1}|\d_c|\|_{\mP_t^{(2)}} \ls \gamma^2/\cnorm 
\end{align*} 
and similar for the other two terms. Therefore, the total contribution from (\ref{eq:bound11}), (\ref{eq:bound12}) in $\tpi$ norm is $O(\gamma^2)$.

For the $\frac{1}{2} \Tau_t^{-\frac32}\Delta_{\tau_t}\mP_t^{(2)}\Tau_t^{-\frac{1}{2} }$ term in (\ref{eq:bound1}) we can use Lemma \ref{lemma:decomp2} to write (omitting the $\frac{1}{2} $)
\begin{align}
& ~ \Tau^{-1}\Tau_t^\frac{1}{2} \left(\mI-\left(1-\frac2p\right)\overline{\mLambda_t}\right)^{-1} \Tau_t^{-\frac32}\Delta_{\tau_t}\mP_t^{(2)}\Tau_t^{-\frac{1}{2} } \left(\mI-\left(1-\frac2p\right)\overline{\mLambda_t}\right)^{-1} \Tau_t^\frac{1}{2}  \mC_t^{-1} \d_c \notag \\
= & ~ \Tau^{-1}\Tau_t(\mD_t'+\mK_t')\Tau_t^{-2}\Delta_{\tau_t}\mP_t^{(2)}(\mD_t'+\mK_t')\mC_t^{-1}\d_c \notag \\
= & ~ \Tau^{-1}\Tau_t\mD_t'\Tau_t^{-2}\Delta_{\tau_t}\mP_t^{(2)}\mD_t'\mC_t^{-1}\d_c + \Tau^{-1}\Tau_t\mK_t'\Tau_t^{-2}\Delta_{\tau_t}\mP_t^{(2)}\mD_t'\mC_t^{-1}\d_c \label{eq:bound21} \\
+ & ~ \Tau^{-1}\Tau_t\mD_t'\Tau_t^{-2}\Delta_{\tau_t}\mP_t^{(2)}\mK_t'\mC_t^{-1}\d_c + \Tau^{-1}\Tau_t\mK_t'\Tau_t^{-2}\Delta_{\tau_t}\mP_t^{(2)}\mK_t'\mC_t^{-1}\d_c. \label{eq:bound22}
\end{align}
We can bound using $\|\Tau_t^{-1}\mP_t^{(2)}\|_\infty \le 1$ (Lemma \ref{lemma:lemma47})
\begin{align*}
&\left\|\Tau^{-1}\Tau_t^\frac{1}{2} \left(\mI-\left(1-\frac2p\right)\overline{\mLambda_t}\right)^{-1} \Tau_t^{-\frac32}\Delta_{\tau_t}\mP_t^{(2)}\Tau_t^{-\frac{1}{2} } \left(\mI-\left(1-\frac2p\right)\overline{\mLambda_t}\right)^{-1} \Tau_t^\frac{1}{2}  \mC_t^{-1} \d_c\right\|_\infty \ls \gamma^2.
\end{align*}
To bound the $\tau$-norm we can bound the first term in (\ref{eq:bound21}) with
\begin{align*} 
\|\E[\Tau^{-1}\Tau_t\mD_t'\Tau_t^{-2}\Delta_{\tau_t}\mP_t^{(2)}\mD_t'\mC_t^{-1}\d_c]\|_\tau \ls \|\Tau_t^{-1}\d_{\tau_t}\|_\infty \E\|\mC_t^{-1}|\d_c|\|_{\mP_t^{(2)}} \ls \gamma^2/\cnorm. 
\end{align*}
The other terms in (\ref{eq:bound21}), (\ref{eq:bound22}) follow similarly, e.g. using $\|\Tau_t^{-1}\mP_t^{(2)}\|_\tau \le 1$
\begin{align*}
\E [ \|\Tau^{-1}\Tau_t\mD_t'\Tau_t^{-2}\Delta_{\tau_t}\mP_t^{(2)}\mK_t'\mC_t^{-1}\d_c\|_\tau ]
\ls & ~ \E [ \|\Tau_t^{-1}\Delta_{\tau_t}\Tau_t^{-1}\mP_t^{(2)}\mK_t'\mC_t^{-1}\d_c\|_\tau ] \\
\ls & ~ \|\Tau_t^{-1}\d_{\tau_t}\|_\infty \E [ \|\mK_t'\mC_t^{-1}\d_c\|_\tau ] \\
\ls & ~ \gamma \E [ \|\mC_t^{-1}|\d_c|\|_{\mP_t^{(2)}} ] \\
\ls & ~ \gamma^2/\cnorm.
\end{align*}
Therefore, the total contribution from (\ref{eq:bound21}), (\ref{eq:bound22}) in $\tpi$ norm is $O(\gamma^2).$ The $\frac{1}{2} \Tau_t^{-\frac{1}{2} }\mP_t^{(2)}\Tau_t^{-\frac32}\Delta_{\tau_t}$ term in (\ref{eq:bound1}) can be handled equivalently.

We turn to the $2\Tau_t^{-\frac{1}{2} }\mP_t^{(2)}\mL_t^{-1}\Delta_{\ell_t}\Tau_t^{-\frac{1}{2} }$ in (\ref{eq:bound2}), which can be handled similarly to previous bounds. We can use Lemma \ref{lemma:decomp2} to write (omitting the $2$)
\begin{align}
& ~ \Tau^{-1}\Tau_t^\frac{1}{2} \left(\mI-\left(1-\frac2p\right)\overline{\mLambda_t}\right)^{-1} \Tau_t^{-\frac{1}{2} }\mP_t^{(2)}\mL_t^{-1}\Delta_{\ell_t}\Tau_t^{-\frac{1}{2} } \left(\mI-\left(1-\frac2p\right)\overline{\mLambda_t}\right)^{-1} \Tau_t^\frac{1}{2}  \mC_t^{-1} \d_c \notag \\
= & ~ \Tau^{-1}\Tau_t(\mD_t'+\mK_t')\Tau_t^{-1}\mP_t^{(2)}\mL_t^{-1}\Delta_{\ell_t}(\mD_t'+\mK_t')\mC_t^{-1}\d_c \notag \\
= & ~ \Tau^{-1}\Tau_t\mD_t'\Tau_t^{-1}\mP_t^{(2)}\mL_t^{-1}\Delta_{\ell_t}\mD_t'\mC_t^{-1}\d_c + \Tau^{-1}\Tau_t\mK_t'\Tau_t^{-1}\mP_t^{(2)}\mL_t^{-1}\Delta_{\ell_t}\mD_t'\mC_t^{-1}\d_c \label{eq:bound41} \\
+ & ~ \Tau^{-1}\Tau_t\mD_t'\Tau_t^{-1}\mP_t^{(2)}\mL_t^{-1}\Delta_{\ell_t}\mK_t'\mC_t^{-1}\d_c + \Tau^{-1}\Tau_t\mK_t'\Tau_t^{-1}\mP_t^{(2)}\mL_t^{-1}\Delta_{\ell_t}\mK_t'\mC_t^{-1}\d_c. \label{eq:bound42}
\end{align}
We can bound using $\|\Tau_t^{-1}\mP_t^{(2)}\|_\infty \le 1$ (Lemma \ref{lemma:lemma47})
\begin{align*}
&\left\|\Tau^{-1}\Tau_t^\frac{1}{2} \left(\mI-\left(1-\frac2p\right)\overline{\mLambda_t}\right)^{-1} \Tau_t^{-\frac{1}{2} }\mP_t^{(2)}\mL_t^{-1}\Delta_{\ell_t}\Tau_t^{-\frac{1}{2} } \left(\mI-\left(1-\frac2p\right)\overline{\mLambda_t}\right)^{-1} \Tau_t^\frac{1}{2}  \mC_t^{-1} \d_c\right\|_\infty \ls \gamma^2.
\end{align*}
To bound the $\tau$-norm we can bound the first term in (\ref{eq:bound41}) with
\begin{align*} 
\|\E[\Tau^{-1}\Tau_t\mD_t'\Tau_t^{-1}\mP_t^{(2)}\mL_t^{-1}\Delta_{\ell_t}\mD_t'\mC_t^{-1}\d_c]\|_\tau \ls \|\mL_t^{-1}\d_{\ell_t}\|_\infty \E\|\mC_t^{-1}|\d_c|\|_{\mP_t^{(2)}} \ls \gamma^2/\cnorm. 
\end{align*}
The other terms in (\ref{eq:bound21}), (\ref{eq:bound22}) follow similarly, e.g. using $\|\Tau_t^{-1}\mP_t^{(2)}\|_\tau \le 1$
\begin{align*}
\E\|\Tau^{-1}\Tau_t\mD_t'\Tau_t^{-1}\mP_t^{(2)}\mL_t^{-1}\Delta_{\ell_t}\mK_t'\mC_t^{-1}\d_c\|_\tau
\ls & ~ \E\|\Tau_t^{-1}\mP_t^{(2)}\mL_t^{-1}\Delta_{\ell_t}\mK_t'\mC_t^{-1}\d_c\|_\tau \\
\ls & ~ \E\|\mL_t^{-1}\Delta_{\ell_t}\mK_t'\mC_t^{-1}\d_c\|_\tau \\
\ls & ~ \|\mL_t^{-1}\d_{\ell_t}\|_\infty \E\|\mK_t'\mC_t^{-1}\d_c\|_\tau \\
\ls & ~ \gamma \E\|\mC_t^{-1}|\d_c|\|_{\mP_t^{(2)}} \ls \gamma^2/\cnorm.
\end{align*}
Therefore, the total contribution from the $2\Tau_t^{-\frac{1}{2} }\mP_t^{(2)}\mL_t^{-1}\Delta_{\ell_t}\Tau_t^{-\frac{1}{2} }$ term in (\ref{eq:bound2}) $O(\gamma^2)$. The $2\Tau_t^{-\frac{1}{2} }\mL_t^{-1}\Delta_{\ell_t}\mP_t^{(2)}\Tau_t^{-\frac{1}{2} }$ term in (\ref{eq:bound2}) can be handled equivalently.

Finally we bound the $4\Tau_t^{-\frac{1}{2} }(\mP_t \circ (\mP_t \mL_t^{-1}\Delta_{\ell_t} \mP_t))\Tau_t^{-\frac{1}{2} }$ term in (\ref{eq:bound2}). We start by using Lemma \ref{lemma:decomp2} to write (omitting the factor of $4$)
\begin{align}
& ~ \Tau^{-1}\Tau_t^\frac{1}{2} \left(\mI-\left(1-\frac2p\right)\overline{\mLambda_t}\right)^{-1} \Tau_t^{-\frac{1}{2} }(\mP_t \circ (\mP_t \mL_t^{-1}\Delta_{\ell_t} \mP_t))\Tau_t^{-\frac{1}{2} } \left(\mI-\left(1-\frac2p\right)\overline{\mLambda_t}\right)^{-1} \Tau_t^\frac{1}{2}  \mC_t^{-1} \d_c \notag \\
= & ~ \Tau^{-1}\Tau_t(\mD_t'+\mK_t')\Tau_t^{-1}(\mP_t \circ (\mP_t \mL_t^{-1}\Delta_{\ell_t} \mP_t))(\mD_t'+\mK_t')\mC_t^{-1}\d_c \notag \\
= & ~ \Tau^{-1}\Tau_t\mD_t'\Tau_t^{-1}(\mP_t \circ (\mP_t \mL_t^{-1}\Delta_{\ell_t} \mP_t))\mD_t'\mC_t^{-1}\d_c + \Tau^{-1}\Tau_t\mK_t'\Tau_t^{-1}(\mP_t \circ (\mP_t \mL_t^{-1}\Delta_{\ell_t} \mP_t))\mD_t'\mC_t^{-1}\d_c \label{eq:bound61} \\
+ & ~ \Tau^{-1}\Tau_t\mD_t'\Tau_t^{-1}(\mP_t \circ (\mP_t \mL_t^{-1}\Delta_{\ell_t} \mP_t))\mK_t'\mC_t^{-1}\d_c + \Tau^{-1}\Tau_t\mK_t'\Tau_t^{-1}(\mP_t \circ (\mP_t \mL_t^{-1}\Delta_{\ell_t} \mP_t))\mK_t'\mC_t^{-1}\d_c. \label{eq:bound62}
\end{align}
Using Lemma \ref{lemma:infbound} and \ref{lemma:projfact} we get that
\begin{align*}
& ~ \left\|\Tau^{-1}\Tau_t^\frac{1}{2} \left(\mI-\left(1-\frac2p\right)\overline{\mLambda_t}\right)^{-1} \Tau_t^{-\frac{1}{2} }(\mP_t \circ (\mP_t \mL_t^{-1}\Delta_{\ell_t} \mP_t))\Tau_t^{-\frac{1}{2} } \left(\mI-\left(1-\frac2p\right)\overline{\mLambda_t}\right)^{-1} \Tau_t^\frac{1}{2}  \mC_t^{-1} \d_c\right\|_\infty \\
\ls & ~ \left\|\Tau_t^{-1}(\mP_t \circ (\mP_t \mL_t^{-1}\Delta_{\ell_t} \mP_t))\Tau_t^{-\frac{1}{2} } \left(\mI-\left(1-\frac2p\right)\overline{\mLambda_t}\right)^{-1} \Tau_t^\frac{1}{2}  \mC_t^{-1} \d_c\right\|_\infty \\
\ls & ~ \|\mL_t^{-1}\Delta_{\ell_t}\|_\infty^2 + \left\|\Tau_t^{-\frac{1}{2} } \left(\mI-\left(1-\frac2p\right)\overline{\mLambda_t}\right)^{-1} \Tau_t^\frac{1}{2}  \mC_t^{-1} \d_c\right\|_\infty^2 \ls \gamma^2
\end{align*}
by (\ref{eq:l}) and $\gamma$-boundedness (Definition \ref{def:gammabounded} \eqref{eq:gammaboundedpart1}). We now bound the $\tau$-norm of all four terms in (\ref{eq:bound61}) and (\ref{eq:bound62}). For the first term in (\ref{eq:bound61}) we bound using Lemma \ref{lemma:projfact} and (\ref{eq:l})
\begin{align}
& ~ \E\|\Tau^{-1}\Tau_t\mD_t'\Tau_t^{-1}(\mP_t \circ (\mP_t \mL_t^{-1}\Delta_{\ell_t} \mP_t))\mD_t'\mC_t^{-1}\d_c\|_\tau \notag \\
\ls & ~ \E\|\Tau_t^{-1}(\mP_t \circ (\mP_t \mL_t^{-1}\Delta_{\ell_t} \mP_t))\mD_t'\mC_t^{-1}\d_c\|_\tau \label{eq:follow1} \\
\ls & ~ \E\|(\mL_t^{-1}\d_{\ell_t})^2\|_{\mP_t^{(2)}} + \E\|(\mD_t'\mC_t^{-1}\d_c)^2\|_{\mP_t^{(2)}} \notag \\
\ls & ~ \|\mL_t^{-1}\d_{\ell_t}\|_\infty \E\|\mL_t^{-1}|\d_{\ell_t}|\|_{\mP_t^{(2)}} + \|\mC_t^{-1}\d_c\|_\infty \E\|\mC_t^{-1}|\d_c|\|_{\mP_t^{(2)}} \notag \\
\ls & ~ \gamma^2/\cnorm \label{eq:follow2}
\end{align}
For the second term in (\ref{eq:bound61}) we also get
\begin{align*}
& ~ \E [ \|\Tau^{-1}\Tau_t\mK_t'\Tau_t^{-1}(\mP_t \circ (\mP_t \mL_t^{-1}\Delta_{\ell_t} \mP_t))\mD_t'\mC_t^{-1}\d_c\|_\tau ] \\ 
\ls & ~ \E [ \|\Tau_t^{-1}(\mP_t \circ (\mP_t \mL_t^{-1}\Delta_{\ell_t} \mP_t))\mD_t'\mC_t^{-1}\d_c\|_\tau ] \ls \gamma^2/\cnorm
\end{align*}
exactly as from (\ref{eq:follow1}) to (\ref{eq:follow2}). For the first term in (\ref{eq:bound62}) we use Lemma \ref{lemma:projfact} and (\ref{eq:l}) to bound
\begin{align}
&\E\|\Tau^{-1}\Tau_t\mD_t'\Tau_t^{-1}(\mP_t \circ (\mP_t \mL_t^{-1}\Delta_{\ell_t} \mP_t))\mK_t'\mC_t^{-1}\d_c\|_\tau \notag \\
\ls & ~ \E\|\Tau_t^{-1}(\mP_t \circ (\mP_t \mL_t^{-1}\Delta_{\ell_t} \mP_t))\mK_t'\mC_t^{-1}\d_c\|_\tau \label{eq:follow3} \\
\ls & ~ \E\|(\mL_t^{-1}\d_{\ell_t})^2\|_{\mP_t^{(2)}} + \E\|(\mK_t'\mC_t^{-1}\d_c)^2\|_{\mP_t^{(2)}} \notag \\
\le & ~ \|\mL_t^{-1}\d_{\ell_t}\|_\infty \E \|\mL_t^{-1}\d_{\ell_t}\|_{\mP_t^{(2)}} + \E\|(\mK_t'\mC_t^{-1}\d_c)^2\|_\tau \notag \\
\le & ~ \gamma^2/\cnorm + \|\mK_t'\mC_t^{-1}\d_c\|_\infty \E\|\mK_t'\mC_t^{-1}\d_c\|_\tau \notag  \\
\le & ~ \gamma^2/\cnorm + \gamma \cdot \E\|\mC_t^{-1}|\d_c|\|_{\mP_t^{(2)}} \ls \gamma^2/\cnorm \label{eq:follow4}
\end{align}
where we have used $\|\mK_t'\|_\infty \ls 1$ and Lemma \ref{lemma:tauchange1}. For the second term in (\ref{eq:bound62}), we exactly follow (\ref{eq:follow3}) and (\ref{eq:follow4}) to bound
\begin{align*}
& ~ \E [ \|\Tau^{-1}\Tau_t\mK_t'\Tau_t^{-1}(\mP_t \circ (\mP_t \mL_t^{-1}\Delta_{\ell_t} \mP_t))\mK_t'\mC_t^{-1}\d_c\|_\tau ] \\
\ls & ~ \E [ \|\Tau_t^{-1}(\mP_t \circ (\mP_t \mL_t^{-1}\Delta_{\ell_t} \mP_t))\mK_t'\mC_t^{-1}\d_c\|_\tau ] \ls \gamma^2/\cnorm.
\end{align*}
Therefore, the total contribution from the $\tpi$-norm from the $4\Tau_t^{-\frac{1}{2} }(\mP_t \circ (\mP_t \mL_t^{-1}\Delta_{\ell_t} \mP_t))\Tau_t^{-\frac{1}{2} }$ term in (\ref{eq:bound2}) is at most $O(\gamma^2)$. Therefore, the $\frac{ \mathrm{d} }{ \mathrm{d} t } \mN_t$ term has total contribution of $O(\gamma^2)$. Combining everything, this gives our desired bound.
\end{proof}

\subsection{Initial and Final Point}
\label{subsec:initialpoint}
We will require some basic properties of self-concordant functions.
\begin{lemma}[Theorem 4.1.7, Lemma 4.2.4 in \cite{Nes98}]
\label{lemma:selfconnes}
Let $\phi$ be a $\nu$-self-concordant function on the domain $\mathcal{X}$. Then for any $x,y \in \mathcal{X}$ we have that
\begin{itemize}
\item $\phi(x)^\top(y-x) \le \nu$.
\item $(\g \phi(y) - \g \phi(x))^\top(y-x) \ge \frac{\|y-x\|_{\g^2\phi(x)}^2}{1+\|y-x\|_{\g^2\phi(x)}}$.
\end{itemize}
\end{lemma}

\finalpoint*
\begin{proof}
	Set $s^{\final}=s$ and $x^{\final}=x-\bar{\Tau}^{-1}\Phi''(x)^{-1}\mA(\mA^{\top}\bar{\Tau}^{-1}\Phi''(x)^{-1}\mA)^{-1}(b-\mA^{\top}x).$
	This obviously satisfies the first point. For the second point, there
	are two steps: We first prove the claim below, and then use the claim
	to prove the second point.
	\begin{claim}
		\label{claim:finalpoint_technical}$\|\Phi''(x)^{\frac{1}{2}}(x^{\final}-x)\|_{\infty}\ls\eps\text{ and }\left\Vert \frac{s^{\final}+\mu\tau(x)\phi'(x^{\final})}{\mu\tau(x)\sqrt{\phi''(x^{\final})}}\right\Vert _{\infty}\ls\eps.$ 
	\end{claim}
	
	\begin{proof}
		We start by using Lemma \ref{lemma:normbound} to get that 
		\begin{align*}
		\|\Phi''(x)^{\frac{1}{2}}(x^{\final}-x)\|_{\infty}
		\le & ~ \|\bar{\Tau}^{-1}\Phi''(x)^{-\frac{1}{2}}\mA(\mA^{\top}\bar{\Tau}^{-1}\Phi''(x)^{-1}\mA)^{-1}(b-\mA^{\top}x)\|_{\infty}\\
		\ls & ~ \|\mA^{\top}x-b\|_{(\mA^{\top}\Tau^{-1}\Phi''(x)^{-1}\mA)^{-1}} \\
		\le & ~ \eps\gamma/\cnorm\le\eps.
		\end{align*}
		1-self-concordance gives us that 
		\begin{align*}
		\left\Vert \Phi''(x)^{-\frac{1}{2}}\left(\phi'(x^{\final})-\phi'(x)\right)\right\Vert _{\infty}\ls\|\Phi''(x)^{-\frac{1}{2}}\Phi''(x)(x^{\final}-x)\|_{\infty}\ls\eps.
		\end{align*}
		Also, we know that $\phi''(x)\approx_{2}\phi''(x^{\final})$, hence
		\begin{align*}
		\left\Vert \frac{s^{\final}+\mu\tau(x)\phi'(x^{\final})}{\mu\tau(x)\sqrt{\phi''(x^{\final})}}\right\Vert _{\infty}\ls\left\Vert \frac{s^{\final}+\mu\tau(x)\phi'(x^{\final})}{\mu\tau(x)\sqrt{\phi''(x)}}\right\Vert _{\infty}\ls\eps.
		\end{align*}
	\end{proof}
Now, we are ready to prove the second point. Define $x^{*}\defeq\argmin_{\substack{\mA^{\top}x=b\\
		\ell_i \le x_i \le u_i \forall i
	}
}c^{\top}x$, and $x_{t}=tx^{\final}+(1-t)x^*$ for $t\in[0,1]$. Let $v=\frac{s^{\final}+\mu\tau(x)\phi'(x^{\final})}{\mu\tau(x)\sqrt{\phi''(x^{\final})}}$,
so that $\|v\|_{\infty}\ls\eps.$ We will use that $\|v\|_{\infty}\le1.$
This gives us 
\begin{align}
 c^{\top}(x^{\final}-x^{*})
= & ~ (s^{\final}-\mA y)^{\top}(x^{\final}-x^{*}) \notag \\
= & ~ (s^{\final})^{\top}(x^{\final}-x^{*})\notag \\
= & ~ \mu(\sqrt{\phi''(x^{\final})}v-\phi'(x^{\final}))^{\top}\Tau(x)(x^{\final}-x^{*})\notag \\
\le & ~ \mu\sum_{i\in [m]} \tau(x)_{i}\left|\sqrt{\phi_{i}''(x_{i}^{\final})}(x_{i}^{\final}-x_{i}^{*})\right|-\mu\phi'(x^{\final})^{\top}\Tau(x)(x^{\final}-x^{*}).\label{eq:almostdone1}
\end{align}
By Lemma \ref{lemma:selfconnes} above coordinate-wise on the $1$-self-concordant
functions $\phi_{i}$, we can bound 
\begin{align*}
& ~ \phi'(x^{\final})^{\top}\Tau(x)(x^{\final}-x^{*}) \\
= & ~ 2\phi'(x^{\final})^{\top}\Tau(x)(x^{\final}-x_{1/2})\\
= & ~  2(\phi'(x^{\final})-\phi'(x_{1/2}))^{\top}\Tau(x)(x^{\final}-x_{1/2})+2\phi'(x_{1/2})^{\top}\Tau(x)(x^{\final}-x_{1/2})\\
= & ~  2(\phi'(x^{\final})-\phi'(x_{1/2}))^{\top}\Tau(x)(x^{\final}-x_{1/2})+2\phi'(x_{1/2})^{\top}\Tau(x)(x_{1/2}-x^{*})\\
\ge & ~ 2\sum_{i \in [m]} \tau(x)_{i}\frac{\left|\sqrt{\phi_{i}''(x_{i}^{\final})}(x_{i}^{\final}-(x_{1/2})_{i})\right|^{2}}{1+\left|\sqrt{\phi_{i}''(x_{i}^{\final})}(x_{i}^{\final}-(x_{1/2})_{i})\right|} - 2\sum_{i\in [m]}\tau(x)_{i}\\
= & ~ \sum_{i \in [m]} \tau(x)_{i}\frac{\left|\sqrt{\phi_{i}''(x_{i}^{\final})}(x_{i}^{\final}-x_{i}^{*})\right|^{2}}{2+\left|\sqrt{\phi_{i}''(x_{i}^{\final})}(x_{i}^{\final}-x_{i}^{*})\right|}-2\sum_{i \in [m]}\tau(x)_{i}
\end{align*}
Applying this to the expression in (\ref{eq:almostdone1}) we get
that 
\begin{align*}
& c^{\top}(x^{\final}-x^{*})\le(\ref{eq:almostdone1})\\
& \le\mu\left(\sum_{i \in [m]}\tau(x)_{i}\left|\sqrt{\phi_{i}''(x_{i}^{\final})}(x_{i}^{\final}-x_{i}^{*})\right|-\sum_{i}\tau(x)_{i}\frac{\left|\sqrt{\phi_{i}''(x_{i}^{\final})}(x_{i}^{\final}-x_{i}^{*})\right|^{2}}{2+\left|\sqrt{\phi_{i}''(x_{i}^{\final})}(x_{i}^{\final}-x_{i}^{*})\right|}+2\sum_{i \in [m]}\tau(x)_{i}\right)\\
& =\mu\left(\sum_{i \in [m]}\tau(x)_{i}\frac{2\left|\sqrt{\phi_{i}''(x_{i}^{\final})}(x_{i}^{\final}-x_{i}^{*})\right|}{2+\left|\sqrt{\phi_{i}''(x_{i}^{\final})}(x_{i}^{\final}-x_{i}^{*})\right|}+2\sum_{i \in [m]}\tau(x)_{i}\right)\le4\mu\sum_{i \in [m]}\tau(x)_{i}\ls n\mu.
\end{align*}

\end{proof}

\subsection{Sampling Schemes}
\label{subsec:proofssampling}
\independent*
\begin{proof}
The (Expectation) condition is clear by definition, and (Covariance) follows by independence. The (Matrix approximation) condition follows by \cite[Lemma 4]{clmmps15} and that $q_i \ge \csample\log(m)\gamma^{-2}\sigma(\bar{\mA})$.

For the (Variance) condition, note that the variance is $0$ if $q_i = 1$. Otherwise, $q_i \ge \cvalid^2\gamma^{-1}|(\d_r)_i|$ so we have
\begin{align*} 
\Var[\mR_{ii}(\d_r)_i] \le \frac{(\d_r)_i^2}{q_i} \le \frac{\gamma|(\d_r)_i|}{\cvalid^2}
\end{align*}
Also, $q_i \ge \csample\sigma(\bar{\Tau}^{-\frac{1}{2} }\Phi''(\bar{x})^{-\frac{1}{2} }\mA)_i\log(m)\gamma^{-2}$ so $\E[\mR_{ii}^2] \le 2\sigma(\bar{\mA})^{-1}$.

For the (Maximum) condition, we perform cases on $q_i$. It is trivial for $q_i = 1$. Otherwise, $q_i \ge \cvalid^2\gamma^{-1}|(\d_r)_i|$, so $|q_i^{-1}(\d_r)_i| \le \frac{\gamma}{\cvalid^2}$.
\end{proof}

\prop*
\begin{proof}
The (Expectation) condition follows directly. For the (Covariance) condition, note that
\begin{align*}  
\E[\mR_{ii}\mR_{jj}] = \sum_{1 \le t_1, t_2 \le C_0S} \frac{1}{C_0^2q_iq_j}\Pr[X_{t_1} \text{ picks } i]\Pr[X_{t_2} \text{ picks } j] \le C_0^2S^2\frac{1}{C_0^2q_iq_j}\frac{q_iq_j}{S^2} = 1. 
\end{align*} 
Because variance is additive over independent samples, we use that $q_i \ge |(\d_r)_i|$ to get
\begin{align*}  
\Var[\mR_{ii}(\d_r)_i] \le \frac{(\d_r)_i^2}{\cvalid^2\gamma^{-2}q_i} \le \frac{\gamma^2|(\d_r)_i|}{\cvalid^2} \le \frac{\gamma|(\d_r)_i|}{\cvalid^2}. 
\end{align*} 
Also, $q_i \ge \sigma(\bar{\mA})_i$, so we get that
\begin{align*}  
\E[\mR_{ii}^2] \le 1+\Var[\mR_{ii}] \le 1+q_i^{-1} \le 2\sigma(\bar{\mA})^{-1}. 
\end{align*} 
To show the (Maximum) condition, we recall that $\mR = \sum_{j=1}^{C_0S} C_0^{-1}X_j$, and use Bernstein's inequality.
Note that the maximum possible value of $C_0^{-1}X_j$ is $(C_0q_i)^{-1}$, and that $\Var[\mR_{ii}] \le (C_0q_i)^{-1}$.
Therefore, by Bernstein's inequality we know that
\begin{align*}  
\Pr[|\mR_{ii}-1| \ge t] \le \exp\left(-\frac{t^2/2}{\frac{1}{C_0q_i} + \frac{t}{3C_0q_i}}\right). 
\end{align*} 
For $t = \frac{\gamma}{\cvalid^2|(\d_r)_i|}$ we have that
\begin{align*}  
\exp\left(-\frac{t^2/2}{\frac{1}{C_0q_i} + \frac{t}{3C_0q_i}}\right) \le \exp\left(-\frac{1}{\frac{(\d_r)_i^2}{50q_i\log(m)} + \frac{|(\d_r)_i| \gamma}{150\cvalid^2 q_i\log(m)}}\right). 
\end{align*}
Now, because $q_i \ge |(\d_r)_i|$, we have that 
\begin{align*}
\frac{(\d_r)_i^2}{50q_i\log(m)} \le \frac{|(\d_r)_i|}{50\log(m)} \le \frac{1}{50\log(m)} 
\end{align*} 
as $\|\d_r\|_\infty \ls \gamma$ by Lemma \ref{lemma:xchange}. From the definition of $q_i$, we also calculate that
\begin{align*}
\frac{|(\d_r)_i| \gamma}{150\cvalid^2 q_i\log(m)} \le \frac{1}{150\cvalid^2 \log(m)}.
\end{align*}
Therefore, we have that
\begin{align*}
& ~ \exp\left(-\frac{1}{\frac{(\d_r)_i^2}{50q_i\log(m)} + \frac{(\d_r)_i \gamma}{150\cnorm\cvalid q_i\log(m)}}\right) \\
\le & ~ \exp\left(\frac{1}{\frac{1}{50\log(m)} + \frac{1}{150\cvalid^2 \log(m)}} \right) \\
\le & ~ \exp(-25\log(m)) \le m^{-10}. 
\end{align*}
Finally, to show the (Matrix concentration) result, we will use the matrix Freedman Inequality \cite{Tropp11}. Let $a_i$ be the rows of the matrix $\bar{\mA}$, so we have that
\begin{align*}
(\bar{\mA}^\top\bar{\mA})^{-1/2}\bar{\mA}^\top \mR\bar{\mA}(\bar{\mA}^\top\bar{\mA})^{-1/2} = \sum_{j=1}^{C_0S} (\bar{\mA}^\top\bar{\mA})^{-1/2}C_0^{-1}q_{i_j}^{-1}a_{i_j}a_{i_j}^\top(\bar{\mA}^\top\bar{\mA})^{-1/2},
\end{align*}
where $i_j$ is the $i$-th selected nonzero entry for $\mR$. We now bound the maximum and the variance. For the maximum, we know that
\begin{align*}
\left\|(\bar{\mA}^\top\bar{\mA})^{-1/2}C_0^{-1}q_{i_j}^{-1}a_{i_j}a_{i_j}^\top(\bar{\mA}^\top\bar{\mA})^{-1/2}\right\|_2 \le C_0^{-1}q_{i_j}^{-1}\sigma(\bar{\mA})_{i_j} \le C_0^{-1}, 
\end{align*}
 as $q_i \ge \sigma(\bar{\mA})_i$ for all $i \in [m]$. For the variance, by the above calculation using $q_i \ge \sigma(\bar{\mA})_i$ for all $i$, we know that
\begin{align*}
& ~\E_{i_j}\left[\left((\bar{\mA}^\top\bar{\mA})^{-1/2}C_0^{-1}q_{i_j}^{-1}a_{i_j}a_{i_j}^\top(\bar{\mA}^\top\bar{\mA})^{-1/2}\right) \right] \\
= & ~ \sum_j C_0^{-2}q_i^{-2} \cdot q_i/S \cdot \left((\bar{\mA}^\top\bar{\mA})^{-1/2}a_ja_j^\top(\bar{\mA}^\top\bar{\mA})^{-1/2}\right)^2 \\
\pe & ~ C_0^{-2}S^{-1}\sum_j (\bar{\mA}^\top\bar{\mA})^{-1/2}a_ja_j^\top(\bar{\mA}^\top\bar{\mA})^{-1/2} = C_0^{-2}S^{-1}\mI.
\end{align*}
As there are $C_0S$ total terms, the variance is bounded by $C_0^{-1}\mI$. Hence, the matrix Freedman inequality tells us that
\begin{align*} 
\Pr\left[\bar{\mA}^\top\mR\bar{\mA} \not\approx_{\gamma} \bar{\mA}^\top\bar{\mA}\right] \le m \cdot\exp\left(-\frac{\gamma^2/2}{C_0^{-1} + \gamma C_0/3}\right) \le m \cdot \exp(-25\log(m)) \le m^{-10} 
\end{align*} 
as desired, by the choice of $C_0$.
\end{proof}

\mixture*
\begin{proof}
Note that by the choice of $\alpha = \frac{1}{4\log(4m/n)}$ we have that \[ \sigma(\bar{\Tau}^{-\frac{1}{2} }\Phi''(\bar{x})^{-\frac{1}{2} }\mA)_i \approx_1 \sigma(\bar{\Tau}^{-\frac{1}{2} -\frac{1}{1-\alpha}}\Phi''(\bar{x})^{-\frac{1}{2} }\mA)_i = \tau(\bar{x})_i, \] which gives the bound for the  part with $\tau_i$. For the other piece, note by the AM-GM inequality
\begin{align*}
C_1 \sqrt{n}(\delta_r)_i^2 + C_2/\sqrt{n} \ge 2\sqrt{C_1C_2}|\delta_r|_i \ge \cvalid^2\gamma^{-1}|\delta_r|_i.
\end{align*}
Now we bound $\sum_i p_i$. By Lemma \ref{lemma:xchange} we have that
\[ \|\d_r\|_2^2 \le \frac{m}{n}\|\d_r\|_\tau^2 \le 2m\gamma^2/n \le m/n \]
because $\gamma \le 1/2$.
Therefore, we have that
\begin{align*}
\sum_{i \in [m]} p_i &= \sum_{i \in [m]} C_1 \sqrt{n}(\delta_r)_i^2 + C_2/\sqrt{n} + C_3 \tau_i\gamma^{-2}\log m \\
&= C_1\sqrt{n}\|\delta_r\|_2^2 + C_2 \frac{m}{\sqrt{n}} + C_3 n\gamma^{-2}\log m \\
&\le (C_1+C_2)\frac{m}{\sqrt{n}} + C_3n\gamma^{-2}\log m.
\end{align*}
\end{proof}

\subsection{Additional IPM Properties}
\label{subsec:proofsadditional}

\morestablex*
\begin{proof}
Define $\hx^{(1)}= x^{(1)}$. Define the \emph{stability potential}, analogous to the centrality potential in Definition \ref{def:potential}, as
\begin{align*} 
\Psi_\stab(x, \hx) \defeq \sum_{i \in [m]} \cosh\left(\lambda_\stab \phi''(\hx^{(k)}_i)^\frac{1}{2} _i(\hx^{(k)}_i-x^{(k)}_i)\right) 
\end{align*} 
for $\lambda_\stab = C\log(mT)/\beta$ for sufficiently large constant $C$. We will choose $\hx^{(k+1)}$ using gradient descent against the potential, and will analyze the procedure using Lemma \ref{lemma:potentialhelper}. Precisely, fix $\eps_\stab = \frac{1}{C\lambda_\stab}$ for $C$ as the same constant as in Algorithm \ref{algo:lsstep}. Chosen this way, we can see that $\gamma\beta \le \eps_\stab$ because $\eps_\stab = \frac{\beta}{C^2\log(mT)}$ and $\gamma = \eps/(\lambda C) \le 1/(C^3\log m)$ by the choice of parameters in Algorithm \ref{algo:lsstep}.

Define $\d_{\hx} = \eps_\stab \g_{\hx} \Psi(x^{(k)}, \hx^{(k)})^{\flat(\tau(\hx^{(k)}))}$ and
\begin{align*} 
\hx^{(k+1)} = \hx^{(k)} - \E[x^{(k+1)} - x^{(k)}] - \Phi''(x^{(k)})^{-\frac{1}{2} } \d_{\hx} = \hx^{(k)} - \E[\bar{\d}_x] - \Phi''(x^{(k)})^{-\frac{1}{2} } \d_{\hx} 
\end{align*}
where $\bar{\d}_x$ is defined in Algorithm \ref{algo:lsstep} line \ref{line:ipm:delta_x}. We will now verify the conditions of Lemma \ref{lemma:potentialhelper} and apply it. Here, we will choose $y = \hx^{(k)}-x^{(k)}$ and $u^{(1)} = \Phi''(x^{(k)})^\frac{1}{2} $. In the notation of Lemma \ref{lemma:potentialhelper}, for simplicity we will just write $c = u^{(1)}$, and $\d_c = \d^{(1)},$ as in the notation of Section \ref{subsec:mainbounds}. By Lemma \ref{lemma:pchange} we know that $\|\mC^{-1}\d_c\|_\infty \le 2\gamma \le 1/100$.

Throughout the proof, we will use that $\tau(\hx^{(k)}) \approx_{O(\beta)} \tau(x^{(k)})$ because $\Phi''(\hx^{(k)}) \approx_{O(\beta)} \Phi''(x^{(k)})$ by induction and Lemma \ref{lemma:lewisapprox}. In particular, for any vector $h$ we have that $\|h\|_\tkpi \approx_{0.1} \|h\|_\tpi$.

By induction, we know that $\|\Phi''(x^{(k)})^\frac{1}{2} (\hx^{(k)}-x^{(k)})\|_\infty \le \beta/2 \le 1/50$, and thus
\begin{align*} 
\|\mC^{-1}\d_c\|_\infty\|\Phi''(x^{(k)})^\frac{1}{2} (\hx^{(k)}-x^{(k)})\|_\infty \le \beta\gamma \le \frac{1}{100\lambda_\stab} 
\end{align*} 
by the choice of $\gamma$. Also, we know that
\begin{align*} 
\eta = \hx^{(k+1)}-\hx^{(k)}-(x^{(k+1)}-x^{(k)})-\Phi''(x^{(k)})^{-\frac{1}{2} }\d_{\hx} = \Phi''(x^{(k)})^{-\frac{1}{2} }\left(\mR\d_r - \d_r - \d_{\hx}\right). 
\end{align*} 
Therefore, by the (Maximum) condition of Definition \ref{def:validdistro} and $\|\d_{\hx}\|_{\tkpi} \le \eps_\stab$ we know that
\begin{align*} 
\|\mW\eta\|_\infty \le \|\Phi''(x^{(k)})^\frac{1}{2} (\hx^{(k+1)}-\hx^{(k)})\|_\infty \le \frac{\gamma}{\cvalid^2} + \eps_\stab \le \gamma\beta + \eps_\stab \le \frac{1}{100\lambda_\stab}, 
\end{align*} so all the conditions of Lemma \ref{lemma:potentialhelper} are satisfied. Now we bound the terms of \eqref{eq:firstline}, \eqref{eq:secondline}, \eqref{eq:thirdline} in expectation over $\mR$. To bound \eqref{eq:firstline}, note that
\begin{align*} 
\E[\psi'(v)^\top \mW\eta] = \psi'(v)^\top \d_{\hx} = -\eps_\stab \|\psi'(v)\|_{\tkpi}^* 
\end{align*}
and by Lemma \ref{lemma:pchange}
\begin{align*}
\E[\psi'(v)^\top \mV \mC^{-1}\d_c] &\le \|v\|_\infty \|\psi'(v)\|_{\tkpi}^* \|\mC^{-1}\E[\d_c]\|_{\tkpi} \\ &\le \beta\gamma \|\psi'(v)\|_{\tkpi}^* \le \frac14\eps_\stab\|\psi'(v)\|_{\tkpi}^*.
\end{align*}
For \eqref{eq:secondline} we have that
\begin{align*}
\E[8\|\mW\eta\|_{\psi''(v)}^2] 
\ls & ~ \|\Var(\d_r)^{1/2}\|_{\psi''(v)}^2 + \|\d_{\hx}\|_{\psi''(v)}^2 \\ 
\le & ~ \gamma\beta^2\|\d_r\|_\tkpi\|\psi''(v)\|_\tkpi^* + \|\d_{\hx}^2\|_\tkpi \|\psi''(v)\|_\tkpi^* \\
\ls & ~ (\gamma^2\beta^2 + \eps_\stab^2)\|\psi''(v)\|_\tkpi^* \\
\ls & ~ \eps_\stab^2\|\psi''(v)\|_\tkpi^*.
\end{align*}
where the first step is via the triangle inequality and the definition of $\eta$, the second step is by the (Variance) condition of Definition \ref{def:validdistro} for $\cvalid \ge \beta^{-2}$ and the definition of $\|\cdot\|_\tpi^*$, the third step is from Lemma \ref{cor:xchange} Part 2 and $\|\d_{\hx}\|_\tkpi \le \eps_\stab$ by Definition, and the final step is by $\gamma\beta \le \eps_\stab$.

Also, we can bound that
\begin{align*}
& ~ \E[8(1+\|c\|_1)\|v\|_\infty^2\sum_{j\in[k]}|c_j|\|(\mU^{(j)})^{-1}\d^{(j)}\|_{\psi''(v)}^2] \\
\ls & ~ \beta^2 \E[\|\mC^{-1}\d_c\|_{\psi''(v)}^2] \\
= & ~ \beta^2 \|\E[\mC^{-2}\d_c^2]^{1/2}\|_{\psi''(v)}^2 \\
\le & ~ \beta^2 \|\E[\mC^{-2}\d_c^2]\|_\tkpi \|\psi''(v)\|_\tkpi^* \\
\ls & ~ \beta^2\gamma^2 \|\psi''(v)\|_\tkpi^* \\
\le & ~ \eps_\stab^2 \|\psi''(v)\|_\tkpi^*.
\end{align*}
where the first step follows from $\|v\|_\infty \le \beta$ by induction, $|c_j| = O(1)$ for all $j$, and the definition of $\d^{(j)}$, the third step follows from the definition of the dual norm $\| \cdot \|_\tpi^*$, the fourth step follows from Lemma \ref{lemma:pchange} Part 3, and the final step follows from $\beta\gamma \le \eps_\stab$.

For \eqref{eq:thirdline} we can bound using the Cauchy-Schwarz inequality and the above computations that
\begin{align}
& ~ \E[8\|\mw\eta\|_{|\psi'(v)|}\sum_{j\in [k]} |c_j|\|(\mU^{(j)})^{-1}\d^{(j)}\|_{|\psi'(v)|}] \notag\\ 
\ls & ~ \E[\|\mw\eta\|_{|\psi'(v)|}^2]^{1/2}\E\left[\left(\sum_{j \in [k]} |c_j|\|(\mU^{(j)})^{-1}\d^{(j)}\|_{|\psi'(v)|}\right)^2\right]^{1/2} \notag \\
\ls & ~ (\eps_\stab^2 \beta\gamma^2)^{1/2} \|\psi'(v)\|_\tkpi^*. \label{eq:epsstab1}
\end{align}
Also, we have that
\begin{align}
& ~\E\left[8(1+\|c\|_1)\|v\|_\infty\sum_{j\in[k]}|c_j|\|(\mU^{(j)})^{-1}\d^{(j)}\|_{|\psi'(v)|}^2\right] \notag \\
\ls & ~ \beta\|\E[\mC^{-2}\d_c^2]^{1/2}\|_|\psi'(v)|^2 \notag \\
\le & ~ \beta \|\E[\mC^{-2}\d_c^2]\|_\tkpi \|\psi'(v)\|_\tkpi^* \notag \\
\ls & ~ \beta\gamma^2 \|\psi'(v)\|_\tkpi^*. \label{eq:epsstab2}
\end{align}
For sufficiently small choice of $\gamma$, we have that even with the suppressed constants in \eqref{eq:epsstab1}, \eqref{eq:epsstab2} that
\begin{align*} 
O\left(\beta\gamma^2 + (\eps_\stab^2 \beta\gamma^2)^{1/2}\right) \le \frac14\eps_\stab. 
\end{align*} 
Therefore combining everything, we get that
\begin{align*} 
\E[\Psi(x^{(k+1)}, \hx^{(k+1)})] \le \Psi(x^{(k)}, \hx^{(k)}) - \frac{1}{4}\eps_\stab\|\psi'(v)\|_\tkpi^* + O(\eps_\stab^2)\|\psi''(v)\|_\tkpi^*. 
\end{align*}
As in the proof of Lemma \ref{cor:finaldrop}, by \cite[Lemma 4.36]{BrandLN+20}, and the fact that $\|1\|_\tkpi \le 4\cnorm\sqrt{n}$, we get that
\begin{align*}
\E[\Psi(x^{(k+1)}, \hx^{(k+1)})] &\le \left(1 - \frac{\lambda_\stab\eps_\stab}{4\cnorm\sqrt{n}}\right)\Psi(x^{(k)}, \hx^{(k)}) + m \\ &\le \left(1 - \frac{1}{4\cnorm C\sqrt{n}}\right)\Psi(x^{(k)}, \hx^{(k)}) + m.
\end{align*}
As $\Psi(x^{(1)}, \hx^{(1)}) = m$, we have by induction that $\E[\Psi(x^{(k)}, \hx^{(k)})] \le 4\cnorm Cm\sqrt{n} \le m^2$ by induction for all $k$. Therefore, with probability $1-m^{-12}$ we have that $\Psi(x^{(k)}, \hx^{(k)}) \le m^{14}$ for all $k$. By the choice of $\lambda_\stab$ this implies that $\|\Phi''(x^{(k)})^\frac{1}{2} (\hx^{(k)}-x^{(k)})\|_\infty \le \beta/2$ as desired.

To finish we must verify the other two conditions. The second item follows from self-concordance and the first. To check the third condition, we have that
\begin{align*}
\|\Phi''(x^{(k)})^\frac{1}{2} (\hx^{(k+1)}-\hx^{(k)})\|_\tkpi = \left\|\Phi''(x^{(k)})^\frac{1}{2} \E[\bar{\d}_x] - \d_{\hx}\right\|_\tkpi \le 1.1\gamma + \eps_\stab \le \gamma
\end{align*}
where the first inequality follows by the triangle inequality and Lemma \ref{lemma:xchange} Part 1.
as $\eps \le \frac{\beta}{C\log(mT)} \le 0.1\gamma$ as $\beta \le \gamma$.
\end{proof}

\morestablerest*
\begin{proof}
The first two items directly follow from $1$-self-concordance, specifically that
\begin{align*} 
|\phi''(\hx^{(k)})^{-\frac{1}{2} } - \phi''(x^{(k)})^{-\frac{1}{2} }| \le |\hx^{(k)} - x^{(k)}| 
\end{align*}
and Lemma \ref{lemma:morestablex}.

For the third item, define $p_0 = \phi''(x^{(k)})^{-\frac{1}{2} }$, $p_1 = \phi''(x^{(k+1)})^{-\frac{1}{2} }$, and $\d_p = p_1 - p_0$. Define $p_t = p_0 + t \cdot \d_p$ and $\tau_t = w(p_t).$ Note that $\tau_t \approx_{0.04} \tau_0 = \tau(\hx^{(k)})$ for all $t \in [0,1]$ as $p_t \approx_{0.01} p_0$ by the second item of this lemma (Lemma \ref{lemma:morestablerest}), and Lemma \ref{lemma:lewisapprox}. Therefore we can compute
\begin{align*}
\|\Tau(\hx^{(k)})^{-1}(\tau(\hx^{(k+1)}) - \tau(\hx^{(k)}))\|_{\tau(\hx^{(k)})} 
= & ~ \left\|\int_0^1 \Tau(\hx^{(k)})^{-1}\mJ_{w(p_t)}\d_p \right\|_\tkpi \\
\ls & ~ \int_0^1 \left\|\Tau_t^{-1}\mJ_{w(p_t)}\d_p \right\|_{\tau_t+\infty} \\
\ls & ~ \|\mP_t^{-1}\d_p\|_{\tau_t+\infty} \\
\ls & ~ \|\mP^{-1}\d_p\|_\tkpi \ls \gamma
\end{align*}
where we have used Lemma \ref{lemma:matrixbound}.

For the fourth item, we once again use that $\tau_1 \approx_{0.04} \tau_0$ and $p_t \approx_{0.01} p_0$. This gives us
\begin{align*}
& ~ \|(\hat{\mW}^{(k)})^{-1}(\hw^{(k+1)} - \hw^{(k)})\|_\tkpi \\
= & ~ \|\Phi''(\hx^{(k)})^\frac{1}{2} \Tau(\hx^{(k)})^{\frac1p-\frac{1}{2} }\left(\phi''(\hx^{(k+1)})^{-\frac{1}{2} }\tau(\hx^{(k+1)})^{\frac{1}{2} -\frac1p} - \phi''(\hx^{(k)})^{-\frac{1}{2} }\tau(\hx^{(k)})^{\frac{1}{2} -\frac1p}\right)\|_\tkpi \\
\le & ~ \|\Phi''(\hx^{(k)})^\frac{1}{2} \Tau(\hx^{(k)})^{\frac1p-\frac{1}{2} }\left(\phi''(\hx^{(k+1)})^{-\frac{1}{2} }-\phi''(\hx^{(k)})^{-\frac{1}{2} }\right)\tau(\hx^{(k+1)})^{\frac{1}{2} -\frac1p}\|_\tkpi \\
+ & ~ \|\Phi''(\hx^{(k)})^\frac{1}{2} \Tau(\hx^{(k)})^{\frac1p-\frac{1}{2} }\phi''(\hx^{(k)})^{-\frac{1}{2} }\left(\tau(\hx^{(k+1)})^{\frac{1}{2} -\frac1p} - \tau(\hx^{(k)})^{\frac{1}{2} -\frac1p}\right)\|_\tkpi \\
\ls & ~ \|\Phi''(\hx^{(k)})^\frac{1}{2} \left(\phi''(\hx^{(k+1)})^{-\frac{1}{2} }-\phi''(\hx^{(k)})^{-\frac{1}{2} }\right)\|_\tkpi \\
+ & ~ \|\Tau(\hx^{(k)})^{\frac1p-\frac{1}{2} }\left(\tau(\hx^{(k+1)})^{\frac{1}{2} -\frac1p} - \tau(\hx^{(k)})^{\frac{1}{2} -\frac1p}\right)\|_\tkpi \\
\ls & ~ \gamma + \|(\tau(\hx^{(k+1)})/\tau(\hx^{(k)}))^{\frac{1}{2} -\frac1p}-1\|_\tkpi \\
\le & ~ \gamma + \|(\tau(\hx^{(k+1)})/\tau(\hx^{(k)}))-1\|_\tkpi \\
= & ~ \gamma + \|\Tau(\hx^{(k)})^{-1}(\tau(\hx^{(k+1)}) - \tau(\hx^{(k)}))\|_\tkpi \ls \gamma.
\end{align*}
Here we have used that $|x^{\frac{1}{2} -\frac1p}-1| \le |x-1|$ for all $x \in [0.9, 1.1]$, and items two and three of this lemma (Lemma \ref{lemma:morestablerest}).
\end{proof}

\paramchange*
\begin{proof}
Note that the smallest possible entry of $W$ and $W'$ is at least $\min_i (u_i-\ell_i)^{-1}$, as $\phi_i''(x) \ge \frac{1}{(u_i-\ell_i)^2}$ for all $x \in (\ell_i, u_i)$. By Lemma \ref{lemma:finalpoint} and Claim \ref{claim:finalpoint_technical}, for any $\eps$-centered $(x, s, \mu)$ encountered in the algorithm, we can find a point $x^\final$ such that $(x^\final, s, \mu)$ is also $\eps$-centered, $\Phi''(x^\final) \approx_1 \Phi''(x)$, and $\mA^\top x^\final = b$, i.e. $x^\final$ is exactly feasible. Thus, it suffices to control the largest entry of $\Phi''(x^\final)$ to bound $\log W'$.

Let $v = \frac{s+\mu\tau \phi'(x^\final)}{\mu\tau \phi''(x^\final)^{1/2}}$, so that $\|v\|_\infty \ls \eps$. We will use that $\|v\|_\infty \le 1/100$ say. Also, we know that $s = \mA z + c$ for some $z \in \R^n$. This gives us that
\[ \mA z + c + \mu\tau\phi'(x^\final) - v\mu\tau\phi''(x^\final)^{1/2} = 0. \] Computing an inner product with $(x^\final - x^\init)$ gives that
\begin{align}
0 &= (x^\final - x^\init)^\top \mA z + c^\top (x^\final - x^\init) \\ &+ \mu\sum_i \tau_i\left(\phi_i'(x^\final_i) - v_i\phi''_i(x^\final_i)^{1/2}\right)(x^\final_i - x^\init_i) \nonumber \\
&= c^\top (x^\final - x^\init) + \mu\sum_i \tau_i\left(\phi_i'(x^\final_i) - v_i\phi''_i(x^\final_i)^{1/2}\right)(x^\final_i - x^\init_i) \label{eq:polyeq}.
\end{align}
We claim that $\left(\phi_i'(x^\final_i) - v_i\phi''_i(x^\final_i)^{1/2}\right)(x^\final_i - x^\init_i) \ge -1$ for all $i$. To show this, we without loss of genearlity assume that $x^\final_i \ge x^\init_i$. Note that $\phi_i'(x^\final_i)(x^\final_i - x^\init_i)$ by our choice of
$\phi_i(x) = -\log(u_i-x)-\log(x-\ell_i)$ and $x^\init = (\ell+u)/2$. 
If $x^\final_i \ge (x^\init_i + u_i)/2$, then note that
\[ \phi_i'(x^\final_i) - v_i\phi''_i(x^\final_i)^{1/2} \ge \frac{1}{2(u_i-x^\final_i)} - v_i\sqrt{\frac{2}{(u_i-x^\final_i)^2}} \ge \frac{1}{4(u_i-x^\final_i)} > 0, \] so the claim is trivially true. So the remaining case if $x^\init_i \le x^\final_i \le (x^\init_i + u_i)/2$.
In this case, we have that
\begin{align*}
&\left(\phi_i'(x^\final_i) - v_i\phi''_i(x^\final_i)^{1/2}\right)(x^\final_i - x^\init_i) \ge -\frac{1}{100}|\phi''_i(x^\final_i)^{1/2}(x^\final_i - x^\init_i)| \\ &\ge -\frac{1}{10}\sqrt{\frac{1}{(u_i-\ell_i)^2}} (u_i-\ell_i) \ge -1.
\end{align*}
Going back to \eqref{eq:polyeq}, we have for all $j \in [m]$ that
\begin{align}
&\mu\tau_j\left(\phi_j'(x^\final_j) - v_j\phi''_j(x^\final_j)^{1/2}\right)(x^\final_j - x^\init_j) \label{eq:boundit} \\
&= -c^\top (x^\final - x^\init) - \mu\sum_{i \in [m]\setminus \{i\}} \tau_i\left(\phi_i'(x^\final_i) - v_i\phi''_i(x^\final_i)^{1/2}\right)(x^\final_i - x^\init_i) \nonumber \\
&\le m\|c\|_\infty \|u-\ell\|_\infty + \mu \sum_{j} \tau_j \le m\|c\|_\infty \|u-\ell\|_\infty + \mu n. \nonumber
\end{align}
We now will bound $\phi''_j(x^\final_j)$. We will assume without loss of generality that $x^\final_j \ge x^\init_j$. If $x^\final_j \le (x^\init_j + u_j)/2$, then we know that $\phi''_j(x^\final_j) \ls (u_i-\ell_i)^{-2}$, as desired. Otherwise, the above arguments give that
\[ \left(\phi_j'(x^\final_j) - v_j\phi''_j(x^\final_j)^{1/2}\right) \ge \frac{1}{4(u_j-x^\final_j)}. \]
Therefore, we get that
\begin{align*}
(u_j - x^\final_j)^{-1} &\ls \mu^{-1}\tau_j^{-1}(x^\final_j-x^\init_j)^{-1}(m\|c\|_\infty \|u-\ell\|_\infty + \mu n) \\
&\le \mu^{-1}\frac{m}{n}(u_j-\ell_j)^{-1}(m\|c\|_\infty \|u-\ell\|_\infty + \mu n).
\end{align*}
Using that $\phi''(x^\final_j) \ls (u_j-x^\final_j)^{-2}$ completes the proof.
\end{proof}     %
\section{Matrix Data Structures}
\label{sec:matrix_data_structures}

In this section we provide the existence
of the required data structure for solving general linear programs.
We start by citing the \textsc{HeavyHitter}- (\Cref{def:heavyhitter})
and \textsc{InverseMaintenance}-data structure (\Cref{def:inverse_maintenance})
proven in \cite{blss20}.
We then prove the existence of a 
\textsc{HeavySampler}-data structure (\Cref{def:heavysampler}).

\begin{lemma}[{\cite[Section 6.1]{blss20}}]
\label{lem:lp:heavyhitter}
There exists a $(P,c,Q)$-\textsc{HeavyHitter} data structure (\Cref{def:heavyhitter})
with $P = \tilde{O}(\nnz(\mA))$, $c_i = \tilde{O}(n)$ for all $i \in [m]$, 
and $Q = \tilde{O}(n)$.
\end{lemma}

\begin{lemma}[{\cite[Theorem 9]{blss20}}]
\label{lem:lp:inversemaintenance}
There exists a $(P,c,Q)$-\textsc{InverseMaintenance}
data structure (\Cref{def:inverse_maintenance}) with $P = \tilde{O}(\nnz(\mA)+n^\omega)$, $c_i = O(1)$ for all $i \in [m]$,
and $Q = \tilde{O}(n^2 + n ^{\omega-1/2})$.
\end{lemma}

We now construct a data structure for the \textsc{HeavySampler}-problem (\Cref{def:heavysampler}). We first provide the data structure \Cref{alg:LP_ds} with guarantees given by the following, \Cref{lem:matrix_heavy_hitter}. 
We then show in \Cref{lem:lp:heavysampler} how the data structure of \Cref{lem:matrix_heavy_hitter} can be used to solve the \textsc{HeavySampler}-problem.

\begin{lemma}\label{lem:matrix_heavy_hitter}
There exists a data structure (Algorithm~\ref{alg:LP_ds}) that supports the following operations.
\begin{itemize}
    \item $\textsc{Initialize}(\mA \in \R^{m \times n},v \in \R_{\ge0}^m, g \in \R^m )$: The data structure is given a matrix $\mA \in \R^{m \times n}$, additive vector $v \in \R^m_{\ge0}$,
    and a scaling vector $g \in \R^m$. 
    It initializes in $\widetilde{O}(\nnz(\mA))$ time.
    \item $\textsc{Scale}(i \in [m], s \geq 0, b \geq 0)$: 
    Update $g_i \leftarrow s$ and $v_i \leftarrow b$ in $\widetilde{O}(\nnz(a_i))$ time.
    \item $\textsc{Sample}(h \in \R^n, U \in \R_{>0})$: 
    	If $U \ge e^4 \|\mG \mA h\|_2^2$ then, with high probability, in $\widetilde{O}(n)$ time, 
    	the data structure returns a random $i \in [m]$ 
    	with $\P[i = j] = (\mG\mA h)_j^2/ U$, and with probability $1-\|\mG\mA h\|_2^2/U$, returns nothing.
    	Each random index returned by a call to \textsc{Sample} is independent from the previous call.
\end{itemize}
\end{lemma}

\begin{algorithm2e}[ht!]
\caption{Sampler data structure.}\label{alg:LP_ds}
\SetKwProg{Proc}{procedure}{}{}
\SetKwProg{Member}{members}{}{}
\Comment{$[\mM]^{l,j}$ denotes the matrix formed by rows $(j-1)m/2^l+1,\cdots,jm/2^l$ of $\mM$.}
\Member{}{
	$\mA \in \R^{m \times n}$, $g \in \R^m$ \Comment{Assume $m$ is a power of 2 for code simplicity} \\
	$c = 1/\log(4m)$, $k = O(c^{-2} \log m)$ \\
	\For{$l \in \{ 0, 1, \cdots, \log m \}$}{
		\For{ $j \in [2^l]$}{
			$\mJ^{l,j} \in \R^{k \times m/2^l}$ 
			\Comment{$\mJ^{l,j} = \mathrm{JL}(c, m/2^l)$ in \Cref{lem:JL_lemma}} \\
			$\mQ^{l,j} \in \R^{k \times d}$ \Comment{$\mQ^{l,j} = \mJ^{l,j} \cdot [\mG \mA]^{l,j}$}
		}
	}
}

\Proc{\textsc{Initialize}$(\mA \in \R^{m \times n}, g \in \R^m)$}{
	$\mA \leftarrow \mA$, $g \leftarrow g$, $v \leftarrow v$ \\
	\For{$l \in \{ 0, 1, \cdots, \log m \}$}{
		\For{ $j \in [2^l]$}{
			$\mJ^{l,j} \leftarrow  \mathrm{JL}(c, m/2^l)$ \Comment{See Lemma \ref{lem:JL_lemma}} \\
			$\mQ^{l,j} \leftarrow \mJ^{l,j} \cdot [\mG \mA]^{l,j}$ \\
		}
	}
}

\Proc{\textsc{Scale}$(i \in [m], s \geq 0)$}{
	\For{$l \in \{ 0, 1, \cdots, \log m \}$}{
		$j = \lceil i \cdot 2^{l} / m \rceil$ \\
		$\mQ^{l,j} \leftarrow \mQ^{l,j} + (s - g_i) \mJ^{l,j} \cdot [ \mathbf{1}_{i}  \mathbf{1}_i^{\top} \mA ]^{l,j}$ \label{line:LP_ds_rescale} \\
	}
	$g_i \leftarrow s$
}

\Proc{\textsc{Sample}$(h \in \R^n, U \in \R_{>0})$}{
	\Return $\textsc{SampleInternal}(h, 1, 1, 0, U)$ \label{line:lp:sample}
}

\Proc{\textsc{SampleInternal}$(h \in \R^n, Z \geq 0, j, \ell, U \in \R_{>0})$}{
	\eIf{$\ell = \log m$}{
		With probability $Z^{-1} \cdot \frac{( \mG \mA \cdot h )_j^2}{U}$ \Return $j$,
		otherwise \Return $\emptyset$ \label{line:lp:reject}
	}{
		$r_1 \leftarrow \|\mQ^{l+1,2j-1} h\|_2^2$, $r_2 \leftarrow \|\mQ^{l+1,2j} h\|_2^2$ \label{line:lp:norm}\\
		With probability $r_1 / (r_1 + r_2)$, \\
		\qquad \Return $\textsc{SampleInternal}(h, Z r_1 / (r_1 + r_2), 2j-1, \ell+1)$ \\
		otherwise \\
		\qquad \Return $\textsc{SampleInternal}(h, Z r_2 / (r_1 + r_2), 2j, \ell+1)$
	}
}
\end{algorithm2e}

First, we state a lemma used in the proof of Lemma \ref{lem:matrix_heavy_hitter}.
\begin{lemma}[Johnson--Lindenstrauss (JL) \cite{jl84}]
\label{lem:JL_lemma} 
There exists a function $\mathrm{JL}(\epsilon,m)$
that given $\epsilon>0$ returns a matrix $\mj\in\R^{k\times m}$
with $k = O(\epsilon^{-2}\log m)$ in $O(km)$ time. For any $v \in \R^{m}$
this matrix $\mj$ satisfies with high probability in $m$ that $\|\mj v\|_{2}\approx_{\epsilon}\|v\|_{2}$.
\end{lemma}
Using \Cref{lem:JL_lemma} allows us to give an algorithm for sampling coordinates proportional to their $\ell_2$ weight (\Cref{lem:matrix_heavy_hitter_internal}). First, we show how to use this to prove \Cref{lem:matrix_heavy_hitter} by analyzing the runtime costs and showing how to apply a \textsc{Scale} operation. Then we show \Cref{lem:matrix_heavy_hitter_internal}, where the high-level approach is to build a binary tree and walk down the tree using the JL lemma to sample a node proportional to the $\ell_2$ norm of its coordinates.

\begin{lemma}\label{lem:matrix_heavy_hitter_internal}
A call to \textsc{SampleInternal}$(h, 1, 1, 0, U)$ for $e^4 \|\mG\mA h\|_2^2 \le U$ randomly returns a single index $i \in [m]$, or no index at all.
Index $i \in [m]$ is returned with probability 
$$p_i = \frac{(\mG \mA h)_i^2}{U}.$$
\end{lemma}

\begin{proof}[Proof of Lemma \ref{lem:matrix_heavy_hitter}]
  The implementation of our data structure is given in Algorithm~\ref{alg:LP_ds}. We now analyze the correctness and efficiency of the operations.
  
{\bf \textsc{Initialize}.} We bound the running time of \textsc{Initialize}. For each $l \in {0,1,\cdots\log m}$ and $j \in [2^l]$ we instantiate a matrix $\mJ^{l,j}=\mathrm{JL}(\epsilon,m/2^{l})$ with $\epsilon = 1/\log{m}$, so each matrix $\mJ^{l,j}$ has $k = O(c^{-2} \log m)$ rows. Creating a $k\times m/2^{l}$ sized JL matrix takes $O(k m / 2^{l})$ time (Lemma~\ref{lem:JL_lemma}), so in total creating the JL matrices takes
\begin{align*}
    \sum_{l=0}^{\log{m}} 2^{l} \cdot O(k m / 2^l) = O(km \log^3 m) = \widetilde{O}(m).
\end{align*}

We use $[\mG \mA]^{l,j}$ to denote the submatrix of $\mG\mA$ consisting of the $\{(j-1)m/2^{l} + 1, \cdots, j \cdot m/ 2^{l}\}$-th rows of $\mG\mA$.
We also use $V^{l,j}$ to denote the sum of $v_i$ for $i \in \{(j-1)m/2^{l} + 1, \cdots, j \cdot m/ 2^{l}\}$
which can be computed in $\tilde{O}(m)$ time.
 For any fixed $l$, computing the matrices $\mQ^{l,j} = \mJ^{l,j} \cdot [\mG \mA]^{l,j} \in \R^{k \times n}$ 
 for all $j \in [m/2^l]$ takes $O(k \cdot \nnz(\mA))$ time since $\mJ^{l,j} \in \R^{k \times m/2^{l}}$ and the $[\mG \mA]^{l,j} \in \R^{m/2^{l} \times n}$ are just a decomposition of the matrix $\mG\mA$. 
 Since we create $\mQ^{l,j}$ for every $l \in \{0,\cdots,\log m\}$ and $j \in [2^{l}]$, in total constructing all the $\mQ$ matrices takes time
\begin{align*}
    \sum_{l =0}^{\log{m}} O(k \cdot \nnz(\mA)) = O(k \nnz(\mA) \cdot \log^3 m ) = \widetilde{O}(\nnz(\mA)).
\end{align*}
Thus initialization can be done in $\widetilde{O}(\nnz(\mA))$ time.

{\bf \textsc{Scale}.} We first prove $\mQ^{l,j}=\mJ^{l,j} \cdot [\mG \mA]^{l,j}$ is still satisfied after updating $g_i$ to be $s$. For any $l \in \{0,\cdots,\log m\}$, we only need to update the $\mQ^{l,j}$ with $j = \lceil \frac{i}{m / 2^{l}} \rceil$ (since $[\mG \mA]^{l,j}$ has rows of $\mG\mA$ in set $\{(j-1)m/2^{l} + 1, \cdots, j \cdot m/ 2^{l}\}$. We have
\begin{align*}
    \mQ^{l,j} = &~ \mJ^{l,j}\cdot [\mG \mA]^{l,j} + (s-g_i) \mJ^{l,j} \cdot [\mathbf{1}_i \mathbf{1}_i^{\top} \mA]^{l,j} \\
    = &~ \mJ^{l,j} \cdot [(\mG + (s-g_i)\mathbf{1}_i \mathbf{1}_i^{\top}) \mA]^{l,j},
\end{align*}
where the first step follows from how we update $\mQ^{l,j}$ (Line~\ref{line:LP_ds_rescale} of Algorithm~\ref{alg:LP_ds}), the second step follows from merging terms.
And $(\mG + (s-g_i)\mathbf{1}_i \mathbf{1}_i^{\top})$ is indeed the updated scaling vector whose $i$-th coordinate is $s$.
Note that we also update $V^{l,j}$ to be the partial sum of the $v_i$ for $i \in \{(j-1)m/2^{l} + 1, \cdots, j \cdot m/ 2^{l}\}$.

Next we bound the running time of \textsc{Scale}. We need to compute $\mJ^{l,j}\cdot [\mathbf{1}_i \mathbf{1}_i^{\top} \mA]^{l,j}$ for $l \in \{0,\cdots,\log m\}$ and one $j$ that depends on $l,i$. Since $\mJ^{l,j} \in \R^{k \times m/2^{l}}$ and $[\mathbf{1}_i \mathbf{1}_i^{\top} \mA]^{l,j} \in \R^{m/2^{l} \times n}$ only has one non-zero row, consisting of $\nnz(a_i)$ non-zero entries, 
this multiplication takes $O(k\nnz(a_i))$ time. Thus in total computing the multiplication for all $l$ takes time
\begin{align*}
    O(\log n) \cdot O(k\nnz(a_i)) = \widetilde{O}(\nnz(a_i)),
\end{align*}
which follows from $k = O(c^{-2} \log m)$.

{\bf \textsc{Sample}.}
By \Cref{line:lp:sample} and \Cref{lem:matrix_heavy_hitter_internal} 
we return a randomly sampled index $j$ according to the distribution
$\P[j = i] = (\mG \mA h)_i^2 / U$ for all $i \in [m]$.

A call to \textsc{SampleInternal} requires $\tilde{O}(n)$ time,
because we compute $O(\log m)$ norms which each require a matrix-vector-product 
with an $\tilde{O}(1)\times n$ matrix in \Cref{line:lp:norm}.

\end{proof}

\begin{proof}[Proof of \Cref{lem:matrix_heavy_hitter_internal}]
Consider the binarytree where nodes are labeled by sub-intervalls of $[m]$.
The root is labeled by $[m]$. For each node labeled by some $[L,R]$, 
the left child is labeled by $[L,\lfloor(L+R)/2\rfloor]$ 
and the right node is labeled by $[\lceil(L+R+1)/2\rceil,R]$.
We also identify the $j$-th node on level $l$ via the tuple $(l,j)$.

Note that the execution of \textsc{SampleInternal} can be seen
as a path from the root of this tree to one of its leaves.
For each node $(i,j)$ it picks the left child with probability
$r_1/(r_1+r_2)$ and otherwise the right child.

We arrive at the index $j \in [m]$ on level $\ell = \log(m)$ (i.e. a leaf) with probability as follows, 
note that we are basically looking at the binary representation of $j$ 
and multiplying together the probabilities that the half containing $j$ is selected in each level.
\begin{align*}
p & =
\prod_{l=1}^{\log m}
	\frac{\| \mQ^{l,\lceil j/(m/2^{l})\rceil}\cdot h\|_{2}^{2}}
	{	\| \mQ^{l,2\lceil j/(m/2^{l-1})\rceil}\cdot h\|_{2}^{2}
		+\| \mQ^{l,2\lceil j/(m/2^{l-1})\rceil-1}\cdot h\|_{2}^{2}}\\
 & =\prod_{l=1}^{\log m}
 e^{\pm 4c}\frac{\|[ \mG \mA h]^{l,\lceil j/(m/2^{l})\rceil}\|_{2}^{2}}
 	{	\|[ \mG \mA h]^{l,2\lceil j/(m/2^{l-1})\rceil}\|_{2}^{2} 
 		+\|[ \mG \mA h]^{l,2\lceil j/(m/2^{l-1})\rceil-1}\|_{2}^{2} }\\
 & =e^{\pm 4c\cdot\log m}\prod_{l=1}^{\log m}
 	\frac{\|[ \mG \mA h]^{l,\lceil j/(m/2^{l})\rceil}\|_{2}^{2}}
 	{\|[ \mG \mA h]^{l-1,\lceil j/(m/2^{l-1})\rceil}\|_{2}^{2}}\\
 & =e^{\pm 4}\cdot\frac{\|[ \mG \mA h]^{\log m,j}\|_{2}^{2}}
 	{\|[ \mG \mA h]^{0,1}\|_{2}^{2} + V^{0,1}}
   =e^{\pm 4}\frac{( \mG \mA h)_{j}^{2} + v_j}{\| \mG \mA h\|_{2}^{2}}
\end{align*}
where the second step follows from the JL matrices guarantee that $\forall l$, $\forall j$, with high probability 
\begin{align*}
    e^{-2c} \| [ \mG \mA]^{l,j} h\|_2^2 \leq \| \mQ^{l,j} \cdot h\|_2^2 = \| \mJ^{l,j}\cdot [ \mG \mA ]^{l,j} h\|_2^2 \leq e^{2c} \|[\mG \mA]^{l,j} h\|_2^2,
\end{align*}
and the fourth step follows from $c = 1/\log(4m)$. 
The other steps simply follow the definition of $[\mG\mA h]^{l,j}$.
Note that the variable $Z$ tracks exactly this probability so in \Cref{line:lp:reject} we have $Z = p$.
This also means $(\mG\mA h)_j^2/U \le e^{-4}\cdot(\mG\mA h)_j^2/\|\mG \mA h\|_2^2 \le Z$
so the rejection sampling defined in \Cref{line:lp:reject} is well defined.
For any $j$, the probability of it being returned by \textsc{SampleInternal}$(h,1,1,0)$
is thus $(\mG\mA h)_j^2/U$
\end{proof}

We now show how to combine the data structure of \Cref{lem:matrix_heavy_hitter} with an additional sampling step to obtain a \textsc{HeavySampler} data structure.

\begin{corollary}
\label{lem:lp:heavysampler}
There exists a $(P,c,Q)$-\textsc{HeavySampler}
data structure for matrices $\mA \in \R^{m \times n}$
with $P = \tilde{O}(\nnz(\mA))$, $c_i= \tilde{O}(n)$ for all $i \in [m]$, 
and $Q = \tilde{O}(n^2+m\sqrt{n})$.
\end{corollary}

\begin{proof}
By \Cref{lemma:prop} and \Cref{cor:l2_sampling} we need to be able to efficiently construct the following distribution:
For some constants $C_1,C_2,C_3$ as in \Cref{cor:l2_sampling}, 
let $q \in \R^m$, $S \in \R$ such that
\begin{align}
q_i \ge C_1 \sqrt{n} (\mG \mA h)_i^2 + C_2 / \sqrt{n} + C_3 \tau_i \gamma^2 \log m \label{eq:matrix:q_i}
\end{align}
and $S \ge \sum_{i=1}^m q_i$.
Let $X$ be a random variable with $X = q_i^{-1} \unitvec_i$ with probability $q_i/S$ for all $i$, and $0$ otherwise.

\paragraph{Generating $X$}
We now describe how to efficiently generate this random $X$.
Let $v \in \R^m_{\ge0}$ with 
$$
v_i = \frac{1}{4} \left( \frac{1}{m+n^{1.5}} + 1.5\frac{\otau_i}{\frac{m}{\sqrt{n}} + n} \right)
$$
where $\otau$ is the current approximation of the Lewis weights.
Further define $U = m/n + \sqrt{n}$ and note that $\|\mG \mA h \|_2^2 \le m/n < U$ by guarantee of \Cref{def:heavysampler} (the definition of \textsc{HeavySampler}).

Flip a balanced coin and then either sample an index $j$ with $P[i=j] = (\mG \mA h)_i^2 / U$ via \Cref{lem:matrix_heavy_hitter}, 
or sample the index by $P[i = j] = v_i$.
We then return 
$$
X = \unitvec_i \cdot (C (\sqrt{n} (\mG \mA h)_i^2 + \frac{1}{4\sqrt{n}} + 1.5\otau_i/4))^{-1}
$$
for $C := \max\{C_1,C_2,C_3 \gamma^2 \log m\}$.
Note that we might sample no index at all, in which case we return $X = 0$.
This procedure can be implemented to take $\tilde{O}(n)$ time by \Cref{lem:matrix_heavy_hitter}.

\paragraph{Correctness}
We now prove that the $X$ generated in the paragraph above satisfies that $X = q_i^{-1} \unitvec_i$ with probability $q_i/S$.
For that we define the following
\begin{align*}
q_i &:= C (\sqrt{n} (\mG \mA h)_i^2 + \frac{1}{4\sqrt{n}} + 1.5\otau_i/4) \\
S &:= 2 C U\sqrt{n}.
\end{align*}
Here we clearly have \eqref{eq:matrix:q_i} by definition of $C := \max\{C_1,C_2,C_3 \gamma^2 \log m\}$.
We also have $S \ge \sum_i q_i$ by
\begin{align*}
\sum_{i=1}^m q_i
=
\sum_{i=1}^m C (\sqrt{n} (\mG \mA h)_i^2 + \frac{1}{4\sqrt{n}} + 1.5\otau_i/4)
=
C (\sqrt{n} \|\mG \mA h\|_2^2 + \frac{m}{4\sqrt{n}} + n/2)
\le
2C (\frac{m}{\sqrt{n}} + n)
\le 
S.
\end{align*}

Note that the random $X$ we generated in the previous paragraph is $X = q^{-1}_i \unitvec_i$ 
with probability $1/2 \cdot ((\mG \mA h)_i^2 / U + v_i)$.
We now show that this probability is $q_i/S$:
\begin{align*}
q_i/S
=&~
S^{-1} \cdot C(\sqrt{n} (\mG \mA h)_i^2 + \frac{1}{4\sqrt{n}} + 1.5\otau_i/4) \\
=&~
1/2 \cdot ((\mG \mA h)_i^2 / U + \frac{1}{4Un} + 1.5\frac{\otau_i}{4U\sqrt{n}}) \\
=&~
1/2 \cdot ((\mG \mA h)_i^2 / U + \frac{1}{4(m+n^{1.5})} + 1.5\frac{\otau_i}{4(\frac{m}{\sqrt{n}} + n)}) \\
=&~
1/2 \cdot ((\mG \mA h)_i^2 / U + v_i),
\end{align*}
where the first step uses the definition of $q_i$, the second step uses the definition of $S$, the third step uses the definition of $U$ and the last step uses the definition of $v_i$.
Thus we indeed have $X = q^{-1}_i \unitvec_i$ with probability $q_i/S$.

\paragraph{Final Complexity}

The complexity parameters $P := \tilde{O}(\nnz(A))$ and $c_i := \tilde{O}(n)$ stem from \Cref{lem:matrix_heavy_hitter}.
Generating one random $X$ takes $\tilde{O}(n)$ time 
and we must generate $\tilde{O}(S) = \tilde{O}(m/\sqrt{n} + n)$ many independent copies of $X$ by \Cref{lemma:prop}, 
which results in a total sampling complexity of $Q := \tilde{O}(m\sqrt{n} + n^2)$.

\end{proof}

\section{Leverage Score}
\label{sec:leverage_score_maintenance}

In this section we show how to efficiently maintain an approximation of the leverage scores $\sigma(\mV \mA)$
under updates to $\mV = \mdiag(v)$.
The data structure is obtained via reduction to a \textsc{HeavyHitter}-data structure.

\begin{theorem}\label{thm:leverage_score_maintenance}
Assume there exists a $(P,c,Q)$-\textsc{HeavyHitter} data structure (\Cref{def:heavyhitter}).
Then there exists a Monte-Carlo data-structure (\Cref{alg:leverage_score_maintenance}), 
that works against an adaptive adversary, 
with the following procedures:
\begin{itemize}
\item \textsc{Initialize}$(\mA \in \R^{m \times n}, v \in \R^m, z\in\R^m, \epsilon \in (0,1) )$: 
	The data structure initializes on a matrix $\mA \in \R^{m \times n}$, 
	scaling $v \in \R^m$,
	target accuracy $\epsilon>0$,
	and regularization parameter $z \in \R^m$ with $z \ge n/m + nc/\|c\|_1$
	(where $c$ is the parameter of the heavy hitter data structure)
	and returns a vector $\osigma \in \R^m$ with $\osigma \approx_\epsilon \sigma(\mV \mA) + z$.
\item \textsc{Scale}$(i \in [m], c \in \R_{\ge0})$: 
	For given $c \approx_{0.25} v_i$, set $v_{i} \leftarrow c$.
\item \textsc{Query}$()$: 
	W.h.p.~in $n$ the data-structure outputs a vector $\osigma \in \R^m$
	such that $\osigma \approx_{\epsilon} \sigma(\mV \mA) + z$.
	The vector $\osigma$ is returned as a pointer and the data structure also returns a set $I \subset [m]$ of indices $i$ where $\osigma_i$ has changed compared to the last call to \textsc{Query}.
\end{itemize}

\end{theorem}

The amortized complexities of our data structure depends on the parameters $P,c,Q$ of the \textsc{HeavyHitter} data structure. Further, our complexity bounds require that some additional properties are satisfied.
These properties and the resulting amortized complexities are stated in \Cref{thm:ls:complexity}.

\begin{theorem}\label{thm:ls:complexity}
Consider the data structure of \Cref{thm:leverage_score_maintenance} (\Cref{alg:leverage_score_maintenance})
and let $P,c,Q$ be the parameters of the \textsc{HeavyHitter} data structure (\Cref{def:heavyhitter}).
Let $v\t$ be the vector $v$ in \Cref{thm:lewis_weight_maintenance} during the $t$-th call to \textsc{Query} (and $v^{(0)}$ during the initialization).
Further, assume the following:
\begin{enumerate}
\item \label{item:ls:solver}
For any given $\omv \in \R^{m}_{\ge 0}$ we can solve linear systems in $ (\mA^\top \omV \mA)^{-1}$
with $\epsilon/(64n)$ accuracy 
(i.e. for input $b$ we can output $\omH b$ for some $\omH \approx_{\epsilon/(64 n)} (\mA^\top \omV \mA)^{-1}$)
in $\tilde{O}(P + \Psi + \nnz(\omV\mA))$ time.
If 
$$\mA^\top \omV \mA \approx_{1/(64 \log n)} \mA^\top (\mV\t)^2 \mA$$
for some $t$, then the required time is only $\tilde{O}(\Psi + \nnz(\omV\mA))$.
\item \label{item:ls:stable}
There exists a sequence $\tv\t$ such that for all $t$
\begin{align}
v\t \in (1\pm1/(64 \log n)) \tv\t \label{eq:ls:nearby_sequence}\\
\| (\tv\t)^{-1}(\tv\t - \tv^{(t+1)})\|_{\sigma(\tmV\t \mA)} = O(1) \label{eq:ls:nearby_sequence_stable}
\end{align}
for all $t$.
\end{enumerate}
Then the following time complexities hold:
\begin{itemize}
	\item \textsc{Initialize} takes $\tilde{O}(P + \epsilon^{-2}(\Psi + \nnz(\mA)))$ time.
	\item \textsc{Scale}$(i, \cdot)$ takes $\tilde{O}(
		\frac{\|c\|_1}{n\epsilon^4} \sigma(\mV \mA)_i + \frac{c_i}{\epsilon^2}
	)$ amortized time.
	\item \textsc{Query} takes
	$
	\tilde{O}(
	\Psi \epsilon^{-2}
	+ 
	\epsilon^{-4} n (\max_i \nnz(a_i))
	+ 
	\epsilon^{-2}\sqrt{\frac{P\|c\|_1}{n}}
	+ 
	Q
	).
	$
	amortized time.
\end{itemize}
\end{theorem}

The high-level idea of \Cref{alg:leverage_score_maintenance}
(\Cref{thm:leverage_score_maintenance}) is as follows:
For a set $I \subset [m]$, the method \textsc{UpdateIndices}$(I)$
computes for all $i \in I$ an approximation $\osigma'_i \approx \sigma(\mV \mA)_i + z_i$.
If $\osigma_i \not\approx \osigma'_i$, then we set $\osigma_i \leftarrow \osigma'_i$.
So for all $i \in I$ we have $\osigma_i \approx \sigma(\mV \mA)_i + z_i$
after a call to \textsc{UpdateIndices}$(I)$, as proven in \Cref{lem:ls:update_indices}.
Thus if the set $I$ contains all indices where $\osigma_i \not\approx \sigma(\mV \mA)_i + z_i$,
then $\osigma$ will be a valid approximation after the execution of \textsc{UpdateIndices}$(I)$.

The set $I$ is constructed in method \textsc{FindIndices} as follows: 
Fix some $T \in \N$. For $j=0,...,\log T$,
the data structure checks every $2^j$ calls to \textsc{Query},
if $\sigma(\mV \mA)_i + z_i$ changed a lot 
compared to its value $2^j$ calls to \textsc{Query} in the past.
This claim is proven in \Cref{lem:ls:find_indices}
and we prove in \Cref{lem:ls:correct_output} that such a set $I$ suffices 
to maintain $\osigma \approx \sigma(\mV \mA) + z$
for up to $T$ iterations.
After $T$ iterations the data structure restarts.

\begin{algorithm2e}
\caption{Data structure for maintaining leverage scores \label{alg:leverage_score_maintenance}}
\SetKwProg{Members}{members}{}{}
\SetKwProg{Proc}{procedure}{}{}
\Members{}{
$t,T,r \in \N$, $v \in \R^m$, $\Delta^{(j)} \in \R^m$ and $S_j,C_j \subset [m]$ for $j=1,...,0.5 \log n$.
}
\Proc{\textsc{Initialize}$(\mA, v^{\init}, z, \epsilon)$}{
	$t \leftarrow 0$, 
	$v \leftarrow v^{\init}$, 
	$z \leftarrow z$,
	$T \leftarrow \lceil\epsilon^2\sqrt{Pn/\|c\|_1}\rceil$ \\
	Compute $\osigma \approx_\epsilon \sigma(\mV \mA) + z$ \label{line:ls:init_sigma} \\
	Let $r = O(\log n)$ be such that a $r \times m$ JL-matrix yields a $1/2$-approximation. \\
	\For{$j=0,...,T$}{
		$D_j.\textsc{Initialize}(\mA,v\cdot z^{-1/2})$ \\
		$S_j\leftarrow \emptyset$\\
		$C_j \leftarrow \emptyset$ \\
		$\Delta^{(j)} \leftarrow \zerovec_m$ \label{line:ls:init_Delta_j}
	}
	\Return $\osigma$
}
\Proc{\textsc{FindIndices}$(h\in\R^n)$}{
	$I \leftarrow \emptyset$\\
	\For{$j=\log T,...,0$}{
		\If{$2^j | t$}{
			$\mS_{i,i}=1$ for $i \in S_j \cup C_j$, 
			and for other $i$ we set $\mS_{i,i} = 1/p_i$ with probability 
			$p_i = \min(1,~c \osigma_i \epsilon^{-2} \log n \log\log n)$ 
			for some large constant $c > 0$, 
			and $\mS_{i,i} = 0$ otherwise. \label{line:ls:S_findindices}\\
			$\mR \leftarrow$ $m \times r$ JL-matrix \label{line:ls:JL_find_indices}\\
			Let $\mM \approx_{1/(64 n)}(\mA^\top(\mV - \Delta^{(j)})^2 \mS^2 \mA)^{-1}$ and $\mM \approx_{1/(64 n)} (\mA^\top\mV^2 \mS^2 \mA)^{-1}$ \\
			$\mH \leftarrow (\mM' \mA^\top (\mV - \Delta^{(j)}) \mS \mR 
				- \mM \mA^\top \mV \mS \mR$ \label{line:ls:H_difference}\\
			$I \leftarrow I \cup D_j.\textsc{HeavyQuery}(\mH \unitvec_k, \epsilon /(48 r\log n))$ for all $k \in [r]$ \label{line:ls:add_query}\\
			$\Delta^{(j)} \leftarrow \zerovec_m$ \label{line:ls:reset_Delta_j}\\
			$D_j.\textsc{Scale}(i, v_i z_i^{-1/2})$ for $i \in S_j$ \label{line:ls:scale_D_j}\\
			$I \leftarrow I \cup S_j$, $S_j \leftarrow \emptyset$, $C_j \leftarrow \emptyset$ \label{line:ls:add_trivial_sets}
		}
	}
	\Return $I$
}
\Proc{\textsc{UpdateIndices}$(I)$}{
	$\mS_{i,i} = 1/\sqrt{p_i}$ with probability $p_i = \min(1,~ c \epsilon^{-2} \osigma_i \log n)$ for some large enough constant $c > 0$, and $\mS_{i,i} = 0$ otherwise. \\
	$\mR \leftarrow$ $\exp(\pm\epsilon/16)$-accurate JL-matrix \\
	$\mH \leftarrow \mM \mA^\top \mV \mS \mR$ for any $\mM \approx_{\epsilon/(64 \log n)} (\mA^\top \mV^2 \mS^2 \mA)^{-1}$ \label{line:ls:H_updateindices}\\
	$I' \leftarrow \emptyset$ \\
	\For{$i \in I$}{
		\If{$\osigma_i \not\approx_{3\epsilon/8} \|e_i^\top \mV \mA \mH\|_2^2 + z_i$}{ \label{line:ls:check_change}
			$\osigma_i \leftarrow \| e_i^\top \mV \mA \mH \|_2^2 + z_i$ \label{line:ls:update_osigma}\\
			$C_j \leftarrow C_j \cup \{ i \}$ for $j = 0,...,\log n$ \label{line:ls:add_C}\\
			$I' \leftarrow I' \cup \{ i \}$
		}
	}
	\Return $I'$
}
\Proc{\textsc{Scale}$(i, c)$}{
	\For{$j=0,...,\log T$}{
		$S_j \leftarrow S_j \cup \{i\}$ \label{line:ls:add_S}\\
		$D_j.\textsc{Scale}(i, 0)$ \label{line:ls:scale_D_j_0}\\
		$\Delta^{(j)} \leftarrow \Delta^{(j)} + c - v_i$ \Comment{Maintain $v^{(t)} = v^{(t_j)} + \Delta^{(j)}$} \label{line:ls:update_Delta_j}\\
	}
	$v_i \leftarrow c$ \\
}
\Proc{\textsc{Query}$()$}{
	\lIf{$t = T$}{\Return $[m]$, \textsc{Initialize}$(\mA, v, w, \epsilon)$}
	$t \leftarrow t + 1$\\
	$I \leftarrow \textsc{FindIndices}()$ \label{line:findIndices}\\
	$I \leftarrow \textsc{UpdateIndices}(I)$ \label{line:verify_I}\\
	\Return $I$, $\osigma$
}
\end{algorithm2e}

\subsection{Correctness}

As the vector $v$ for which we want to maintain $\sigma(\mV \mA)$ changes over time,
we use the following notation during the proof of \Cref{thm:leverage_score_maintenance} 
to specify which instance of $v$ we refer to:

\begin{definition}
Write $v\t$ for the vector $v$ during the $t$-th call to \textsc{Query} 
and $v^{(0)}$ for the vector $v$ as given during the initialization.
Likewise write $\osigma\t$ for the vector $\osigma$ returned by the $t$-th call to \textsc{Query}.
\end{definition}

So we want to prove that $\osigma\t \approx_\epsilon \sigma(\mV\t \mA) + z$.
We will prove this via induction, 
so let the following \Cref{prop:ls:approximate_sigma} 
be the induction hypothesis.
Note that for $t=0$ the claim is true 
as seen by the initialization procedure 
(\Cref{line:ls:init_sigma}).

\begin{proposition}
\label{prop:ls:approximate_sigma}
During the $t$-th call to \textsc{Query}
we have $\osigma^{(t-1)} \approx_{\epsilon} \sigma(\mV^{(t-1)} \mA) + z$,
i.e. the result of the \emph{previous} call to \textsc{Query} was correct.
\end{proposition}

During the $t$-th call to \textsc{Query},
the algorithm needs access to past $v^{(k)}$ for $k < t$.
In order to not process an $m$-dimensional vector in every iteration,
these vectors are maintained implicitly by the data structure
as stated in \Cref{lem:ls:delta_j}.

\begin{lemma}
\label{lem:ls:delta_j}
When executing \Cref{line:ls:H_difference} we have $v\t = v^{(t-2^j)} + \Delta^{(j)}$.
\end{lemma}

\begin{proof}
Let $t_j$ be the last time we set $\Delta^{(j)} \leftarrow \zerovec_m$.
We set $\Delta^{(j)} \leftarrow \zerovec_m$ either during initialization (so $t_j = 0$) 
or in \Cref{line:ls:reset_Delta_j} (\textsc{FindIndices}) when $2^j | t$,
so right after executing \Cref{line:ls:reset_Delta_j} we have $t_j = t - 2^j$.

Further, whenever $v_i$ is increased by some $\delta \in \R$,
we add $\delta$ also to $\Delta^{(j)}_i$ for all $j$
in \Cref{line:ls:update_Delta_j}, 
so we always have $v\t_i = v^{(t_j)} + \Delta^{(j)}$.

Thus in \Cref{line:ls:H_difference} we have $v\t_i = v^{(t-2^j)} + \Delta^{(j)}$.
\end{proof}

Next, we prove the previous claim of the data structure's outline regarding \textsc{FindIndices}.
We claimed that every $2^j$ calls to \textsc{FindIndices} for any $j=1,...,\log T$,
the function returns a set $I \subset [m]$ of indices $i$ where $\sigma(\mV\t \mA)_i + z_i$ changed a lot 
compared to $\sigma(\mV^{(t-2^j)} \mA)_i + z_i$.

\begin{lemma}
\label{lem:ls:find_indices}
Let $I \subset [m]$ be the set returned by \textsc{FindIndices}$(h \in \R^n)$.
Then w.h.p set $I$ contains all indices $i$ where for any $0 \le j \le \log T$ we have 
$2^j|t$
and 
\begin{align}
\sigma(\mV\t \mA)_i +z_i \not\approx_{\epsilon/(4\log n)} \sigma(\mV^{(t-2^j)})_i +z_i.
\label{eq:ls:sigma_difference}
\end{align}
\end{lemma}

\begin{proof}
Note that for $j$ with $w^j | t$ the set $I$ contains all indices $i \in S_j$
by \Cref{line:ls:add_trivial_sets}.
So we only need to consider the remaining indices $i \notin S_j$.
Let $\mZ = \mdiag(z)$ and let $\mF$ be diagonal with $\mF_{i,i} = 1$ for $i \notin S_j$.
$I$ contains all indices $i$ where for some $k$ we have $(\mF \mZ^{-1/2} \mV^{(t-2^j)} \mA \mH \unitvec_k)_i > \epsilon/(48r\log n)$.
(because we called $D_j.\textsc{Scale}(i,0)$ whenever $i$ was added to $S_j$ \Cref{line:ls:scale_D_j_0},
and otherwise we call $D_j.\textsc{Scale}(i,v_i/\sqrt{z_i})$ for $i \in S_j$ in \Cref{line:ls:scale_D_j}).

When $\mR$ in \Cref{line:ls:JL_find_indices} 
has $r=O(\log n)$ columns (and thus $\mH$ has $r$ columns),
then if
$\|\unitvec_i^\top \mF \mZ^{-1/2} \mV^{(t-2^j)} \mA \mH \|_2^2 > \epsilon^2 /(48^2 \log^2 n)$,
then $|\unitvec_i^\top \mF \mZ^{-1/2} \mV^{(t-2^j)} \mA \mH \unitvec_k| > \epsilon /(48 r\log n)$ for some $k$, 
so $i$ is added to set $I$ in \Cref{line:ls:add_query}.
So we must show that if \eqref{eq:ls:sigma_difference} is satisfied, 
then $\|\unitvec_i^\top \mF \mZ^{-1/2} \mV^{(t-2^j)} \mA \mH \|_2^2 > \epsilon^2 /(48^2 \log^2 n)$.

By \Cref{line:ls:H_difference} we have
\begin{align*}
\|\unitvec_i^\top \mF \mZ^{-1/2} \mV^{(t-2^j)} \mA \mH\|_2 
=&~
\|\unitvec_i^\top \mF \mZ^{-1/2} \mV^{(t-2^j)} \mA \left(
\mM' \mA^\top \mV^{(t-2^j)} \mS - \mM \mA^\top \mV\t \mS
\right) \mR \|_2 
\\
\ge&~
0.5
\|\unitvec_i^\top \mF \mZ^{-1/2} \mV^{(t-2^j)} \mA \left(
\mM' \mA^\top \mV^{(t-2^j)} \mS - \mM \mA^\top \mV\t \mS
\right)\|_2 \\
\ge&~
0.5 z_i^{-1/2} \Big|~
(\|\unitvec_i^\top \mV^{(t-2^j)} \mA \mM' \mA^\top \mV^{(t-2^j)} \mS \|_2^2 + z_i)^{0.5} \\
&~ -
(\| \unitvec_i^\top \mV\t \mA \mM \mA^\top \mV\t \mS \|_2^2 + z_i)^{0.5}
\Big|
\end{align*}
For the second step we used that $\mR$ is a JL-matrix 
such that for any $w \in \R^m$ we have w.h.p $\|w^\top \mR\|_2 \ge 0.5 \|w\|_2$.
The third step uses triangle inequality and that we only consider $i \notin S_j$, 
so $v\t_i = v^{(t-2^j)}_i$.

Note that set $C_j$ contains all indices for which the leverage score has changed by \Cref{line:ls:add_C}.
Thus for all $k \notin C_j$ we have that $\osigma_k$ is a constant factor approximation 
of $\sigma(\mV^{(t-1)} \mA)_k + z_k$ and $\sigma(\mV^{(t-2^j)}\mA)_k + z_k$.
Further, since $v\t \approx_1 v^{(t-1)}$ by assumption on \textsc{Scale} 
(see \Cref{thm:leverage_score_maintenance}) 
we have that $\osigma_k$ is also a constant factor approximation of $\sigma(\mV^{(t)} \mA)_k + z_k$.
Hence we have that $\mS_{k,k} = 1/\sqrt{p_k}$ with probability $p_k$ and $0$ otherwise,
where 
$$
p_k \ge \min\{1, \max\{ \sigma(\mV\t \mA)_k , \sigma(\mV^{(t-2^j)} \mA)_k \} c \epsilon^{-2} \log n \log\log n \}
$$
for some large enough constant $c > 0$.
Thus with high probability we have
\begin{align*}
\mM 
\approx_{1/(64\log n)}
(\mA^\top (\mV\t)^2 \mS^2 \mA)^{-1}
\approx_{\epsilon/(64\log n)} 
(\mA^\top (\mV\t)^2 \mA)^{-1}
\end{align*}
and likewise $\mM' \approx_{1/(32 \log n)} (\mA^\top (\mV^{(t-2^j)})^2 \mA)^{-1}$,
because $\mS$ is a valid leverage score sample for both matrices.
Finally, this means that
\begin{align*}
\sigma(\mV\t \mA)_i
=&~ 
\| \unitvec_i^\top \mV\t \mA (\mA^\top (\mV\t)^2 \mA)^{-1} \mA^\top \mV\t \|_2^2 \\
\approx_{\epsilon/(16\log n)}&~
\| \unitvec_i^\top \mV\t \mA \mM \mA^\top \mV\t \mS\|_2^2
\end{align*}
and likewise for $t-2^j$. By $z_i \ge n/m$ we thus have
\begin{align*}
&~
\sigma(\mV\t \mA)_i +z_i \not\approx_{\epsilon/(4\log n)} \sigma(\mV^{(t-2^j)} \mA)_i +z_i\\
\Rightarrow&~
\| \unitvec_i^\top \mV\t \mA (\mA^\top (\mV\t)^2 \mA)^{-1} \mA^\top \mV\t \|_2^2 +z_i \\
&~\not\approx_{\epsilon/(4\log n)} 
\| \unitvec_i^\top \mV^{(t-2^j)} \mA (\mA^\top (\mV^{(t-2^j)})^2 \mA)^{-1} \mA^\top \mV^{(t-2^j)} \|_2^2 +z_i\\
\Rightarrow&~
\| \unitvec_i^\top \mV\t \mA \mM \mA^\top \mV\t \mS\|_2^2 +z_i \\ 
&~\not\approx_{\epsilon/(8\log n)} 
\| \unitvec_i^\top \mV^{(t-2^j)} \mA \mM' \mA^\top \mV^{(t-2^j)} \mS\|_2^2 +z_i\\
\Rightarrow&~
(\| \unitvec_i^\top \mV\t \mA \mM \mA^\top \mV\t \mS\|_2^2 +z_i)^{0.5} \\ 
&~\not\approx_{\epsilon/(16\log n)} 
(\| \unitvec_i^\top \mV^{(t-2^j)} \mA \mM' \mA^\top \mV^{(t-2^j)} \mS\|_2^2 +z_i)^{0.5}\\
\Rightarrow&~
|~
(\| \unitvec_i^\top \mV^{(t-2^j)} \mA \mM' \mA^\top \mV^{(t-2^j)} \mS \|_2^2 + z_i)^{0.5} \\
&~-
(\| \unitvec_i^\top \mV\t \mA \mM \mA^\top \mV\t \mS \|_2^2 + z_i)^{0.5}
~| \ge \epsilon \sqrt{z_i} / (48\log n)
\end{align*}
so the index $i$ is returned by \textsc{FindIndices}.
\end{proof}

Next, we prove the claim that calling \textsc{UpdateIndices}$(I)$ 
indeed guarantees that $\osigma_i \approx_{\epsilon/2} \sigma(\mV\t \mA)_i + z_i$ for all $i \in I$.

\begin{lemma}\label{lem:ls:update_indices}
For $I \subset [m]$ after a call to \textsc{UpdateIndices}$(I)$ 
we have $\osigma_i \approx_{\epsilon/2} \sigma(\mV\t \mA)_i + z_i$ for all $i \in I$ w.h.p in $n$.
\end{lemma}

\begin{proof}
For $i \in I$ the method \textsc{UpdateIndices} computes
\begin{align*}
&~
\| \unitvec_i^\top \mV\t \mA \mH \|_2^2 \\
=&~
\| \unitvec_i^\top \mV\t \mA \mM \mA^\top \mV\t \mS \mR \|_2^2 \\
\approx_{\epsilon/8}&~
\| \unitvec_i^\top \mV\t \mA \mM \mA^\top \mV\t \|_2^2 \\
=&~
\sigma(\mV\t \mA)_i
\end{align*}
where we used that $\mR$ is a JL-matrix that yields an $\exp(\pm\epsilon/16)$-factor approximation of the norm
and that $\osigma_i \approx_\epsilon \sigma(\mV^{(t-1)} \mA) \approx_{4} \sigma(\mV\t \mA)$,
so $\mS^2$ is a leverage score sample and satisfies with high probability $\mA^\top (\mV\t)^2\mS^2 \mA \approx_{\epsilon/32} \mA^\top (\mV\t)^2 \mA$
Thus $\mM \approx_{\epsilon/16} (\mA^\top (\mV\t)^2 \mA)^{-1}$.

So if $\osigma_i \approx_{3\epsilon/8} \| \unitvec_i^\top \mV\t \mA \mH \|_2^2 + z_i$,
then $\osigma_i \approx_{\epsilon/2} \sigma(\mV\t \mA)_i + z_i$.
On the other hand, if the difference is larger, 
then \textsc{UpdateIndices} sets 
$\osigma_i \leftarrow \| \unitvec_i^\top \mV\t \mA \mH \|_2^2 + z_i$,
so then $\osigma_i \approx_{\epsilon/8} \sigma(\mV\t \mA)_i + z_i$.
\end{proof}

At last, we combine the guarantees of \textsc{FindIndices} and \textsc{UpdateIndices}
to show that we always maintain the desired approximation.
For that we will use the following natural lemma about decomposing an interval into powers of two from \cite{BrandLN+20}.

\begin{lemma}[\cite{BrandLN+20}]
\label{lem:transform_t}
Given any $\ot < t$, there exists a sequence $$t = t_0 > t_1 > ... > t_k = \overline{t}$$ 
such that $k\leq 2\log t$ and $t_{z+1} = t_z - 2^{\ell_z}$ where $\ell_z$ satisfies $2^{\ell_z} | t_z $ for all $z=0,\ldots,k-1$. 
\end{lemma}

\begin{lemma}
\label{lem:ls:correct_output}
After the $t$-th call to \textsc{Query} we have $\osigma_i \approx_\epsilon \sigma(\mV \mA)_i + z_i$ w.h.p. in $n$. 
\end{lemma}

\begin{proof}
Fix some $i$. 
We want to prove that $\osigma_i \approx_\epsilon \sigma(\mV\t \mA)_i +z_i$.

Assume index $i$ was last contained in set $I$ when calling \textsc{UpdateIndices}$(I)$
during the $\ot$-th call to \textsc{Query} ($\ot = 0$ if $i$ was never in set $I$).
By \Cref{lem:ls:update_indices}
we have $\osigma_i \approx_{\epsilon/2} \sigma(\mV^\ot \mA)_i +z_i$.
If $\ot = t$, then we are done, so let us focus on the case $\ot < t$ instead.
In that case we are left with showing 
$\sigma(\mV^{\ot} \mA)_i +z_i \approx_{\epsilon/2} \sigma(\mV\t \mA)_i +z_i$.

By \Cref{lem:transform_t} there exists a sequence $t = t_0 > t_1 > ... > t_k = \ot$
with $k \le 2 \log t$ and $t_{z+1} = t_z - 2^{\ell_z}$ 
where $\ell_z$ satisfies $2^{\ell_z} | t_z$ 
for all $z = 0,...,k-1$.

If during the $t_z$-th call to \textsc{Query} we had 
$\sigma(\mV^{(t_{z})}\mA)_i +z_i \not\approx_{\epsilon/(4\log n)} \sigma(\mV^{(t_{z+1})})_i+z_i$,
then by \Cref{lem:ls:find_indices} the index $i$ would have been added to set $I$ in \textsc{FindIndices}.
Since by assumption $i$ was not added to $I$ since $\ot$,
we thus know $\sigma( \mV^{(t_{z})} \mA )_i + z_i \approx_{ \epsilon / (4 \log n) } \sigma( \mV^{(t_{z+1})} )_i + z_i$.
As the sequence of $t_z$ has length $2 \log t$ and $t \le \sqrt{n}$ (by restarting the data structure after $\sqrt{n}$ calls to \textsc{Query}) we have
$\sigma( \mV^{(\ot)} \mA )_i +z_i = \sigma( \mV^{(t_k)} \mA )_i +z_i \approx_{\epsilon 2 \log t /(4\log n)} \sigma( \mV^{(t_0)} )_i +z_i = \sigma( \mV^{(t)} )_i +z_i$,
so $ \sigma( \mV^{(\ot)} \mA )_i \approx_{\epsilon/2} \sigma( \mV^{(t)} )_i + z_i$. 
\end{proof}

\subsection{Complexity}

In this subsection we prove the amortized complexity guarantees of \Cref{thm:leverage_score_maintenance}.
Remember that the idea of our algorithm was to detect whenever a leverage score changes a lot,
and to then recompute the detected scores.
Thus to bound the complexity, we must bound how many $i$ are detected for which $\sigma_i$ might have changed a lot.
For that we will use the following \Cref{lem:stable_large_sequence}, which is proven in \Cref{sec:stabilizer}.

\begin{restatable*}{lemma}{stableLargeSequence}
\label{lem:stable_large_sequence}

Let $\mA\in\R^{m\times n}$. Let $v^{(0)},v^{(1)},\cdots,v^{(T)}\in\R_{+}^{m}$
be a sequence of vectors with $v^{(j)}\approx_{1/4}v^{(j-1)}$ for $j = 1,2,\cdots ,T$. Let
$\mF$ and $\mS$ be diagonal matrices on $\R^{m\times m}$ such that
$\|\mF\|_{2}\leq1$,
\begin{itemize}
\item $\mF\mV^{(j)}=\mF\mV^{(j-1)}$ for $j\in\{1,2,\cdots,T\}$,
\item $(\mV^{(j)}-\mV^{(j-1)})\mS=\mV^{(j)}-\mV^{(j-1)}$ for $j\in\{1,2,\cdots,T\}$,
\item $\mA^{\top}\mS^{2}\mV^{(j)2}\mA\approx_{1/2}\mA^{\top}\mV^{(j)2}\mA$
for $j\in\{0,1,2,\cdots,T\}$.
\end{itemize}
Suppose that there is a sequence $\ov^{(0)},\ov^{(1)},\cdots,\ov^{(T)}\in\R_{+}^{m}$
such that 
\begin{align}
\ov^{(j)}\approx_{1/6}v^{(j)} \label{eq:stab:nearby_sequence_condition}\\
\|(\widetilde{\mV}^{(j-1)})^{-1}(\ov^{(j)}-\ov^{(j-1)})\|_{\sigma(\widetilde{\mV}^{(j-1)}\mA)+\infty}\leq\frac{1}{6} \label{eq:stab:stable_squence_condition}.
\end{align}
for $j = 1,2,\cdots ,T$.
Then, for $\mP^{(j)} \defeq\mV^{(j)}\mA(\mA^{\top}\mV^{(j)2} \mS^2 \mA)^{-1}\mA\mV^{(j)}$,
we have
\[
\|\mF(\mP^{(T)}-\mP^{(0)})\mS\|_{F}^{2}\lesssim T^{2}+\sum_{j=1}^{T}\sum_{v_{i}^{(j)}\neq v_{i}^{(j-1)}}\sigma(\mV^{(j)}\mA)_{i}.
\]
\end{restatable*}

As the complexity of \textsc{UpdateIndices}$(I)$ depends on the size of set $I$,
we first require a bound on the size of these sets.
The following lemma bounds the total size of all sets $I$, when weighting each element $i \in I$
by $z_i$.
This weighting is required because the runtime cost due to one $i \in I$ in \textsc{UpdateIndices}
scales in $c_i$. %

\begin{lemma}\label{lem:ls:size_bound_I}
Let $I\t$ be the set $I$ returned by \textsc{FindIndices} during the $t$-th call to \textsc{Query}.
Then 
\begin{align}
\sum_{t=1}^T \sum_{i \in I\t} c_i
\le
\tilde{O}\left(
T^2 \frac{\|c\|_1}{n \epsilon^{2}}
+
\sum_{t=0}^{T-1} \sum_{v^{(t)}_i \neq v^{(t+1)}_i} \left(c_i + \frac{\|c\|_1}{n\epsilon^2} \sigma(\mV\t \mA)_i \right)
\right)
\label{eq:ls:size_bound_I}
\end{align}
\end{lemma}

\begin{proof}
An index $i$ is added to $I$, if it is contained in set $S_j$ for some $j=0,...,\log n$.
Thus one call to \textsc{Scale}$(i ,\cdot)$ results in $i$ being in $\log n$ many $I$.
This is the $\sum_{t=0}^{T-1} \sum_{v^{(t)}_i \neq v^{(t+1)}_i} c_i$ term in \eqref{eq:ls:size_bound_I}.

Next, an index $i$ is added to $I$, if it was returned by $D_j.\textsc{HeavyQuery}$ in \Cref{line:ls:add_query}.
The number of returned indices (weighted by $c_i$) is bounded by 
\begin{align}
&~
\sum_{k\in[r]} \sum_{i\in[m]} c_i \cdot \mathbf{1}_{|(\mF \mZ^{-1/2} \mV^{(t-2^j)}\mA \mH \unitvec_k)_i|> \epsilon/(48 r \log n)}
\notag\\
=&~
\tilde{O}\left( 
	\sum_{k\in[r]} \|\mF \mZ^{-1/2} \mV^{(t-2^j)}\mA \mH \unitvec_k\|_c^2 \epsilon^{-2} 
\right) \notag\\
=&~
\tilde{O}\left( 
\frac{\|c\|_1}{n} \| \mF \mV^{(t-2^j)} \mA \left(\mM'\mA^\top\mV^{(t-2^j)}\mS - \mM \mA^\top \mV^{(t)}\mS\right) \mR \|_F^2 \epsilon^{-2} 
\right) \notag\\
\le&~
\tilde{O}\left(
\frac{\|c\|_1}{n} \left(2^{2j} + \sum_{j=t-2^j}^{t} \sum_{v^{(j)}_i \neq v^{(j+1)}_i} \sigma(\mV^{(j)} \mA)_i\right) \epsilon^{-2}
\right), \label{eq:ls:frobenius_bound}
\end{align}
where the first step uses $n\cdot c / \|c\|_1 \le z$
and the last step uses that $\mR$ is a JL-matrix,
that $\mM \approx_{1/(64n)}(\mA^\top \mV\t \mS^2 \mA)^{-1}$
and $\mM' \approx_{1/(64n)} (\mA^\top \mV^{(t-2^j)} \mS^2 \mA)^{-1}$,
and \Cref{lem:stable_large_sequence}.
Note that we can apply \Cref{lem:stable_large_sequence},
because \eqref{eq:stab:nearby_sequence_condition} is satisfied by \eqref{eq:ls:nearby_sequence},
\eqref{eq:stab:stable_squence_condition} is satisfied by \eqref{eq:ls:nearby_sequence_stable},
the conditions on $\mF$ are satisfied by definition,
and the conditions on $\mS$ are satisfied by $\mS$ being a leverage score sample 
where we used probability $p_i = 1$ for $i \in S_j \cup C_j$.

As $D_j.\textsc{HeavyQuery}$ is performed once every $j$ iterations for $j=0,...,T$,
we obtain \eqref{eq:ls:size_bound_I}.
\end{proof}

Using \Cref{lem:ls:size_bound_I} we can now bound the total runtime of all calls to \textsc{UpdateIndices}.

\begin{lemma}\label{lem:ls:amortized_updateindices}
The amortized cost of the $t$-th call to \textsc{UpdateIndices} is
\begin{align*}
\tilde{O}\left(
\Psi \epsilon^{-2}
+ \epsilon^{-4} n (\max_i \nnz(a_i))
+ T\frac{\|c\|_1}{n\epsilon^4} 
+ \sum_{v^{(t)}_i \neq v^{(t-1)}_i} \left(\frac{c_i}{\epsilon^2} + \frac{\|c\|_1}{n\epsilon^4} \sigma(\mV^{(t-1)} \mA)_i \right)
\right).
\end{align*} 
\end{lemma}

\begin{proof}
The matrix $\mS$ has $O(\epsilon^{-2} n \log n)$ non-zero entries
and the matrix $\mR$ has $O(\epsilon^{-2} \log n)$ columns.
Thus we can compute $\mH$ in \Cref{line:ls:H_updateindices} 
in $\tilde{O}(\Psi \epsilon^{-2} + \epsilon^{-4} n \max_i \nnz(a_i))$ time.

Next, we compute $\|\unitvec_i^\top \mV \mA \mH\|_2^2$ for all $i \in I$.
The time required for that is
\begin{align*}
O\left(\sum_{i \in I} \nnz(a_i) \epsilon^{-2} \log n\right)
=
\tilde{O} \left( \epsilon^{-2} \sum_{i \in I} c_i \right)
\end{align*}
which according to \Cref{lem:ls:size_bound_I} is an amortized cost of
\begin{align*}
\tilde{O}(T\frac{\|c\|_1}{n \epsilon^4} + \sum_{v^{(t)}_i \neq v^{(t-1)}_i} (\frac{c_i}{\epsilon^2} + \frac{\|c\|_1}{n \epsilon^4} \sigma(\mV^{(t-1)} \mA)_i) )
\end{align*}
for the $t$-th call to \textsc{UpdateIndices}.
\end{proof}

Next, we want to analyze the amortized cost of a call to \textsc{FindIndices}.
Note that the complexity will depend on the size of the sets $C_j$,
because for each $i \in C_j$ the matrix $\mS$ in \Cref{line:ls:S_findindices} will be more dense.
An index $i$ is added to $C_j$ in \Cref{line:ls:add_C}, if we changed $\osigma_i$.
Thus before bounding the complexity of \textsc{FindIndices},
we will first bound how often $\osigma_i$ is changed.
We show that we can bound the number of times we change any entry of $\osigma$
with respect to how often the function \textsc{Scale} is called.

\begin{lemma}\label{lem:ls:bound_output}
Let $\osigma\t$ be the output after the $t$-th call to \textsc{Query}
and $v\t$ be the vector $v$ during the $t$-th call to \textsc{Query}.
Then
$$
\sum_{i=1}^m \sum_{t=1}^T (\sigma(\mV\t \mA)_i+z_i) \mathbf{1}_{\osigma\t_i \neq \osigma^{(t-1)}_i}
\le
O\left(
\sum_{i=1}^m \sum_{t=1}^T \sigma(\mV\t \mA)_i \mathbf{1}_{v\t_i \neq v^{(t-1)}_i} \epsilon^{-1}
\right).
$$
\end{lemma}

\begin{proof}
Note that when we update $\osigma_i$ in \Cref{line:ls:update_osigma},
then we have $\osigma_i \approx_{\epsilon/8} \sigma(\mV \mA) + z_i$,
because $\|\unitvec_i^\top \mV \mA \mH \|_2^2 \approx_{\epsilon/8} \sigma(\mV \mA)$ (see proof of \Cref{lem:ls:update_indices}).
Further note that the $i$th entry of the output vector $\osigma$ can only change,
whenever $\osigma_i \not\approx_{3\epsilon/8} \|\unitvec_i^\top \mV \mA \mH \|_2^2 + z_i$
because of \Cref{line:ls:check_change}.
Thus in order for $\osigma_i$ to change, 
$\sigma(\mV \mA)_i +z_i$ must have changed by at least a $\exp(\pm\epsilon/8)$-factor,
so we can bound
\begin{align*}
&~
\sum_{i=1}^m \sum_{t=1}^T (\sigma(\mV\t \mA)_i+z_i) \mathbf{1}_{\osigma\t_i \neq \osigma^{(t-1)}_i} \\
=&~
O\left(
\sum_{i=1}^m \sum_{t=1}^T (\sigma(\mV\t \mA)_i+z_i)\cdot \frac{|(\sigma(\mV\t \mA)_i+z_i) - (\sigma(\mV^{(t-1)} \mA)_i + z_i)|}{(\sigma(\mV\t \mA)_i+z_i)\epsilon}
\right) \\
=&~
O\left(
\sum_{i=1}^m \sum_{t=1}^T |\sigma(\mV\t \mA)_i - \sigma(\mV^{(t-1)} \mA)_i| \epsilon^{-1}
\right).
\end{align*}
Here the difference of the two leverage scores can be bounded as follows
\begin{align*}
&~
|\sigma(\mV\t \mA)_i - \sigma(\mV^{(t-1)} \mA)_i| \\
=&~
|(\mV\t_i)^2 (\mA (\mA^\top (\mV\t)^2 \mA)^{-1} \mA^\top)_{i,i} 
  - (\mV^{(t-1)}_i)^2 (\mA (\mA^\top (\mV^{(t-1)})^2 \mA)^{-1} \mA^\top)_{i,i}| \\
\le&~
|\left((\mV\t_i)^2 - (\mV^{(t-1)}_i)^2\right) (\mA (\mA^\top (\mV\t)^2 \mA)^{-1} \mA^\top)_{i,i} | \\
&~ + | (\mV^{(t-1)}_i)^2 (\mA \left((\mA^\top (\mV\t)^2 \mA)^{-1} \mA)^{-1} - (\mA^\top (\mV^{(t-1)})^2 \mA)^{-1} \right) \mA^\top)_{i,i}|
\end{align*}
Here the first term can be bounded by
\begin{align*}
&~
|\left((\mV\t_i)^2 - (\mV^{(t-1)}_i)^2\right) (\mA (\mA^\top (\mV\t)^2 \mA)^{-1} \mA^\top)_{i,i} | \\
=&~
|\left(1 - (\mV^{(t-1)}_i/\mV\t_i)^2\right) (\mV\t_i)^2(\mA (\mA^\top (\mV\t)^2 \mA)^{-1} \mA^\top)_{i,i} | \\
=&~
3 \sigma(\mV\t \mA)
\end{align*}
by using $v\t \approx_{1/2} v^{(t-1)}$.
The second term can be bounded as follows.
For 
$$\mH_x := \mA^\top ((1-x)(\mV^{(t-1)})^2 + x (\mV\t)^2) \mA$$ 
we have
$$
(\mA^\top (\mV\t)^2 \mA)^{-1} \mA)^{-1} - (\mA^\top (\mV^{(t-1)})^2 \mA)^{-1}
=
\int_0^1 \mH_x^{-1} \mA^\top ((\mV\t)^2 - (\mV^{(t-1)})^2)  \mA \mH_x^{-1} ~ dx
$$
so 
\begin{align*}
&~
| (\mV^{(t-1)}_i)^2 (\mA \left((\mA^\top (\mV\t)^2 \mA)^{-1} \mA)^{-1} - (\mA^\top (\mV^{(t-1)})^2 \mA)^{-1} \right) \mA^\top)_{i,i}| \\
=&~
| \int_0^1 (\mV^{(t-1)}_i)^2 (\mA \mH_x^{-1} \mA^\top ((\mV\t)^2 - (\mV^{(t-1)})^2)  \mA \mH_x^{-1} \mA^\top)_{i,i} ~ dx| \\
\le&~
 2\int_0^1 |
 (\mV\t \mA \mH_x^{-1} \mA^\top \mV\t (1 - (\mV^{(t-1)}\mV^{-1})^2)  \mV\t \mA \mH_x^{-1} \mA^\top \mV\t)_{i,i} 
 |~ dx \\
=&~
 2\int_0^1
 \|(1 - (\mV^{(t-1)}\mV^{-1})^2)^{1/2} \mV\t \mA \mH_x^{-1} \mA^\top \mV\t \unitvec_i \|_2^2 
 ~ dx
\end{align*}
Note that when taking the sum over all $i\in[m]$ this can be bounded by
\begin{align*}
&~
\sum_{i\in[m]}
\int_0^1
 \|(1 - (\mV^{(t-1)}\mV^{-1})^2)^{1/2} \mV\t \mA \mH_x^{-1} \mA^\top \mV\t \unitvec_i \|_2^2 
 ~ dx \\
=&~
\int_0^1
 \|(1 - (\mV^{(t-1)}\mV^{-1})^2)^{1/2} \mV\t \mA \mH_x^{-1} \mA^\top \mV\t \|_F^2 
 ~ dx \\
\le&~
 \|(1 - (\mV^{(t-1)}\mV^{-1})^2)^{1/2} \mP \|_F^2 
 ~ dx
\end{align*}
where $\mP := \mV\t \mA (\mA^\top (\mV\t)^2 \mA)^{-1} \mA^\top \mV\t$.
This can be written as
\begin{align*}
\|(1 - (\mV^{(t-1)}/\mV\t)^2)^{1/2} \mP \|_F^2 
&=
\sum_{i\in[m]} (1 - (v^{(t-1)}_i/v\t_i)^2)^{1/2} \sum_{j \in [m]} \mP_{i,j}^2 \\
&=
\sum_{i\in[m]} (1 - (v^{(t-1)}_i/v\t_i)^2)^{1/2} \sigma(\mV\t \mA)_i \\
&\le
\sum_{i\in[m]} \mathbf{1}_{v\t_i \neq v^{(t-1)}_i} \sigma(\mV\t \mA)_i
\end{align*}
By combining we obtain
$$
\sum_{i=1}^m \sum_{t=1}^T (\sigma(\mV\t \mA)_i+z_i) \mathbf{1}_{\osigma\t_i \neq \osigma^{(t-1)}_i}
\le
O(
\sum_{i=1}^m \sum_{t=1}^T \sigma(\mV\t \mA)_i \mathbf{1}_{v\t_i \neq v^{(t-1)}_i} \cdot \epsilon^{-1}
).
$$
\end{proof}

\begin{lemma}\label{lem:ls:amortized_findindices}
The amortized cost of the $t$-th call to \textsc{FindIndices} is
\begin{align*}
\tilde{O}\left(
\Psi
+ n \epsilon^{-2} (\max_i \nnz(a_i))
+ \frac{\|c\|_1}{n\epsilon^2} \cdot T
+ Q
+ \sum_{v^{(t)}_i \neq v^{(t-1)}_i} (\frac{\|c\|_1}{n \epsilon^2} \sigma(\mV^{(t-1)} \mA)_i + c_i )
\right).
\end{align*}
\end{lemma}

\begin{proof}
The cost of \textsc{FindIndices} is dominated by $D_j.\textsc{HeavyQuery}(w)$ for $w = \mH \unitvec_k$, $k=1,...,r$,
and the computation of $\mH$.

\paragraph{Cost from QueryHeavy}

The cost of $D_j.\textsc{QueryHeavy}$ is bounded by
\begin{align*}
&~
\tilde{O}( \sum_{k\in[r]} \left( \|\mF\mZ^{-1/2}\mV^{(t-2^j)} \mA \mH \unitvec_k\|_c^2 \epsilon^{-2} + Q \right)) \\
=&~
\tilde{O}( \frac{\|c\|_1}{n\epsilon^2} \cdot \|\mF\mV^{(t-2^j)} \mA \mH\|_F^2 + Q) \\
=&~
\tilde{O}(
\frac{\|c\|_1}{n\epsilon^2} \cdot \left(2^{2j} + \sum_{j=t-2^j}^{t} \sum_{v^{(j)}_i \neq v^{(j+1)}_i} \sigma(\mV^{(j)} \mA)_i\right)
+ Q )
\end{align*}
where the last step bounds the frobenius-norm via \eqref{eq:ls:frobenius_bound}.
Note that this cost is paid every $2^j \le T$ iterations,
so we can interpret this as amortized cost for the $t$-th call
\begin{align*}
\tilde{O}\left(
\frac{\|c\|_1}{n\epsilon^2} \cdot \left(T + \sum_{v^{(t)}_i \neq v^{(t-1)}_i} \sigma(\mV^{(t-1)} \mA)_i\right)
+ Q \right).
\end{align*}

\paragraph{Cost from matrix $\mH$}

Fix some $j$ for which we have to compute $\mH$ in \Cref{line:ls:H_difference}.
Here the matrix $\mR$ has $O(\log n)$ columns,
so we must compute $O(\log n)$ products with $\omA := \mS \mV \mA$ 
and $\omA' := \mS (\mV - \Delta^{(j)}) \mA$,
and further solve $O(\log n)$ linear systems in
$\omA^\top \omA$ and $\omA'^\top \omA$.

Note that both $\omA$ and $\omA'$ have the same number of non-zero entries,
so the time for computing $\mH$ is bounded by
$\tilde{O}(\Psi + \nnz(\omA))$.
By definition of $\mS$, the number of non-zero entries in both matrices is bounded by
\begin{align*}
\nnz(\omA)
&=
\tilde{O}\left(
n \epsilon^{-2} (\max_i \nnz(a_i))
+ \sum_{i \in S_j} \nnz(a_i)
+ \sum_{i \in C_j} \nnz(a_i)
\right) \\
&=
\tilde{O}\left(
n \epsilon^{-2} (\max_i \nnz(a_i))
+ \sum_{k=t-2^j}^t \sum_{v^{(t)}_i \neq v^{(t-1)}_i} \nnz(a_i)
+ \sum_{k=t-2^j}^t \sum_{\osigma^{(t-1)}_i \neq \osigma^{(t-2)}_i} \nnz(a_i)
\right).
\end{align*}
As the cost for computing $\mH$ is paid once every $2^j$ calls to \textsc{FindIndices},
the total cost over all $T$ calls to \textsc{FindIndices} is bounded by
\begin{align*}
&~
\tilde{O}\left(
T\cdot\Psi
+ T n \epsilon^{-2} (\max_i \nnz(a_i))
+ \left(\sum_{t=1}^T \sum_{v^{(t)}_i \neq v^{(t-1)}_i} \nnz(a_i)\right)
+ \left(\sum_{t=1}^T \sum_{\osigma^{(t-1)}_i \neq \osigma^{(t-2)}_i} \nnz(a_i)\right)
\right)\\
=&~
\tilde{O}\left(
T\cdot\Psi
+ T n \epsilon^{-2} (\max_i \nnz(a_i))
+ \left(\sum_{t=1}^T \sum_{v^{(t)}_i \neq v^{(t-1)}_i} c_i\right)
+ \left(\frac{\|c\|_1}{n} \sum_{t=1}^T \sum_{\osigma^{(t-1)}_i \neq \osigma^{(t-2)}_i} (\sigma(\mV^{(t-2)} \mA)_i + z_i) \right)
\right)\\
=&~
\tilde{O}\left(
T\cdot\Psi
+ T n \epsilon^{-2} (\max_i \nnz(a_i))
+ \sum_{t=1}^T \sum_{v^{(t)}_i \neq v^{(t-1)}_i} (\frac{\|c\|_1}{n \epsilon} \sigma(\mV^{(t-1)} \mA)_i + c_i)
\right)
\end{align*}
by $\|c\|_1 z_i / n \ge c_i \ge \nnz(a_i)$ and \Cref{lem:ls:bound_output}.

\paragraph{Amortized Cost}

The amortized cost for the $t$-th call to \textsc{FindIndices} is thus
\begin{align*}
\tilde{O}\left(
\Psi
+ n \epsilon^{-2} (\max_i \nnz(a_i))
+ \frac{\|c\|_1}{n\epsilon^2} \cdot T
+ Q
+ \sum_{v^{(t)}_i \neq v^{(t-1)}_i} \frac{\|c\|_1}{n \epsilon^2} (\sigma(\mV^{(t-1)} \mA)_i + c_i )
\right).
\end{align*}

\end{proof}

We now charge the terms in the complexities of \textsc{FindIndices} and \textsc{UpdateIndices}
as amortized cost to \textsc{Scale} and \textsc{Query} to obtain the complexities as stated in \Cref{thm:ls:complexity}.

\begin{proof}[Proof of \Cref{thm:ls:complexity}]
We now charge each term of the amortized cost of \textsc{FindIndices} and \textsc{UpdateIndices}
to \textsc{Scale} and \textsc{Query}.

We charge the amortized cost of \textsc{UpdateIndices} (\Cref{lem:ls:amortized_updateindices})
as follows
$$
\tilde{O}\left(
\underbrace{
\Psi \epsilon^{-2}
+ \epsilon^{-4} n (\max_i \nnz(a_i))
+ T\frac{\|c\|_1}{n\epsilon^4}
}_{\textsc{Query}}
+ \underbrace{
\sum_{v^{(t)}_i \neq v^{(t-1)}_i} \frac{c_i}{\epsilon^2} + \frac{\|c\|_1}{n\epsilon^4} \sigma(\mV^{(t-1)} \mA)_i
}_{\textsc{Scale}}
\right).
$$

And for \textsc{FindIndices} (\Cref{lem:ls:amortized_findindices}) we charge
\begin{align*}
\tilde{O}&\left(
\underbrace{
\Psi
+ n \epsilon^{-2} (\max_i \nnz(a_i))
+ \frac{\|c\|_1}{n\epsilon^2} \cdot T
+ Q
}_{\textsc{Query}}
+
\underbrace{
\sum_{v^{(t)}_i \neq v^{(t-1)}_i} \frac{\|c\|_1}{n \epsilon^2} \sigma(\mV^{(t-1)} \mA)_i + c_i 
}_{\textsc{Scale}}
\right).
\end{align*}
Note that we reinitialize the data structure after $T$ iterations.
Thus the amortized complexity of \textsc{Query} also depends on the cost of initializing the data structure,
which we analyze next.

\paragraph{Initialization}
During initialization we have to compute $\exp(\pm\epsilon)$-approximate leverage scores.
It is known that this can be done in time $\tilde{O}(\nnz(\mA) + S)$,
where $S$ is the time required to solve a linear system in a matrix of the form $\mA^\top \mW \mA$
by using the Johnson-Lindenstrauss lemma \cite{jl84,SS08}.
By Condition \ref{item:ls:solver} of \Cref{thm:ls:complexity}
computing the leverage score thus takes $\tilde{O}(P + \Psi + \nnz(\mA))$ time.
We further initialize $\tilde{O}(1)$ instances of the assumed \textsc{HeavyHitter} data structure,
which takes $\tilde{O}(P)$ time.

Note that re-initializing the data structure every $T$ iterations
adds $\tilde{O}(\left(\Psi + \nnz(\mA) + P\right) / T)$ amortized cost to \textsc{Query}.

\paragraph{Query}
By collecting all the terms that we charge to \textsc{Query},
we obtain the following amortized cost per call to \textsc{Query}:
\begin{align*}
&~
\tilde{O}\left(
\Psi \epsilon^{-2}
+ 
\epsilon^{-4} n (\max_i \nnz(a_i))
+ 
T \frac{\|c\|_1}{n\epsilon^4}
+ 
Q
+ P / T
\right) \\
=&~
\tilde{O}\left(
\Psi \epsilon^{-2}
+ 
\epsilon^{-4} n (\max_i \nnz(a_i))
+ 
\epsilon^{-2}\sqrt{\frac{P\|c\|_1}{n}}
+ 
Q
\right)
\end{align*}
by choice of $T = \epsilon^2\sqrt{Pn/\|c\|_1}$.

\paragraph{Scale}
When calling \textsc{Scale}$(i,\cdot)$, the data structure calls $D_j.\textsc{Scale}(i,\cdot)$ for $j=1,...,\log n$
which takes $\tilde{O}(c_i)$ time.
Together with the amortized complexity terms that we charged to \textsc{Scale} 
we obtain an amortized complexity per call to \textsc{Scale} of
$$
\tilde{O}\left(
\frac{\|c\|_1}{n\epsilon^4} \sigma(\mV^{(t)} \mA)_i + \frac{c_i}{\epsilon^2}
\right).
$$
Here we used that the sum $\sum_{v^{(t)}_i \neq v^{(t-1)}_i} \sigma(\mV^{(t-1)} \mA)_i + c_i$
has one term for each call to \textsc{Scale}.

\end{proof}

\subsection{Stabilizer}
\label{sec:stabilizer}

In this section we want to prove \Cref{lem:stable_large_sequence}.
To prove the main lemma, we will first prove a weaker variant that involves only one step on $\mP$.
\begin{lemma}
\label{lem:one_step_P_frob}Let $\mA\in\R^{m\times n}$. Given $v,v'\in\R_{+}^{m}$
with $v\approx_{1/2}v'$. Let $\mF$ and $\mS$ be diagonal matrices
on $\R^{m\times m}$ such that $\mF\mV=\mF\mV'$, $(\mV-\mV')\mS=\mV-\mV'$,
$\|\mF\|_{2}\leq1$ and $\mA^{\top}\mS^{2}\mV^{2}\mA\approx_{1/2}\mA^{\top}\mV^{2}\mA$
and $\mA^{\top}\mS^{2}\mV'^{2}\mA\approx_{1/2}\mA^{\top}\mV'^{2}\mA$.
Let $\mP=\mV\mA(\mA^{\top}\mV^{2}\mS^{2}\mA)^{-1}\mA\mV$ and $\mP'=\mV'\mA(\mA^{\top}\mV'^{2}\mS^{2}\mA)^{-1}\mA\mV'$,
we have
\[
\|\mF(\mP'-\mP)\mS\|_{F}\lesssim\|\mV^{-1}(v'-v)\|_{\sigma(\mV\mA)}.
\]
\end{lemma}
\begin{proof}
Let $\Delta=\diag(\frac{v'-v}{v})$. Then, we have
\begin{align}
\mP'= & \Delta\mV\mA(\mA^{\top}\mV'^{2}\mS^{2}\mA)^{-1}\mA\mV'+\mV\mA(\mA^{\top}\mV'^{2}\mS^{2}\mA)^{-1}\mA\mV\Delta+\mV\mA(\mA^{\top}\mV'^{2}\mS^{2}\mA)^{-1}\mA\mV\label{eq:stable_P_diff}
\end{align}
By Woodbury matrix identity, we have 
\begin{align*}
(\mA^{\top}\mV'^{2}\mS^{2}\mA)^{-1} & =(\mA^{\top}\mV\mS(\mI+\Delta)^{2}\mS\mV\mA)^{-1}\\
 & =(\mA^{\top}\mV^{2}\mS^{2}\mA)^{-1}-(\mA^{\top}\mV^{2}\mS^{2}\mA)^{-1}\mA^{\top}\mV\mS((2\Delta+\Delta^{2})^{-1}+\mP)^{-1}\mS\mV\mA(\mA^{\top}\mV^{2}\mS^{2}\mA)^{-1}.
\end{align*}
Applying this into (\ref{eq:stable_P_diff}), we have 
\begin{align*}
\mP'-\mP= & \Delta\mV\mA(\mA^{\top}\mV'^{2}\mS^{2}\mA)^{-1}\mA\mV'+\mV\mA(\mA^{\top}\mV'^{2}\mS^{2}\mA)^{-1}\mA\mV\Delta+\mP\mS((2\Delta+\Delta^{2})^{-1}+\mP)^{-1}\mS\mP
\end{align*}
Since $\mF\mV=\mF\mV'$ and $(\mV-\mV')\mS=\mV-\mV'$, we have $\mF\Delta=0$
and $\Delta\mS=\Delta$. Together with $\|\mF\|_{2}\leq1$, we have
\begin{align*}
\|\mF(\mP'-\mP)\mS\|_{F}\leq & \|\mF\mV\mA(\mA^{\top}\mV'^{2}\mS^{2}\mA)^{-1}\mA\mV\Delta\|_{F}+\|\mF\mP\mS((2\Delta+\Delta^{2})^{-1}+\mP)^{-1}\mS\mP\mS\|_{F}\\
\leq & \|\mV\mA(\mA^{\top}\mV'^{2}\mS^{2}\mA)^{-1}\mA\mV\Delta\|_{F}+\|\mP\mS((2\Delta+\Delta^{2})^{-1}+\mP)^{-1}\mS\mP\mS\|_{F}.
\end{align*}
For the first term, we have 
\begin{align*}
\|\mV\mA(\mA^{\top}\mV'^{2}\mS^{2}\mA)^{-1}\mA\mV\Delta\|_{F}^{2} & =\Tr\Delta\mV\mA(\mA^{\top}\mV'^{2}\mS^{2}\mA)^{-1}\mA\mV^{2}\mA(\mA^{\top}\mV'^{2}\mS^{2}\mA)^{-1}\mA\mV\Delta\\
 & \lesssim\Tr\Delta\mV\mA(\mA^{\top}\mV'^{2}\mS^{2}\mA)^{-1}\mA\mV\Delta\\
 & \lesssim\Tr\Delta\mV\mA(\mA^{\top}\mV{}^{2}\mS^{2}\mA)^{-1}\mA\mV\Delta\\
 & =\|\mV^{-1}(v'-v)\|_{\sigma(\mV\mA)}^{2}
\end{align*}
where we used $\mA^{\top}\mV'^{2}\mS^{2}\mA\approx\mA^{\top}\mV'^{2}\mA\approx\mA^{\top}\mV{}^{2}\mA$
at the first and second inequality.

For the second term, we have
\begin{align*}
\mP\mS^{2}\mP & =\mV\mA(\mA^{\top}\mV^{2}\mS^{2}\mA)^{-1}\mA^{\top}\mS^{2}\mV^{2}\mA(\mA^{\top}\mV^{2}\mS^{2}\mA)^{-1}\mA^{\top}\mV=\mP.
\end{align*}
Hence, we have
\begin{align*}
\|\mP\mS((2\Delta+\Delta^{2})^{-1}+\mP)^{-1}\mS\mP\mS\|_{F}^{2} & \lesssim\|\mP\mS((2\Delta+\Delta^{2})^{-1}+\mP)^{-1}\mS\mP\|_{F}^{2}\\
 & \leq\Tr\mP\mS((2\Delta+\Delta^{2})^{-1}+\mP)^{-2}\mS\mP\\
 & \leq\Tr\mP\mS(2\Delta+\Delta^{2})^{2}\mS\mP\\
 & \leq9\Tr\mP\mS\Delta^{2}\mS\mP
\end{align*}
where we used $\Delta\preceq I$ at the last inequality. Finally,
using $\mS\Delta^{2}\mS=\Delta^{2}$, we have that the last term bounded
by $O(\|\mV^{-1}(v'-v)\|_{\sigma(\mV\mA)}^{2})$. Combining both terms
give the result.
\end{proof}

Now, we can prove the main statement.
\stableLargeSequence

\begin{proof}
Let $I$ be the set of indices $i$ such that $v_{i}^{(j)}$ changed
during the $j\in\{0,1,\cdots,T\}$. The proof involves defining a
new weight sequence 
\[
w_{i}^{(j)}=\begin{cases}
\ov_{i}^{(t)}v_{i}^{(0)}/\ov_{i}^{(0)} & \text{if }i\in I\\
v_{i}^{(0)} & \text{elses}
\end{cases}
\]
and $w_{i}^{(T+1)}=v_{i}^{(t)}$. We define
\[
\widetilde{\mP}^{(j)}=\mW^{(j)}\mA(\mA^{\top}\mW^{(j)2} \mS^2 \mA)^{-1}\mA\mW^{(j)}.
\]
Since $w_{i}^{(0)}=v_{i}^{(0)}$ and $w_{i}^{(T+1)}=v_{i}^{(T)}$,
we have 
\[
\|\mF(\mP^{(T)}-\mP^{(0)})\mS\|_{F}\leq\|\mF(\widetilde{\mP}^{(T+1)}-\widetilde{\mP}^{(0)})\mS\|_{F}\leq\sum_{j=1}^{T+1}\|\mF(\widetilde{\mP}^{(j)}-\widetilde{\mP}^{(j-1)})\mS\|_{F}.
\]
Now we check the conditions of Lemma \ref{lem:one_step_P_frob}. Since
$\ov^{(j)}\approx_{1/6}v^{(j)}$ and $\|(\widetilde{\mV}^{(j-1)})^{-1}(\ov^{(j)}-\ov^{(j-1)})\|_{\infty}\leq\frac{1}{6}$,
we have $w^{(j)}\approx_{1/2}w^{(j-1)}$ for all $j$. For any $i\in I$,
we have $\mF_{ii}=0$ and $\mS_{ii}=1$ (due to the condition $\mF\mV^{(j)}=\mF\mV^{(j-1)}$
and $(\mV^{(j)}-\mV^{(j-1)})\mS=\mV^{(j)}-\mV^{(j-1)}$). For $i\notin I$,
$w_{i}^{(j)}$ is a constant. Hence, this verifies the conditions of  Lemma \ref{lem:one_step_P_frob}: $\mF\mW^{(j)}=\mF\mW^{(j-1)}$,
$(\mW^{(j)}-\mW^{(j-1)})\mS=\mW^{(j)}-\mW^{(j-1)}$ and $\mA^{\top}(\mS^{(j-1)})^{2}(\mV^{(j-1)})^{2}\mA\approx_{1/2}\mA^{\top}(\mV^{(j-1)})^{2}\mA$.

For $j\leq T$, Lemma \ref{lem:one_step_P_frob} shows that
\begin{align*}
\|\mF(\widetilde{\mP}^{(j)}-\widetilde{\mP}^{(j-1)})\mS\|_{F} & \lesssim\|(\mW^{(j-1)})^{-1}(w^{(j)}-w^{(j-1)})\|_{\sigma(\mW^{(j-1)}\mA)}\\
 & =\|(\widetilde{\mV}^{(j-1)})^{-1}(\ov^{(j)}-\ov^{(j-1)})\|_{\sigma(\mW^{(j-1)}\mA)}\\
 & \lesssim\|(\widetilde{\mV}^{(j-1)})^{-1}(\ov^{(j)}-\ov^{(j-1)})\|_{\sigma(\widetilde{\mV}^{(j-1)}\mA)}\lesssim1
\end{align*}
where we used the formula of $w^{(j)}$ and $w^{(j-1)}$ in the first
equality, $w^{(j-1)}\approx_{1/6}\ov^{(j-1)}$ and the assumption
on $\ov$ at the last line.

For $j=T+1$, Lemma \ref{lem:one_step_P_frob} shows that
\begin{align*}
\|\mF(\widetilde{\mP}^{(T+1)}-\widetilde{\mP}^{(T)})\mS\|_{F}^{2} & \lesssim\|(\mW^{(T)})^{-1}(w^{(T+1)}-w^{(T)})\|_{\sigma(\mW^{(T)}\mA)}^{2}\\
 & \lesssim\sum_{w_{i}^{(j)}\neq w_{i}^{(j-1)}}\sigma_{i}(\mW^{(T)}\mA)\lesssim\sum_{i\in I}\sigma_{i}(\widetilde{\mV}^{(T)}\mA)
\end{align*}
where we used $\|(\mW^{(j-1)})^{-1}(w^{(j)}-w^{(j-1)})\|_{\infty} \lesssim 1$
in the second inequality, we used $\ov^{(T)}\approx_{1/6}w^{(T)}$ and
$I$ contains all changing indices at the end.

Combining both cases, we have
\[
\|\mF(\mP^{(T)}-\mP^{(0)})\mS\|_{F}^{2}\lesssim T^{2}+\sum_{i\in I}\sigma_{i}(\widetilde{\mV}^{(T)}\mA).
\]
Hence, it suffices to prove that
\[
\sum_{i\in I}\sigma_{i}(\widetilde{\mV}^{(T)}\mA)\lesssim\sum_{i\in I}\sigma_{i}(\widetilde{\mV}^{(t_{i})}\mA)+T^{2}.
\]
where $t_{i}$ be the time such that $v_{i}^{(t_{i})}\neq v_{i}^{(t_{i}-1)}$. To prove this, we define $\sigma_{i}^{(t_{i})}=\sigma_{i}(\widetilde{\mV}^{(t_{i})}\mA)$
and for all $j\geq t_{i}$, we define
\[
\sigma_{i}^{(j+1)}=\begin{cases}
\frac{\sigma_{i}(\widetilde{\mV}^{(j+1)}\mA)}{\sigma_{i}(\widetilde{\mV}^{(j)}\mA)}\cdot\sigma_{i}^{(j)} & \left|\frac{\sigma_{i}(\widetilde{\mV}^{(j+1)}\mA)-\sigma_{i}(\widetilde{\mV}^{(j)}\mA)}{\sigma_{i}(\widetilde{\mV}^{(j)}\mA)}\right|\geq\frac{1}{T}\\
\sigma_{i}^{(j)} & \text{elses}
\end{cases}.
\]
Intuitively, $\sigma_{i}^{(j)}$ is essentially the same as $\sigma_{i}(\widetilde{\mV}^{(j)}\mA)$
except we ignore all smaller than $\frac{1}{T}$ relative movement.
Since there are only $T$ steps, we have $\sigma_{i}^{(j)}\approx_{2}\sigma_{i}(\widetilde{\mV}^{(j)}\mA)$.
By the assumption $\|(\widetilde{\mV}^{(j-1)})^{-1}(\ov^{(j)}-\ov^{(j-1)})\|_{\sigma(\widetilde{\mV}^{(j-1)}\mA)+\infty}\leq\frac{1}{6}$
and \cite[Lemma 14]{LeeS15}, we have that 
\[
\|\sigma(\widetilde{\mV}^{(j-1)}\mA)^{-1}(\sigma(\widetilde{\mV}^{(j)}\mA)-\sigma(\widetilde{\mV}^{(j-1)}\mA))\|_{\sigma(\widetilde{\mV}^{(j-1)}\mA)}\lesssim1.
\]
Using the definition of $\sigma_{i}^{(j)}$, we have a similar bound
for $\sigma_{i}^{(j)}$:
\[
\|(\sigma^{(j-1)})^{-1}(\sigma^{(j)}-\sigma^{(j-1)})\|_{\sigma^{(j-1)}}\lesssim1.
\]
Since $\sigma^{(j)}$ only makes relative movement at least $\frac{1}{T}$,
this implies $\sum_{i}|\sigma_{i}^{(j)}-\sigma_{i}^{(j-1)}|\lesssim T$.
Hence, we have
\[
\sum_{i\in I}\sigma_{i}(\widetilde{\mV}^{(T)}\mA)\lesssim\sum_{i\in I}\sigma_{i}^{(T)}\lesssim\sum_{i\in I}\sigma_{i}^{(t_{i})}+T^{2}\lesssim\sum_{i\in I}\sigma_{i}(\widetilde{\mV}^{(t_{i})}\mA)+T^{2}.
\]
\end{proof}

\section{Primal and Gradient Maintenance}
\label{sec:gradient_maintenance}

\gradientMaintenance
This is almost exactly Theorem 7.1 in~\cite{BrandLN+20} except here we include an additional per-coordinate accuracy parameter $w$, and in $\textsc{QuerySum}$ the guarantee becomes $\|w^{-1}(\ox - x^{(t)})\|_\infty \le \epsilon$ instead of $\|\ox^{-1}(\ox - x^{(t)})\|_\infty \le \epsilon$. The implementation and analysis of the data strucutre largely follows from~\cite{BrandLN+20}, and we include it for completeness. The main idea is to maintain a $O(\epsilon^{-2}\log n)$-dimensional approximation $\nabla\Psi(\oz)^{\flat(\otau)}$ of the $m$-dimensional exact gradient $\nabla\Psi(z)^{\flat(\ttau)}\in \R^m$ by slightly perturbing $\ttau$ and $z$. The approximation $\nabla\Psi(\oz)^{\flat(\otau)}$ is formally still $m$-dimensional, but we say $O(\epsilon^{-2} \log n)$ dimension in the sense that its $m$ entries
can be put into $O(\epsilon^{-2} \log n)$ buckets, and entries in the same bucket share a common value. The proof of the theorem follows from two sub data structures which we specify below. The first lemma addresses the construction and maintenance of the low dimensional approximation of the exact gradient, and we use the corresponding result from~\cite{BrandLN+20} without modification.
\begin{lemma}[{\cite[Lemma 7.2]{BrandLN+20}}]
\label{lem:gradient_reduction} 
There exists a deterministic data-structure
that supports the following operations
\begin{itemize}
\item $\textsc{Initialize }(\mA\in\R^{m\times n},g\in\R^{m},\ttau\in\R^{m},z\in\R^{m},\epsilon>0)$:
	The data-structure preprocesses the given matrix $\mA\in\R^{m\times n}$,
	vectors $g,\ttau,z\in\R^{m}$, and accuracy parameter $\epsilon>0$
	in $O(\nnz(\mA))$ time. 
	The data-structure assumes $0.5\le z\le2$ and $n/m\le\ttau\le2$.
	The output is a partition 
	$\bigcup_{k=1}^K I_k = [m]$ 
	with $K = O(\epsilon^{-2} \log n)$.
\item $\textsc{Update}(i \in [m], a \in \R, b \in \R, c \in \R)$: 
	Sets $g_{i}=a$, $\ttau_{i}=b$ and $z_i=c$ in $O(\nnz(a_i))$ time.
	The data-structure assumes $0.5\le z\le2$ and $n/m\le\ttau\le2$.
	The index $i$ might be moved to a different set, so the data-structure returns $k$ such that $i \in I_k$.
\item $\textsc{Query}()$: 
	Returns $\mA^{\top}\mG\nabla\Psi(\oz)^{\flat(\otau)} \in \R^n$ for some $\otau \in \R^m$, 
	$\oz \in \R^m$ with $\otau \approx_\epsilon \ttau$ 
	and $\|\oz-z\|_{\infty}\le \epsilon$,
	where $x^{\flat(\otau)} := \argmax_{\|w\|_{\otau + \infty} \le 1} \langle x, w \rangle$.
	The data-structure further returns the low dimensional representation $s\in\R^K$ such that
	\begin{align*}\sum_{k=1}^K s_k \mathbf{1}_{i \in I_k} = \left( \nabla\Psi(\oz)^{\flat(\otau)}\right)_i\end{align*}
	for all $i \in [m]$,
	in $O(n\epsilon^{-2}\log n )$ time. 
\item $\textsc{Potential}()$
	Returns $\Psi(z)$ in $O(1)$ time.
\end{itemize}
\end{lemma}
The next lemma maintains the desired sum for $\textsc{QuerySum}$, which we need to slightly modify the corresponding result (Lemma 7.6 in~\cite{BrandLN+20}) to accomodate our per-coordinate accuracy requirement. The implementation of the data structure for the lemma is given in~\Cref{alg:gradient_accumulator}. The input to the data structure is the $O(\epsilon^{-2} \log n)$ dimensional representations of the gradients of each iteration, which is computed by the data structure of the previous lemma.
Here for a set $I \subset [m]$ we also use $I$ as $0/1$-vector with $I_{i} = 1$ when $i \in I$ and $I_i = 0$ otherwise.
\begin{lemma}
\label{lem:gradient_accumulator} 
There exists a deterministic data-structure
that supports the following operations
\begin{itemize}
\item $\textsc{Initialize }(x^{\init}\in\R^m,g\in\R^m,(I_k)_{1\le k \le K},\epsilon\in (0,1]^m)$:
	The data-structure initialized on the given vectors $x^{\init},g\in\R^m$,
	the partition $\bigcup_{k=1}^K I_k = [m]$ where $K=O(\epsilon^{-2}\log n)$, and the per-coordinate accuracy parameter $\epsilon \in (0,1]^m$	in $O(m)$ time.
\item $\textsc{Scale}(i \in [m], a \in \R)$: 
	Sets $g_{i}\leftarrow a$ in $O(\log n)$ amortized time.
\item $\textsc{Move}(i \in [m],k\in[1,K])$: 
	Moves index $i$ to set $I_k$ in $O(\log n)$ amortized time.
\item $\textsc{SetAccuracy}(i \in [m],\delta\in (0,1])$: 
	Sets $\epsilon_i\leftarrow \delta$ in $O(\log n)$ amortized time.
\item $\textsc{Query}(s \in \R^K, h \in \R^m)$: 
	Let $g^{(\ell)}$ and $\epsilon^{(\ell)}$ be the state of vector $g$ and $\epsilon$ respectively during the $\ell$-th call to \textsc{Query} 
	and let $s^{(\ell)}$ and $h^{(\ell)}$ be the input arguments of the respective call. The vector $h$ will always be provided as a sparse vector so that we know where are the non-zeros in the vector. Define $y^{(\ell)} = \mG^{(\ell)} \sum_{k=1}^K I_k^{(\ell)} s_k^{(\ell)}$ and $x^{(t)} = x^{\init} + \sum_{\ell=1}^t h^{(\ell)} + y^{(\ell)}$,
	then the $t$-th call to \textsc{Query} returns a vector $\ox$ satisfying $ |\ox_i- x^{(t)}_i|\leq \epsilon^{(t)}_i$ for all $i\in[m]$.
	After $T$ calls to \textsc{Query}, the total time of all $T$ calls is bounded by
	\begin{align*}O\left(
		TK+\log n\cdot \sum_{\ell=0}^T \|h^{(\ell)}\|_0 + T \log n\cdot \sum_{\ell=1}^T \|y^{(\ell)} / \epsilon^{(\ell-1)}\|_2^2 \right).\end{align*}
	The vector $\ox \in \R^m$ is returned as a pointer 
	and additionally a set $J \subset [m]$ is returned 
	that contains the indices where $\ox$ changed compared to the result of the previous \textsc{Query} call.
\item $\textsc{ComputeExactSum}()$:
	Returns the current exact vector $x^{(t)}$ in $O(m\log n)$ time.
\end{itemize}
\end{lemma}

\begin{algorithm2e}[p!]
\caption{Algorithm for accumulating $\mG \nabla\Psi(\ov)^{\flat}$ 
(\Cref{lem:gradient_accumulator}) \label{alg:gradient_accumulator}}
\SetKwProg{Members}{members}{}{}
\SetKwProg{Proc}{procedure}{}{}
\SetKwProg{Priv}{private procedure}{}{}
\Members{}{
	$I_1,...,I_K$ \tcp*{Partition $\bigcup_k I_k = [m]$}
	$t \in \N, \ox \in \R^m$ \tcp*{\textsc{Query} counter and approximation of $x^{(t)}$}
	$\hat{\ell}\in\N^{m}$ \tcp*{$\hat{\ell}_i$ is value of $t$ when we last update $\ox_i \leftarrow x_i$}
	$f^{(t)} \in \R^K$ \tcp*{Maintain $f^{(t)}=\sum_{k=1}^t s^{(k)}$}
	$\Delta^{(high)},\Delta^{(low)} \in \R^m$ \tcp*{Maintain $\Delta_i = f^{(\hat{\ell}_i)}_k\pm |\epsilon_i / (10 g_i)| $ if $i\in I_k$}
}
\Proc{\textsc{Initialize}$(x^{\init} \in \R^m, g \in \R^m, (I_k)_{1\le k \le K}, \epsilon \in (0,1]^m)$}{
	$\ox \leftarrow x^{\init}$, 
	$(I_k)_{1\le k \le K} \leftarrow (I_k)_{1\le k \le K}$,
	$t \leftarrow 0$,
	$f^{(0)} \leftarrow \zerovec_K$,
	$g \leftarrow g$,
	$\epsilon \leftarrow \epsilon$
}
\Priv{\textsc{ComputeX}$(i,h_i)$}{
	Let $k$ be such that $i \in I_k$ \\
	$\ox_i \leftarrow \ox_i + g_i \cdot (f^{(t)}_k -  f^{(\hat{\ell}_i)}_k)+ h_i$  \label{line:update_bin}\\
	$\hat{\ell}_i \leftarrow t$ \\
	$J \leftarrow J \cup \{ i \}$\\
}
\Priv{\textsc{UpdateDelta}$(i)$}{
	Let $k$ be such that $i \in I_k$. \\
	$\Delta^{(high)}_i \leftarrow  f^{(\hat{\ell}_i)}_k+|\epsilon_i / (10 g_i)| $ \\
	$\Delta^{(low)}_i \leftarrow  f^{(\hat{\ell}_i)}_k -|\epsilon_i / (10 g_i)| $
}
\Proc{\textsc{Move}$(i\in[m], k)$}{
	\textsc{ComputeX}$(i,0)$ \\
	Move index $i$ to set $I_k$ \\
	$\textsc{UpdateDelta}(i)$
}
\Priv{\textsc{Scale}$(i, a)$}{
	\textsc{ComputeX}$(i,0)$ \\
	$g_i \leftarrow a$ \\
	$\textsc{UpdateDelta}(i)$
}
\Priv{\textsc{SetAccuracy}$(i, \delta)$}{ %
	\textsc{ComputeX}$(i,0)$ \\
	$\epsilon_i \leftarrow \delta$ \\
	$\textsc{UpdateDelta}(i)$
}
\Proc{\textsc{Query}$(s \in \R^K, h \in \R^m)$}{
	$t \leftarrow t + 1$,	$J \leftarrow \emptyset$ \tcp*{Collect all entries that have changed since the last call to \textsc{Query}} 
	$f^{(t)} \leftarrow f^{(t-1)} + s$\label{line:group_bin_ft} \\
	\For{$i$ such that $h_i \neq 0$}{
		\textsc{ComputeX}$(i,h_i)$, $\textsc{UpdateDelta}(i)$ \label{line:nonzero_hi}
	}
	\For{$k=1,\ldots,K$}{
		\For{$i \in I_k$ with $f^{(t)}_k > \Delta^{(high)}_i$ or $f^{(t)}_k < \Delta^{(low)}_i$}{
			\textsc{ComputeX}$(i,0)$, $\textsc{UpdateDelta}(i)$ \label{line:large_change_x_f}
		}
	}
	\Return $\ox$, $J$
}
\Proc{\textsc{ComputeExactSum}$()$}{
	\For{$i\in[m]$ and $\hat{\ell}_i<t$}{
		$\textsc{ComputeX}(i,0)$, $\textsc{UpdateDelta}(i)$
	}
	\Return $\ox$
}
\end{algorithm2e}
The proof follows a trivial adaptation of the proof for Lemma 7.6 in~\cite{BrandLN+20} and we reproduce it below for completeness.
\begin{proof}[Proof of \Cref{lem:gradient_accumulator}]
	We start by analyzing the correctness.
	\paragraph{Invariant}
	Let $s^{(t)},h^{(t)},g^{(t)},I_k^{(t)},\epsilon^{(t)}$ be the state of $s,h,g,I_k,\epsilon$ during the $t$-th call to \textsc{Query}
	and by definition of $x^{(t)}$ we have for any index $i$ that
	\begin{align*}
	x^{(t)}_i = x^{\init}_i +\sum_{\ell=1}^{t}g^{(\ell)}_i \left(\sum_{k=1}^K  s^{(\ell)}_k \mathbf{1}_{i\in I_k^{(\ell)}}\right) + h^{(\ell)}_i.
	\end{align*}
	It is easy to check that $\hat{\ell}_i$ always store the most recent iteration when $\ox_i$ is updated by \textsc{ComputeX}$(i,h_i)$. We first prove by induction that this update is always calculated correct, that is, the data-structure maintains the invariant
	$\ox_i = x^{(\hat{\ell}_i)}_i$.

We see from \Cref{line:group_bin_ft} that the data-structure maintains $f^{(t)} = \sum_{k=1}^t s^{(k)}$.
	Further note that \textsc{ComputeX}$(i,h_i)$ is called 
	whenever $h_i$ is non-zero (\Cref{line:nonzero_hi}), $g_i$ or $\epsilon_i$ is changed, or $i$ is moved to a different $I_k$. Thus if $\hat{\ell}_i < t$, we know none of these events happened during iteration $\ell\in (\hat{\ell}_i,t]$ and the only moving part is the $s^{(\ell)}$'s over these iterations, which is exactly $f^{(t)}-f^{(\hat{\ell}_i)}$.
	Thus, if $k$ is the set $I_k$ where $i$ belongs to over iterations $(\hat{\ell}_i,t]$,   the execution of \Cref{line:update_bin} gives
	\begin{align*}
	&~
	g_i \cdot (f^{(t)}_k - f^{(\hat{\ell}_i)}_k) + h^{(t)}_i
	=
	g_i \sum_{\ell=\hat{\ell}_i+1}^t s^{(\ell)}_k + h^{(t)}_i 
	=
	h^{(t)}_i + \sum_{\ell=\hat{\ell}_i+1}^t g^{(\ell)}_i s^{(\ell)}_k \\
	=&~
	h^{(t)}_i + \sum_{\ell=\hat{\ell}_i+1}^t g^{(\ell)}_i \left(\sum_{k=1}^K  s^{(\ell)}_k \mathbf{1}_{i\in I_k^{(\ell)}}\right) 
	=
	\sum_{\ell=\hat{\ell}_i+1}^t \left(g^{(\ell)}_i \left(\sum_{k=1}^K  s^{(\ell)}_k \mathbf{1}_{i\in I_k^{(\ell)}}\right) + h^{(\ell)}_i\right)
	\end{align*}
	where the first equality uses $ f^{(t)}=\sum_{\ell=1}^t s^{(\ell)}$ and the second equality uses $g^{(\ell)}_i=g^{(t)}_i$ for all $\hat{\ell}_i < \ell \le t$,
	because \textsc{ComputeX}$(i,h_i)$ is called whenever $g_i$ is changed.
	The third equality is because \textsc{ComputeX}$(i,h_i)$ is called whenever $i$ is moved to a different set, 
	so $i \in I_k^{(\ell)}$ for the same $k$ for all $\hat{\ell}_i < \ell \le t$.
	The last equality uses $h^{(\ell)}_i = 0$ for $\hat{\ell}_i < \ell < t$, because \textsc{ComputeX}$(i,h_i)$ is called whenever $h_i$ is non-zero.
	Thus by induction over the number of calls to \textsc{ComputeX}$(i,h_i)$, when $\hat{\ell}_i$ is increased to $t$ we have
	\begin{align*}
	\ox_i = 	
	x^{(t)}_i = x^{\init}_i +\sum_{\ell=1}^{t} \left(g^{(\ell)}_i \left(\sum_{k=1}^K  s^{(\ell)}_k \mathbf{1}_{i\in I_k^{(\ell)}}\right) + h^{(\ell)}_i\right)	 = x^{(t)}_i,
	\end{align*}
	so the invariant is always maintained.
	\paragraph{Correctness of Query.}
	We claim that the function \textsc{Query} returns a vector $\ox$ such that for all $i$ 
	\begin{align*}\ox_i \in [x^{(t)}_i\pm \epsilon^{(t)}_i] := \left[x^{\init}_i + \sum_{\ell=1}^{t} \left(g^{(\ell)}_i \left(\sum_{k=1}^K  s^{(\ell)}_k \mathbf{1}_{i\in I_k^{(\ell)}}\right) + h^{(\ell)}_i \right) \pm \epsilon^{(t)}_i\right].\end{align*}
	Given the invariant discussed above, we only need to guarantee \textsc{ComputeExact}$(i,h_i)$ is called whenever the approximation guarantee is violated for some $i$. Moreover, same as when we proved the invariant above, we only need to guarantee this in the case that since iteration $\hat{\ell}_i$, $h_i$ is always $0$, $g_i,\epsilon_i$ remain constant and $i$ remains in the same $I_k$ for some $k$. Thus, the task is equivalent to detect whenever 
	\begin{align*}
	\left| \sum_{\ell=\hat{\ell}_i+1}^{t} g^{(\ell)}_i s^{(\ell)}_k \right| = |g_i \cdot (f^{(t)}_k - f^{(\hat{\ell}_i)}_k)|> \frac{\epsilon^{(t)}_i}{10}.
	\end{align*}
	which is the same as 
	\begin{align*}
	f^{(t)}_k \notin [f^{(\hat{\ell}_i)}_k - |\epsilon_i/(10 g_i)|,f^{(\hat{\ell}_i)}_k + |\epsilon_i/(10 g_i)|] 
	\end{align*}
	Note that the lower and upper limits in the above range are exactly $\Delta^{(high)}_i$ and $\Delta^{(low)}_i$	as maintained by \textsc{UpdateDelta}$(i)$, which will be called whenever any of the terms involved in the calculation of these limits changes.
	Thus \Cref{line:large_change_x_f} makes sure that we indeed maintain
	\begin{align*}|\ox_i -  x^{(t)}_i|\leq \epsilon_i \qquad \forall i.\end{align*} Also it is easy to check the returned set $J$ contains all $i$'s such that $\ox_i$ changed since last \textsc{Query}.
	
	\paragraph{Complexity}
	The call to \textsc{Initialize} takes $O(m+K)$ as we initialize a constant number of $K$ and $m$ dimensional vectors, and this reduces to $O(m)$ since there can be at most $m$ non-empty $I_k$'s.
	A call to \textsc{ComputeX} takes $O(1)$ time.
	
	To implement \Cref{line:large_change_x_f} efficiently without enumerating all $m$ indices, we maintain for each $k\in[K]$ two sorted lists of the $i$'s in $I_k$, 
	sorted by $\Delta_i^{(high)}$ and $\Delta_i^{(low)}$ respectively.
	Maintaining these sorted lists results in $O(\log n)$ time per call for \textsc{UpdateDelta}.
	Hence \textsc{Move}, \textsc{Scale} and \textsc{SetAccuracy} also run in $O(\log n)$ time.
	To implement the loop for \Cref{line:large_change_x_f} we can go through the two sorted lists in order, but stop as soon as the check condition no longer holds. This bounds the cost of the loop by $O(K)$ plus $O(\log n)$ times the number of indices $i$ satisfying $f^{(t)}_k > \Delta^{(high)}_i$ or $f^{(t)} < \Delta^{(low)}_i$,
	i.e. $|f^{(t)}_k - f^{(\hat{\ell}_i)}_k| > \Theta(\epsilon^{(\hat{\ell}_i)}_i/g^{(\hat{\ell}_i)}_i)$.
	Note if a \textsc{ComputeX} and \textsc{UpdateDelta} is triggered by this condition for any $i$, $h_i$ must be $0$ during $(\hat{\ell}_i,t]$ iterations. Thus, let $z^{(t)} := x^{\init} + \sum_{\ell=1}^t \mG^{(\ell)} \sum_k I_k^{(\ell)} s_k^{(\ell)}$, we can rewrite that condition as
	$|z^{(t)}_i - z^{(\hat{\ell}_i)}_i| > \Theta(|\epsilon^{(\hat{\ell}_i)}_i|)$.
	Throughout $T$ calls to \textsc{Query}, 
	we can bound the total number of times where $i$ satisfies 
	$|z^{(t)}_i - z^{(\hat{\ell}_i)}_i| > \Theta(|\epsilon^{(\hat{\ell}_i)}_i|)$
	by
	\begin{align*}
	O\left(T \sum_{\ell=1}^T \| \mG^{(\ell)}(\sum_k I_k^{(\ell)} s_k^{(\ell)}) / \epsilon^{(\ell-1)} \|_2^2 \right).
	\end{align*}
	The number of times \textsc{ComputeX} and \textsc{UpdateDelta} are triggered due to $h^{(t)}_i\neq 0$ is $\|h^{(t)}\|_0$ each iteration, and updating $f^{(t)}$ takes $O(K)$ time. So the total time for $T$ calls to \textsc{Query} can be bounded by
	\begin{align*}
	O(TK + \log(n)\cdot\sum_{\ell=0}^T \|h^{(\ell)}\|_0 + \log(n) \cdot T\sum_{\ell=1}^T \|\mG^{(\ell)} (\sum_k s_k^{(\ell)}I_k^{(\ell)}) / \epsilon^{(\ell-1)}\|_2^2 / \epsilon^2).
	\end{align*}
	The time for \textsc{ComputeExactSum} is $O(m\log n)$ since it just calls \textsc{ComputeX} and \textsc{UpdateDelta} on all $m$ indices.
\end{proof}
\begin{proof}[Proof of \Cref{thm:gradient_maintenance}] 
The data-structure for \Cref{thm:gradient_maintenance} follows directly by combining
\Cref{lem:gradient_reduction} and \Cref{lem:gradient_accumulator}.
The result for \textsc{QueryProduct} is obtained from \Cref{lem:gradient_reduction},
and the result for \textsc{QuerySum} is obtained from \Cref{lem:gradient_accumulator} using the vector $s \in \R^K$ returned by \Cref{lem:gradient_reduction} 
as input to \Cref{lem:gradient_accumulator} and $\epsilon_i$ being $w_i\epsilon$. Note we charge the cost incurred by calling the data structure of \Cref{lem:gradient_accumulator} to the \textsc{QuerySum} complexity. 
\end{proof}

\clearpage
\newcommand{\ProjHeavyHitter}{\textsc{HeavyHitter}\xspace}
\section{Dual Slack Maintenance}
\label{sec:dual_maintenance}

\dualSlackMaintenance
The data structure for the theorem above is given in~\Cref{alg:dual_maintenance}. Both the implementation and analysis follow straightforward adaptations of Algorithm 4 and Theorem 6.1 in~\cite{BrandLN+20}, and we reproduce the analysis below for completeness. Our new version makes the formal reduction to HeavyHitters (\Cref{def:heavyhitter}) more clear and adds the \textsc{SetAccuracy} method.
\begin{algorithm2e}[p!]
\caption{\label{alg:dual_maintenance}Algorithm for \Cref{thm:dual_maintenance}} %
\SetKwProg{Members}{members}{}{}
\SetKwProg{Proc}{procedure}{}{}
\SetKwProg{PrivateProc}{private procedure}{}{}
\Members{}{
$\hat{f}, \ov\in \R^m, w\in\R^m, t\in \N$\\
$D_j, T = \sqrt{nP/\|z\|_1}$ \Comment{$D_j$ are $(P,z,Q)$-\textsc{HeavyHitter} (\Cref{def:heavyhitter})} \\
$f^{(j)} \in \R^n$ and $F_j\subset [m]$ for $0\le j\le \log T$ \\
}
\Proc{\textsc{Initialize}$(\mA, v^{\init}, w^{\init},\epsilon$)}{
	$\ov \leftarrow v^{\init}$, $\hat{f} \leftarrow \zerovec_n$, $w \leftarrow w^{\init}$,
	$t \leftarrow 0$\\
	\For{$j=0,\dots,\log T$}{
		$D_j.\textsc{Initialize}(\mA,w^{-1})$ (\Cref{def:heavyhitter})\\
		$f^{(j)} \leftarrow \zerovec_n$,
		$F_j\leftarrow \emptyset$\\
	}
	\Return $\mA v^{\init}$
}
\PrivateProc{\textsc{FindIndices}$(h\in\R^n)$}{
	$I \leftarrow \emptyset$\\
	\For{$j=\log T,...,0$}{
		$f^{(j)} \leftarrow f^{(j)} + h$ 
		\Comment{When $2^j | t$, then $f^{(j)} = \sum_{k=t-2^j+1}^t h^{(k)}$}\\
		\If{$2^j | t$}{
			$I \leftarrow I \cup D_j.\textsc{HeavyQuery}(f^{(j)}, 0.2 \epsilon/\log n)$ \label{line:detect_changes}\\
			$f^{(j)} \leftarrow \zerovec_n$
		}
	}
	\Return $I$
}
\Proc{\textsc{SetAccuracy}$(i, \delta$)}{%
	$w_i \leftarrow \delta$ \\
	\lFor{$j=0,...,\log T$}{
		$F_j \leftarrow F_j \cup \{ i \}$, $D_j.\textsc{Scale}(i, 0)$
	}
}
\PrivateProc{\textsc{VerifyIndex}$(i)$}{
	\If{$|\ov_i - (v^\init + \mA \hat{f})_i| \ge 0.2 w_i \epsilon /\log n$}{ \label{line:check_change}
		$\ov_i \leftarrow (v^\init + \mA \hat{f})_i$ \label{line:set_vi}\\
		\For{$j=0,...,\log T$}{
			$F_j \leftarrow F_j \cup \{ i \}$ \Comment{Notify other $D_j$'s to stop tracking $i$.}\\
			$D_j.\textsc{Scale}(i, 0)$ \label{line:reinsert_i} 
		}
		\Return True
	}
	\Return False
}
\Proc{\textsc{Add}$(h\in\R^n)$}{
	\lIf{$t = T$}{\Return $\textsc{Initialize}(\mA, \mA (\hat{f}+h), w, \epsilon)$}%
	$t \leftarrow t + 1$, $\hat{f} \leftarrow \hat{f} + h$,
	$I \leftarrow \textsc{FindIndices}(h)$ \label{line:dual:findIndices}\\
	$I \leftarrow \left\{i|i\in I \mbox{ and }\textsc{VerifyIndex}(i)=True\right\}$ \label{line:dual:verify_I}\\
	\lFor{$j:2^j | t$}{
		$I \leftarrow I \cup \left\{i|i\in F_j \mbox{ and }\textsc{VerifyIndex}(i)=True\right\}$ \label{line:verify_F}
	}
	\For{$j:2^j | t$}{
		\For{$i\in I \cup F_j$}{
			$D_j.\textsc{Scale}(i, 1/w_i)$  \label{line:reweight}
		}
		$F_j \leftarrow \emptyset$ \label{line:empty_F}\\
	}
	\Return $I$, $\ov$
}
\Proc{\textsc{ComputeExact}$()$}{
	\Return $v^\init + \mA \hat{f}$
}
\end{algorithm2e}
Throughout this section we denote $h^{(t)}$ the input vector $h$ 
of the $t$-th call to \textsc{Add} (or equivalently referred to as the $t$-th iteration),
and let $v^{(t)} = v^\init + \mA \sum_{k=1}^t h^{(k)}$ be the state of the exact solution $v$ 
(as defined in \Cref{thm:dual_maintenance}) 
for the $t$-th call to \textsc{Add}.

In our algorithm (see \Cref{alg:dual_maintenance}) we maintain a vector $\hat{f}$ 
which is the sum of all past input vectors $h$, 
so we can retrieve the exact value of $v^{(t)}_i = v^\init_i + (\mA \hat{f})_i$ for any $i$ efficiently.
This value is computed and assigned to $\ov_i$ whenever the approximation $\ov$ that we maintain no longer satisfies the error guarantee for some coordinate $i$. 
As to how we detect when this may happen, we know the difference between $v^{(t)}$ and the state of $v$ at an earlier $t'$-th \textsc{Add} call is
\[
v^{(t)} - v^{(t')} = \mA \left(\sum_{k=t'+1}^t h^{(t)}\right),
\]
and thus we can detect all coordinates $i$ 
that changes above certain threshold from $t'$ to $t$-th \textsc{Add} call 
using the $(P,z,Q)$-HeavyHitter data structure of \Cref{def:heavyhitter}
(by querying it with $\sum_{k=t'+1}^t h^{(t)}$ as the parameter $h$). 
Note since the error guarantee we want is multiplicative in $w$ 
(i.e., $\ov^{(t)}_i \in v^{(t)}_i\pm \epsilon w_i$ for all $i$), 
while the threshold $\epsilon$ in \Cref{def:heavyhitter} is absolute and uniform, 
we give $w^{-1}$ as the scaling vector to the \ProjHeavyHitter data structure to accommodate this. 

Since the most recent updates on $\ov_i$ for different indices $i$'s happen at different iterations, 
we need to track accumulated changes to $v_i$'s over different intervals 
to detect the next time an update is necessary for each $i$. 
Thus, it is not sufficient to just have one copy of the \ProjHeavyHitter. 
On the other hand, keeping one individual copy of the \ProjHeavyHitter for each $0 \le t' < t$ will be too costly in terms of running time. 
We handle this by instantiating $\log T$ copies of the \ProjHeavyHitter data structure $D_j$ for $j=0,...,\log T$ where $T=\sqrt{nP/\|z\|_1}$, and each copy takes charge of batches with increasing number of iterations. 
In particular, the purpose of $D_j$ is to detect all coordinates $i$ in $v$ 
with large accumulated change over batches of $2^j$ iterations 
(see how we update and reset $f^{(j)}$ in \textsc{FindIndices} in \Cref{alg:dual_maintenance}). 
Each $D_j$ has its local copy of a scaling vector, which is initialized to be $w^{-1}$, 
we refer to it as $\hat{g}^{(j)}$, and the cost to query $D_j$ is proportional to $\|\hat{\mG}^{(j)} \mA f^{(j)}\|_z^2+Q$. 
Note $f^{(j)}$ accumulates updates over $2^j$ iterations, 
and $\|\sum_{k=1}^{2^j} h^{(k)}\|_z^2$ can be as large as $2^j\sum_{k=1}^{2^j} \|h^{(k)}\|_z^2$. 
Since we want to bound the cost of our data structure by the sum of the squares of updates 
(which can in turn be bounded by our IPM method) 
instead of the square of the sum of updates, 
querying $D_j$ incurs an additional $2^j$ factor overhead. 
Thus for efficiency purposes, if $v_i$ would take much less than $2^j$ iterations to accumulate a large enough change,
we can safely let $D_j$ stop tracking $i$ during its current batch 
since $v_i$'s change would have been detected by a $D_{j'}$ of appropriate (and much smaller) $j'$ 
so that $\ov_i$ would have been updated to be the exact value 
(see implementation of \textsc{VerifyIndex}). 
Technically, we keep a set $F_j$ to store all indices $i$ that $D_j$ stops tracking for its current batch of iterations and set $\hat{g}^{(j)}_i$ to $0$ so we don't pay for coordinate $i$ when we query $D_j$. Also note that whenever the accuracy requirement $w_i$ for coordinate $i$ changes, we add $i$ to all $F_j$'s (see implementation of \textsc{SetAccuracy}), and this will make sure we explicitly check whether $\ov_i$ is within the approximation requirement when \textsc{Add} is called since we will call \textsc{VerifyIndex} on $i$ (see \Cref{line:dual:verify_I}).
At the start of a new batch of $2^j$ iterations for $D_j$, we add back all indices in $F_j$ to $D_j$ (\Cref{line:reweight}) and reset $F_j$. 
As a result, only those $i$'s that indeed would take (close to) $2^j$ iterations 
to accumulate a large enough change are necessary to be tracked by $D_j$, 
so we can query $D_j$ less often for large $j$ to offset its large cost. 
In particular, we query each $D_j$ every $2^j$ iterations (see \Cref{line:detect_changes}). 

We start our formal analysis with the following lemma, which adapts Lemma 6.2 in~\cite{BrandLN+20}. 
\begin{lemma}\label{lem:findIndices}
Suppose we perform the $t$-th call to \textsc{Add}. 
 Then the call to \textsc{FindIndices} in \Cref{line:dual:findIndices}
returns a set $I \subset [m]$ containing all $i \in [m]$ 
such that there exists some $j$ with $2^j|t$ satisfying both
$i \notin F_j$ and 
$|v^{(t-2^j)}_i - v^{(t)}_i| \ge 0.2 w^{(t)}_i \epsilon / \log n$.
\end{lemma}

\begin{proof}
Pick any $j$ with $2^j|t$, if
$|v^{(t-2^j)}_i - v^{(t)}_i| \ge 0.2 w^{(t)}_i \epsilon / \log n$,
then
$$
\left|
e_i^\top \mA \sum_{k=t-2^j+1}^t h^{(k)}
\right|
\ge 0.2 w^{(t)}_i \epsilon / \log n.
$$
We will argue why \textsc{FindIndices} detect all $i$'s satisfying this condition.

Note that we have $f^{(j)} = \sum_{k=t-2^j+1}^t h^{(k)}$, and thus by guarantee of \Cref{def:heavyhitter} 
when we call $D_j.$\textsc{HeavyQuery}$(f^{(j)},0.2\epsilon/\log n)$ (in \Cref{line:detect_changes}), we obtain for every $j$ with $2^j|t$ and all $i\in[m]$ with
$$
\left|\hat{g}^{(j)}_i e_i^\top  \mA \sum_{k=t-2^j+1}^t h^{(k)} \right| \ge 0.2 \epsilon / \log n.
$$
Here $\hat{g}^{(j)}_i = 0$ if $i \in F_j$ due to change of $w_i$ in \textsc{SetAccuracy} or in \Cref{line:reinsert_i}, 
which happens whenever $\ov_i$ is changed in \Cref{line:set_vi}. 
Thus by \Cref{line:reweight} we have $\hat{g}^{(j)}_i = 1 / w^{(t)}_i$ for all $i \notin F_j$. 
Equivalently, we obtain all indices $i \notin F_j$ satisfying the following condition, which proves the lemma.
\begin{align*}
\left| e_i^\top \mA \left(\sum_{k=t-2^j+1}^t h^{(k)}\right) ~ \right|
\ge
0.2 w^{(t)}_i \epsilon / \log n
\end{align*}
\end{proof}
To guarantee that the approximation $\ov$ we maintain is within the required $\epsilon$ error bound of the exact vector $v$, we need to argue that the $D_j$'s altogether are sufficient to detect all potential events that would cause $\ov_i$ to become outside of $v\pm\epsilon w$. It is easy to see that if an index $i$ is included in the returned set $I$ of \textsc{FindIndices} (\Cref{line:dual:findIndices}), then our algorithm will follow up with a call to \textsc{VerifyIndex$(i)$}, which will guarantee that $\ov_i$ is close to the exact value $v_i$ (or $\ov_i$ will be updated to be $v_i$). Thus, if we are in iteration $t$, and $\ot$ is the most recent time \textsc{VerifyIndex$(i)$} is called, we know $\ov_i$ satisfies the approximation guarantee for $v^{(\ot)}_i$, the value of $\ov_i$ remains the same since iteration $\ot$, and the index $i$ is not in the result of \textsc{FindIndices} for any of the iterations after $\ot$. We will demonstrate the last condition is sufficient to show $v^{(\ot)}_i \approx v^{(t)}_i$, which in turn will prove $\ov_i \approx v^{(t)}_i$. To start, we first need to argue that for any two iterations $\ot < t$, the interval can be partitioned into a small number of batches such that each batch is exactly one of the batches tracked by some $D_j$. We cite without proof the following lemma from~\cite{BrandLN+20}.
\begin{lemma}[{\cite[Lemma 6.3]{BrandLN+20}}]
\label{lem:dual:transform_t}
Given any $\ot < t$, there exists a sequence $$t = t_0 > t_1 > ... > t_k = \overline{t}$$ 
such that $k\leq 2\log t$ and $t_{x+1} = t_x - 2^{\ell_x}$ where $\ell_x$ satisfies $2^{\ell_x} | t_x $ for all $x=0,\ldots,k-1$. 
\end{lemma}
Now we can argue $\ov$ stays in the desired approximation range around $v$ with the same argument (with slight adaptation) of Lemma 6.4 in~\cite{BrandLN+20}.
\begin{lemma}[Correctness of \Cref{thm:dual_maintenance}]
Assume we perform the $t$-th call to \textsc{Add},
the returned vector $\ov$ satisfies 
$|v^{(t)}_i - \ov_i| \le \epsilon w^{(t)}_i$ for all $i\in[m]$, 
and $I$ contains all indices that have changed since the $(t-1)$-th \textsc{Add} call.
\end{lemma}
\begin{proof}
By $\hat{f} = \sum_{j=1}^t h^{(j)}$ we have $v^\init +\mA \hat{f} = v^{(t)}$ in \Cref{line:check_change}.
So after a call to $\textsc{VerifyIndex}(i)$ we know that $|\ov_i - v^{(t)}_i| < 0.2 \epsilon w^{(t)}_i / \log n$,
either because the comparison $|\ov_i - (v^\init + \mA \hat{f})_i | \ge 0.2 \epsilon w_i /\log n$ in \Cref{line:check_change} returned false,
or because we set $\ov_i \leftarrow (v^\init + \mA \hat{f})_i$ in \Cref{line:set_vi}. 
Note this is also the only place we may change $\ov_i$. 
So consider some time $\overline{t} \le t$ 
when $\textsc{VerifyIndex}(i)$ was called for the last time 
(alternatively $\ot = 0$).
Then $\ov_i$ has not changed during the past $t - \ot$ calls to \textsc{Add}, 
and we know $|\ov_i - v^{(\ot)}_i| \le 0.2  \epsilon w_i / \log n$.
We now want to argue that $|v^{(\ot)} - v^{(t)}| \le 0.2 \epsilon w_i$, 
which via triangle inequality would then imply $|\ov_i - v^{(t)}_i| \le \epsilon w_i $. Note here we can omit the superscript of $w_i$ since it has not changed during the past $t-\ot$ calls to \textsc{Add}, since otherwise $i$ would have been added to all $D_j$'s (particularly $D_0$), which would have triggered a call to \textsc{VerifyIndex} in~\Cref{line:verify_F}. 

For $t = \ot$ this is obvious, so consider $\ot < t$. We know from \Cref{lem:dual:transform_t} the existence of a sequence
$$t = t_0 > t_1 > ... > t_k = \overline{t}$$ 
with $2^{\ell_x} | t_x $ and $t_{x+1} = t_x - 2^{\ell_x}$. In particular, this means that the interval between iteration $t_{x+1}$ and $t_x$ correspond to exactly a batch tracked by $D_{\ell_x}$. Thus, at iteration $t_x$ when \textsc{FindIndices} is called, $D_{\ell_x}.\textsc{HeavyQuery}$ is executed in \Cref{line:detect_changes}. This gives us $|v^{(t_x)}_i-v^{(t_{x+1})}_i|< 0.2 \epsilon w_i / \log n$ for all $x$, because by \Cref{lem:findIndices} 
the set $I \cup (\bigcup_j F_j)$ contains all indices $i$ 
which might have changed by $0.2 w_i \epsilon / \log n$
over the past $2^\ell$ iterations for any $2^\ell | t$,
and because $\textsc{VerifyIndex}(i)$ is called for all $i \in I \cup (\bigcup_j F_j)$ 
in \Cref{line:dual:verify_I} and \Cref{line:verify_F}.

Note we can assume $\log t \leq \log T$ since we reset the data-structure every $T=\sqrt{nP/\|z\|_1}$ iterations, and this bounds the length of the sequence $k\leq 2\log T$. This then yields the bound
\[
|v^{(\ot)}_i - v^{(t)}_i|  = |v^{(t_k)}_i - v^{(t_0)}_i| 
\le~ \sum_{x=1}^{k} |v^{(t_x)}_i - v^{(t_{x-1})}_i|
\le~ k\cdot 0.2 \epsilon w_i / \log T 
\le ~ 0.4\epsilon w_i
\]
Thus we have $|\ov_i - v^{(t)}_i| \le (0.4 \epsilon+0.2 \epsilon/ \log n) w_i$,
which satisfies our approximation guarantee. 
It is also straightforward to check that when we return the set $I$ at the end of \textsc{Add}, 
$I$ contains all the $i$'s where $\textsc{VerifyIndex}(i)$ is called and returned true in this iteration, 
which are exactly all the $i$'s where $\ov_i$'s are changed in \Cref{line:set_vi}.
\end{proof}
Now we proceed to the complexity of our data structure. We start with the cost of \textsc{FindIndices}, which is mainly on the cost of querying $D_j$'s. As we discussed at the beginning, there can be a large overhead for large $j$, but this is compensated by querying large $j$ less frequently. The following is a straightforward adaptation of Lemma 6.5 in~\cite{BrandLN+20}. Note that \textsc{SetAccuracy} also triggers \textsc{VerifyIndex} (indirectly in \text{Add} since we add $i$ to all $F_j$'s whenever we change $w_i$). We attribute the cost of these \textsc{VerifyIndex} calls to the amortized running time of \textsc{SetAccuracy} instead of counting it in the time spent in \textsc{VerifyIndex}.
\begin{lemma}\label{lem:dual_query_complexity}
After $T$ calls to \textsc{Add}, the total time spent in \textsc{FindIndices} and \textsc{VerifyIndex} is bounded by
\[
\tilde{O}\left(
T\epsilon^{-2}\sum_{t=1}^T \| (v^{(t)}-v^{(t-1)})/w^{(t)}\|_z^2 
+ TQ
\right)
\]
\end{lemma}

\begin{proof}
We start with the cost of \textsc{FindIndices}. 
Every call to \textsc{Add} invokes a call to \textsc{FindIndices}, 
so we denote the $t$-th call to \textsc{FindIndices} as the one associated with the $t$-th \textsc{Add}. 
Fix any $j$ and consider the cost for $D_j$. 
We update $f^{(j)}$ once in each call, which takes $O(n)=O(Q)$ ($Q\geq n$ because the Heavy Hitter data structure needs to read the input which costs $O(n)$ time). 
Every $2^j$ calls would incur the cost to $D_j.\textsc{HeavyQuery}(f^{(j)})$. 
Without loss of generality we consider the cost of the first time this happens (at iteration $2^j$) since the other batches follow the same calculation. 
We denote $\hat{g}^{(j)}$ as the scaling vector in $D_j$ when the query happens. 
We know $\hat{g}^{(j)}_i = 0$ if $i \in F_j$, 
and $\hat{g}^{(j)}_i = 1 / w_i$ otherwise. 
Note here we can skip the superscript indicating the iteration number, 
since if $i\notin F_j$ it must be that $\ov_i$ and $w_i$ have not changed over the $2^j$ iterations. 
The cost to query $D_j$ in \Cref{line:detect_changes} can then be bounded by.
\begin{align*}
\tilde O(\|\hat{\mG}^{(j)} \mA f^{(j)} \|_z^2 \epsilon^{-2}  + Q)
\le&~
\tilde O(\|\sum_{t=1}^{2^j} \mdiag(1/w^{(t)}) \mA h^{(t)}\|_z^2 \epsilon^{-2} + Q)\\
\le&~
\tilde O\left(\left(\sum_{t=1}^{2^j}\| \mdiag(1/w^{(t)}) \mA h^{(t)}\|_z \right)^2 \epsilon^{-2}  + Q\right)\\
\le&~
\tilde O\left(2^j\cdot\sum_{t=1}^{2^j}\| \mdiag(1/w^{(t)}) \mA h^{(t)}\|_z^2 \epsilon^{-2} + Q\right)
\end{align*}
The first line is by \Cref{def:heavyhitter}, and note we use $0.2\epsilon/\log n$ as the error parameter in the call. 
The first inequality is by looking at $\hat{g}^{(j)}_i$ for each index $i$ separately. 
The value is either $0$ so replacing it by $1/w^{(t)}_i$ only increase the norm, 
or we know $i\notin F_j$, so $\hat{g}^{(j)}_i=1/w^{(t)}_i$ for all $t\in[1,2^j]$. 
The second inequality uses the triangle inequality, and the third inequality uses Cauchy-Schwarz.
The cost of all subsequent queries to $D_j$ follows similar calculation, 
and as this query is only performed once every $2^j$ iterations, the total time after $T$ iterations is
\begin{align*}
\tilde{O}\left(
T\epsilon^{-2}\sum_{t=1}^T \| \mdiag(1/w^{(t)}) \mA h^{(t)}\|_z^2 + T2^{-j}Q
\right)
=
\tilde{O}\left(
T\epsilon^{-2}\sum_{t=1}^T \| (v^{(t)}-v^{(t-1)})/w^{(t)}\|_z^2 + T2^{-j}Q
\right)
\end{align*}
Note that the equality follows the definition of $v^{(t)}$ in \Cref{thm:dual_maintenance}. 
We can then sum over the total cost for all $D_j$'s as well as updating $f^{(j)}$ for $j=1,\ldots, \log T$ 
to get the final running time bound in the lemma statement.

As to the cost of \textsc{VerifyIndex}$(i)$, 
each call computes $(v^\init + \mA \hat{f})_i$
which takes $O(\nnz(a_i))\leq z_i$ time as each row of $\mA$ has $\nnz(a_i)$ non-zero entries. 
Further, the updates to $F_j$'s and calls to $D_j.\textsc{Scale}$ for all $j$'s take $O(z_i\log T)$ time. 
Now we need to bound the total number of times we call \textsc{VerifyIndex} for some $i$, 
which can only happen in two cases. 
The first case (\Cref{line:dual:verify_I}) is when $i$ is returned by \textsc{FindIndices} in \Cref{line:dual:findIndices}, 
and the total time we spent in \textsc{VerifyIndex} is bounded by $O(\sum_{i\in I}z_i\log T)$. By the guarantee of $\sum_{i\in I}z_i$ given in~\Cref{def:heavyhitter} for the \textsc{QueryHeavy} calls, we know the total cost over $T$ iterations of such \textsc{VerifyIndex} calls is bounded by the toal running time of \textsc{FindIndices} (up to a $\log T$ factor).
The second case (\Cref{line:verify_F}) is when $i$ is in some $F_j$ 
because $\ov_i$ was updated due to $v_i$ changing by more than $0.2\epsilon w_i/\log n$ (or $w_i$ is updated, but we count such cost separately in \textsc{SetAccuracy}), 
and the total cost can be bounded by 
\[
\tilde{O}\left(
T\epsilon^{-2}\sum_{t=1}^T \| (v^{(t)}-v^{(t-1)})/w^{(t)}\|_z^2 
\right).
\]
Adding up the total cost of \textsc{VerifyIndex} and \textsc{FindIndices} proves the lemma.
\end{proof}

We proceed to prove the complexity bounds in \Cref{thm:dual_maintenance}. We will set $T$ to be $\sqrt{nP/\|z\|_1}$ and re-initialize the data structure every $T$ \textsc{Add} calls. 
\paragraph{\textsc{Initialize}.} 
The main work is to initialize the data structures $D_j$ for $j=1,...,\log T$, 
which takes $\tilde{O}(P)$ time in total by \Cref{def:heavyhitter}.
\paragraph{\textsc{SetAccuracy}.} The main work is to call $D_j.$\textsc{Scale} for all $j$'s, and this takes $\tilde{O}(z_i)$ time. Recall we also add $i$ to all $F_j$'s, which triggers \textsc{VerifyIndex} later when we call \textsc{Add}. Technically for $i$'s that are added to all $F_j$'s due to \textsc{SetAccuracy} we can flag them and skip the inner for loop in \textsc{VerifyIndex} since the loop does nothing to these $i$'s. Thus, the total time to just update $\ov_i$ is $O(\nnz(a_i))$ which is bounded by $O(z_i)$. We have another $D_j.$\textsc{Scale} cost incurred in \textsc{Add} when we restart $D_j$ every $2^j$ iterations (see line~\Cref{line:reweight}), and such cost is bounded by $\tilde{O}(z_i)$.
\paragraph{\textsc{Add}.} 
The cost not associated with any \textsc{FindIndices} and \textsc{VerifyIndex} is $O(n)$. 
Together with \Cref{lem:dual_query_complexity} gives the bound of total time of $T$ calls to \textsc{Add} 
\[
\tilde{O}\left(
T\epsilon^{-2}\sum_{t=1}^T \| (v^{(t)}-v^{(t-1)})/w^{(t)}\|_2^2 
+ TQ
\right)
\]
Moreover, the cost to reinitialize the data structure (once every $T$ iterations) is $\tilde{O}(P)$. Together with the bound above we get the amortized time specified in the theorem statement.
\paragraph{\textsc{ComputeExact}.} 
This just takes $O(\nnz(A))$ to compute the matrix-vector product.

\section{Graph Data Structures}
\label{sec:graph_data_structures}
In this section we formally state the data structure results we need to efficient implement our interior point method to show Theorem \ref{thm:mincost flow} in Section \ref{sec:mincostflow}.
\begin{lemma}[{\cite[Lemma 5.1]{BrandLN+20}}]
\label{lem:graph:heavyhitter}
There exists a $(P,c,Q)$-\textsc{HeavyHitter} data structure (\Cref{def:heavyhitter}) for matrices $\mA$, 
where $\mA$ is obtained by removing one column from an incidence matrix 
of a directed graph with $n+1$ vertices and $m$ edges, $P = \tilde{O}(m)$, $c_i = \tilde{O}(1)$ for all $i \in [m]$, 
and $Q = \tilde{O}(n \log W)$ 
where $W$ is the ratio of the largest to smallest non-zero entry in $g$.
\end{lemma}

\begin{proof}
If $\mA$ is an incidence matrix, then \cite[Lemma 5.1]{BrandLN+20}
yields a $(P,c,Q)$-\textsc{HeavyHitter} with $P,c,Q$ as stated in \Cref{lem:graph:heavyhitter}.
If we remove one column from $\mA$, then the remaining matrix can be considered an incidence matrix 
with at most $n$ additional rows that contain only a single non-zero entry (i.e~$+1$ or $-1$).
For the rows that form an incidence matrix we use \cite[Lemma 5.1]{BrandLN+20},
while for the remaining $n$ rows we compute $(\mG \mA h)_i$ explicitly
and return the index $i$ if $|(\mG \mA h)_i| > \epsilon$.
This additional explicit computation requires only $O(n)$ time which is subsumed by $Q$.
\end{proof}

\begin{lemma}
\label{lem:graph:inversemaintenance}
There exists a $(P,c,Q)$-\textsc{InverseMaintenance}
data structure for matrices $\mA \in \R^{m \times n}$,
where $\mA$ is obtained by removing one column from an incidence matrix 
of a directed graph with $n+1$ vertices and $m$ edges, $P = \tilde{O}(m)$, $c_i = 1$ for all $i \in [m]$,
and $Q = n$.
\end{lemma}

\begin{proof}
For any diagonal matrix $\mV \in \R_{\ge0}^{m\times m}$
we have that $\mA^\top \mV \mA$ is a symmetric diagonally dominant matrix.
For such matrices there exist nearly-linear time solvers, e.g.~\Cref{lem:solver}.
Thus \textsc{Initialize} consists of reading the matrix.
During \textsc{Update} no operation is performed
and for \textsc{Solve} we use \Cref{lem:solver}.
\end{proof}

We use the following lemma which follows from directly from Lemma 5.1 and 8.2 in \cite{BrandLN+20}.

\begin{lemma}[{\cite[Lemma 5.1 and 8.2]{BrandLN+20}}]
\label{lem:graph:heavysampler_l2}
There exists the following data structure
\begin{itemize}
\item \textsc{Initialize}$(\mA \in \R^{m \times n}, g \in \R^m_{>0}, \otau)$
	Initializes the data structure on the given matrix $\mA$ in $\tilde{O}(m)$ time,
	where $\mA$ is obtained by removing one column from an incidence matrix 
	of a directed graph with $n+1$ vertices and $m$ edges.
\item \textsc{Scale}$(i, a, b)$
	Sets $g_i \leftarrow a$ and $\otau_i \leftarrow i$ in $\tilde{O}(1)$ time.
\item \textsc{Sample}$(h \in \R^n, C_1, C_2, C_3)$
	Returns a random diagonal matrix $\mR \in \R^{m \times m}$
	where independently for all $i$ we have 
	$\mR_{i,i} = 1/p_i$ with probability $p_i$ and $\mR_{i,i} = 0$ otherwise
	for
	$$p_i \ge \min \left\{1, C_1 \frac{m}{\sqrt{n}} \cdot \frac{(\mG\mA h)_i^2}{\|\mG\mA h\|_2^2} + C_2 \frac{1}{\sqrt{n}} + C_3 \otau_i \right\}.$$
	With high probability the time 
	(and thus also the output size of $\mR$) 
	is bounded by $\tilde{O}((C_1+C_2)m/\sqrt{n} + C_3n \log W)$
	where $W$ is a bound on the ratio of largest to smallest entry in $g$.
\end{itemize}
\end{lemma}

Note that by \Cref{cor:l2_sampling} the data structure of \Cref{lem:graph:heavysampler_l2}
yields the following $(P,c,Q)$-\textsc{HeavySampler}.

\begin{corollary}
\label{lem:graph:heavysampler}
There exists a $(P,c,Q)$-\textsc{HeavySampler}
data structure for edge-vertex incience matrices $\mA \in \R^{m \times n}$
with $P = \tilde{O}(m)$, $c_i = \tilde{O}(1)$ for all $i \in [m]$, 
and $Q = \tilde{O}(m/\sqrt{n} + n \log W)$,
where $W$ is the ratio of the largest to smallest non-zero entry in $g$.
\end{corollary}

\end{document}